\documentclass[letterpaper,12pt]{article}

\usepackage[margin=1in]{geometry}

\usepackage{amsmath,amssymb,amsthm}
\usepackage{newtxtext,newtxmath}

\usepackage{setspace}
\usepackage{graphicx} 
\usepackage[T1]{fontenc}
\usepackage[utf8]{inputenc}
\usepackage{color}
\usepackage{float}
\usepackage{amsbsy}
\usepackage{amstext}
\usepackage{amsmath,amssymb,amsthm}
\usepackage{bm} 
\usepackage{dsfont}
\usepackage{graphicx}
\usepackage{comment}
\usepackage{ulem}
\usepackage{subcaption}
\usepackage{microtype}

 \usepackage[unicode=true,pdfusetitle,
 bookmarks=true,bookmarksnumbered=true,bookmarksopen=true,
 breaklinks=true,pdfborder={0 0 0},pdfborderstyle=UseOutlines,backref=page,colorlinks=true]
 {hyperref}
 
\usepackage[noabbrev,nameinlink,capitalise]{cleveref}
\usepackage{booktabs} 
\usepackage[authoryear,longnamesfirst]{natbib}
\usepackage{tikz}
\usepackage{threeparttable}
\usepackage{xparse}
\usepackage{array}
\usepackage{longtable}
\usepackage{pdflscape}
\usepackage{siunitx}
\usepackage{algorithm2e} 
\usepackage{enumitem}
\usepackage{refcount}
\usepackage{etoc}



\theoremstyle{plain}
\newtheorem{assumption}{\protect\assumptionname}[section]
\newtheorem{remark}{Remark}[section]
        \let\oldremark\remark
        \let\endoldremark\endremark
        \renewenvironment{remark}[1][]{%
        \pushQED{$\blacksquare$}
        \oldremark[#1]\normalfont
        }{
        \popQED\endoldremark
    }
\theoremstyle{definition}
\newtheorem{definition}{\protect\definitionname}[section]
\theoremstyle{plain}
\newtheorem{theorem}{\protect\theoremname}[section]
\theoremstyle{plain}
\newtheorem{lemma}{\protect\lemmaname}[section]
\theoremstyle{plain}
\newtheorem{corollary}{\protect\corollaryname}[section]
\theoremstyle{plain}
\newtheorem{example}{\protect\examplename}[section]
\theoremstyle{plain}

\crefname{table}{Table}{Tables}
\crefname{subtable}{Table}{Tables}
\crefname{part}{Part}{Parts}
\crefname{chapter}{Chapter}{Chapters}
\crefname{section}{Section}{Sections}
\crefname{subsection}{Section}{Sections}
\crefname{subsubsection}{Section}{Sections}
\crefname{appendix}{Appendix}{Appendices}
\crefname{theorem}{Theorem}{Theorems}
\crefname{lemma}{Lemma}{Lemmas}
\crefname{algorithm}{Algorithm}{Algorithms}
\crefname{algocf}{Algorithm}{Algorithms}
\crefname{listing}{Listing}{Listings}
\crefname{figure}{Figure}{Figures}
\crefname{assumption}{Assumption}{Assumptions}
\crefname{equation}{}{}
\crefformat{equation}{(#2#1#3)}
\crefmultiformat{equation}{(#2#1#3)}{ and~(#2#1#3)}{, (#2#1#3)}{ and~(#2#1#3)}

\renewcommand{\ref}[1]{\cref{#1}}

\providecommand{\assumptionname}{Assumption}
\providecommand{\corollaryname}{Corollary}
\providecommand{\definitionname}{Definition}
\providecommand{\lemmaname}{Lemma}

\providecommand{\theoremname}{Theorem}
\providecommand{\examplename}{Example}

\newcounter{example@save}

\usepackage{xcolor}
\definecolor{darkblue}{rgb}{0.0, 0.0, 0.55}
\definecolor{teal}{rgb}{0.0, 0.5, 0.5}
\hypersetup{linkcolor=teal, urlcolor=blue, citecolor=darkblue}
\renewcommand*{\backrefalt}[4]{%
    \ifcase #1 %
    \or
        [Cited on page #2]
    \else
        [Cited on pages #2]
    \fi}

\usetikzlibrary{shapes.geometric, arrows.meta,positioning,calc} 
\usetikzlibrary{overlay-beamer-styles,fit}
\tikzset{
  startstop/.style = {rectangle, rounded corners, draw, fill=gray!10, align=center, text width=4cm, minimum height=1cm, font=\small},
  decision/.style  = {diamond, draw, align=center, text width=4cm, inner sep=2pt, font=\small, fill=orange!10,aspect=2},
  process/.style   = {rectangle, draw, fill=blue!5, align=center, text width=4cm, minimum height=1cm, font=\small},
  arrow/.style     = {thick, -{Stealth}}
}

\NewDocumentCommand{\sumto}{ O{i} O{n} }{%
  \sum_{#1=1}^{#2}%
}
\NewDocumentCommand{\prodto}{ O{i} O{n} }{%
  \prod_{#1=1}^{#2}%
}
\NewDocumentCommand{\intto}{ O{i} O{n} }{%
  \int_{#1=1}^{#2}%
} 
\NewDocumentCommand{\limto}{ O{i} O{\infty} }{%
  \lim_{#1\to #2}%
}

\DeclareMathOperator{\lin}{span}

\DeclareMathOperator{\rank}{rank}

\DeclareMathOperator*{\argmin}{arg\,min}

\newcommand{\indep}{\mathrel{\perp\mspace{-10mu}\perp}} 

\newcommand{\Var}{\operatorname{Var}}

\renewcommand{\bar}[1]{\overline{#1}}
\renewcommand{\tilde}[1]{\widetilde{#1}}
\renewcommand{\hat}[1]{\widehat{#1}}

\definecolor{myblue}{rgb}{0.000,0.478,0.718}   
\definecolor{myorange}{rgb}{0.878,0.416,0.231} 

\DeclareFontFamily{U}{mathx}{}
\DeclareFontShape{U}{mathx}{m}{n}{<-> mathx10}{}
\DeclareSymbolFont{mathx}{U}{mathx}{m}{n}
\DeclareMathAccent{\widecheck}{0}{mathx}{"71}
\renewcommand{\check}[1]{\widecheck{#1}}

\AddToHook{build/page/before}{%
	\ifnum\value{page}=\getpagerefnumber{tab:dml-cluster-summary}\relax
		\thispagestyle{empty}%
	\fi
}

\newcommand{\so}{\operatorname{so}}
\newcommand{\ta}{\operatorname{ta}}
\newcommand{\sub}{\operatorname{sub}}

\let\originalPr\Pr
\renewcommand{\Pr}{\originalPr\nolimits}

\shortcites{bajari2025hedonic}
\shortcites{sellam2022the}

\hypersetup{pdftitle={Econometrics with Pre-Trained Embeddings for Unstructured Data},pdfauthor={Yuya Shimizu}}

\begin{document}
\etocdepthtag.toc{main}
\title{\vspace{-50pt} Econometrics with Pre-Trained Embeddings for~Unstructured~Data}
\vspace{-30pt}
\author{Yuya Shimizu\footnote{\href{mailto:yuya.shimizu@wisc.edu}{yuya.shimizu@wisc.edu}, Department of Economics, University of Wisconsin-Madison.
I am grateful to Bruce Hansen, Jack Porter, and Xiaoxia Shi for their invaluable advice and encouragement. 
I thank Naoki Aizawa, Yong Cai, Harold Chiang, Junho Choi, Hugo Freeman, Woosik Gong, Hide Ichimura, Hiro Kaido, Tetsuya Kaji, Hiro Kasahara, Toru Kitagawa, Lixiong Li, Lorenzo Magnolfi, Yaroslav Mukhin, Whitney Newey, Ashesh Rambachan, Kensuke Sakamoto, Suproteem Sarkar, Chris Taber, Jing Tao, Keyon Vafa, Ken West, Famy Xia, Kohei Yata, and seminar participants at Wisconsin, ESIF-AIML, Booth-AIE, MEG, CESG, and SEA for helpful comments and suggestions.
I thank the authors of Ahrens et al.\ \citeyearpar{ahrens2026introduction}, especially Achim Ahrens, for sharing their code and insights.
This work is supported by the dissertation fellowship and the summer research fellowship from Wisconsin.}
}
\vspace{-30pt}
\date{This version: November 25, 2026\\
\small{(First version: July 19, 2026)}\\
\href{https://yshimizu-econ.github.io/assets/pdf/Shimizu_JMP_Embedding.pdf}{[Click here for the latest version]}
}


\maketitle 
\vspace{-30pt}

\begin{abstract}
	\setstretch{1.1}
	\fontsize{11}{13.6}\selectfont
	Unstructured data, such as images and text, are increasingly used in empirical economics. Since training machine learning models on unstructured data is costly, economists often use off-the-shelf pre-trained deep learning models developed by computer scientists to extract potentially high-dimensional embeddings, which are then used as covariates in target economic analyses. Despite the popularity of this practice, its theoretical foundations remain limited. There are two main difficulties. First, pre-trained models are typically trained on different datasets and for different tasks, making it unclear when they can be used validly for the target task. Second, the embedding function is not uniquely identified, complicating the analysis of its estimation error and the effect of that error on the target task. We provide sufficient conditions to overcome these difficulties. A key condition, which we call transferability, governs the convergence rate we derive. This rate accounts for source pre-training estimation error adjusted for the strength of transferability, as well as target estimation error with the pre-trained embeddings treated as standard covariates. To assess transferability, we develop a computationally feasible bootstrap test that does not require re-estimating the embeddings or nuisance functions. The transferability condition also provides practical guidance for applied researchers on selecting the source embedding function and target model classes. Our theory applies to a broad class of econometric analyses using double machine learning, including partially linear regression with unstructured controls, price elasticity estimation with product quality measured from images and text, missing-data imputation under selection on unstructured observables, and average treatment effect estimation with unstructured confounders. In an empirical application, we estimate the labor supply elasticity on Amazon Mechanical Turk, an online labor market platform, using embeddings of job descriptions as controls.
\end{abstract}



\newpage
\section{Introduction}
Empirical economists increasingly use pre-trained deep learning models to extract embeddings from images and text and then include these embeddings as covariates in downstream economic analyses.\footnote{In the machine learning literature, embeddings are also commonly referred to as extracted features or learned representations.} An embedding function maps unstructured inputs into vectors in a lower-dimensional, though potentially still high-dimensional, latent space that summarizes information from the original data. However, the function is typically trained using a dataset and an objective function that both differ from those of the economic application. It therefore remains unclear when the resulting embeddings can be used validly in the economic task. From an econometric perspective, this workflow constitutes a two-step estimation problem (\citealp{pagan1984econometric,murphy1985estimation}), in which first-step estimation errors may affect second-step approximation unless additional conditions are satisfied. It also raises an identification problem: embedding functions are generally not uniquely identified, complicating the analysis of their convergence rates and downstream effects. This paper establishes sufficient conditions under which this workflow yields valid econometric inference and develops a computationally feasible bootstrap test of a key condition, which we call transferability. The transferability condition also guides applied researchers in selecting the pre-trained embedding function and economic machine learning model classes.

The use of embeddings is motivated by the difficulty of incorporating unstructured data directly into standard empirical specifications.
Images and text can contain economically relevant information, but their raw representations are ultra-high-dimensional and their individual components generally lack direct economic meaning since they do not capture context.\footnote{
	For example, a color image with a resolution of $224 \times 224$ and three RGB channels contains $224 \times 224 \times 3 = 150{,}528$ raw input values. A text representation with one coordinate per vocabulary token can have more than $30{,}000$ dimensions. Individual pixel values in an image or naive aggregations of token vectors in text are generally not economically meaningful variables. \cite{dell2025deep} provides an excellent review of image and text data and relevant deep learning models for economic applications.
}
Embedding functions organize these inputs into lower-dimensional structured representations that can serve as covariates in econometric models.

Pre-trained embeddings can provide valuable information in empirical economics. In some applications, adding embeddings to a specification can materially change estimates of parameters of interest, including their sign and statistical significance.\footnote{
	The following is an incomplete list of applications that use embeddings as covariates. \cite{jean2016combining} use satellite image embeddings to predict local poverty status.
	\cite{avivi2025patent} uses patent-text embeddings to control for application quality when studying gender bias in patent examination.
	\cite{dube2020monopsony} estimate the elasticity of job vacancy filling speed with respect to rewards, using job description embeddings as controls.
    \cite{bach2024adventures,bajari2025hedonic,compiani2025demand} use product text and image embeddings to control for product quality in hedonic regression and demand estimation.
    \mbox{\cite{klaassen2024doublemldeep}} study causal effects with multimodal data.
    Across these applications, a common assumption is that embeddings capture economically relevant information in unstructured inputs and can be used as covariates to control or to impute missing variables.
    Related applied work uses product embeddings not primarily as covariates, but to measure similarity in product space (\citealp{han2025copyright,magnolfi2025triplet}). In macroeconomics, \mbox{\cite{casella2026structural}} propose a framework to incorporate textual information into structural macroeconomic models using topic models and text embeddings.
}
Since training such representation models from scratch typically requires an extremely large sample and substantial computational resources, economists often rely on pre-trained embeddings as structured summaries of raw inputs.

Despite the empirical popularity of this workflow, its theoretical foundations remain limited.
There are two main theoretical difficulties.
First, the pre-trained model is usually trained on a different dataset and for a different task than the economic application; thus, it is unclear whether the pre-trained representation is valid for the target economic task.
Second, the embedding function is usually not identified: for example, we can insert any invertible matrix and its inverse between the last hidden layer of a deep neural network and the output layer without changing the output.

To address these difficulties, we first formalize the setup containing a source task for the computer scientist and a target task for the economist. The embedding function is estimated from the source task and then used in the target task. This corresponds to the practical scenario where the economist can access the pre-trained embedding function from the computer scientist via standard deep learning libraries such as PyTorch, but it is difficult or impossible for the economist to observe the source sample used to train the embedding function.
Our theory allows for this practical setting by treating the source sample as unobserved for the economist.\footnote{
    In the machine learning literature, this is known as a transfer-learning setup in which knowledge learned in the source task is reused in the target task (\citealp{pan2009survey}).
}

We then establish two conditions that jointly suffice to justify the two-step workflow of using pre-trained embeddings in the target task.
The first condition requires the existence of a shared best-approximating embedding function that yields small approximation errors for both the source and target tasks.
This condition formalizes the common (implicit) assumption in the empirical literature that the pre-trained embedding function summarizes the relevant information in the unstructured data for the target task.
The second is a transferability condition under which a neighborhood of the source-task identified set is included in the corresponding neighborhood of the target-task identified set.
When the embedding function is the last hidden layer of a deep neural network and the economist uses a linear model on top of it, transferability is satisfied if the coefficient of the target task lies in the row span of the source-task coefficients associated with the embedding function in population, together with additional conditions relating the marginal distributions of unstructured inputs across source and target tasks.

To illustrate the importance of the transferability condition, \cref{fig:sim_plm_hist} compares double machine learning (DML) estimates for a partially linear model with pre-trained embeddings under strong and weak source signals about the components relevant to the target task. The target equations adapt the nonlinear functional form in \cite{chernozhukov2018double} to the population embeddings. See \cref{subsec:simulation_ej18} for the simulation design.\footnote{The histograms use the XGB 3 nuisance learner with source sample size $m=20{,}000$, target sample size $n=10{,}000$, and $1{,}000$ Monte Carlo replications per scenario. We compare $\rho=0$ (left) and $\rho=1-n^{-1/4}$ (right) in the source model described in \cref{subsec:simulation_ej18}. The true coefficient is $\theta^*=0.5$.}
Under strong source signals, the estimates are centered near the true parameter. Under weak source signals, they exhibit noticeable Monte Carlo bias. This comparison motivates a transferability condition that quantifies how informative the source task is for learning embeddings that retain information relevant to the target task.
\begin{figure}[htbp]
	\centering
	\includegraphics[width=0.95\linewidth]{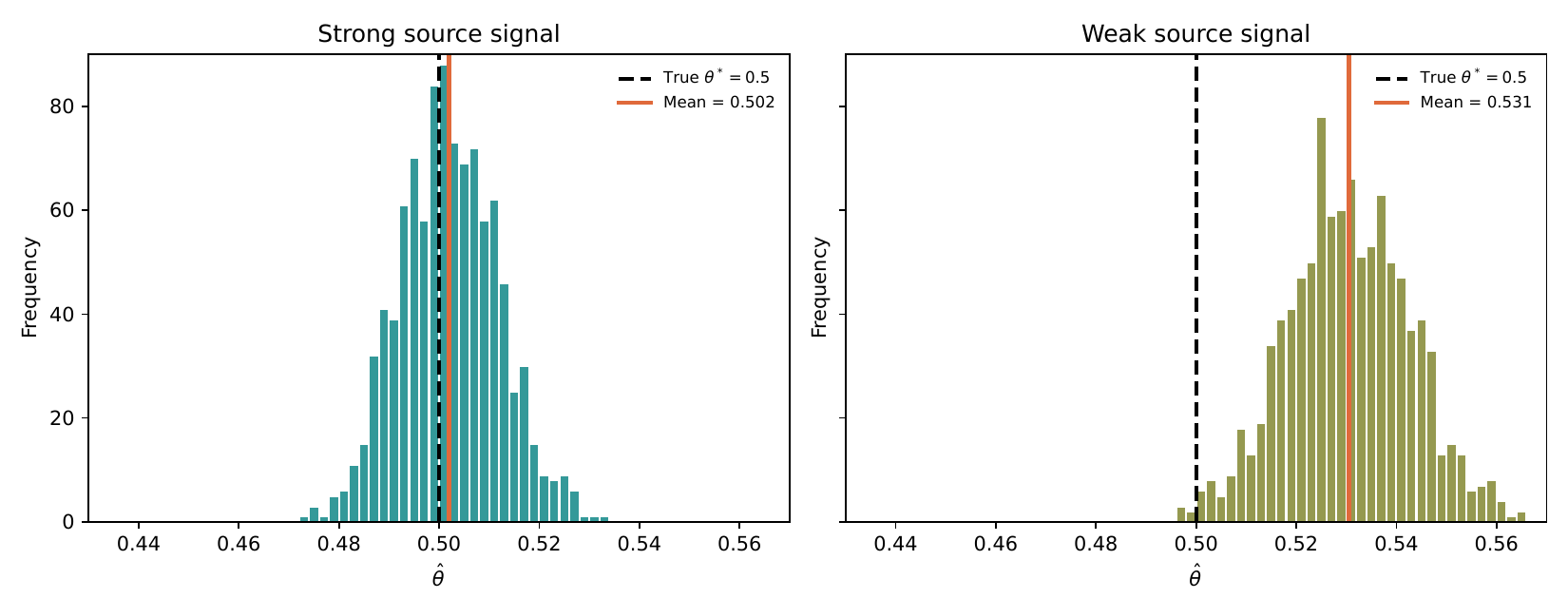}
	\caption{Histograms of the DML estimator $\hat{\theta}$ under strong (left) and weak (right) source signals about target-relevant components of the embedding function}
	\label{fig:sim_plm_hist}
\end{figure}

Under these conditions, we derive the convergence rate of transfer learning for the target task estimator that uses pre-trained embeddings.
Our convergence rate accounts for randomness in both the source and target samples and consists of three components: (i) target-task estimation error when the pre-trained embeddings are treated as standard covariates; (ii) source-task pre-training estimation error, adjusted for the strength of transferability; and (iii) approximation error of the target model class using embeddings as inputs. Each of the first two components is bounded by a rate proportional to the square root of the ratio of the corresponding model-class complexity to its sample size. Since these complexity measures are well characterized for commonly used machine learning models, researchers can readily evaluate the resulting convergence rate. The overall convergence rate for the target task is governed by the slowest-decaying of the three components. This rate can be faster than $n^{-1/4}$, thereby satisfying the rate condition for valid inference under the double machine learning (DML) framework of \cite{chernozhukov2018double} with a target sample of size $n$.

Our transferability theory also provides practical guidance for applied researchers. 
First, the transferability condition is easier to justify for later hidden layers of a deep neural network than for earlier hidden layers, because earlier hidden layers introduce more complex partial identification issues. Thus, the last hidden layer is preferable as the embedding function if the economist is otherwise indifferent between earlier and last hidden layers.
Second, even under the sufficient conditions we impose and when the last hidden layer is selected as the embedding function, the embedding function is still not point identified, and it is only identified up to an invertible linear transformation.
Partial identification of the embedding function makes variable selection, such as LASSO, fragile. Instead, we recommend using ridge regression or shallow neural networks for the target model since they are robust to invertible linear transformations of the embeddings.

We also develop a computationally feasible multiplier bootstrap test for transferability.
The test compares DML estimators across embeddings pre-trained for the same source task using different pre-training procedures, such as different random seeds for the initial values.
It computes critical values from weighted sums of estimated influence functions; thus, bootstrap draws do not require re-estimating embeddings or nuisance functions.
Under the regularity conditions, rejection provides evidence against transferability. 

Finally, as an empirical illustration, we apply the framework to Amazon Mechanical Turk, an online labor market platform, using the data of \cite{dube2020monopsony}.
We use pre-trained embeddings of job descriptions as controls to estimate the elasticity of job filling speed with respect to posted rewards.
We compare alternative text representations and nuisance learners and use the bootstrap test to assess robustness to embeddings trained with different random seeds. 
The empirical application motivates two empirical practices: jointly selecting model classes for the pre-trained embedding function and downstream nuisance learners based on out-of-sample performance in the target task, and assessing and reporting the robustness of DML estimates across embeddings pre-trained for the same source task using different random seeds.

\paragraph{Related Literature}
The literature most closely related to ours can be grouped into three strands.
First, several recent papers use pre-trained embeddings with theoretical justifications for downstream inference.
\cite{vafa2025estimating} study wage-gap decomposition with foundation models.
\cite{bach2024adventures} show that multimodal embeddings from text and images can improve demand estimation and elasticity analysis.
These two papers establish square-root-$n$ consistency for the parameter of interest under high-level conditions.
\cite{christensen2026unstructured} develop a bias correction in demand counterfactuals when product differentiation is proxied by finite-dimensional nuisance parameters reparameterized from the embedding function and economic parameters.\footnote{Such reparameterization is not available for the general DML setup we consider in this paper, which includes \cref{ex:partially_linear,ex:demand,ex:imputation,ex:ate}.}
While the scope of \cite{christensen2026unstructured} is different from this paper, their theory takes convergence of the nuisance parameters as given, and our theory can complement their theory by providing sufficient conditions for the convergence in probability.
Our paper is also related to \cite{fernandez2026pragmatic}, who build on the omitted-variable-bias framework of \cite{chernozhukov2026long} to study inference with AI-learned embeddings. Their analysis accommodates both fixed and fine-tuned embeddings within the target sample but does not explicitly analyze a separate source training stage. We analyze this stage under explicit transferability conditions, accounting for source task estimation error and its randomness and the nonidentification of the embedding function.
The closest paper is \cite{schulte2025adjustment}. They study the estimation of the average treatment effect (ATE) with unstructured confounders and point out the identification issues of the embedding function.
However, they ignore source-task estimation error and assume transferability across tasks implicitly.
Our paper differs by deriving a general convergence rate theory under explicit transferability conditions and by incorporating source task estimation error. Estimation error in the source task can substantially affect the convergence rate of the target task estimator and, in turn, the validity of the downstream inference.

Second, our transferability condition is related to the transfer learning literature in machine learning.
\cite{tripuraneni2020theory} introduce task diversity as a key condition for learning a shared representation that transfers well to new tasks, and \cite{watkins2023optimistic} derive excess-risk bounds for multi-task representation learning that can yield faster rates under related diversity conditions.
Our rates match the corresponding fast rates in \cite{watkins2023optimistic} up to logarithmic factors when expressed in terms of VC dimensions (\cref{sec:diversity_lit}).
However, their bounds attain these rates when the minimum average source risk and the optimal target risk are sufficiently small, a condition that can be restrictive in economic applications.
Our paper uses a similar task diversity condition to ensure that the source task representation is valid for the target task, and we provide an econometric characterization of the condition through set relations between neighborhoods of the source- and target-task identified sets.
Moreover, we can establish an $o_P(n^{-1/4})$ convergence rate for the target-task nuisance predictor using pre-trained embeddings even when both optimal risks are bounded away from zero, providing a rate condition directly relevant to valid econometric inference.

Third, our paper also contributes to the growing literature on inference with variables generated from unstructured data.
\cite{fong2021machine,angelopoulos2023prediction,egami2023using,battaglia2025inference,carlson2025unifying,duan2026inference,li2026econometric,ludwig2026large,zhang2026debiasing} provide a framework to correct bias arising from variables generated from unstructured data in the context of missing data.
They typically assume that missingness does not induce selection bias or that the propensity score is known.\footnote{Except for \cite{battaglia2025inference}, these papers use a missing-variable-imputation framework to handle missing data in the context of \cite{chen2008semiparametric}.}
As discussed by \cite{carlson2025unifying}, the known propensity score assumption can be relaxed if we can estimate it with sufficiently fast convergence rates.
However, the convergence rate of propensity-score estimators based on unstructured data is not studied in these papers.
Beyond missing variable imputation, \cite{rambachan2024program,carlson2025making,modarressi2025causal} study settings in which unstructured data are used to construct outcome measures.
Our paper departs from this literature by studying a general double machine learning setting in which generated embeddings enter the nuisance learning stage as inputs for both outcome imputation and propensity score estimation (e.g., \cref{ex:imputation} below).

\paragraph{Structure of the Paper}
\cref{sec:setup} introduces the DML framework, motivating applications, and the transfer learning setup.
\cref{sec:transfer_theory} develops the general convergence theory and characterizes task diversity.
\cref{sec:transfer_convergence} derives convergence rates under lower-level conditions, including results for high-dimensional embeddings.
\cref{sec:test_transferability} develops the multiplier bootstrap test for transferability.
\cref{sec:simulations} presents simulation results, and \cref{sec:empirical_monopsony} applies the framework to the Amazon Mechanical Turk data.
\cref{sec:conclusion} concludes.
The supplementary appendix contains proofs, additional theoretical results, an extension to multimodal transfer learning, and further empirical results.

\paragraph{Notation}
For deterministic sequences $a_n$ and $b_n$, we write $a_n \lesssim b_n$ if there exists a constant $C>0$ such that $a_n \leq C b_n$ for large $n$, and we write $a_n \asymp b_n$ if $a_n \lesssim b_n$ and $b_n \lesssim a_n$.
We define the $L^r(P)$ norm as $\|a\|_{P,r}=\{\sum_{t=1}^T \int |a_t(v)|^r dP(v)\}^{1/r}$ for a function $a:\mathcal{V}\mapsto\mathbb{R}^T$, where $a_t$ is the $t$-th element of $a$, and $\|a\|_{\infty} = \sup_{v\in\mathcal{V}}|a(v)|$ for a function $a:\mathcal{V}\to\mathbb{R}$.
For a vector $v\in\mathbb{R}^T$, we define the $\ell^r$ norm as $\|v\|_r=\{\sum_{t=1}^T |v_t|^r\}^{1/r}$, where $v_t$ is the $t$-th element of $v$.
For a matrix $A\in\mathbb{R}^{T\times K}$, we denote the $t$-th largest singular value of $A$ by $\sigma_t(A)$, the largest singular value of $A$ by $\sigma_{\max}(A)=\sigma_1(A)$, and the smallest singular value of $A$ by $\sigma_{\min}(A)=\sigma_{\min\{T,K\}}(A)$.
We also define the smallest nonzero singular value of $A$ by $\sigma_{\min}^+(A)$.
We define the spectral norm (or operator norm) as $\|A\|_{\mathrm{sp}}=\sigma_{\max}(A)$, the Frobenius norm as $\|A\|_F=\sqrt{\sum_{t=1}^T\sum_{k=1}^K A_{t,k}^2}$, and the nuclear norm (or trace norm) as $\|A\|_*=\sum_{t=1}^{\min\{T,K\}}\sigma_t(A)$.
We denote the Moore-Penrose pseudo-inverse of a matrix $A\in\mathbb{R}^{T\times K}$ by $A^+\in\mathbb{R}^{K\times T}$.
For a square matrix $A\in\mathbb{R}^{T\times T}$, we denote the $t$-th largest eigenvalue of $A$ by $\mu_t(A)$, the largest eigenvalue of $A$ by $\mu_{\max}(A)=\mu_1(A)$, and the smallest eigenvalue of $A$ by $\mu_{\min}(A)=\mu_T(A)$ when all eigenvalues are real.

\section{Setup}
\label{sec:setup}
We first introduce the general DML setup, then four empirical applications, and finally formalize the transfer learning setup for the machine learning components.
In the following examples, we consider the setting where the economist has a sample of size $n$ for the target task and the computer scientist provides the pre-trained embedding function $\check{h}(\cdot)$ estimated independently from a different source task.
See \cref{sec:transfer_setup} for the detailed discussion of the embedding functions.

\subsection{Double Machine Learning (DML)}
\label{sec:dml}
Let $\psi(W_i;\theta,\gamma,\alpha)$ be a moment function, where $W_i$ is an observation from an i.i.d. sample, $\theta$ is the parameter of interest, and $(\gamma,\alpha)$ is an infinite-dimensional nuisance parameter vector.
We consider the moment condition:
\begin{equation}
    \mathbb{E}[\psi(W_i;\theta^*,\gamma^*,\alpha^*)]=0,
\end{equation}
where $\theta^*$ is the true value of $\theta$, and $(\gamma^*,\alpha^*)$ is the true value of $(\gamma,\alpha)$.
We assume that $\gamma^*$ and $\alpha^*$ are unknown functions of $Z_i$, which is a subvector of $W_i$ containing unstructured or ultra-high-dimensional variables such as images and text.
The remaining components of $W_i$ are low-dimensional and structured, such as an outcome variable and a treatment variable.

We assume that the computer scientist provides a pre-trained embedding function $\check{h}(z)$ estimated from a different source task.
In this setting, it is common for the economist to estimate the nuisance parameters $(\gamma^*,\alpha^*)$ using the pre-trained embeddings $\check{h}(Z_i)$ as regressors in the target stage models.
Let $\hat{f}_{\gamma}$ and $\hat{f}_{\alpha}$ denote the estimated target heads. For each nuisance component $\diamond\in\{\gamma,\alpha\}$, we allow any known Lipschitz prediction map $\Lambda_{\diamond}:\mathbb{R}\to\mathbb{R}$ that converts the best target predictor under the training loss to $\diamond^*$. Its Lipschitz constant is assumed to be an absolute constant. Examples include the identity map when the target head estimates a conditional mean by squared loss and $\Lambda_{\diamond}(u)=\Lambda(u):=1/(1+\exp(-u))$ when the head estimates log odds by binary logistic loss.
The nuisance estimators entering the DML score are
\begin{equation*}
	\hat{\gamma}=\Lambda_{\gamma}\circ\hat{f}_{\gamma}\circ\check{h},
	\qquad
	\hat{\alpha}=\Lambda_{\alpha}\circ\hat{f}_{\alpha}\circ\check{h}.
\end{equation*}
Here, $\Lambda_{\gamma}\circ\hat{f}_{\gamma}\circ\check{h}$ is the composite function defined as $\Lambda_{\gamma}\circ\hat{f}_{\gamma}\circ\check{h}(z)=\Lambda_{\gamma}(\hat{f}_{\gamma}(\check{h}(z)))$. For vector-valued nuisance functions, the prediction maps are applied componentwise.
The DML estimator $\hat{\theta}$ is the solution of the sample moment equation
\begin{equation}
	\frac{1}{n}\sum_{i=1}^n \psi(W_i;\hat{\theta},\hat{\gamma},\hat{\alpha})=0.
\end{equation}

To use the DML theory, we typically verify three key conditions: (i) Neyman orthogonality with respect to $(\gamma,\alpha)$,\footnote {
    \cite{escanciano2026automatic} study locally robust estimation with generated regressors.
    While their theory suggests that the generated regressors affect inference on the parameter of interest in general, this effect of generated regressors does not arise in our setup.
	We discuss the relationship to their results in \cref{sec:escanciano}.
}
(ii) sample splitting, and (iii) a sufficiently fast convergence rate of the nuisance parameters:\footnote{
We implicitly assume that the Neyman orthogonal score function is smooth enough to use a second-order Taylor expansion. See also Assumptions 3.3(b) and 3.4(c) in \cite{chernozhukov2018double} and the discussion after that.
}
\begin{equation}
	\|\gamma^*-\hat{\gamma}\|_{P,2}=o_P(n^{-1/4}) \quad \text{and} \quad \|\alpha^*-\hat{\alpha}\|_{P,2}=o_P(n^{-1/4}),
\end{equation}
where $\|\cdot\|_{P,2}$ is the $L^2(P)$ norm for the distribution of $W_i$, and the $o_P$ notation is with respect to the source and target samples used to estimate $(\hat{f}_\gamma,\hat{f}_\alpha)$ and $\check{h}$.

The first condition holds by construction of the score, and the second holds by the sample-splitting procedure. Because the source sample is independent of the target sample and is used only to estimate nuisance parameters, the sample-splitting condition holds once the target sample is split between nuisance estimation and parameter estimation (\cref{fig:split}).
\begin{figure}
    \centering
    \includegraphics[width=0.5\linewidth]{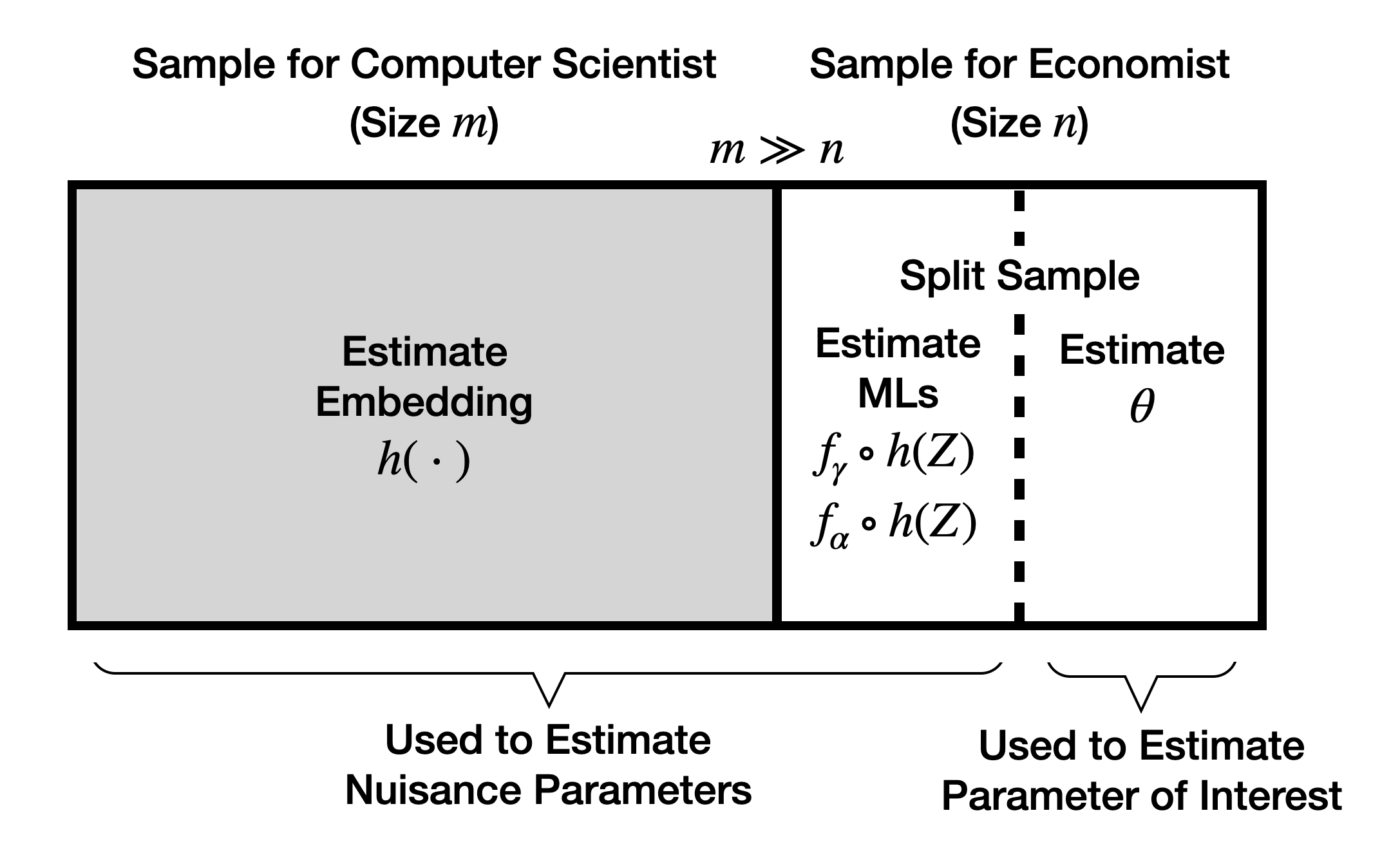}
    \caption{Sample Splitting with Transfer Learning}
    \label{fig:split}
\end{figure}
Alternatively, for some machine learning models, we can use the leave-one-out stability condition in \cite{chen2022debiased} to avoid sample splitting.
Since the remaining DML arguments are the same as in \cite{chernozhukov2018double}, we omit the details.
Under additional regularity conditions on the moment function in \cite{chernozhukov2018double}, the asymptotic normality and consistency of the asymptotic variance estimator are guaranteed. Thus, this paper focuses on the third condition: the convergence rate of the nuisance parameters.

\begin{remark}
    We implicitly assume that the two target stage machine learning models $f_{\gamma}$ and $f_{\alpha}$ use the same embedding function $h$ from the source task for notational simplicity.
    However, we can easily extend our theory to the case where two target stage machine learning models use different embeddings $h_{\gamma}$ and $h_{\alpha}$ from different source tasks.
\end{remark}

\subsection{Applications}
We include four empirical applications to illustrate the DML setup with pre-trained embeddings. Our theory can cover any DML model having the same structure as these examples.

\begin{example}[Partially Linear Model with Unstructured Controls.]
    \label{ex:partially_linear}
	We consider the partially linear model $Y_i = X_i\theta^* +\kappa^*(Z_i) + U_i$ with $\mathbb{E}[U_i \mid X_i,Z_i]=0$,
	where $Y_i$ is the outcome, $X_i$ is a scalar variable, $Z_i$ is a vector of unstructured or ultra-high-dimensional control variables such as images and text, $\kappa^*(\cdot)$ is an unknown nonparametric function, and $U_i$ is the error term.
    The parameter of interest is the coefficient $\theta^*$ of $X_i$.
    This is the classical partially linear setup (\citealp{robinson1988root}), now with controls represented by unstructured-data embeddings.
	We assume that an economist has a sample of size $n$, $\{(Y_i,X_i,Z_i)\}_{i=1}^n$, for the target task.
    The parameter $\theta^*$ is estimated by solving the Neyman orthogonal score $\mathbb{E}[\psi(W_i;\theta,\gamma^*,\alpha^*)]=0$,
    where $W_i=(Y_i,X_i,Z_i)$, $\psi(W_i;\theta,\gamma,\alpha)=\{(Y_i-\gamma(Z_i))-(X_i-\alpha(Z_i))\theta\}(X_i-\alpha(Z_i))$, $\gamma^*(Z_i)=\mathbb{E}[Y_i|Z_i]$, and $\alpha^*(Z_i)=\mathbb{E}[X_i|Z_i]$.\footnote{
	Alternatively, we can use the score $\psi(W_i;\theta,\kappa,\alpha)=(Y_i-X_i\theta-\kappa(Z_i))(X_i-\alpha(Z_i))$ with nuisance parameters $(\kappa^*,\alpha^*)$. \cite{chernozhukov2018double} show that both scores are Neyman orthogonal.
    }
    The DML estimator is given by
	\begin{equation}
		\hat{\theta}=\left(\frac{1}{n}\sum_{i=1}^n \left(X_i-\hat{\alpha}(Z_i)\right)^2\right)^{-1}\left(\frac{1}{n}\sum_{i=1}^n \left(X_i-\hat{\alpha}(Z_i)\right)\left(Y_i-\hat{\gamma}(Z_i)\right)\right).
	\end{equation}
\end{example}

\begin{example}[Demand Estimation with Text and Image Data of Products]
    \label{ex:demand}
	Consider a variant of the demand estimation model in \cite{berry1994estimating}, motivated by the recent findings by \cite{bajari2025hedonic} and \cite{compiani2025demand} that text and image embeddings of products are informative about consumer choice. Let $\theta^*$ be the parameter of interest, which is the price coefficient.
	Let $Y_i$ be the logarithm of the relative market share of product $i$, $X_i$ be the price of product $i$, and $Z_i$ be the unstructured or ultra-high-dimensional control variables of product $i$, such as product images and descriptions.
    A simple demand estimation model is given by $Y_i = X_i\theta^* +\kappa^*(Z_i) + U_i$ and $\mathbb{E}[U_i \mid X_i,Z_i]=0$,
    where $\kappa^*(Z)$ is an unknown nonparametric function and $U_i$ is an error term.
	Under the exogenous price assumption, we can use the partially linear model as in \cref{ex:partially_linear} to estimate the price coefficient $\theta^*$.
    The exogenous price assumption can also be relaxed. Suppose that the price $X_i$ is endogenous due to the brand effect, and the economist observes an instrumental variable $V_i$ for the price $X_i$, such as a cost shifter.
	Consider the demand model given by $Y_i = X_i\theta^* +\gamma^*(Z_i) + U_i,$ $\mathbb{E}[U_i \mid Z_i,V_i]=0$,
    $V_i = \alpha^*(Z_i) + e_i$, and $\mathbb{E}[e_i \mid Z_i]=0$.
    The parameter of interest $\theta^*$ is estimated as the solution of the Neyman orthogonal moment equation
    $\mathbb{E}[\psi(W_i;\theta,\gamma^*,\alpha^*)]=0$,
    where $W_i=(Y_i,X_i,V_i,Z_i)$, $\psi(W_i;\theta,\gamma,\alpha)=(Y_i-X_i\theta-\gamma(Z_i))(V_i-\alpha(Z_i))$, and $\alpha^*(Z_i)=\mathbb{E}[V_i|Z_i]$.
\end{example}

\begin{example}[Imputation with Unstructured Data]
	\label{ex:imputation}
	Suppose that an economist observes a sample of size $n$, $\{(Y_i,D_iX_i,D_i,Z_i)\}_{i=1}^n$, and is interested in the regression coefficient $\theta^*$, which is the solution of the moment $\mathbb{E}[(Y_i-X_i\theta)X_i]=0$, where $Y_i$ is the fully observed outcome and $X_i$ is a scalar covariate with missing values.
	$X_i$ is observed if $D_i=1$ and missing if $D_i=0$.
	The economist can use the unstructured or ultra-high-dimensional variables $Z_i$ to predict the moments involving $X_i$.
	We impose the missing at random assumption, $(X_i,Y_i)\indep D_i\mid Z_i$ (\citealp{rubin1976inference}).
	Let $\alpha^*(Z_i)=P[D_i=1\mid Z_i]$ be the propensity score and assume that $\alpha^*(Z_i)$ is bounded away from zero.
	The parameter $\theta^*$ is estimated by solving the Neyman orthogonal moment equation
	$\mathbb{E}[\psi(W_i;\theta,\gamma^*,\alpha^*)]=0$,
	where $W_i=(Y_i,D_iX_i,D_i,Z_i)$,
	$\psi(W_i;\theta,\gamma,\alpha)=(D_i/\alpha(Z_i))(Y_iX_i-X_i^2\theta - (\gamma_1(Z_i)-\gamma_2(Z_i)\theta)) + \gamma_1(Z_i)-\gamma_2(Z_i)\theta$,
	$\gamma^*_1(Z_i)=\mathbb{E}[X_iY_i|Z_i]$,
	$\gamma^*_2(Z_i)=\mathbb{E}[X_i^2|Z_i]$, and $\gamma^*=(\gamma_1^*,\gamma_2^*)$.
	When the propensity is trained by binary logistic loss, the score uses $\hat{\alpha}=\Lambda\circ\hat{f}_{\alpha}\circ\check{h}$ on the probability scale.
\end{example}

\begin{example}[Average Treatment Effect with Unstructured Confounders]
	\label{ex:ate}
	Let $D_i\in\{0,1\}$ denote treatment, $Y_i(d)$ the potential outcome under treatment status $d\in\{0,1\}$, and $Z_i$ the unstructured or ultra-high-dimensional confounders.
	The average treatment effect (ATE) is $\theta^*=\mathbb{E}[Y_i(1)-Y_i(0)]$.
	Assume that the observed outcome satisfies $Y_i=D_iY_i(1)+(1-D_i)Y_i(0)$, and impose unconfoundedness, $(Y_i(0),Y_i(1))\indep D_i\mid Z_i$ (\citealp{rosenbaum1983central}).
	Let $\alpha^*(Z_i)=P[D_i=1\mid Z_i]$ be the propensity score and assume the overlap condition: $\varepsilon\leq\alpha^*(Z_i)\leq1-\varepsilon$ almost surely for some fixed $\varepsilon\in(0,1/2)$. Let $\gamma^*(d,Z_i)=\mathbb{E}[Y_i\mid D_i=d,Z_i]=\mathbb{E}[Y_i(d)\mid Z_i]$.
	The ATE $\theta^*=\mathbb{E}[\gamma^*(1,Z_i)-\gamma^*(0,Z_i)]$ is estimated using DML with the Neyman orthogonal moment equation
	$\mathbb{E}[\psi(W_i;\theta,\gamma^*,\alpha^*)]=0$,
	where $W_i=(Y_i,D_i,Z_i)$ and $\psi(W_i;\theta,\gamma,\alpha)=(D_i(Y_i-\gamma(1,Z_i))/(\alpha(Z_i))-(1-D_i)(Y_i-\gamma(0,Z_i))/(1-\alpha(Z_i)))+(\gamma(1,Z_i)-\gamma(0,Z_i)-\theta)$.
	If the propensity is trained by binary logistic loss, its fitted log odds are converted to $\hat{\alpha}=\Lambda\circ\hat{f}_{\alpha}\circ\check{h}$ before evaluating this score.
\end{example}

\subsection{Transfer Learning}
\label{sec:transfer_setup}
In this subsection, we explain the transfer-learning setup for the DML estimator.
We explicitly consider the source task for the computer scientist, together with the target task for the economist.

Hereafter, let $P_{\so}$ and $P_{\ta}$ denote the source and target task distributions, and let $\mathbb{E}_{\so}$ and $\mathbb{E}_{\ta}$ denote the corresponding expectations.
Let $P$ and $\mathbb{E}$ denote the joint distribution of the source and target tasks and its expectation, respectively.
Thus, the expectations and probabilities used in the above subsections are under the target task $\mathbb{E}_{\ta}$ and $P_{\ta}$ except for the $o_P$ notation, which is with respect to both the source and target tasks.

We consider the setting where a source sample of size $m$ and a target sample of size $n$ share the same best-approximating embedding function $h_m$ (\cref{def:approx_embedding}).
Our asymptotics let both the source sample size $m$ and the target sample size $n$ go to infinity, where the rate at which $m$ diverges may depend on $n$.
Formally, we treat the source sample size $m=m_n$ as a sequence indexed by the target sample size $n$ and let $m_n\to\infty$ as $n\to\infty$.
Suppose that these samples are generated by 
$(S_j,Z_j)\sim_{i.i.d.}P_{\so}$ for $j=1,\dots,m$ and $(Y_i,Z_i)\sim_{i.i.d.}P_{\ta}$ for $i=1,\dots,n$, where $(S_j,Z_j)$ and $(Y_i,Z_i)$ are independent and $Z_j$ and $Z_i$ have different marginal probability distributions.
Here, $S_j\in\mathcal{S}\subset\mathbb{R}$ is the outcome for the source task, $Y_i\in\mathcal{Y}\subset\mathbb{R}$ is the outcome for the target task, and $Z_j,Z_i\in\mathcal{Z}\subset\mathbb{R}^d$ denote unstructured or ultra-high-dimensional covariates.

\paragraph{Source Task for Computer Scientist}
We focus on the most empirically relevant setting where the source task outcome is multi-valued, i.e., $\mathcal{S}=\{1,2,\dots,T+1\}$ for some integer $T\geq1$.
Then, the source task is multi-class classification with $T+1$ classes and the computer scientist trains the pre-trained model $g\circ h(z):\mathcal{Z}\mapsto\mathbb{R}^{T}$ to provide the predicted log-odds for each class $t=1,\ldots,T$ against the base class $T+1$.
Let $h:\mathcal{Z}\mapsto\mathcal{V}\subset\mathbb{R}^{K}$ be an embedding function returning $h(z)\in\mathcal{V}$ for some integer $K\geq1$, and let
$g=(g_1,\ldots,g_T):\mathcal{V}\mapsto\mathbb{R}^{T}$ be a model for the source task, where $g_t(v)$ is the $t$-th element of $g(v)$ for $t=1,\ldots,T$. We normalize the score of the base class as $g_{T+1}(v)=0$ for every $v\in\mathcal{V}$.
Let $\mathcal{G}_m$ denote the function class of $g$, $\mathcal{G}_{t,m}$ denote the function class of $g_t$ for each $t=1,\ldots,T$, and $\mathcal{H}_m$ denote the function class of $h$ for the source task.
For example, in a deep neural network model, the model can be decomposed as $g\circ h$, where $h$ is some hidden layer of a deep neural network and $g$ is the remaining transformation from the hidden layer to the output layer.
For a given loss function $\ell_{\so}:\mathbb{R}^{T}\times\mathcal{S}\mapsto\mathbb{R}$, let the set of population minimizers be the solution of the population risk minimization:\footnote{
    As in \cite{farrell2021deep}, we assume that minimizer(s) exist for both population and sample risk for simplicity. For the statistical convergence rate theorem, exact existence is not crucial as we can always replace the infimum by a function within a functional class with arbitrarily close risk.
}
\begin{equation}
    \mathcal{R}_m = \argmin_{g\in\mathcal{G}_m,h\in\mathcal{H}_m}\mathbb{E}_{\so}[\ell_{\so}(g\circ h(Z),S)].
\end{equation}
For a score vector $a_{\so}=(a_{\so,1},\ldots,a_{\so,T})'\in\mathbb{R}^T$, we similarly normalize the base-class score as $a_{\so,T+1}=0$. The most common loss function for multi-class classification is the multinomial logit loss:\footnote{This loss function is also common in the machine learning literature, where it is often called the cross-entropy loss with a softmax output layer.}
\begin{equation*}
	\ell_{\so}(a_{\so},S)
	= -\sum_{t=1}^{T}\mathds{1}\{S=t\}a_{\so,t}
	+\log\left(1+\sum_{t=1}^{T}\exp(a_{\so,t})\right).
\end{equation*}
The key observation here is that the source task is essentially trained on $T$ different tasks $g_1,\ldots,g_T$ with shared embeddings $h(Z)\in\mathbb{R}^K$ to classify labels with $T+1$ classes.

The computer scientist observes a sample of size $m$, $\{(S_j,Z_j)\}_{j=1}^m\sim P_{\so}^m$, for the source task and estimates the pre-trained model $(\check{g},\check{h})$ by some estimation method, such as empirical risk minimization:
\begin{equation}
    \label{eq:sample_minimizer_so}
    (\check{g},\check{h})\in\argmin_{g\in\mathcal{G}_m,h\in\mathcal{H}_m}\frac{1}{m}\sum_{j=1}^m\ell_{\so}(g\circ h(Z_j),S_j).
\end{equation}

\paragraph{Target Task for Economist}
Next, we consider the target task for the economist.
The economist will train the target model $f\circ h(z):\mathcal{Z}\mapsto\mathbb{R}$ to predict the target outcome $Y\in\mathcal{Y}$ using learned embeddings from the source task.
Let $f:\mathcal{V}\mapsto\mathbb{R}$ be the model for the target task, where $\mathcal{F}_n$ is the function class of $f$ for the target task.
Let $\ell_{\ta}:\mathbb{R}\times\mathcal{Y}\mapsto\mathbb{R}$ be the loss function for the target task such as the least squares loss for regression and the logistic loss for binary classification.
Given an embedding function $h_m$ from the source task, let $\mathcal{F}_n^{\sub}(h_m)$ denote the set of solutions to the population risk minimization problem:
\begin{equation}
	\mathcal{F}_n^{\sub}(h_m) = \argmin_{f\in\mathcal{F}_n^{\sub}}\mathbb{E}_{\ta}[\ell_{\ta}(f\circ h_m(Z),Y)],
\end{equation}
where $\mathcal{F}_n^{\sub}\subseteq\mathcal{F}_n$ is a subset of a functional class $\mathcal{F}_n$ for the target task, possibly subject to restrictions such as a sparsity restriction for LASSO in a linear model $\mathcal{F}_n$.
We often have $\mathcal{F}_n^{\sub}=\mathcal{F}_n$ without additional restrictions, for example for linear regression or shallow neural networks.

The economist observes a sample of size $n$, $\{Y_i,Z_i\}_{i=1}^{n}\sim P_{\ta}^n$, and knows the pre-trained model $\check{h}$ given by the computer scientist, but does not observe the source sample.
Since the embedding $\check{h}(z)$ is often high-dimensional, the economist can use various machine learning models for $f$ in the target task.
For example, the economist uses a pre-trained convolutional neural network (CNN) to extract high-dimensional embeddings $\check{h}(z)$ (e.g., embeddings of dimension 4{,}096 from VGG19's final hidden fully connected layer after activation) from ultra-high-dimensional image data (e.g., pixel information with 150{,}528 dimensions).
The extracted embeddings are then used as regressors, for example, in a ridge regression of the outcome $Y_i$ on $\check{h}(Z_i)$.

Finally, the economist estimates the target model $\hat{f}$ by empirical risk minimization:\footnote{
		Throughout the paper, we impose suitable conditions ensuring that minimizers indexed by random variables exist and are measurable.
	}
\begin{equation}
    \label{eq:sample_minimizer_ta}
    \hat{f}\in\argmin_{f\in\mathcal{F}_n}\frac{1}{n}\sum_{i=1}^n\ell_{\ta}(f\circ \check{h}(Z_i),Y_i).
\end{equation}

In our theory, we allow small optimization errors for both estimators as shown in \cref{eq:sample_minimizer_ta,eq:sample_minimizer_so} in \cref{thm:target_convergence,thm:source_convergence}.

In summary, a computer scientist provides the pre-trained model $(\check{g},\check{h})$ estimated from the source sample $\{(S_j,Z_j)\}_{j=1}^m$, and an economist uses the embeddings $\check{h}(Z_i)$ as regressors to estimate the target function $\hat{f}$ in the target sample $\{(Y_i,Z_i)\}_{i=1}^n$.

Below, we present four examples of the transfer learning setup.
Further details of pre-trained deep learning models are provided in \cref{sec:pretrained_dl}.

\begin{example}[DNN Embeddings for Ultra-High-Dimensional Data]
    \label{ex:DNN}
	Consider the setting where the computer scientist trains a deep neural network (DNN) on the source task and the economist uses the learned embeddings as inputs to machine learning models in the target DML task.
	In this case, the embedding function $h$ is given by the last hidden layer of the DNN, and the pre-trained model $g$ is given by the output layer of the DNN (linear regression). Economists can use various machine learning models, such as shallow neural networks and ridge regression, for the function $f$ in the target task.
    The data $Z$ can be ultra-high-dimensional, and the DNN can extract high-dimensional embeddings of the data using a source task nonparametric regression problem, where the outcome $S$ is a continuous variable and the loss function $\ell_{\so}$ is the mean squared error loss.
\end{example}

An autoencoder is a special case of the DNN in \cref{ex:DNN}, where the source task is the nonparametric regression task with the same input and output, i.e., $S=Z$, and the loss function $\ell_{\so}$ is a distance function between the input and output, such as the mean squared error loss (\citealp{hinton2006reducing}).

\begin{example}[CNN Embeddings for Image Data]
	The computer scientist trains a CNN, such as VGG19 (\citealp{simonyan2015very}) or ResNet50 (\citealp{he2016deep}), on the source task\footnote{Pre-trained models for both CNNs are available in the torchvision package in PyTorch.}, and the economist uses the learned embeddings as inputs to machine learning models in the target DML task as in \cref{ex:DNN}.
	The embedding function $h$ is the output of the final hidden fully connected layer after activation in VGG19 and the flattened output of the final average-pooling layer in ResNet50. In both cases, the remaining classifier is affine, and the source head $g$ returns its logit contrasts relative to a baseline class, as in \cref{sec:transfer_setup}.
	The covariate $Z$ is an image, and the CNN can extract high-dimensional embeddings of dimension 4{,}096 for VGG19 and 2{,}048 for ResNet50 from the image using the nonparametric classification task in the source task, where the outcome $S$ is an image label (1{,}000 classes in ImageNet) and the loss function $\ell_{\so}$ is the cross-entropy loss (negative log likelihood).
\end{example}

\begin{example}[Transformer Embeddings for Text Data]
	The computer scientist trains a transformer model, such as BERT (\citealp{devlin2019bert}), on the source task\footnote{Pre-trained models for BERT are available in the Hugging Face Transformers package in PyTorch.}, and the economist uses the learned embeddings as inputs to machine learning models in the target DML task as in \cref{ex:DNN}.
	BERT is pre-trained on two tasks: masked language modeling (MLM) and next sentence prediction (NSP).
	The last transformer block returns a contextualized state for each token in the input text. Masked language modeling predicts the original token at randomly selected positions, including selected tokens that remain unchanged in the corrupted input. Next sentence prediction uses the state of the special token \texttt{[CLS]} to predict whether the second sentence follows the first, in which the special token \texttt{[SEP]} separates the sentences.
	Both tasks use cross-entropy loss, with approximately 30{,}000 token labels for masked language modeling and two labels for next sentence prediction.
	For a text-level embedding of dimension 1{,}024 in BERT LARGE, we consider either the transformed \texttt{[CLS]} state or the mean of the transformed MLM states over positions eligible for token prediction, excluding \texttt{[CLS]}, \texttt{[SEP]}, and \texttt{[PAD]}.
\end{example}

See \cref{sec:pretrained_dl} for details of the functional form of the pre-trained models in the examples above.

\begin{example}[Multimodal Embeddings for Image and Text Data]
    The economist can also use the multimodal embeddings learned by the computer scientist.
	For example, the economist can use images and text as unstructured inputs and extract embeddings separately using a pre-trained CNN and a pre-trained transformer.
\end{example}

\section{Transfer Learning Theory}
\label{sec:transfer_theory}
This section and the next present our main theoretical results on transfer learning for DML.
For this purpose, we first introduce the best predictors for the source and target tasks.
Suppose that the source and target tasks admit best predictors satisfying
$$a_{\so}^*\in\argmin_{a_{\so}:\text{ square-integrable}}\mathbb{E}_{\so}[\ell_{\so}(a_{\so}(Z),S)]$$
 and
$$a_{\ta}^*\in\argmin_{a_{\ta}:\text{ square-integrable}}\mathbb{E}_{\ta}[\ell_{\ta}(a_{\ta}(Z),Y)],$$
respectively.
For example, suppose that the source loss function $\ell_{\so}$ is the multinomial logit loss, $P_{\so}(S=t\mid Z)>0$ for every $t=1,\ldots,T+1$, $P_{\so}$-a.s., and the corresponding log-odds functions are square-integrable under $P_{\so}$. Under the normalization $a_{\so,T+1}^*(Z)=0$ and suitable moment conditions, the best source predictor is almost surely unique and given by the log-odds of the class probabilities, $a_{\so,t}^*(Z) = \log(P_{\so}(S=t\mid Z) / P_{\so}(S=T+1\mid Z))$ for $t=1,\ldots,T$.
For the target task, if the loss is the squared loss $\ell_{\ta}(a_{\ta}(Z),Y)=(1/2)\times(Y - a_{\ta}(Z))^2$ and suitable moment conditions are satisfied, the best predictor is almost surely unique and given by the conditional expectation of the outcome as $a_{\ta}^*(Z)=\mathbb{E}_{\ta}[Y\mid Z]$, and if the loss is the binary logit loss $\ell_{\ta}(a_{\ta}(Z),Y)=-Y\log(\exp(a_{\ta}(Z))/(1+\exp(a_{\ta}(Z))))-(1-Y)\log(1/(1+\exp(a_{\ta}(Z))))$ and suitable moment conditions are satisfied, the best predictor is almost surely unique and given by the log-odds of the conditional probability as $a_{\ta}^*(Z)=\log(P_{\ta}(Y=1\mid Z) / P_{\ta}(Y=0\mid Z))$.
In DML applications, these two loss functions are the most common for the target task.\footnote{For binary logistic loss, the probability nuisance is $\Lambda\circ a_{\ta}^*$, and its estimator is $\Lambda\circ\hat{f}\circ\check{h}$. Since $0<\Lambda'(u)=\Lambda(u)(1-\Lambda(u))\leq1/4$ for every $u\in\mathbb{R}$,
\begin{equation}
	\left\|\Lambda\circ\hat{f}\circ\check{h}-\Lambda\circ a_{\ta}^*\right\|_{P_{\ta},2}
	\leq
	\frac{1}{4}\left\|\hat{f}\circ\check{h}-a_{\ta}^*\right\|_{P_{\ta},2}.
	\label{eq:logit_probability_rate}
\end{equation}
Thus, the convergence rates below for the fitted log odds also bound the probability-scale nuisance error required by DML.}

In this section, we develop a general convergence rate theory for the target model $\hat{f}\circ \check{h}$ by relating it to the source task performance of the pre-trained model $\check{g}\circ\check{h}$ relative to the best predictor $a_{\so}^*$.
In \cref{sec:key_assumptions}, we introduce two key concepts: a best-approximating embedding function and task diversity.
In \cref{sec:convergence_rate_general}, we present the theorem on the convergence rate of the target model in general form.
\cref{sec:conditions_diversity} discusses sufficient and necessary conditions for task diversity.

\subsection{Key Concepts}
\label{sec:key_assumptions}
\subsubsection{Best-Approximating Embedding Function}
\label{sec:approx_embedding}
We first introduce the concept of a best-approximating embedding function, which is a key definition for the transfer learning theory.
In general, $a^*_{\so}$ and $a^*_{\ta}$ need not be representable as $g_m \circ h_m$ and $f_n \circ h_m$.
This mismatch arises from restrictions on the source classes $\mathcal{G}_m,\mathcal{H}_m$ and the target class $\mathcal{F}_n^{\sub}$ even if there exists an embedding function that captures the common features of the source and target tasks well.
The following definition quantifies the approximation error for both source and target tasks.

\begin{definition}[Approximation Error]
    \label{def:approx_embedding}
    For sequences of functions $(g_m,h_m)\in\mathcal{R}_m$ and $f_n\in\mathcal{F}_n^{\sub}(h_m)$, we define the approximation errors for the source and target tasks as
    \begin{align}
        \|g_m \circ h_m - a_{\so}^*\|_{P_{\so},2}
        & =: \epsilon_{\so,m},\\
        \|f_n \circ h_m - a_{\ta}^*\|_{P_{\ta},2}
        & =: \epsilon_{\ta,n}.
    \end{align}
\end{definition}

\begin{remark}
    This definition allows $\epsilon_{\so,m}$ and $\epsilon_{\ta,n}$ to depend on $(g_m,h_m)$ and $(f_n,h_m)$, respectively, but we suppress this dependence in the notation for simplicity.
    If there exists a best embedding function $h^*$ such that $h^*\in\mathcal{H}_m$ eventually and there exist some measurable functions $g^*$ and $f^*$ such that
    $a_{\so}^*=g^*\circ h^*$, $a_{\ta}^*=f^*\circ h^*$, then we can interpret $h^*$ as the best embedding function and $g^*$ and $f^*$ as the best source and target predictors, respectively.
    Although it is common to assume the existence of the best embedding function in the transfer learning literature and set $\epsilon_{\so,m}$ and $\epsilon_{\ta,n}$ exactly to zero, we allow the approximate embedding function to cover more general cases where there is no best embedding function, but the source and target models can be well approximated by some embedding function.
\end{remark}

Under the assumptions of \cref{thm:transfer_learning_general}, target prediction consistency follows when the target estimation error, the source estimation error adjusted for transferability, and the target approximation error all decay to zero.
The approximation errors depend on the richness of the function classes $\mathcal{G}_m$ and $\mathcal{F}_n^{\sub}$ and how well $h_m$ captures the common features of the source and target tasks. Suppose that both tasks use squared loss, $\mathcal{G}_m$ and $\mathcal{F}_n^{\sub}$ are unrestricted linear head classes, and the outcomes and coordinates of $h_m$ are square-integrable under their respective task distributions. Then $g_m\circ h_m$ and $f_n\circ h_m$ are the $L^2(P_{\so})$ and $L^2(P_{\ta})$ projections of $a_{\so}^*$ and $a_{\ta}^*$ onto the span of the constant and the coordinates of $h_m$, respectively, with the source projection interpreted coordinatewise. This parallels series regression (\citealp{newey1997convergence}), with learned embeddings replacing fixed basis functions. Note that constrained or nonlinear head classes need not yield this projection.

\subsubsection{Task Diversity}
\label{sec:diversity}
Next, we introduce the concept of task diversity, which is another key definition.
We define task diversity in a way that makes its connection to identification explicit.
The reader can refer to \cref{sec:diversity_lit} for a detailed comparison with existing definitions of task diversity in the machine learning literature.
For functions $(g_m,h_m)\in\mathcal{R}_m$ and $f_n\in\mathcal{F}_n^{\sub}(h_m)$, define the following near-identified sets of $h_m$ for $\rho\geq0$:
\begin{equation*}
    \mathcal{H}_{\so,m}(\rho):=\left\{h\in\mathcal{H}_m: \min_{g\in\mathcal{G}_{m}}\left\|g\circ h-g_m\circ h_m\right\|_{P_{\so},2}\leq \rho\right\}
\end{equation*}
and 
\begin{equation*}
    \mathcal{H}_{\ta,m}(\rho):=\left\{h\in\mathcal{H}_m: \min_{f\in\mathcal{F}_n^{\sub}}\left\|f\circ h-f_n\circ h_m\right\|_{P_{\ta},2} \leq \rho\right\},
\end{equation*}
where we assume that the minimum exists for both the source and target tasks for simplicity.
Note that $\mathcal{H}_{\so,m}(0)$ and $\mathcal{H}_{\ta,m}(0)$ are the identified sets for $h_m$ given the population predictions $g_m\circ h_m$ and $f_n\circ h_m$, respectively.
If $h_m$ is identified, these sets shrink to $\{h_m\}$ as $\rho\downarrow 0$. In general, however, $h_m$ need not be identified, and these sets can contain multiple elements. It is known that the hidden layers of the deep neural network are identified only up to some transformations, such as the permutation of the order of the hidden units, the scaling of the hidden units, and the rotation of the hidden units (\citealp{grigsby2023hidden}). See also \cref{ex:counter} for another example of partial identification of the embedding function for middle hidden layers of a deep neural network.

\begin{definition}[Task Diversity]
    \label{def:diversity}
    Fix functions $(g_m,h_m)\in\mathcal{R}_m$ and $f_n\in\mathcal{F}_n^{\sub}(h_m)$.
    For function classes $\mathcal{G}_m$, $\mathcal{H}_m$, and $\mathcal{F}_n^{\sub}$ and some $\rho_{\so,m}\geq0$ and $\rho_{\ta,n}\geq0$, we say that $g_m$ is {\it $(\rho_{\so,m},\rho_{\ta,n})$-diverse over $f_n$ with respect to $h_m$} if
	\begin{equation}    
        \label{eq:diversity}
        \mathcal{H}_{\so,m}(\rho_{\so,m}) \subseteq \mathcal{H}_{\ta,m}(\rho_{\ta,n}).
    \end{equation}
\end{definition}

Thus, the task diversity condition requires that if $h$ is in a near-identified set of the source task, then $h$ is also in a near-identified set of the target task, where the distances are measured by the $L^2(P_{\so})$- and $L^2(P_{\ta})$-norms, respectively.

The name ``task diversity'' comes from the intuition that if the $T$ source tasks are sufficiently diverse, then any embedding function $h$ close to the source task identified set should also be close to the target task identified set.
This intuition and sufficient and necessary conditions for task diversity are formalized in \cref{sec:conditions_diversity}.

Next, we introduce the concept of $\nu_n$-transferability, which is defined in terms of the task diversity condition.
\begin{definition}[$\nu_n$-Transferability]
	\label{def:transferability}
	Fix a nonnegative sequence $\{\delta_m\}_{m=1}^{\infty}$, a sequence $\{\nu_n\}_{n=1}^{\infty}$ satisfying $\nu_n>0$ for every $n$, sequences of functions $(g_m,h_m)\in\mathcal{R}_m$, and $f_n\in\mathcal{F}_n^{\sub}(h_m)$.
	For sequences of function classes $\{\mathcal{G}_m\}_{m=1}^{\infty}$, $\{\mathcal{H}_m\}_{m=1}^{\infty}$, and $\{\mathcal{F}_n^{\sub}\}_{n=1}^{\infty}$, we say that $g_m$ is {\it $\nu_n$-transferable to $f_n$ with respect to $h_m$ at rate $\delta_m$} if for every nonnegative sequence $\{\rho_{\so,m}\}_{m=1}^{\infty}$ such that $\rho_{\so,m}\asymp\delta_m$, the $(\rho_{\so,m},\rho_{\ta,n})$-diversity condition holds with $\rho_{\ta,n}=\rho_{\so,m}/\nu_n$ for all sufficiently large $n$.
\end{definition}

$\nu_n$-transferability at rate $\delta_m$ requires diversity at source radius $\rho_{\so,m}\asymp\delta_m$, meaning that $c\delta_m\leq\rho_{\so,m}\leq C\delta_m$ for some constants $0<c\leq C<\infty$ and all sufficiently large $m$.
In particular, the condition applies to every fixed positive multiple of $\delta_m$, with the threshold in $n$ allowed to depend on the radius sequence. It need not imply $(0,0)$-diversity.
The quantity $\nu_n$ captures the identification strength of the target task informed by the source task.
Larger $\nu_n$ implies that the task $g_m$ is more transferable to $f_n$ since it implies a smaller target task near-identified set given the near-identified set of the source task.

We will relate the sequences $\rho_{\so,m}$ to the convergence rates of the source task introduced below.
Let $f_n^{\bullet}\in\argmin_{f\in\mathcal{F}_n^{\sub}} \|f\circ \check{h}-f_n \circ h_m\|_{P_{\ta},2}$ be the best approximation of the population target model $f_n\circ h_m$ by the estimated embedding function $\check{h}$.
Formally, $f_n^{\bullet}$ depends on $\check{h}$, but we suppress the dependence for notational simplicity.
\begin{assumption}[Convergence Rate of Models]
    \label{asm:convergence}\ 
    \begin{enumerate}[label=(\roman*)]
        \item (Source Model) There exists some sequence $\delta_{\so,m}\geq0$ such that
        \begin{equation}
            \label{eq:source_convergence}
            \left\|\check{g}\circ \check{h}-g_m\circ h_m\right\|_{P_{\so},2}=O_{P_{\so}}(\delta_{\so,m}).
        \end{equation}
		\item (Target Model given $\check{h}$) For every sequence of estimated embedding functions $\check{h}\in\mathcal{H}_m$ satisfying \cref{eq:source_convergence}, there exists some sequence $r_{\ta,n}\geq0$ such that
        \begin{equation}
            \label{eq:target_convergence}
            \left\|\hat{f}\circ \check{h}-f_n^{\bullet}\circ \check{h}\right\|_{P_{\ta},2}=O_{P}(r_{\ta,n}).
        \end{equation}
    \end{enumerate}
\end{assumption}

The first part of this assumption states that the pre-trained model $\check{g}\circ \check{h}$ converges to the population source model $g_m\circ h_m$ at some rate $\delta_{\so,m}$.
The second part states that the target model $\hat{f}\circ \check{h}$ converges to the population target model $f_n^{\bullet}\circ \check{h}$ at some rate $r_{\ta,n}$ given the estimated embedding function $\check{h}$.
Note that the $O_P$ notation in the second part is with respect to both the source and target samples since $\check{h}$ depends on the source sample.
This assumption is high-level, and we will provide more primitive sufficient conditions in \cref{sec:transfer_convergence} to avoid assuming this directly.

Now, we assume the $\nu_n$-transferability condition with the convergence rate of the source task $\delta_{\so,m}$.

\begin{assumption}[Transferability with Convergence Rates]
    \label{asm:transferability}
    For sequences of function classes $\{\mathcal{G}_m\}_{m=1}^{\infty}$, $\{\mathcal{H}_m\}_{m=1}^{\infty}$, and $\{\mathcal{F}_n^{\sub}\}_{n=1}^{\infty}$, there exist $(g_m,h_m)\in\mathcal{R}_m$ and $f_n\in\mathcal{F}_n^{\sub}(h_m)$ such that the task $g_m$ is $\nu_n$-transferable to $f_n$ with respect to $h_m$ at rate $\delta_{\so,m}$, where $\delta_{\so,m}$ is the convergence rate of the source task defined in \cref{asm:convergence}.
\end{assumption}

\cref{asm:transferability} and the quantity $\nu_n$ play a key role in determining the convergence rate of the target model $\hat{f}\circ \check{h}$ as we will see in the next subsection.

\subsection{Convergence Rate of Target Model}
\label{sec:convergence_rate_general}
Our main theorem on the convergence rate of the target model is as follows.
\begin{theorem}[Convergence Rate of Target Model]
    \label{thm:transfer_learning_general}
    Suppose that \cref{asm:convergence,asm:transferability} hold.
        Then, for every sequence of estimated embedding functions $\check{h}\in\mathcal{H}_m$ satisfying \cref{eq:source_convergence}, we have
    \begin{equation}
        \left\|\hat{f}\circ \check{h}-a^*_{\ta}\right\|_{P_{\ta},2}=O_{P}\left(r_{\ta,n}+\delta_{\so,m}/\nu_n+\epsilon_{\ta,n}\right).
    \end{equation}
\end{theorem}

This theorem states that the convergence rate of the target model $\hat{f}\circ \check{h}$ to the best predictor $a_{\ta}^*$ is determined by three components: (i) the convergence rate $r_{\ta,n}$ of the target model given the estimated embedding function $\check{h}$, (ii) the convergence rate of the source task adjusted by the transferability strength $\nu_n$, and (iii) the approximation error $\epsilon_{\ta,n}$ due to a best-approximating embedding function.
Thus, the rate of the target model is determined by the slowest of these three components.
The target-stage rate $r_{\ta,n}$ is achieved when transferability is strong, the source task converges sufficiently fast, and the approximation error is sufficiently small.
In particular, under the theorem's assumptions, $r_{\ta,n}+\delta_{\so,m}/\nu_n+\epsilon_{\ta,n}=o(1)$ suffices for consistency of the target predictor.

\begin{remark}
	\label{rem:transferability}
    Applied work using source task embeddings often ignores the estimation error of the pre-trained model and directly applies the DML theory for the target task by treating the estimated embedding function $\check{h}$ as the best-approximating embedding function $h_m$.
    This theorem provides a theoretical justification for this common practice by showing that the estimation error of the pre-trained model can be negligible if the transferability is strong, the source task converges sufficiently fast, and the approximation error is sufficiently small.

    However, there are two important differences between the current practice and the theoretical result in this theorem.
    First, current practice often ignores the randomness of the estimated embedding function $\check{h}$, while our convergence rate is with respect to joint randomness under $P$, rather than target sample randomness under $P_{\ta}$ alone.\footnote{
        Our rate is pointwise with respect to the sequence of estimated embedding functions $\check{h}\in\mathcal{H}_m$ satisfying \cref{eq:source_convergence}. 
		Extending the result to a uniform convergence rate over the estimated set of embedding functions $h_m$ is an interesting direction for future research, but doing so would require additional assumptions and is beyond the scope of this paper.
    }
	Second, applied work often ignores nonidentification of $h_m$ and the associated task diversity condition. Without considering the nonidentification issue, the convergence point of machine learning models in the target task can depend on a specific optimization algorithm including the choice of initialization and the order of observations in the optimization steps used by the computer scientist to estimate the pre-trained model in the source task, which is not desirable for the economist.
	
	The $\nu_n$-transferability condition ensures that, for any nonnegative sequence $\rho_m\asymp\delta_{\so,m}$ and all sufficiently large $n$, the source task near-identified set $\mathcal{H}_{\so,m}(\rho_m)$ is included in the target task near-identified set $\mathcal{H}_{\ta,m}(\rho_m/\nu_n)$.
	While the sequence $\check{h}$ does not converge in probability to one specific function due to nonidentification, as $\check{h}$ eventually belongs to the source task near-identified set $\mathcal{H}_{\so,m}(\rho_m)$ with high probability, the $\nu_n$-transferability condition ensures that it also eventually belongs to the target task near-identified set $\mathcal{H}_{\ta,m}(\rho_m/\nu_n)$ with high probability (\cref{fig:transferability}).
	Together with \cref{asm:convergence}, this ensures that the target model $\hat{f}\circ\check{h}$ attains the rate in \cref{thm:transfer_learning_general} for every pre-trained embedding sequence satisfying \cref{eq:source_convergence}, regardless of the optimization algorithm used to obtain it.
	\begin{figure}[htbp]
	\centering
	\resizebox{0.6\linewidth}{!}{%
		\begin{tikzpicture}[
			x=1cm,
			y=1cm,
			line cap=round,
			line join=round,
			every node/.style={font=\small, fill=white, inner sep=1.5pt}
		]
			\path[use as bounding box] (0,0) rectangle (10.4,4.8);

			\draw[
				myorange,
				fill=myorange!6,
				line width=1.1pt,
				dash pattern=on 3pt off 2.3pt
			]
				(0.25,2.10)
				.. controls (0.35,3.55) and (1.75,4.55) .. (4.15,4.72)
				.. controls (6.70,4.90) and (9.15,4.12) .. (10.05,2.75)
				.. controls (10.48,2.08) and (10.28,1.55) .. (9.68,1.45)
				.. controls (9.12,1.35) and (8.68,1.78) .. (8.22,1.28)
				.. controls (7.80,0.82) and (8.88,0.70) .. (8.80,0.28)
				.. controls (8.70,-0.10) and (7.10,0.02) .. (5.72,0.18)
				.. controls (4.18,-0.08) and (2.48,0.05) .. (1.45,0.62)
				.. controls (0.58,1.10) and (0.18,1.58) .. cycle;

			\draw[
				myblue,
				fill=myblue!9,
				line width=1.1pt,
				dash pattern=on 3pt off 2.3pt
			]
				(2.22,1.15)
				.. controls (2.02,2.28) and (3.02,3.40) .. (4.78,3.67)
				.. controls (6.42,3.92) and (7.92,3.25) .. (8.06,2.15)
				.. controls (8.12,1.70) and (7.85,1.35) .. (7.58,1.08)
				.. controls (7.35,0.85) and (8.15,0.67) .. (8.08,0.42)
				.. controls (7.98,0.08) and (6.78,0.03) .. (5.62,0.16)
				.. controls (4.45,-0.02) and (3.25,0.08) .. (2.62,0.52)
				.. controls (2.38,0.68) and (2.25,0.90) .. cycle;

			\draw[myorange, line width=1.1pt]
				(0.52,2.15)
				.. controls (0.62,3.35) and (2.02,4.30) .. (4.20,4.43)
				.. controls (6.52,4.57) and (8.72,3.90) .. (9.58,2.72)
				.. controls (9.98,2.18) and (9.82,1.72) .. (9.25,1.62)
				.. controls (8.72,1.52) and (8.30,1.82) .. (7.92,1.40)
				.. controls (7.55,0.98) and (8.38,0.76) .. (8.18,0.42)
				.. controls (7.90,0.03) and (6.55,0.25) .. (5.50,0.40)
				.. controls (4.10,0.18) and (2.75,0.28) .. (1.72,0.82)
				.. controls (0.92,1.25) and (0.45,1.65) .. cycle;

			\draw[myblue, line width=1.1pt]
				(2.43,1.20)
				.. controls (2.25,2.21) and (3.16,3.20) .. (4.81,3.47)
				.. controls (6.32,3.70) and (7.64,3.09) .. (7.84,2.11)
				.. controls (7.94,1.64) and (7.65,1.30) .. (7.39,1.06)
				.. controls (7.15,0.84) and (7.87,0.67) .. (7.78,0.46)
				.. controls (7.64,0.16) and (6.61,0.17) .. (5.57,0.31)
				.. controls (4.47,0.13) and (3.44,0.22) .. (2.85,0.61)
				.. controls (2.61,0.77) and (2.48,0.98) .. cycle;

			\coordinate (first-1) at (9.35,4.05);
			\coordinate (first-2) at (8.65,3.78);
			\coordinate (first-3) at (7.96,3.55);
			\coordinate (first-4) at (7.35,3.35);
			\coordinate (first-5) at (6.82,3.20);
			\coordinate (first-6) at (6.38,3.10);
			\coordinate (first-7) at (6.02,3.03);
			\coordinate (first-8) at (5.75,2.99);
			\coordinate (first-9) at (5.52,2.96);

			\draw[black, line width=0.9pt]
				(first-1) -- (first-2) -- (first-3) -- (first-4)
				-- (first-5) -- (first-6) -- (first-7)
				-- (first-8) -- (first-9);
			\foreach \i in {1,...,9} {
				\fill[black] (first-\i) circle (1.4pt);
			}
			\draw[black, ->, line width=0.9pt]
				(first-4) -- (first-5);
			\node[text=black] at (9.65,4.55) {$\check{h}_{2,m}$};

			\coordinate (second-1) at (0.55,0.25);
			\coordinate (second-2) at (0.87,0.23);
			\coordinate (second-3) at (1.16,0.23);
			\coordinate (second-4) at (1.42,0.26);
			\coordinate (second-5) at (1.65,0.31);
			\coordinate (second-6) at (1.85,0.38);
			\coordinate (second-7) at (2.02,0.46);
			\coordinate (second-8) at (2.16,0.55);
			\coordinate (second-9) at (2.28,0.64);
			\coordinate (second-10) at (2.39,0.72);
			\coordinate (second-11) at (2.49,0.79);

			\draw[black, line width=0.9pt]
				(second-1) -- (second-2) -- (second-3) -- (second-4)
				-- (second-5) -- (second-6) -- (second-7)
				-- (second-8) -- (second-9) -- (second-10) -- (second-11);
			\foreach \i in {1,...,11} {
				\fill[black] (second-\i) circle (1.4pt);
			}
			\draw[black, ->, line width=0.9pt]
				(second-4) -- (second-5);
			\node[text=black] at (0.72,0.55) {$\check{h}_{1,m}$};

			\node[text=myorange] at (3.05,4.08)
				{$\mathcal{H}_{\ta,m}(\rho_m/\nu_n)$};
			\node[text=myorange] at (1.92,2.92)
				{$\mathcal{H}_{\ta,m}(0)$};
			\node[text=myblue] at (5.42,2.30)
				{$\mathcal{H}_{\so,m}(0)$};
			\node[text=myblue] at (1.68,2.02)
				{$\mathcal{H}_{\so,m}(\rho_m)$};
		\end{tikzpicture}
	}
	\caption{Transferability and embedding nonidentification.}
	\label{fig:transferability}
	\captionsetup{font=footnotesize,justification=justified,singlelinecheck=false}
	\caption*{\textit{Notes:} The figure illustrates the concept of transferability and embedding nonidentification. The black dots represent the sequences $\{\check{h}_{1,m}\}_m$ and $\{\check{h}_{2,m}\}_m$ converging to two different functions, the blue sets represent the source task identified set (solid line) and near-identified set (dashed line), and the orange sets represent the target task identified set (solid line) and near-identified set (dashed line). Transferability requires inclusion of the near-identified sets and does not require the identified sets to be nested.}
\end{figure}

\end{remark}

The next theorem gives a useful necessary condition based on the source task identified set $\mathcal{H}_{\so,m}(0)$. We will revisit this necessary condition in \cref{sec:lasso} for the specific case of the LASSO.
\begin{theorem}[Necessary Condition for Convergence of Target Model]
    \label{thm:conv_nec}
	Fix sequences of functions $(g_m,h_m)\in\mathcal{R}_m$ and $f_n\in\mathcal{F}_n^{\sub}(h_m)$.
	Let $f_{n,h_m^\dagger}^{\bullet}\in\argmin_{f\in\mathcal{F}_n^{\sub}} \|f\circ h_m^\dagger-f_n\circ h_m\|_{P_{\ta},2}$ be the best approximation of the population target model $f_n\circ h_m$ by the embedding function $h_m^\dagger$.
	Let $\widehat f_{n,h_m^\dagger}\in\argmin_{f\in\mathcal{F}_n^{\sub}}\frac{1}{n}\sum_{i=1}^n\ell_{\ta}(f\circ h_m^\dagger(Z_i),Y_i)$ be the corresponding estimator of the target model.
    Suppose that, for every deterministic sequence
	$h_m^\dagger\in\mathcal{H}_{\so,m}(0)$,
	\begin{align*}
		\|\widehat f_{n,h_m^\dagger}\circ h_m^\dagger
		-f_{n,h_m^\dagger}^{\bullet}\circ h_m^\dagger\|_{P_{\ta},2}
		&=o_{P_{\ta}}(1),\\
		\|\widehat f_{n,h_m^\dagger}\circ h_m^\dagger-a_{\ta}^*\|_{P_{\ta},2}
		&=o_{P_{\ta}}(1).
	\end{align*}
    Then, $\mathcal{H}_{\so,m}(0)\subseteq \mathcal{H}_{\ta,m}(\rho_{\ta,n})$ for some sequence $\rho_{\ta,n}=o(1)$.
\end{theorem}

\begin{remark}
    The $\nu_n$-transferability condition in \cref{asm:transferability} is stronger than necessary for convergence of the target model for simplicity of presentation.
    A weaker sufficient condition is that $\min_{f\in\mathcal{F}_n^{\sub}}\left\|f\circ \check{h}-f_n\circ h_m\right\|_{P_{\ta},2}=O_{P}(\delta_{\so,m}/\nu_n)$ for every sequence of estimated embedding functions $\check{h}\in\mathcal{H}_m$ satisfying \cref{eq:source_convergence}.
    Indeed, this weaker condition is necessary in the following sense: if $\epsilon_{\ta,n}=O(\delta_{\so,m}/\nu_n)$, $r_{\ta,n}=O(\delta_{\so,m}/\nu_n)$, and
    $\left\|\hat{f}\circ \check{h}-a^*_{\ta}\right\|_{P_{\ta},2}=O_{P}\left(\delta_{\so,m}/\nu_n\right)$ for every sequence of estimated embedding functions $\check{h}\in\mathcal{H}_m$ satisfying \cref{eq:source_convergence}, then $\min_{f\in\mathcal{F}_n^{\sub}}\left\|f\circ \check{h}-f_n\circ h_m\right\|_{P_{\ta},2}=O_{P}(\delta_{\so,m}/\nu_n)$. The proof is similar to the one for \cref{thm:conv_nec}.
\end{remark}

\subsection{Closer Look at Task Diversity}
\label{sec:conditions_diversity}
In this subsection, we discuss sufficient and necessary conditions for task diversity in \cref{def:diversity}.
\subsubsection{Sufficient Condition}
We first assume that the density ratio of $Z$ in the source and target tasks is bounded away from zero.
\begin{assumption}[Prediction Bound for Source Model]
	\label{asm:support}
	There exists $\underline{w}_n>0$ such that $p_{\so}(z)/p_{\ta}(z)\geq\underline{w}_n^2$, $P_{\ta}$-a.s., where $p_{\so}(z)$ and $p_{\ta}(z)$ are the densities of $Z$ in the source and target samples with respect to a common dominating measure.
\end{assumption}

\cref{asm:support} holds only if the support of $Z$ in the target task is included in the support of $Z$ in the source task.
We allow $\underline{w}_n$ to decrease slowly to zero as $n$ increases, but not too quickly.
Although it is commonly assumed in the machine learning literature (e.g., \citealp{johansson2019support}), the density ratio condition is a strong assumption especially when $Z$ is unstructured data.
Relaxing this assumption is beyond the scope of this paper and is an important direction for future research.

\begin{theorem}[Sufficient Condition for Transferability on the Mean Squared Loss]
    \label{thm:suf_transfer}
    Suppose that there exists a Lipschitz function $\Gamma:\mathbb{R}^{T}\to\mathbb{R}$ with Lipschitz constant $L_{\Gamma}$ such that $\|f_n\circ h_m - \Gamma \circ g_m \circ h_m\|_{P_{\ta},2}\leq\varrho_n$ for some tolerance $\varrho_n\geq0$.
	For this $\Gamma$, we assume that for every $h\in\mathcal{H}_{\so,m}(\rho_{\so,m})$, there exists $f\in\mathcal{F}_n^{\sub}$ such that $\left\|f \circ h - \Gamma \circ g^{\bullet}_m \circ h\right\|_{P_{\ta},2} \leq\varrho_n$, where $g^{\bullet}_m\in\argmin_{g\in\mathcal{G}_{m}} \left\|g \circ h - g_m \circ h_m\right\|_{P_{\so},2}$.
    Then, under \cref{asm:support}, $g_m$ is $(\rho_{\so,m},\rho_{\ta,n})$-diverse over $f_n$ with respect to $h_m$ for every $\rho_{\so,m}\geq0$ and $\rho_{\ta,n}=L_{\Gamma}\rho_{\so,m}/\underline{w}_n+2\varrho_n$.
\end{theorem}

This theorem provides a sufficient condition for $(\rho_{\so,m},\rho_{\ta,n})$-diversity based on the existence of a Lipschitz function $\Gamma$ that relates the target model $f_n\circ h_m$ to the source model $g_m\circ h_m$.
The first condition requires that the target model can be approximated by applying the Lipschitz function $\Gamma$ to the source model with some small error $\varrho_n$.
The second condition requires that for every embedding function $h$ close to the source task identified set, there exists a target model $f$ that can approximate $\Gamma \circ g^{\bullet}_m \circ h$ with some small error $\varrho_n$.
Under these conditions and the density ratio condition in \cref{asm:support}, the task $g_m$ is $(\rho_{\so,m},L_{\Gamma}\rho_{\so,m}/\underline{w}_n+2\varrho_n)$-diverse over $f_n$ with respect to $h_m$ for every $\rho_{\so,m}\geq0$.
Note that transferability is stronger if the Lipschitz constant $L_{\Gamma}$ is small, the approximation error $\varrho_n$ is small, and the constant $\underline{w}_n$ related to the infimum of the density ratio is large.
In the extreme case $L_\Gamma=0$ (i.e., $\Gamma$ is constant), we have $(\rho_{\so,m},2\varrho_n)$-diversity.
Holding the remaining quantities fixed, a smaller $L_\Gamma$ gives a tighter sufficient bound.

In the next example, we verify the sufficient condition in \cref{thm:suf_transfer} when both the computer scientist and the economist use linear models.

\begin{example}[Linear Model]
    \label{ex:transferability_linear}
    Consider the simple linear models for both $f_n$ and $g_m$, i.e., $f_n \circ h_m(z) = \alpha_{\ta,n} + \beta_{\ta,n}' h_m(z)$ and $g_m \circ h_m(z) = \alpha_{\so,m} + B_{\so,m} h_m(z)$, where $\alpha_{\ta,n}\in\mathbb{R}$, $\beta_{\ta,n}\in\mathbb{R}^{K}$, $\alpha_{\so,m}=(\alpha_{\so,m,1},\ldots,\alpha_{\so,m,T})'\in\mathbb{R}^{T}$, and $B_{\so,m}=(\beta_{\so,m,1},\ldots,\beta_{\so,m,T})'\in\mathbb{R}^{T\times K}$ with $\beta_{\so,m,t}\in\mathbb{R}^{K}$ for $t=1,\ldots,T$.
	Let $\mathcal{G}_m$ be the class of linear source heads $g_{\alpha,B}(v)=\alpha+Bv$ with $\alpha\in\mathbb{R}^{T}$ and $B\in\mathbb{R}^{T\times K}$ and $\mathcal{F}_n^{\sub}$ be the class of linear target heads $f_{\alpha,\beta}(v)=\alpha+\beta'v$ with $\alpha\in\mathbb{R}$ and $\beta\in\mathbb{R}^{K}$.
    If $B_{\so,m}$ has full column rank, i.e., $\operatorname{rank}(B_{\so,m})=K$, then
    \begin{equation*}
        \Gamma(x)=\alpha_{\ta,n}+\beta_{\ta,n}'B_{\so,m}^{+}(x-\alpha_{\so,m})
    \end{equation*}
    satisfies $f_n\circ h_m=\Gamma \circ g_m \circ h_m$ for any linear $f_n$, where $B_{\so,m}^{+}:= (B_{\so,m}'B_{\so,m})^{-1}B_{\so,m}'$ is the Moore-Penrose inverse of $B_{\so,m}$.

    On the other hand, if $\operatorname{rank}(B_{\so,m})<K$, $g_m$ is not diverse for all directions.
    An affine function $\Gamma$ with an intercept satisfying $f_n\circ h_m=\Gamma \circ g_m \circ h_m$ exists for every value of $h_m$ if and only if $\beta_{\ta,n}$ belongs to the row span of $B_{\so,m}$. That is, there exists $b\in\mathbb{R}^{T}$ such that $\beta_{\ta,n}'=b' B_{\so,m}$.\footnote{
			More generally, by allowing some approximation error $\varrho_n\geq0$, we can instead assume that there exists $b\in\mathbb{R}^{T}$ such that $\|(\beta_{\ta,n}'-b' B_{\so,m})h_m(Z)\|_{P_{\ta},2}\leq \varrho_n$. Choosing $\Gamma(x)=\alpha_{\ta,n}+b'(x-\alpha_{\so,m})$ gives Lipschitz constant $L_{\Gamma}=\|b\|_2$ and approximation error at most $\varrho_n$. In this case, $g_m$ is $(\rho_{\so,m},L_{\Gamma}\rho_{\so,m}/\underline{w}_n+2\varrho_n)$-diverse over $f_n$ with respect to $h_m$, provided that the functional class condition holds.
		Note that the approximation error $\varrho_n$ induced by $\Gamma$ is always zero in the full-column-rank case.
    }
		Then, we can choose $\Gamma(x)=\alpha_{\ta,n}+\beta_{\ta,n}'B_{\so,m}^{+}(x-\alpha_{\so,m})$ by the definition of the Moore-Penrose inverse $B_{\so,m}B_{\so,m}^+B_{\so,m}=B_{\so,m}$ and $\beta_{\ta,n}'=b' B_{\so,m}$.
	Since the intercept does not affect the Lipschitz constant, $L_{\Gamma}=\|\beta_{\ta,n}'B_{\so,m}^{+}\|_2 \leq \|\beta_{\ta,n}\|_2 /\sigma_{\min}^+(B_{\so,m})$ in both cases.\footnote{
        The inequality follows from \cref{lem:l2_ineq,lem:spectral_inverse}.
    }
	Thus, by \cref{thm:suf_transfer}, $g_m$ is $(\rho_{\so,m},L_{\Gamma}\rho_{\so,m}/\underline{w}_n)$-diverse over $f_n$ with respect to $h_m$ if, for every $h\in\mathcal{H}_{\so,m}(\rho_{\so,m})$ and every corresponding best linear source head $(\alpha_{\so,m}^{\bullet},B_{\so,m}^{\bullet})\in\argmin_{\alpha\in\mathbb{R}^{T},B\in\mathbb{R}^{T\times K}} \left\|\alpha+Bh-(\alpha_{\so,m}+B_{\so,m}h_m)\right\|_{P_{\so},2}$, the linear model
    \begin{equation*}
        v\mapsto \alpha_{\ta,n}+\beta_{\ta,n}'B_{\so,m}^{+}
        (\alpha_{\so,m}^{\bullet}-\alpha_{\so,m}+B_{\so,m}^{\bullet}v)
    \end{equation*}
    belongs to $\mathcal{F}_n^{\sub}$. In particular, the displayed linear model belongs to $\mathcal{F}_n^{\sub}$ if $\mathcal{F}_n^{\sub}$ contains all linear models.
\end{example}

\cref{ex:transferability_linear} shows that, when the target class contains all linear heads with intercepts, the row-span condition is sufficient for transferability with $\varrho_n=0$. Under the regularity conditions in \cref{thm:nec_transfer_linear_pre}, the same condition is also necessary for $(0,0)$-diversity. Intuitively, if the coefficients of the source tasks are informative enough to represent the coefficients of the target task, the transferability condition holds.
Thus, under those regularity conditions, $(0,0)$-diversity for every possible target coefficient requires at least $K$ source tasks ($T\geq K$).
The next corollary states sufficient conditions for general Lipschitz target models, which include the linear target case as a special case.

\begin{corollary}[Sufficient Condition for Transferability for Linear Source Models]
    \label{cor:transferability_full_rank}
    Let $\mathcal{G}_m$ be the class of linear source heads $g_{\alpha,B}(v)=\alpha+Bv$ with $\alpha\in\mathbb{R}^{T}$ and $B\in\mathbb{R}^{T\times K}$, and suppose that $g_m \circ h_m = \alpha_{\so,m} + B_{\so,m} h_m$, where $\alpha_{\so,m}\in\mathbb{R}^{T}$ and $B_{\so,m}\in\mathbb{R}^{T\times K}$.
	Consider any Lipschitz target model class $\mathcal{F}_n^{\sub}$ with Lipschitz constant $L_{\mathcal{F}}$. Suppose that, for every $\rho_{\so,m}\geq0$, every $h\in\mathcal{H}_{\so,m}(\rho_{\so,m})$, and every corresponding best linear source head
    \begin{equation*}
        (\alpha_{\so,m}^{\bullet},B_{\so,m}^{\bullet})
        \in
        \argmin_{\alpha\in\mathbb{R}^{T},\,B\in\mathbb{R}^{T\times K}}
        \left\|\alpha+Bh-(\alpha_{\so,m}+B_{\so,m}h_m)\right\|_{P_{\so},2},
    \end{equation*}
    there exists $f\in\mathcal{F}_n^{\sub}$ such that
    \begin{equation*}
        f\circ h
        =
        f_n\left(B_{\so,m}^{+}
        (\alpha_{\so,m}^{\bullet}-\alpha_{\so,m}+B_{\so,m}^{\bullet}h)\right),
        \qquad P_{\ta}\text{-a.s.}
    \end{equation*}
    Assume \cref{asm:support} holds. If $B_{\so,m}$ has full column rank, then $g_m$ is $(\rho_{\so,m},\rho_{\ta,n})$-diverse over $f_n$ with respect to $h_m$ for every $\rho_{\so,m}\geq0$, where
    \begin{equation*}
        \rho_{\ta,n}
        =
        \frac{L_{\mathcal{F}}}
        {\underline{w}_n\sigma_{\min}(B_{\so,m})}\rho_{\so,m}.
    \end{equation*}
\end{corollary}
The corollary implies that, when the source task is linear and the target task is Lipschitz, the $\nu_n$-transferability condition holds if the source task coefficients are informative enough to represent the target task coefficients.
Moreover, the transferability measure $\nu_n = \underline{w}_n\sigma_{\min}(B_{\so,m})/L_{\mathcal{F}}$ is larger if the source task has a larger minimum singular value $\sigma_{\min}(B_{\so,m})$, the target task has a smaller Lipschitz constant $L_{\mathcal{F}}$, and the density ratio of $Z$ in the source and target tasks is bounded away from zero with a larger value of $\underline{w}_n$.
\cref{fig:task_diversity} illustrates the minimum singular value of the source task coefficient matrix $B_{\so,m}$, writing $\beta_t=\beta_{\so,m,t}$ for $t=1,2,3$.

\begin{figure}[H]
	\centering
	\captionsetup[subfigure]{justification=centering}
	\newcommand{\taskdiversitypanel}[2]{%
		\begin{tikzpicture}[
			font=\small,
			axis/.style={-{Latex[length=2.2mm]},draw=black!58,line width=0.65pt},
			sourceVector/.style={-{Latex[length=2.7mm]},draw=myblue,line width=1.45pt},
			coefficientFace/.style={fill=teal,fill opacity=0.10,draw=teal,draw opacity=0.38,line width=0.50pt,line join=round},
			coordinateGrid/.style={draw=black!15,line width=0.25pt},
			gridBoundary/.style={draw=black!18,line width=0.35pt},
			endpoint/.style={circle,inner sep=1.20pt,fill=myblue}
		]
			\path[use as bounding box] (-3.25,-1.95) rectangle (3.25,2.65);
			\begin{scope}[
				shift={(-1.50cm,-0.90cm)},
				x={(2.992829cm,-0.045432cm)},
				y={(-0.207305cm,-0.655900cm)},
				z={(0.000000cm,2.927069cm)}
			]
				\coordinate (O) at (0,0,0);
				\coordinate (B1) at ({cos(15)},{sin(15)},0);
				\coordinate (B2) at ({cos(75)},{sin(75)},0);
				\coordinate (B3) at (#1);

				\foreach \t in {0.25,0.50,0.75} {
					\draw[coordinateGrid] (0,\t,0) -- (1,\t,0);
					\draw[coordinateGrid] (\t,0,0) -- (\t,1,0);
					\draw[coordinateGrid] (0,0,\t) -- (0,1,\t);
					\draw[coordinateGrid] (0,\t,0) -- (0,\t,1);
					\draw[coordinateGrid] (0,0,\t) -- (1,0,\t);
					\draw[coordinateGrid] (\t,0,0) -- (\t,0,1);
				}
				\draw[gridBoundary] (O) -- (1,0,0) -- (1,1,0) -- (0,1,0) -- cycle;
				\draw[gridBoundary] (O) -- (0,1,0) -- (0,1,1) -- (0,0,1) -- cycle;
				\draw[gridBoundary] (O) -- (0,0,1) -- (1,0,1) -- (1,0,0) -- cycle;

				\filldraw[coefficientFace] (O) -- (B1) -- (B2) -- cycle;
				\filldraw[coefficientFace] (O) -- (B1) -- (B3) -- cycle;
				\filldraw[coefficientFace] (O) -- (B2) -- (B3) -- cycle;

				\draw[axis] (O) -- (1.08,0,0) node[below right] {};
				\draw[axis] (O) -- (0,1.08,0) node[above right] {};
				\draw[axis] (O) -- (0,0,1.08) node[above] {};

				\foreach \i in {1,2,3} {
					\draw[sourceVector] (O) -- (B\i);
					\node[endpoint] at (B\i) {};
				}
				\node[text=myblue,below right,xshift=1pt,yshift=-1pt] at (B1) {$\beta_1$};
				\node[text=myblue,below right,xshift=1pt,yshift=-1pt] at (B2) {$\beta_2$};
				\node[text=myblue,#2] at (B3) {$\beta_3$};
			\end{scope}
		\end{tikzpicture}%
	}
	\begin{subfigure}[t]{0.5\linewidth}
		\centering
		\taskdiversitypanel
			{{1/sqrt(6)},{1/sqrt(6)},{2/sqrt(6)}}
			{above left,xshift=-1pt,yshift=1pt}
		\caption[Diversified tasks]{Diversified tasks 
			($\sigma_{\min}(B_{\so,m})\approx 0.70711$)}
		\label{fig:task_diversity_diversified}
	\end{subfigure}%
	\begin{subfigure}[t]{0.5\linewidth}
		\centering
		\taskdiversitypanel
			{{10/sqrt(201)},{10/sqrt(201)},{1/sqrt(201)}}
			{below right,xshift=-5pt,yshift=-3pt}
		\caption[Non-diversified tasks]{Non-diversified tasks 
			($\sigma_{\min}(B_{\so,m})\approx 0.05467$)}
		\label{fig:task_diversity_nondiversified}
	\end{subfigure}
	\caption{Minimum singular values of the source coefficient matrices $B_{\so,m}=(\beta_1,\beta_2,\beta_3)'$.}
	\label{fig:task_diversity}
	\captionsetup{font=footnotesize,justification=justified,singlelinecheck=false}
	\caption*{\textit{Notes:} In both panels, $\beta_1=(\cos 15^{\circ},\sin 15^{\circ},0)'$ and $\beta_2=(\cos 75^{\circ},\sin 75^{\circ},0)'$.
		The third coefficient is $\beta_3=(1,1,2)'/\sqrt{6}$ in panel (a) and $\beta_3=(10,10,1)'/\sqrt{201}$ in panel (b).
		All coefficient vectors have unit norm, and both matrices have full rank.
		The coefficients in panel (b) are nearly linearly dependent.
		Holding $L_{\mathcal{F}}$ and $\underline{w}_n$ fixed, the larger minimum singular value in panel (a) yields a larger transferability measure $\nu_n$.}
\end{figure}

\subsubsection{Necessary Condition}
The following theorems provide necessary condition for $(0,0)$-diversity, which requires $\mathcal{H}_{\so,m}(0)\subseteq\mathcal{H}_{\ta,m}(0)$. We start with the linear case as in \cref{ex:transferability_linear}.
When $B_{\so,m}$ is rank deficient, there is a direction $v(Z)$ in the embedding space that is not captured by any of the source tasks, but it can be relevant for the target task.
Intuitively, if the source tasks are not diverse enough to cover all directions relevant to the target task, $(0,0)$-diversity fails for the uncovered directions of the source task coefficients. See the next theorem for a formal statement.

\begin{theorem}[Necessary Condition for $(0,0)$-Diversity in Linear Models]
    \label{thm:nec_transfer_linear_pre}
    Suppose that the population source and target heads are linear as in \cref{ex:transferability_linear}, so that $g_m\circ h_m=\alpha_{\so,m}+B_{\so,m}h_m$ and $f_n\circ h_m=\alpha_{\ta,n}+\beta_{\ta,n}'h_m$. Assume that $\mathcal{G}_m$ contains every linear source head $v\mapsto\alpha+Bv$, where $\alpha\in\mathbb{R}^{T}$ and $B\in\mathbb{R}^{T\times K}$, and that every $f\in\mathcal{F}_n^{\sub}$ is a linear target head $v\mapsto\alpha+\beta'v$, where $\alpha\in\mathbb{R}$ and $\beta\in\mathbb{R}^{K}$.
    Assume also that $\mathbb{E}_{\ta}[(1, h_m(Z)')' (1, h_m(Z)')]$ is positive definite and that, for every unit vector $v\in\mathbb{R}^{K}$, there exists some $c\in\mathbb{R}^{K}$ such that $(I_K - vv')h_m + c\in\mathcal{H}_m$.
    If $g_m$ is $(0,0)$-diverse over $f_n$ with respect to $h_m$,
    then there exists a vector $b\in\mathbb{R}^{T}$ such that $\beta_{\ta,n}' = b' B_{\so,m}$, and there exists an affine function $\Gamma(x)=\alpha_{\ta,n}+\beta_{\ta,n}'B_{\so,m}^+ (x-\alpha_{\so,m})$ such that $f_n \circ h_m(Z) = \Gamma \circ g_m \circ h_m(Z)$, $P_{\ta}$-a.s.
\end{theorem}

The assumption $(I_K - vv')h_m + c\in\mathcal{H}_m$ is a technical condition that allows us to construct a function $h$ that attains the same source prediction as $h_m$ but a different target prediction from $h_m$.
For some functional classes $\mathcal{H}_m$ such as the class of ReLU-activated hidden layer functions, $\mathcal{H}_m$ only contains nonnegative functions, but it satisfies this assumption by choosing $c$ to be sufficiently large if $\mathcal{H}_m$ is sufficiently rich.

The following theorem provides necessary conditions for $(0,0)$-diversity in a general model with a nonlinear source task under sufficiently rich function classes $\mathcal{H}_m$ and $\mathcal{G}_{m}$ so that the nonlinear analog of the construction of $h$ in the proof of \cref{thm:nec_transfer_linear_pre} is possible.
Note that we do not put any restriction on the function class $\mathcal{F}_n^{\sub}$.

\begin{theorem}[Necessary Condition for $(0,0)$-Diversity in General Models]
	\label{thm:nec_transfer}
	Assume that $f_n\circ h_m\in L^2(P_{\ta})$, where $P_{\ta}$ in this norm denotes the target marginal distribution of $Z$.
	We also assume that there exists $h_0\in\mathcal{H}_{\so,m}(0)$ such that $h_0 \neq h_m$ and $f\circ h_0\in L^2(P_{\ta})$ for every $f\in\mathcal{F}_n^{\sub}$.
	Then, if
	$g_m$ is $(0,0)$-diverse over $f_n$ with respect to $h_m$,
	there exists a measurable function $\Gamma_0:\mathbb{R}^{K}\to\mathbb{R}$ such that $f_n \circ h_m(Z) = \Gamma_0 \circ h_0(Z)$, $P_{\ta}$-a.s.

	In particular, suppose that
	\begin{equation*}
		\sigma(h_0(Z))
		\subseteq
		\overline{\sigma(g_m\circ h_m(Z))}^{P_{\ta}},
	\end{equation*}
	where the right-hand side is the $P_{\ta}$-completion defined in \cref{lem:doob_dynkin_completion}.
	Then there exists a measurable function $\Gamma:\mathbb{R}^{T}\to\mathbb{R}$ such that $f_n \circ h_m(Z) = \Gamma \circ g_m \circ h_m(Z)$, $P_{\ta}$-a.s.
\end{theorem}

The existence of a distinct $h_0$ in \cref{thm:nec_transfer} means that the population source model does not uniquely identify the embedding. Such source nonidentification commonly arises in deep neural networks, where different hidden representations can be paired with compensating source heads. As $(0,0)$-diversity requires $\mathcal{H}_{\so,m}(0)\subseteq\mathcal{H}_{\ta,m}(0)$, the existence of $\Gamma_0$ is needed to adapt the nonidentified embedding function.
The sigma-field inclusion condition is an additional information condition for
the target distribution. For the same $h_0$, it states that, up to $P_{\ta}$-null sets, the population source output $g_m\circ h_m$ contains all the information on $h_0(Z)$. Sufficiently rich classes $\mathcal{G}_m$ and $\mathcal{H}_m$ are needed for the construction of a source-equivalent $h_0$ satisfying this condition.
Under this additional condition, the existence of $\Gamma$ is necessary for $(0,0)$-diversity.

Under the assumptions of \cref{thm:nec_transfer}, including the sigma-field inclusion, measurable factorization $f_n \circ h_m(Z) = \Gamma \circ g_m \circ h_m(Z)$, $P_{\ta}$-a.s., is necessary for $(0,0)$-diversity.
Lipschitz continuity of $\Gamma$ and the target-class representation condition remain additional sufficient restrictions in \cref{thm:suf_transfer}; their necessity is not established by \cref{thm:nec_transfer}.
However, the next example shows that the transferability condition may not hold even if both $f_n$ and $g_m$ are linear but the pre-training model class $\mathcal{G}_m$ is nonlinear.
\begin{example}[Counterexample for Nonlinear Pre-training Model]
    \label{ex:counter}
    Suppose that $h_m(Z)\in\mathbb{R}_+$ and both tasks use the mean squared loss.
	Suppose that $\mathcal{F}_n^{\sub}$ is a set of linear models as in \cref{ex:transferability_linear}, but $\mathcal{G}_{m}$ includes any shallow neural networks with one hidden layer with ReLU activation of width 2.
	Suppose that $Z\sim\textrm{Unif}\ [-1,1]$, $h_m(Z)=|Z|$, $f_n(\cdot) = g_m(\cdot) = (\cdot)$ are identity maps and $T=1$.
	If we choose $h(Z)=Z+1$, then we have
    $$\min_{g\in\mathcal{G}_{m}}\left\|g \circ h(Z) - g_m \circ h_m(Z)\right\|_{P_{\so},2}^2=0,$$
	since $g\in\mathcal{G}_{m}$ can approximate the function $g_m \circ h_m(Z)=|Z|$, for example, by $g(x)=\operatorname{ReLU}(x-1)+\operatorname{ReLU}(-(x-1))$ (see \cref{fig:nonlinear_counterexample}).
    On the other hand,
    $$\min_{f\in\mathcal{F}_n^{\sub}}\left\|f\circ h-f_n \circ h_m\right\|_{P,2}^2=\mathbb{E}[(1/2 - |Z|)^2]=1/12>0,$$ 
    since $f\in\mathcal{F}_n^{\sub}$ is linear and cannot approximate $f_n \circ h_m(Z)=|Z|$ perfectly.
	Thus, $g_m$ is not $(0,0)$-diverse over $f_n$ with respect to $h_m$. Since $h\in\mathcal{H}_{\so,m}(0)$ and its target approximation error is $1/\sqrt{12}$, this example also violates $\nu_n$-transferability at any rate satisfying $\delta_m/\nu_n=o(1)$.
	\begin{figure}[htbp]
	\centering
	\captionsetup[subfigure]{justification=centering}
	\tikzset{
		counterNeuron/.style={circle,draw,line width=0.8pt,minimum size=0.72cm,inner sep=2pt},
		counterEdge/.style={-{Latex[length=1.7mm]},line width=0.8pt},
		counterWeight/.style={font=\footnotesize,text=black,fill=white,inner sep=1pt},
		counterBias/.style={font=\footnotesize,text=black},
		counterBoundary/.style={dashed,draw=black!40,line width=0.6pt}
	}
	\begin{subfigure}[t]{0.5\linewidth}
		\centering
		\begin{tikzpicture}[font=\small]
			\path[use as bounding box] (-0.55,-1.75) rectangle (6.65,2.35);
			\node at (2.05,2.05) {Embedding $h_m$};
			\node at (5.65,2.05) {Head $g_m$};
			\draw[counterBoundary] (4.95,-1.50) -- (4.95,1.65);

			\node[counterNeuron,draw=black!60] (input) at (0,0) {$z$};
			\node[counterNeuron,draw=myblue,fill=myblue!5] (relu-positive) at (1.90,0.85) {$\sigma$};
			\node[counterNeuron,draw=myblue,fill=myblue!5] (relu-negative) at (1.90,-0.85) {$\sigma$};
			\node[counterNeuron,draw=myblue,fill=myblue!5] (embedding) at (3.80,0) {$+$};
			\node[below] at (3.80,-1.05) {$h_m(z)=|z|$};
			\node[counterNeuron,draw=black!65,fill=black!3] (output) at (6.05,0) {};
			\node[below] at (6.05,-0.48) {$|z|$};

			\draw[counterEdge,myblue] (input) -- node[counterWeight,above left] {$1$} (relu-positive);
			\draw[counterEdge,myblue] (input) -- node[counterWeight,below left] {$-1$} (relu-negative);
			\draw[counterEdge,myblue] (relu-positive) -- node[counterWeight,above right] {$1$} (embedding);
			\draw[counterEdge,myblue] (relu-negative) -- node[counterWeight,below right] {$1$} (embedding);
			\draw[counterEdge,black!65] (embedding) -- node[counterWeight,above] {$1$} (output);
		\end{tikzpicture}
		\caption{Reference network}
		\label{fig:nonlinear_counterexample_source}
	\end{subfigure}%
	\begin{subfigure}[t]{0.5\linewidth}
		\centering
		\begin{tikzpicture}[font=\small]
			\path[use as bounding box] (-0.55,-1.75) rectangle (6.65,2.35);
			\node at (1.05,2.05) {Embedding $h$};
			\node at (4.65,2.05) {Head $g$};
			\draw[counterBoundary] (3.10,-1.50) -- (3.10,1.65);

			\node[counterNeuron,draw=black!60] (input) at (0,0) {$z$};
			\node[counterNeuron,draw=myblue,fill=myblue!5] (embedding) at (1.90,0) {$\sigma$};
			\node[counterBias,above] at (1.90,0.48) {$\text{(bias)}=1$};
			\node[below] at (1.90,-1.05) {$h(z)=z+1$};
			\node[counterNeuron,draw=black!65,fill=black!3] (relu-positive) at (3.80,0.85) {$\sigma$};
			\node[counterBias,above] at (4.00,1.28) {$\text{(bias)}=-1$};
			\node[counterNeuron,draw=black!65,fill=black!3] (relu-negative) at (3.80,-0.85) {$\sigma$};
			\node[counterBias,below] at (3.80,-1.28) {$\text{(bias)}=1$};
			\node[counterNeuron,draw=black!65,fill=black!3] (output) at (6.05,0) {$+$};
			\node[below] at (6.05,-0.48) {$|z|$};

			\draw[counterEdge,myblue] (input) -- node[counterWeight,above] {$1$} (embedding);
			\draw[counterEdge,black!65] (embedding) -- node[counterWeight,above left] {$1$} (relu-positive);
			\draw[counterEdge,black!65] (embedding) -- node[counterWeight,below left] {$-1$} (relu-negative);
			\draw[counterEdge,black!65] (relu-positive) -- node[counterWeight,above right] {$1$} (output);
			\draw[counterEdge,black!65] (relu-negative) -- node[counterWeight,below right] {$1$} (output);
		\end{tikzpicture}
		\caption{Equivalent source network}
		\label{fig:nonlinear_counterexample_alternative}
	\end{subfigure}
	\caption[Equivalent Neural Networks with Different Embeddings]{Equivalent Neural Networks with Different Embeddings}
	\label{fig:nonlinear_counterexample}
	\captionsetup{font=footnotesize,justification=justified,singlelinecheck=false}
	\caption*{\textit{Notes:} Both networks satisfy $g_m(h_m(z))=g(h(z))=\sigma(z)+\sigma(-z)=|z|$ for $z\in[-1,1]$. $\sigma(u)=\max\{u,0\}$ is ReLU and $+$ denotes a weighted sum.
	Edge labels give weights and biases are zero unless shown.}
\end{figure}

\end{example}

From this section, we have seen that the transferability condition is substantially stronger for nonlinear source heads $\mathcal{G}_m$ than for linear source heads $\mathcal{G}_m$.
This is because the nonlinear source head $\mathcal{G}_m$ can approximate a much richer class of functions than the linear source head $\mathcal{G}_m$, which makes it more difficult to satisfy the diversity condition.
Therefore, if the economist is indifferent between linear and nonlinear source heads, the linear source head is preferable for guaranteeing the transferability condition.

\section{Convergence Rate of Transfer Learning under Lower-Level Conditions}
\label{sec:transfer_convergence}
\subsection{Low-dimensional Case}
\label{sec:low_dim}
In this section, we provide lower-level primitive conditions for the convergence rates for the transfer-learning problem without the high-level conditions of \cref{asm:convergence}.
We first introduce some notation.
For a function class $\mathcal{F}$, let $\mathsf{V}(\mathcal{F})$ be the VC dimension of the set of subgraphs of functions in $\mathcal{F}$.\footnote{The {\it subgraph} of a function $f: \mathcal{X} \rightarrow \mathbb{R}$ is the subset of $\mathcal{X} \times \mathbb{R}$ given by $\{(x,y): y<f(x)\}$. The VC dimension of a set of subgraphs is defined as the largest integer such that there exists a set of points with cardinality equal to the integer that can be shattered by the subgraphs. See \cite{van2023weak}, Section 2.6 for more details.}
To avoid technical complications concerning the measurability of the supremum, we assume that the functional classes $\mathcal{F}_n$, $\mathcal{F}_n^{\sub}$, and $\mathcal{G}_m\circ\mathcal{H}_m$ are pointwise measurable.\footnote{A measurable functional class $\mathcal{F}$ is {\it pointwise measurable} if it contains a countable subset $\mathcal{F}^0$ such that for every $f \in \mathcal{F}$ there exists a sequence $f_j\in\mathcal{F}^0$ with $f_j(v) \rightarrow f(v)$ for every $v$, where the convergence is with respect to the Euclidean norm. See \cite{van2023weak}, Section 2.3 for more details.}

We start with the convergence rate of the target model given the estimated embedding function $\check{h}$ from the source task.
Before we state the theorem on the convergence rate of the target model $\hat{f}\circ \check{h}$, we assume one more condition on the target loss function.
\begin{assumption}[Conditions on Target Task]
    \label{asm:loss_ta}\ 
    \begin{enumerate}
		\item For every nonnegative sequence $\{\rho_{\so,m}\}_{m=1}^{\infty}$ such that $\rho_{\so,m}\asymp\delta_{\so,m}$, every $f\in\mathcal{F}_n$, and every $h\in\mathcal{H}_{\so,m}(\rho_{\so,m})$, we have
            \begin{equation*}
                \left\|f\circ h - a_{\ta}^*\right\|_{P_{\ta},2}^2 \lesssim \mathbb{E}_{\ta}[\ell_{\ta}(f\circ h(Z),Y)] - \mathbb{E}_{\ta}[\ell_{\ta}(a_{\ta}^*(Z),Y)] \lesssim \left\|f\circ h - a_{\ta}^*\right\|_{P_{\ta},2}^2
            \end{equation*}
            for large enough $m$ and $n$.
        \item The loss function $\ell_{\ta}(\cdot,\cdot)$ is Lipschitz continuous in the first argument with some constant $C_{\ell,\ta}>0$.
    \end{enumerate}
\end{assumption}
This assumption requires that the excess risk of the target loss function is bounded above and below by the squared $L^2$-norm.
This assumption is standard in statistical learning theory and is imposed in many existing works (e.g., \citealp{farrell2021deep}).
\cref{lem:loss_property} verifies the excess-risk comparison and the required Lipschitz comparison over uniformly bounded candidate-score ranges for least squares and binary logistic loss.

\begin{theorem}[Convergence Rate of Target Model]
    \label{thm:target_convergence}
    Let
    \begin{equation}
        r_{\ta,n}=\sqrt{\frac{\mathsf{V}(\mathcal{F}_n)\log(n)}{n}}.
        \label{target_convergence_r}
    \end{equation}
	For every estimated embedding function $\check{h}$ satisfying \cref{asm:convergence} (i), let $\hat{f}\in\mathcal{F}_n$ be an estimated target model satisfying
    \begin{align}
        \frac{1}{n}\sum_{i=1}^{n}\ell_{\ta}(\hat{f}\circ \check{h}(Z_i),Y_i) &\leq \min_{f\in\mathcal{F}_n}\frac{1}{n}\sum_{i=1}^{n}\ell_{\ta}(f\circ \check{h}(Z_i),Y_i)\nonumber\\
        & \quad + O_P(r_{\ta,n}^2 + (\delta_{\so,m}/\nu_n)^2 +\epsilon_{\ta,n}^2).
        \label{eq:target_app_minimizer}
    \end{align}
    Suppose that $\|f\|_{\infty}\leq M_F<\infty$ for every $f\in\mathcal{F}_n$.
    Then, under \cref{asm:convergence} (i), \cref{asm:transferability,asm:support,asm:loss_ta}, we have
    \begin{equation}
        \left\|\hat{f}\circ \check{h}-a_{\ta}^*\right\|_{P_{\ta},2}=O_{P}\left(r_{\ta,n} + \delta_{\so,m}/\nu_n +\epsilon_{\ta,n}\right).
        \label{eq:target_convergence_rate}
    \end{equation}
\end{theorem}
The theorem shows that, under standard conditions on the function class $\mathcal{F}_n$ and the target loss function, the target estimation error term $r_{\ta,n}$ is governed by the VC dimension of $\mathcal{F}_n$ and the target sample size $n$.

\begin{remark}[Fine-tuning]
	Our theory accommodates fine-tuning by allowing $\check{h}$ to differ from the source task estimator. This remains valid provided $\check{h}$ is estimated on a sample independent of the final parameter estimation sample and $\hat{f}\circ\check{h}$ converges to $a_{\ta}^*$ in $L^2(P_{\ta})$ faster than $n^{-1/4}$.
	Thus, the economist can fine-tune the embedding function using the same sample to train $\hat{f}$. Establishing the rate for this joint fit requires a separate argument because \cref{thm:target_convergence} treats the embedding as independent of the sample used to train $\hat{f}$.
	Since fine-tuning does not necessarily improve the accuracy of the target model (\citealp{kumar2022fine}, Table 1), the economist can check by cross-validation whether the loss of $\hat{f}\circ\check{h}$ with fine-tuning is smaller than the loss without fine-tuning. The formal hypothesis testing procedure for this comparison is proposed by \cite{fava2025training}.
\end{remark}

Next, we provide a theorem on the convergence rate of the source model $\check{g}\circ \check{h}$ to the population source model $g_m\circ h_m$.
The rate of convergence for deep learning models is an active area of research and can be improved in future work. Importantly, our transfer learning rate of the target task uses the convergence rate of the source model as an input; thus, any improvement in the convergence rate of the source model can be directly translated into an improved convergence rate of the target model by our theory.

\begin{assumption}[Conditions on Source Task]
    \label{asm:loss_so}\ 
    \begin{enumerate}
        \item For every $g\in\mathcal{G}_m$ and $h\in\mathcal{H}_m$,
            \begin{equation}
                \tau_m^2\left\|g\circ h - a_{\so}^*\right\|_{P_{\so},2}^2 \lesssim \mathbb{E}_{\so}[\ell_{\so}(g\circ h(Z),S)] - \mathbb{E}_{\so}[\ell_{\so}(a_{\so}^*(Z),S)] \lesssim \left\|g\circ h - a_{\so}^*\right\|_{P_{\so},2}^2
            \end{equation}
            with some sequence $\tau_m\in(0,1]$ for large enough $m$ and $n$.
        \item The loss function $\ell_{\so}(\cdot,\cdot)$ is Lipschitz continuous in the first argument with some constant $C_{\ell,\so}>0$.
        \item There exists a sequence $L_{\mathcal{G},m}\geq1$ such that every $g\in\mathcal{G}_m$ is $L_{\mathcal{G},m}$-Lipschitz with respect to the Euclidean norm, uniformly over $\mathcal{G}_m$. That is, for every $v,\widetilde{v}\in\mathcal{V}$,
            \begin{equation*}
                \|g(v)-g(\widetilde{v})\|_2
                \leq L_{\mathcal{G},m}\|v-\widetilde{v}\|_2.
            \end{equation*}
    \end{enumerate}
\end{assumption}
The loss conditions in parts (i)-(ii) are satisfied with $\tau_m = T^{-1}$ if the loss function is the multinomial logistic loss and the Euclidean-norm bound $M$ in \cref{lem:loss_property_multilogit} holds uniformly in $m$ and $T$ over all $g\in\mathcal{G}_m$ and $h\in\mathcal{H}_m$ and for $a_{\so}^*$.
Part (iii) is a regularity condition on the source-head class rather than the loss function. It rules out discontinuous heads and allows the entropy of the composite class $\mathcal{G}_m\circ\mathcal{H}_m$ to be controlled by the coordinatewise entropies of $\mathcal{G}_m$ and $\mathcal{H}_m$.

The next theorem relates the convergence rate of the source model in \cref{asm:convergence} (i) to the VC dimensions of the function classes $\mathcal{H}_m$ and $\mathcal{G}_{m}$.
\begin{theorem}[Convergence Rate of Source Model]
    \label{thm:source_convergence}
    Let
    \begin{equation}
            r_{\so,m}=\sqrt{\frac{\sum_{k=1}^K \mathsf{V}(\mathcal{H}_{m,k}) \log\left(m L_{\mathcal{G},m}K\right)
            +\sum_{t=1}^T\mathsf{V}(\mathcal{G}_{m,t}) \log\left(m T\right)}{\tau_m^2 m}}.
        \label{eq:source_convergence_r}
    \end{equation}
	Let $\check{h}\in\mathcal{H}_m$ and $\check{g}\in\mathcal{G}_m$ be any estimated embedding function and source model estimator satisfying
    \begin{align}
        \frac{1}{m}\sum_{i=1}^m \ell_{\so}(\check{g}\circ \check{h}(Z_i),S_i) 
        & \leq \min_{g\in\mathcal{G}_m, h\in\mathcal{H}_m} \frac{1}{m}\sum_{i=1}^m \ell_{\so}(g\circ h(Z_i),S_i)\nonumber\\
		&\quad + O_{P_{\so}}(r_{\so,m}^2 + \epsilon_{\so,m}^2).
        \label{eq:source_app_minimizer}
    \end{align}
    Suppose that $\sup_{z\in\mathcal{Z}}\|h(z)\|_2\leq M_{H}<\infty$ for every $h\in\mathcal{H}_m$, and $\sup_{v\in\mathcal{V}}\|g(v)\|_2\leq M_{G}<\infty$ for every $g\in\mathcal{G}_m$.
    Then, under \cref{asm:loss_so}, we have
    \begin{align}
        \left\|\check{g}\circ \check{h}-g_m\circ h_m\right\|_{P_{\so},2}
        &=O_{P_{\so}}\left(r_{\so,m}/\tau_m + \epsilon_{\so,m}/\tau_m\right),\\
        \left\|\check{g}\circ \check{h}-a_{\so}^*\right\|_{P_{\so},2}
        &=O_{P_{\so}}\left(r_{\so,m}/\tau_m + \epsilon_{\so,m}/\tau_m\right).\label{eq:source_convergence_rate}
    \end{align}
\end{theorem}

This theorem states that under some standard conditions on the function classes $\mathcal{H}_m$ and $\mathcal{G}_{m}$ and the source loss function, the convergence rate of the source model $\check{g}\circ \check{h}$ to the population model $g_m\circ h_m$ (or the best predictor $a_{\so}^*$) is determined by the coordinatewise VC dimensions of $\mathcal{H}_m$ and $\mathcal{G}_{m}$, the logarithmic factors induced by the embedding dimension $K$, the output dimension $T$, and the head Lipschitz constant $L_{\mathcal{G},m}$, the source sample size $m$, and the curvature term $\tau_m$ of the source loss function. 

\cref{thm:target_convergence,thm:source_convergence} together imply that, under the transferability assumption, the convergence rate of the target model $\hat{f}\circ \check{h}$ to the best predictor $a_{\ta}^*$ is governed mainly by five components: the ratio of the VC dimension of the target model class $\mathcal F_n$ to the target sample size $n$; the ratio of the source complexities for $\mathcal H_m$ and $\mathcal G_m$ to the source sample size $m$; the transferability rate $\nu_n$; the approximation error terms $\epsilon_{\ta,n}$ and $\epsilon_{\so,m}$; and the curvature term $\tau_m$ of the source loss function.
To use the primitive source bound, set $\delta_{\so,m}=(r_{\so,m}+\epsilon_{\so,m})/\tau_m$ and impose the target theorem's assumptions at this rate. Under the assumptions of both theorems, $r_{\ta,n}+(r_{\so,m}+\epsilon_{\so,m})/(\tau_m\nu_n)+\epsilon_{\ta,n}=o(1)$ suffices for target prediction consistency.
For example, $\mathsf{V}(\mathcal{F}_n)=O(K)$ if $\mathcal{F}_n$ is a class of linear models with $K$ regressors, and $\mathsf{V}(\mathcal{F}_n)=O(W_n L_n \log(W_n))$ if $\mathcal{F}_n$ is a class of ReLU DNNs with $W_n$ parameters and depth $L_n$ (\citealp{bartlett2019nearly}, Theorem 7).
In particular, the rate improves when both function class complexities are small relative to sample size and the transferability and approximation error terms are negligible.
If $\mathsf{V}(\mathcal{F}_n)\log(n)=o(\sqrt{n})$ and the transfer error term $\delta_{\so,m}/\nu_n$ and approximation error term $\epsilon_{\ta,n}$ are both $o(n^{-1/4})$, the target model can achieve a convergence rate faster than $n^{-1/4}$.

Under the preceding conditions, our overall target excess risk is $o_P(n^{-1/2})$ by \cref{asm:loss_ta}(i), which gives the $o_P(n^{-1/4})$ prediction rate used in our DML framework.
The general excess-risk bound in Theorem 3 of \cite{tripuraneni2020theory} contains a term of order $n^{-1/2}$  and therefore does not guarantee this prediction rate.
For fixed $T$ and $K$, our overall target excess-risk bound can match the fast VC-based rate of \cite{watkins2023optimistic} up to logarithmic factors (see \cref{sec:diversity_lit}).
Their bound gives this fast rate when the optimal target risk and the minimum average source risk are sufficiently small. Our bound allows both risks to be bounded away from zero.

We can directly apply the above theorems to derive the convergence rates of transfer learning when the target sample size $n$ is large enough compared to its input dimension $K$.
This strategy is taken in, for example, \cite{bajari2025hedonic} and \cite{bach2024adventures}.
Since the VC dimension is well studied for many specific models such as linear models and DNNs (\citealp{bartlett2019nearly}), the above theorems can be applied to derive the convergence rates of transfer learning for these specific models under some conditions on the transferability and approximation errors.
For DNNs, the approximation error can be small if the best predictor is sufficiently smooth and the DNN has enough width and depth (\citealp{yarotsky2017error}).

\subsection{High-dimensional Case}
\label{sec:high_dim}
\subsubsection{Convergence of $h$ for linear source heads $g$}
\label{sec:convergence_h}
For the low-dimensional case, the VC dimension of the function class $\mathcal{F}_n$ is small, and \cref{thm:target_convergence} is directly applicable to derive the convergence rate of the target model $\hat{f}\circ \check{h}$ to the best predictor $a_{\ta}^*$.
An important feature of the VC dimension is that it is a pure functional complexity measure and does not depend on the distribution of the input data.
On the other hand, \cref{thm:target_convergence} is not directly applicable to the high-dimensional case since the VC dimension of $\mathcal{F}_n$ is large, and we need to assume some additional structure on the input data, i.e., embeddings, to derive the convergence rate of the target model.

In this section, we focus on the case in which the source model has a linear head as in \cref{ex:transferability_linear} and $h_m(Z)\in\mathbb{R}^K$ is the last hidden layer of the source model.
In this case, we obtain a convergence rate for the projected estimated embedding toward an affine transformation of a best-approximating embedding function $h_m$.
Recall that the population source model is $g_m \circ h_m(Z) = \alpha_{\so,m} + B_{\so,m} h_m(Z)$ for some vector $\alpha_{\so,m}\in\mathbb{R}^T$, some matrix $B_{\so,m}\in\mathbb{R}^{T\times K}$, and an embedding function $h_m(Z)\in\mathbb{R}^K$.
The pre-trained source model is $\check{g}\circ \check{h}(Z) = \check{\alpha}_{\so} + \check{B}_{\so} \check{h}(Z)$ for some vector $\check{\alpha}_{\so}\in\mathbb{R}^T$, some matrix $\check{B}_{\so}\in\mathbb{R}^{T\times K}$, and an embedding function $\check{h}(Z)\in\mathbb{R}^{K}$ estimated from the source task.
Define a transformed embedding function $\check{h}^{\mathrm{proj}}$ as the projection of $\check{h}$ onto the row space of $\check{B}_{\so}$ as follows:
\begin{equation}
    \label{eq:h_id}
	\check{h}^{\mathrm{proj}} := \check{B}_{\so}^{+} \check{B}_{\so} \check{h}.
\end{equation}
The following lemma gives an $L_2$ convergence rate for the projected embedding $\check{h}^{\mathrm{proj}}$ toward an affine transformation $\check{q}+\check{Q}h_m$ of $h_m$, where $\check{q} := \check{B}_{\so}^{+} (\alpha_{\so,m} - \check{\alpha}_{\so})$ is the shift term and $\check{Q} := \check{B}_{\so}^{+} B_{\so,m}$ is the linear transformation of $h_m$.

\begin{lemma}[Convergence of the projected embedding]
	\label{lem:convergence_h}
	Suppose that the population and pre-trained source heads are linear, as defined in \cref{ex:transferability_linear}, and that \cref{asm:convergence} (i) holds. If, for some deterministic sequence $\underline{\sigma}_m>0$,
	\begin{equation}
		\|\check{B}_{\so}^{+}\|_{\mathrm{sp}}
		=
		O_{P_{\so}}(\underline{\sigma}_m^{-1}),
		\label{eq:convergence_h_condition}
	\end{equation}
	then
	\begin{equation}
		\left\|\check{h}^{\mathrm{proj}}-(\check{q}+\check{Q}h_m)\right\|_{P_{\so},2}
		=
		O_{P_{\so}}(\delta_{\so,m}/\underline{\sigma}_m).
	\end{equation}
	In particular, the condition \cref{eq:convergence_h_condition} holds if \footnote{The event $\{\operatorname{rank}(\check{B}_{\so})\geq1\}$ is only used to ensure that the smallest positive singular value $\sigma_{\min}^+(\check{B}_{\so})$ is well defined.}
	\begin{equation*}
		P_{\so}\left(
		\operatorname{rank}(\check{B}_{\so})\geq1,
		\ \sigma_{\min}^+(\check{B}_{\so})\geq\underline{\sigma}_m
		\right)
		\rightarrow
		1.
	\end{equation*}
\end{lemma}

If $\check{B}_{\so}$ has full column rank, $\check{B}_{\so}^{+} = (\check{B}_{\so}'\check{B}_{\so})^{-1}\check{B}_{\so}'$ is the left inverse of $\check{B}_{\so}$ and $\check{B}_{\so}^{+}\check{B}_{\so} = I_K$ is the identity matrix; thus, $\check{h}^{\mathrm{proj}} = \check{h}$.
On the other hand, we have $\check{B}_{\so}^{+}\check{B}_{\so} \neq I_K$ in general if $\check{B}_{\so}$ does not have full column rank.

We can interpret $\check{h}^{\mathrm{proj}}$ as the projection of $\check{h}$ onto the row space of $\check{B}_{\so}$ since $\check{B}_{\so}^+\check{B}_{\so} = P_{\check{B}_{\so}'} := \check{B}_{\so}'(\check{B}_{\so}\check{B}_{\so}')^+\check{B}_{\so}$ by \cite{harville2008matrix}, Corollary 20.3.8. 
Note that $\check{h}^{\mathrm{proj}}$ and $\check{h}$ return the same output for the source task since $\check{B}_{\so}\check{h}^{\mathrm{proj}}(z) = \check{B}_{\so}\check{h}(z)$. The projection preserves the source output while removing components of the estimated embedding in the null space of $\check{B}_{\so}$.\footnote{An indeterminacy can remain. For every constant vector $c\in\mathbb{R}^{K}$, the pairs $(\check{\alpha}_{\so},\check{h}^{\mathrm{proj}})$ and $(\check{\alpha}_{\so}+\check{B}_{\so}c,\check{h}^{\mathrm{proj}}-c)$ produce the same source output algebraically. Target predictions are preserved if the target class contains $v\mapsto f(v+c)$ for each relevant head $f$ and shift $c$. For an affine head $f(v)=\alpha+\beta'v$, this requires an intercept whose permitted range includes $\alpha+\beta'c$. Applying the target convergence results also requires the translated heads to satisfy their parameter and regularity bounds.}

\cref{lem:convergence_h} controls the distance between the projected embedding $\check{h}^{\mathrm{proj}}$ and $\check{q}+\check{Q}h_m$.
The matrix $\check{Q}$ describes the linear part of this relation for the component of $\check{h}$ in the row space of $\check{B}_{\so}$.
In the special case in which $\check{\alpha}_{\so}=\alpha_{\so,m}$, $\check{B}_{\so} = B_{\so,m}A^{-1}$ and $\check{h} = A h_m$ for an invertible matrix $A\in\mathbb{R}^{K\times K}$, and $B_{\so,m}$ has full column rank, we have $\check{B}_{\so}\check{h} = B_{\so,m}A^{-1}A h_m = B_{\so,m}h_m$ and the exact equality $\check{h}^{\mathrm{proj}} = \check{Q} h_m$, where $\check{Q} = \check{B}_{\so}^{+} B_{\so,m} = A$. Since this holds for every invertible matrix, the embedding function $h_m$ is identified from the source task only up to an invertible linear transformation $A$ such that $A h_m\in\mathcal{H}_m$.
This special case is consistent with the discussion in \cite{schulte2025adjustment} that the embedding function is identified only up to an invertible linear transformation under the restrictive assumption that there is no estimation error for the embedding function.
More generally, when $B_{\so,m}$ or $\check{B}_{\so}$ is rank deficient, this affine relation applies to projected embeddings. Components in the null space of the corresponding source-head matrix remain unidentified by the source output, and source-equivalent embeddings need not be affine transformations of one another.

\subsubsection{LASSO Target Model}
\label{sec:lasso}
We do not recommend using LASSO for the target model for the following reason.
As we have seen in \cref{sec:convergence_h}, $\mathcal{H}_{\so,m}(0)$ includes invertible linear transformations of a best-approximating embedding function $h_m$.
\cref{thm:conv_nec} suggests that, to ensure $\|\hat{f}_{n,h_m^{\dagger}}\circ h_m^{\dagger} - a^*_{\ta}\|_{P_{\ta},2} = o_{P_{\ta}}(1)$ for any sequence $h_m^{\dagger}\in\mathcal{H}_{\so,m}(0)$, it is necessary that $\mathcal{H}_{\so,m}(0)\subseteq \mathcal{H}_{\ta,m}(\rho_{\ta,n})$ for some sequence $\rho_{\ta,n}=o(1)$ when the target model has small estimation and approximation errors.
That is, for every invertible matrix $Q\in\mathbb{R}^{K\times K}$ and every vector $q\in\mathbb{R}^{K}$ such that $q + Q h_m \in \mathcal{H}_m$, there must exist some $f\in\mathcal{F}_n^{\sub}$ such that $\|f_n\circ h_m - f\circ (q + Q h_m)\|_{P_{\ta},2}$ converges to zero as $n\rightarrow\infty$.
Note that linear regression, ridge regression, and DNNs satisfy this condition since they transform the input vector by a linear transformation.

Therefore, our theory recommends avoiding methods that rely on assumptions, such as sparsity, that are not preserved under the unidentified invertible linear transformations of the embedding function $h_m$ (e.g., \citealp{wainwright2019high}, Chapter 7).
Suppose the economist uses LASSO for the target model given the estimated embedding function $\check{h}$ from the source task.
To use LASSO, the existing theory typically requires an approximate-sparsity assumption.
On the other hand, since $\mathcal{H}_{\so,m}(0)$ includes invertible linear transformations of a best-approximating embedding function $h_m$, Theorem 1 in \cite{kolesar2026fragility} implies that the approximate sparsity assumption is violated in general.
Thus, even if the economist can justify the approximate sparsity assumption for an embedding function $h_m$, the approximate sparsity assumption is hard to justify for other $h\in\mathcal{H}_{\so,m}(0)$.
Since we cannot identify the embedding function $h_m$ from the source task, any method relying on variable selection can be fragile.

\subsubsection{Ridge Target Model}
\label{sec:ridge_target}
\cref{sec:lasso} shows that approximate sparsity is not preserved under the linear reparameterizations left unidentified by the source task. We therefore study an $\ell_2$-penalized linear target model as a tractable alternative.

The scale of a linear source head is not identified separately from the scale of its embedding. Thus, the condition in \cref{eq:convergence_h_condition} is not invariant to observationally equivalent rescalings. Multiplying the embedding by a large constant and dividing the source-head matrix by the same constant can make its smallest positive singular value arbitrarily small.
This scale indeterminacy is not a problem for the low-dimensional case since the VC dimension is invariant to scaling. However, the scale indeterminacy is typically a problem for the high-dimensional case since we use the structure on the embeddings and the source head to derive the convergence rate of the target model.
The ridge theory below therefore imposes fixed singular value bounds on the source-head matrices used in the analysis.

Fix constants $C_{\alpha,\so}<\infty$ and $0<\underline{\sigma}\leq\bar{\sigma}<\infty$, and let $\Theta_{\so,m}$ be a deterministic, nonempty, compact subset of
\begin{equation*}
	\left\{
		(\alpha,B)\in\mathbb{R}^{T}\times\mathbb{R}^{T\times K}:
		\|\alpha\|_2\leq C_{\alpha,\so},
		\operatorname{rank}(B)\geq1,
		\|B\|_{\mathrm{sp}}\leq\bar{\sigma},
		\sigma_{\min}^+(B)\geq\underline{\sigma}
	\right\}.
\end{equation*}
We assume that the population source-head parameters $(\alpha_{\so,m},B_{\so,m})$ and the pre-trained source-head parameters $(\check{\alpha}_{\so},\check{B}_{\so})$ belong to $\Theta_{\so,m}$ with probability approaching one.
The transformation in \cref{sec:source_head_transformation} normalizes the positive singular values of a nonzero source-head matrix if one wants to ensure the conditions on the singular value bounds.

For each $(\alpha,B)\in\Theta_{\so,m}$, define
\begin{equation}
	\label{eq:normalized_linear_source_model}
	g_{\alpha,B}:v\mapsto\alpha+Bv.
\end{equation}
For this subsection, let $\mathcal{G}_m := \{g_{\alpha,B}:(\alpha,B)\in\mathbb{R}^{T}\times\mathbb{R}^{T\times K}\}$ be the class of linear source heads, so that $\{g_{\alpha,B}:(\alpha,B)\in\Theta_{\so,m}\}\subseteq\mathcal{G}_m$.
Thus, $\mathcal{G}_m$ may also contain linear heads with parameters outside $\Theta_{\so,m}$. For every $(\alpha,B)\in\Theta_{\so,m}$, \cref{lem:spectral_inverse} gives $\|B^+\|_{\mathrm{sp}}\leq\underline{\sigma}^{-1}$. 
Moreover, on the event that the chosen population and pre-trained source-head parameters belong to $\Theta_{\so,m}$, we have $\|\check{Q}\|_{\mathrm{sp}} = \|\check{B}_{\so}^+B_{\so,m}\|_{\mathrm{sp}}\leq\bar{\sigma}/\underline{\sigma}$.

Define the projected-embedding class used in the ridge analysis by
\begin{equation}
	\label{eq:H_m_proj}
	\mathcal{H}_m^{\mathrm{proj}}
	:=
	\left\{
	B^+Bh:
	h\in\mathcal{H}_m,
	(\alpha,B)\in\Theta_{\so,m}
		\text{ for some }\alpha\in\mathbb{R}^{T}
	\right\}.
\end{equation}
For every $\rho\geq0$, let $\mathcal{H}_{\so,m}^{\mathrm{proj}}(\rho):=\{h\in\mathcal{H}_m^{\mathrm{proj}}: \min_{g\in\mathcal{G}_{m}}\|g\circ h-g_m\circ h_m\|_{P_{\so},2}\leq \rho\}$.
Fix constants $C_{\alpha},C_{\beta}<\infty$, and let
\begin{equation*}
	\Theta_n
	:=
	\left\{
		(\alpha,\beta)\in\mathbb{R}\times\mathbb{R}^{K}:
		|\alpha|\leq C_{\alpha},
		\|\beta\|_2\leq C_{\beta}
	\right\}.
\end{equation*}

Let $\mathcal{F}_n = \mathcal{F}_n^{\sub} = \{f(v) = \alpha+\beta'v: (\alpha,\beta)\in\Theta_n\}$ and let $f_n(v) = \alpha_{\ta,n}+\beta_{\ta,n}'v\in\mathcal{F}_n^{\sub}(h_m)$ be the population target model.
For a deterministic sequence $\lambda_n>0$, define the ridge estimator on the projected embedding by
\begin{equation}
	\label{eq:ridge_estimator}
	\hat{f}_n^{\mathrm{ridge}}(v):=\hat{\alpha}_n^{\mathrm{ridge}}+\hat{\beta}_n^{\mathrm{ridge}'}v
	\in
	\argmin_{\alpha+\beta'v\in\mathcal{F}_n}
	\left\{
		\frac{1}{n}\sum_{i=1}^{n}\ell_{\ta}\left(\alpha + \beta'\check{h}^{\mathrm{proj}}(Z_i),Y_i\right)
		+
		\lambda_n^2\|\beta\|_2^2
	\right\}.
\end{equation}
We assume the following conditions on the target model and the source parameters.

\begin{assumption}[Conditions for Ridge Regression]
	\label{asm:ridge_target}
	\
	\begin{enumerate}[label=(\roman*)]
		\item (Source Parameter and Embedding Restrictions)
		The population source parameters satisfy $(\alpha_{\so,m},B_{\so,m})\in\Theta_{\so,m}$ and $h_m\in\mathcal{H}_m$.
		The pre-trained source parameters satisfy $P_{\so}((\check{\alpha}_{\so},\check{B}_{\so})\in\Theta_{\so,m}, \check{h}\in\mathcal{H}_m) \rightarrow 1$.
		Consequently, $\check{h}^{\mathrm{proj}}\in\mathcal{H}_m^{\mathrm{proj}}$ with probability approaching one.
		For every $m$, there exists a deterministic constant $C_m<\infty$, possibly diverging with $n$, such that
		$\max\{\mathbb{E}_{\so}[\|h_m(Z)\|_2^2],\sup_{h\in\mathcal{H}_m^{\mathrm{proj}}}\mathbb{E}_{\so}[\|h(Z)\|_2^2]\} \leq C_m$.

		\item (Target-Slope Representation for Transferability)
		There exist a sequence $b_n\in\mathbb{R}^{T}$ and a constant $C_b<\infty$ such that $\sup_n\|b_n\|_2\leq C_b$ and $\beta_{\ta,n}'=b_n'B_{\so,m}$, where $B_{\so,m}$ denotes the population source head.

		\item (Target Parameter Restrictions)
		The target parameter bounds satisfy $\sup_n|\alpha_{\ta,n}|+2C_bC_{\alpha,\so}\leq C_{\alpha}$ and $C_b\bar{\sigma} \leq C_{\beta}$.
	\end{enumerate}
\end{assumption}

For any embedding $h$, let
\begin{equation*}
	\Sigma_{h,n}
	:=
	\mathbb{E}_{\ta}\left[(h(Z)-\mathbb{E}_{\ta}[h(Z)])(h(Z)-\mathbb{E}_{\ta}[h(Z)])'\right].
\end{equation*}
In particular, $\Sigma_{h_m,n}$ denotes the covariance matrix of the population embedding.
Define the effective degrees of freedom of the ridge regression based on the embedding $h$ by
\begin{equation}
	\label{eq:effective_degrees_of_freedom}
	\mathcal{N}_{h,n}(\lambda)
	:=
	\operatorname{tr}\left(\Sigma_{h,n}(\Sigma_{h,n}+\lambda^2I_K)^{-1}\right)
	=
	\sum_{j=1}^K\frac{\mu_j(\Sigma_{h,n})}{\mu_j(\Sigma_{h,n})+\lambda^2},
\end{equation}
where $\mu_1(\Sigma_{h,n})\geq\cdots\geq\mu_K(\Sigma_{h,n})\geq0$ are the eigenvalues of $\Sigma_{h,n}$.

\begin{theorem}[Convergence Rate of Ridge Target Model]
	\label{thm:ridge_target_convergence}
	Suppose that \cref{asm:convergence} (i) and \cref{asm:support,asm:ridge_target} hold. We also assume \cref{asm:loss_ta} holds with $\mathcal{H}_{\so,m}(\rho_{\so,m})$ replaced by $\mathcal{H}_{\so,m}^{\mathrm{proj}}(\rho_{\so,m})$. Let $\lambda_n>0$ be deterministic and let $\hat{f}_n^{\mathrm{ridge}}\circ\check{h}^{\mathrm{proj}}$ be defined in \cref{eq:ridge_estimator}. Then,
	\begin{align}
		&\left\|\hat{f}_n^{\mathrm{ridge}}\circ\check{h}^{\mathrm{proj}}-a_{\ta}^*\right\|_{P_{\ta},2}\nonumber\\
		=&
		O_P\left(\sqrt{\frac{1+\mathcal{N}_{h_m,n}(\lambda_n)}{n}}+\lambda_n
		+\frac{\delta_{\so,m}/\underline{w}_n+(\delta_{\so,m}/\underline{w}_n)^{1/2}\operatorname{tr}(\Sigma_{h_m,n})^{1/4}}{\sqrt{n}\lambda_n}
		+\frac{\delta_{\so,m}}{\underline{w}_n}+\epsilon_{\ta,n}
		\right).
	\end{align}
\end{theorem}

The first two terms in the above rate are the estimation errors of ridge regression, and the last two terms are the transferability and approximation errors.
The middle term arises from the difference between $\mathcal{N}_{\check{h}^{\mathrm{proj}},n}(\lambda_n)$ and $\mathcal{N}_{h_m,n}(\lambda_n)$.

For a concrete rate, suppose that, uniformly over $n$ and $j=1,\ldots,K$,
\begin{equation*}
	\mu_j(\Sigma_{h_m,n})
	\leq
	C_{\mu}j^{-2a}
	\qquad
	\text{for constants }C_{\mu}<\infty\text{ and }a>1/2,
\end{equation*}
and $\delta_{\so,m}/\underline{w}_n+\epsilon_{\ta,n}=O(n^{-a/(2a+1)})$. Then $\mathcal{N}_{h_m,n}(\lambda)\lesssim\lambda^{-1/a}$. Choosing $\lambda_n\asymp n^{-a/(2a+1)}$ gives
\begin{equation*}
	\left\|\hat{f}_n^{\mathrm{ridge}}\circ\check{h}^{\mathrm{proj}}-a_{\ta}^*\right\|_{P_{\ta},2}
	=
	O_P\left(
	\max\left\{n^{-a/(2a+1)},n^{-(a+1)/(2(2a+1))}\right\}
	\right).
\end{equation*}
If $a>1/2$, the above rate is $o_P(n^{-1/4})$, which is sufficient for the DML framework.
Thus, ridge is a recommended regularized alternative when the target model is linear and the embedding is nearly identified up to a well-conditioned linear distortion.

In practice, the tuning parameter $\lambda_n$ for ridge is typically selected by $K$-fold cross-validation that directly targets prediction risk; see, for example, \citet{hastie2022surprises} for a recent detailed discussion of cross-validation for ridge regression, including in over-parameterized settings.

A natural question is whether a similar convergence rate can be achieved by neural network models with regularization such as random dropout or $\ell_2$ regularization.\footnote{
	Another line of research in deep learning theory on convergence rates makes use of functional shape assumptions on the function of interest (\citealp{kohler2021rate,bhattacharya2024deep}) or data structure assumptions such as approximate manifold structure (\citealp{jiao2023deep,schulte2025adjustment}). For a VC-based estimation term of the form $\sqrt{\mathsf{V}(\mathcal{F}_n)\log(n)/n}$, raw input dimension may grow provided the resulting complexity grows sufficiently slowly: $\mathsf{V}(\mathcal{F}_n)\log(n)=o(n)$ makes this term vanish, and $\mathsf{V}(\mathcal{F}_n)\log(n)=o(n^{1/2})$ makes it $o(n^{-1/4})$. The other approximation and transfer terms must satisfy their corresponding rate requirements.
}
We conjecture that the answer is positive.
For linear regression, dropout can be interpreted as a form of ridge regularization (\citealp{srivastava2014dropout}, Section 9.1).
\cite{wei2020implicit} demonstrate both theoretically and empirically that dropout can work as regularization for deep neural networks.
However, even when $h_m$ is directly observed, the convergence rate of the regularized neural network estimator is not well understood in the literature at the level of generality of \cite{farrell2021deep}, to the best of our knowledge. Establishing the convergence rate of regularized neural network estimators for the target model in our setting is an important direction for future research.

\section{Test for Transferability}
\label{sec:test_transferability}

As discussed in \cref{rem:transferability} and illustrated in \cref{fig:transferability}, $\nu_n$-transferability enables accurate estimation of target nuisance functions using pre-trained embeddings even when the source task only partially identifies the embedding.
In this section, we develop a formal test using an implication of transferability: under the rate and regularity conditions below, DML estimators based on different embeddings in the source near-identified set should yield similar estimates of the target parameter. We propose a computationally feasible procedure based on the multiplier bootstrap and show its uniform validity under the null hypothesis and its uniform consistency under well-separated alternatives.

The test uses several embeddings pre-trained for the same source task, for example by training the same model architecture with different random seeds in optimization.
We split the target sample into three parts, using one part to fit the nuisance functions for each embedding and the other two to estimate the parameter of interest.
We compare estimators across embeddings and estimation samples to assess their sensitivity to the choice of pre-trained embedding.
Using separate estimation samples prevents the leading sampling randomness from canceling in these comparisons.
We reject the transferability null when the maximum standardized difference exceeds a critical value computed by a multiplier bootstrap, without refitting the nuisance functions or embeddings.
Under the regularity conditions below, rejection provides evidence against transferability.

\subsection{Multiplier Bootstrap}
Consider the setup in \cref{sec:dml}, with a scalar parameter of interest $\theta^*\in\Theta\subset\mathbb{R}$ and a scalar Neyman orthogonal score $\psi(W;\theta,\gamma,\alpha)$ satisfying $\mathbb{E}_{\ta}[\psi(W;\theta^*,\gamma^*,\alpha^*)]=0$.
The nuisance functions $\gamma^*$ and $\alpha^*$ may each have a fixed number of components. In that case, the transferability and rate conditions below apply to each component. While we focus on scalar nuisance functions and scalar $\theta$ for simplicity, the results extend to vector-valued nuisance functions by replacing the scalar norm with an appropriate vector norm.

Suppose that $R\geq2$ pre-trained models $(\check{g}_j,\check{h}_j)\in\mathcal{G}_m\times\mathcal{H}_m$, $j=1,\ldots,R$, are available for the same source task. The number of models $R$ is deterministic and may diverge as $n\to\infty$. The models may share a source sample and differ in their pre-training randomness. Their collection is independent of the target sample.
Such pre-trained variants are available for both text and image embeddings.\footnote{Examples include BERT, GPT-2, and Pythia for text, and ResNet-50, Swin-Tiny, and ViT for images.}

Fix the source distribution $P_{\so}$, function classes $\mathcal{G}_m$ and $\mathcal{H}_m$, and sequences $(g_m,h_m)\in\mathcal{R}_m$ and $\delta_{\so,m}$. For each $\diamond\in\{\gamma,\alpha\}$ and each target distribution $P_{\ta}$, let $a_{\diamond}^*$ denote the best target predictor on the scale of the training loss, so that $\diamond^*=\Lambda_{\diamond}\circ a_{\diamond}^*$ with the prediction maps defined in \cref{sec:dml}. Let $\mathcal{F}_{\diamond,n}^{\sub}$ denote the corresponding target head class and let $f_{\diamond,n}\in\mathcal{F}_{\diamond,n}^{\sub}(h_m)$ be a population target head.
For positive sequences of transferability strengths $\nu_{\diamond,n}$, $\diamond\in\{\gamma,\alpha\}$, the null requires that $g_m$ be $\nu_{\diamond,n}$-transferable to $f_{\diamond,n}$ with respect to $h_m$ at rate $\delta_{\so,m}$, in the sense of \cref{def:transferability}. Equivalently,
\begin{equation}
	\label{eq:transferability_test_null_nu}
	H_{0,n}:\quad
	\sup_{h\in\mathcal{H}_{\so,m}(\rho_{\so,m})}
	\min_{f\in\mathcal{F}_{\diamond,n}^{\sub}}
	\left\|f\circ h-f_{\diamond,n}\circ h_m\right\|_{P_{\ta},2}
	\leq \rho_{\so,m}/\nu_{\diamond,n},
\end{equation}
for each $\diamond\in\{\gamma,\alpha\}$ and every nonnegative sequence $\rho_{\so,m}\asymp\delta_{\so,m}$, for all sufficiently large $n$.
Here, $\mathcal{H}_{\so,m}(\rho_{\so,m})$ is the source task near-identified set defined in \cref{sec:diversity}.
While the null restricts all embeddings in this set, the observed pre-trained embeddings provide a way to test an implication of this restriction.

For each $n$, let $\mathcal{P}_{0,\ta,n}$ be a class of target distributions satisfying the null hypothesis $H_{0,n}$ for the fixed source problem $P_{\so}$. Define the corresponding class of joint source and target laws by $\mathcal{P}_{0,n}=\{P_{\so}\times P_{\ta}:P_{\ta}\in\mathcal{P}_{0,\ta,n}\}$.
Target population objects and deterministic target rates may depend on $P_{\ta}$, but this dependence is suppressed in the notation for simplicity.

In our test, we partition the target sample independently of the data into a nuisance-training sample $\mathcal{I}_3$ and two disjoint estimation samples $\mathcal{I}_1$ and $\mathcal{I}_2$.
Let $n_s=|\mathcal{I}_s|$ for $s\in\{1,2,3\}$, with all three sample proportions converging to positive constants.
For each $\diamond\in\{\gamma,\alpha\}$ and $j=1,\ldots,R$, fit the target head $\hat{f}_{\diamond,j}$ using $\mathcal{I}_3$ and embedding $\check{h}_j$. Write $\hat{\gamma}_j=\Lambda_{\gamma}\circ\hat{f}_{\gamma,j}\circ\check{h}_j$ and $\hat{\alpha}_j=\Lambda_{\alpha}\circ\hat{f}_{\alpha,j}\circ\check{h}_j$.
For $j=1,\ldots,R$ and $s\in\{1,2\}$, define the sample moment by $\hat{M}_{j,s}(\theta)=n_s^{-1}\sum_{i\in\mathcal{I}_s}\psi(W_i;\theta,\hat{\gamma}_j,\hat{\alpha}_j)$ and let $\hat{\theta}_{j,s}$ solve $\hat{M}_{j,s}(\hat{\theta}_{j,s})=0$.

\cref{alg:bootstrap_transferability} constructs the test statistic as the largest absolute standardized difference between estimators based on different embeddings and the two estimation samples $\mathcal{I}_1$ and $\mathcal{I}_2$.
Let the population score derivative be
\begin{equation*}
	J = \left.\frac{\partial}{\partial\theta}\mathbb{E}_{\ta}[\psi(W;\theta,\gamma^*,\alpha^*)]\right|_{\theta=\theta^*}.
\end{equation*}
Define the population influence function and its variance by $\phi(W)=-J^{-1}\psi(W;\theta^*,\gamma^*,\alpha^*)$ and $V=\mathbb{E}_{\ta}[\phi(W)^2]$.
Under $H_{0,n}$ in \cref{eq:transferability_test_null_nu} and \cref{asm:transferability_test} below, the estimators have the expansion
\begin{equation}
	\label{eq:transferability_test_expansion}
	\hat{\theta}_{j,s}-\theta^*
	=\frac{1}{n_s}\sum_{i\in\mathcal{I}_s}\phi(W_i)+o_P(n^{-1/2}),
\end{equation}
uniformly over $j=1,\ldots,R$, $s\in\{1,2\}$, and $P\in\mathcal{P}_{0,n}$.
The centering at $\theta^*$ and the influence function do not depend on the embedding. The leading terms in \cref{eq:transferability_test_expansion} would therefore cancel if two estimators used the same estimation sample. Using separate estimation samples $\mathcal{I}_1$ and $\mathcal{I}_2$ avoids this cancellation and gives the leading difference positive variance $V(n_1^{-1}+n_2^{-1})$ conditional on the partition.

For each bootstrap draw, we approximate the sampling variation in these leading terms using weighted sums of the estimated influence functions, with independent standard normal multipliers across observations.
We use the same multiplier for an observation across all embeddings to preserve dependence across comparisons.
The embeddings, fitted nuisance functions, and sample partition remain fixed throughout the bootstrap; only the multipliers vary. Thus, the bootstrap requires no re-estimation of the nuisance functions or embeddings.
\begingroup
\RestyleAlgo{ruled}
\begin{algorithm}[htbp]
	\small
	\caption{Multiplier bootstrap test of transferability for a DML parameter}
	\label{alg:bootstrap_transferability}
	\begin{minipage}{\hsize}
		\KwIn{An i.i.d. target sample $\{W_i\}_{i=1}^n$; Neyman orthogonal score $\psi(W;\theta,\gamma,\alpha)$; pre-trained embeddings $\check{h}_1,\ldots,\check{h}_R$; significance level $\textsf{a}\in(0,1)$; number of bootstrap draws $B$.}
		\begin{enumerate}[label=\arabic*.,leftmargin=*]
			\item \textbf{Split the target sample.} Randomly select a nuisance-training sample $\mathcal{I}_3$, then partition the remaining observations into $\mathcal{I}_1$ and $\mathcal{I}_2$. Draw the partition independently of the observed data and use it for every embedding. Let $n_s=|\mathcal{I}_s|$ for $s\in\{1,2\}$, with $n_1,n_2\geq2$.
			\item \textbf{Estimate nuisance functions on $\mathcal{I}_3$.} For each $j$, use only $\mathcal{I}_3$ to fit the target heads, including preprocessing and tuning. Set $\hat{\gamma}_j=\Lambda_{\gamma}\circ\hat{f}_{\gamma,j}\circ\check{h}_j$ and $\hat{\alpha}_j=\Lambda_{\alpha}\circ\hat{f}_{\alpha,j}\circ\check{h}_j$, converting fitted log odds to probabilities before evaluating the DML score.
			\item \textbf{Compute estimators and influence functions on $\mathcal{I}_1$ and $\mathcal{I}_2$.} For each $j$ and $s\in\{1,2\}$, solve the estimating equation with respect to $\theta$:
			\begin{equation*}
				\frac{1}{n_s}\sum_{i\in\mathcal{I}_s}\psi(W_i;\theta,\hat{\gamma}_j,\hat{\alpha}_j)=0,
			\end{equation*}
			and denote the solution by $\hat{\theta}_{j,s}$.
			Compute the estimator of the derivative of the score with respect to $\theta$:
			\begin{equation}
				\label{eq:transferability_test_derivative_estimate}
				\hat{J}_{j,s}=\frac{1}{n_s}\sum_{i\in\mathcal{I}_s}\partial_{\theta}\psi(W_i;\hat{\theta}_{j,s},\hat{\gamma}_j,\hat{\alpha}_j).
			\end{equation}
			Require $\hat{J}_{j,s}\neq0$ for every $j,s$; otherwise, stop without returning a test decision. Compute
			\begin{equation}
				\label{eq:transferability_test_influence_weights}
				\hat{\phi}_{ij,s} = -\hat{J}_{j,s}^{-1}\psi(W_i;\hat{\theta}_{j,s},\hat{\gamma}_j,\hat{\alpha}_j),
				\quad
				\hat{q}_{ij,s}=\frac{\hat{\phi}_{ij,s}}{\sqrt{n_s(n_s-1)}}.
			\end{equation}
			Here, $\sum_{i\in\mathcal{I}_s}\hat{q}_{ij,s}^2$ estimates the variance of $\hat{\theta}_{j,s}$.
			\item \textbf{Form the test statistic.} For every $1\leq j<k\leq R$, calculate
			\begin{equation}
				\label{eq:transferability_test_standard_error}
				\hat{\sigma}_{jk}^2 = \sum_{i\in\mathcal{I}_1}\hat{q}_{ij,1}^2+\sum_{i\in\mathcal{I}_2}\hat{q}_{ik,2}^2.
			\end{equation}
			Require every $\hat{\sigma}_{jk}>0$; otherwise, stop without returning a test decision. Then compute
			\begin{equation}
				\label{eq:transferability_test_statistic}
				\mathcal{T} = \max_{1\leq j<k\leq R}\frac{|\hat{\theta}_{j,1}-\hat{\theta}_{k,2}|}{\hat{\sigma}_{jk}}.
			\end{equation}
			\item \textbf{Generate the bootstrap distribution.} For each $b=1,\ldots,B$, draw independent $\xi_i^{(b)}\sim N(0,1)$ for $i\in\mathcal{I}_1\cup\mathcal{I}_2$, independently across draws and of the data. Use the same multiplier for observation $i$ across all embeddings and compute bootstrap test statistics
			\begin{equation}
				\label{eq:transferability_test_bootstrap_statistic}
				\mathcal{T}_b^{(\xi)}=\max_{1\leq j<k\leq R}
				\frac{\left|\sum_{i\in\mathcal{I}_1}\xi_i^{(b)}\hat{q}_{ij,1}-\sum_{i\in\mathcal{I}_2}\xi_i^{(b)}\hat{q}_{ik,2}\right|}{\hat{\sigma}_{jk}}.
			\end{equation}
			\item \textbf{Return the decision and $p$-value.} Let $c_{1-\textsf{a}}^{(\xi)}$ be the $(1-\textsf{a})$-quantile of $\{\mathcal{T}_b^{(\xi)}\}_{b=1}^B$, using the inverse empirical distribution function. Reject the transferability null $H_{0,n}$ in \cref{eq:transferability_test_null_nu} if $\mathcal{T}>c_{1-\textsf{a}}^{(\xi)}$. Report the bootstrap $p$-value $B^{-1}\sum_{b=1}^B\mathds{1}\{\mathcal{T}_b^{(\xi)}\geq\mathcal{T}\}$.
		\end{enumerate}
	\end{minipage}
\end{algorithm}
\endgroup

\subsection{Uniform Size Validity of the Test}
We now state sufficient conditions for validity of this procedure uniformly over $P\in\mathcal{P}_{0,n}$.
For $P=P_{\so}\times P_{\ta}\in\mathcal{P}_{0,n}$, let $\Pr_P$ denote probability under this law, including the randomness of both samples, the partition, pre-training randomness, and bootstrap draws when applicable.
For the asymptotic statements below, set $\hat{\phi}_{ij,s}=\hat{q}_{ij,s}=0$ whenever $\hat{J}_{j,s}=0$, and define $\hat{\sigma}_{jk}$ by \cref{eq:transferability_test_standard_error} using these extended weights. The algorithm still stops without returning a test decision on these events.

For each $\diamond\in\{\gamma,\alpha\}$, let $r_{\diamond,n}$ be a common rate for estimating the target head across embeddings, and let $\epsilon_{\diamond,n}=\|f_{\diamond,n}\circ h_m-a_{\diamond}^*\|_{P_{\ta},2}$ denote the approximation error of the population target model on the training scale. Since the Lipschitz constants of the prediction maps $\Lambda_{\diamond}$ are uniformly bounded, these rates also bound the corresponding nuisance errors up to a fixed multiplicative constant.\footnote{See \cref{eq:logit_probability_rate} for the binary logistic case.}
For each pre-trained embedding function $\check{h}_j$, let
\begin{equation*}
	f_{\diamond,j,n}^{\bullet}\in\argmin_{f\in\mathcal{F}_{\diamond,n}^{\sub}}\left\|f\circ\check{h}_j-f_{\diamond,n}\circ h_m\right\|_{P_{\ta},2}
\end{equation*}
be a target head that best approximates $f_{\diamond,n}\circ h_m$ when paired with $\check{h}_j$. We write the score estimation error at $\theta^*$ as $\Delta\psi_j(W):=\psi(W;\theta^*,\hat{\gamma}_j,\hat{\alpha}_j)-\psi(W;\theta^*,\gamma^*,\alpha^*)$.

\begin{assumption}[Conditions for the Transferability Test]
	\label{asm:transferability_test}
	\leavevmode
	\begin{enumerate}[label=(\roman*)]
		\item \textup{(Sampling and Sample Splitting)}
		The number of pre-trained embedding functions $R\geq2$ is deterministic and may diverge as $n\to\infty$. For each $P_{\ta}\in\mathcal{P}_{0,\ta,n}$, the target observations are i.i.d. from $P_{\ta}$. The pre-trained models and auxiliary pre-training randomness are jointly independent of the target observations. The partition $(\mathcal{I}_1,\mathcal{I}_2,\mathcal{I}_3)$ is independent of all these quantities. Its sample sizes satisfy $n_1,n_2\geq2$ and $n_3\geq1$, with $n_s/n\xrightarrow{P}\pi_s\in(0,1)$ uniformly over $P\in\mathcal{P}_{0,n}$ for $s\in\{1,2,3\}$. The limits $\pi_s$ do not depend on $P$. Each fitted nuisance function depends only on the pre-trained embedding functions, observations in $\mathcal{I}_3$, and auxiliary pre-training randomness.

		\item \textup{(Source and Target Convergence)}
		The convergence bounds in \cref{asm:convergence} hold uniformly over the embedding functions. Specifically, the source models satisfy
		\begin{equation}
			\label{eq:transferability_test_source_rate}
			\max_{1\leq j\leq R}\left\|\check{g}_j\circ\check{h}_j-g_m\circ h_m\right\|_{P_{\so},2}
			=O_{P_{\so}}(\delta_{\so,m}).
		\end{equation}
		For each $\diamond\in\{\gamma,\alpha\}$, the target fits based on $\mathcal{I}_3$ satisfy
		\begin{equation}
			\label{eq:transferability_test_target_rate}
			\max_{1\leq j\leq R}\left\|\hat{f}_{\diamond,j}\circ\check{h}_j-f_{\diamond,j,n}^{\bullet}\circ\check{h}_j\right\|_{P_{\ta},2}
			=O_P(r_{\diamond,n}),
		\end{equation}
		uniformly over $P\in\mathcal{P}_{0,n}$. The deterministic rates satisfy
		\begin{equation}
			\label{eq:transferability_test_rate}
			\max_{\diamond\in\{\gamma,\alpha\}}
			\left(r_{\diamond,n}+\delta_{\so,m}/\nu_{\diamond,n}+\epsilon_{\diamond,n}\right)
			=o(n^{-1/4}),
		\end{equation}
		uniformly over $P_{\ta}\in\mathcal{P}_{0,\ta,n}$.

		\item \textup{(Neyman Orthogonality and Score Regularity)}
		For each $P_{\ta}\in\mathcal{P}_{0,\ta,n}$, the parameter $\theta^*$ is an interior point of $\Theta$ and $\mathbb{E}_{\ta}[\psi(W;\theta^*,\gamma^*,\alpha^*)]=0$. There exist constants $0<\underline{J}\leq\overline{J}<\infty$ and $0<\underline{V}\leq\overline{V}<\infty$ such that $\underline{J}\leq |J|\leq\overline{J}$ and $\underline{V}\leq V\leq\overline{V}$ uniformly over $P_{\ta}\in\mathcal{P}_{0,\ta,n}$. For some fixed $\zeta>0$, the population influence function satisfies the moment condition
		\begin{equation}
			\label{eq:transferability_test_moment_bound}
			\sup_{n\geq1}\sup_{P_{\ta}\in\mathcal{P}_{0,\ta,n}}\mathbb{E}_{\ta}\left[\left|\frac{\phi(W)}{\sqrt{V}}\right|^{2+\zeta}\right]<\infty.
		\end{equation}
		For each $n$ and $P_{\ta}$, there is a deterministic set $\mathcal{N}_n=\mathcal{N}_n(P_{\ta})$ of nuisance pairs containing $(\gamma^*,\alpha^*)$ such that
		\begin{equation}
			\label{eq:transferability_test_nuisance_set}
			\inf_{P\in\mathcal{P}_{0,n}}\Pr_P\left((\hat{\gamma}_j,\hat{\alpha}_j)\in\mathcal{N}_n\text{ for all }j=1,\ldots,R\right)\longrightarrow1.
		\end{equation}
		For every $(\gamma,\alpha)\in\mathcal{N}_n$, the path $(\gamma^*+r(\gamma-\gamma^*),\alpha^*+r(\alpha-\alpha^*))$ remains in the domain of the function $\psi(W;\theta^*,\cdot,\cdot)$ for $r\in[0,1]$. The population moment along this path at $\theta^*$ is continuous on $[0,1]$ and twice continuously differentiable on $(0,1)$. The score satisfies Neyman orthogonality:
		\begin{equation}
			\label{eq:transferability_test_orthogonality}
			\left.\frac{\partial}{\partial r}\mathbb{E}_{\ta}\left[\psi\left(W;\theta^*,\gamma^*+r(\gamma-\gamma^*),\alpha^*+r(\alpha-\alpha^*)\right)\right]\right|_{r=0}=0.
		\end{equation}
		The second derivative satisfies
		\begin{equation}
			\label{eq:transferability_test_second_order}
			\sup_{(\gamma,\alpha)\in\mathcal{N}_n}\sup_{r\in(0,1)}
			\left|\frac{\partial^2}{\partial r^2}\mathbb{E}_{\ta}\left[\psi\left(W;\theta^*,\gamma^*+r(\gamma-\gamma^*),\alpha^*+r(\alpha-\alpha^*)\right)\right]\right|
			=o(n^{-1/2}),
		\end{equation}
		uniformly over $P_{\ta}\in\mathcal{P}_{0,\ta,n}$. The pointwise maximum of the score estimation errors satisfies
		\begin{equation}
			\label{eq:transferability_test_score_envelope}
			\sqrt{\log R}
			\left\|\max_{1\leq j\leq R}|\Delta\psi_j(W)|\right\|_{P_{\ta},2}
			=o_P(1),
		\end{equation}
		uniformly over $P\in\mathcal{P}_{0,n}$.

		\item \textup{(Estimation on Each Subsample)}
		The estimators satisfy $\hat{M}_{j,s}(\hat{\theta}_{j,s})=0$ almost surely for every $j=1,\ldots,R$ and $s\in\{1,2\}$, and $\max_{1\leq j\leq R,s\in\{1,2\}}|\hat{\theta}_{j,s}-\theta^*|=o_P(1)$ uniformly over $P\in\mathcal{P}_{0,n}$.

		\item \textup{(Influence-Function Estimation)}
		For each $n$ and $P\in\mathcal{P}_{0,n}$, there exists a deterministic open interval $\mathcal{U}_n(P)$ containing $\theta^*$. For every $j=1,\ldots,R$, the fitted score $\psi(W;\theta,\hat{\gamma}_j,\hat{\alpha}_j)$ is well defined and continuously differentiable in $\theta$ on $\mathcal{U}_n(P)$, almost surely. The intervals satisfy
		\begin{equation}
			\label{eq:transferability_test_parameter_localization}
			\inf_{P\in\mathcal{P}_{0,n}}\Pr_P\left(\hat{\theta}_{j,s}\in \mathcal{U}_n(P)\text{ for all }j=1,\ldots,R,\ s\in\{1,2\}\right)\longrightarrow1.
		\end{equation}
		The sample derivatives satisfy
		\begin{equation}
			\label{eq:transferability_test_local_derivative}
			\max_{1\leq j\leq R,s\in\{1,2\}}\sup_{\theta\in \mathcal{U}_n(P)}
			\left|\frac{1}{n_s}\sum_{i\in\mathcal{I}_s}\partial_{\theta}\psi(W_i;\theta,\hat{\gamma}_j,\hat{\alpha}_j)-J\right|
			=o_P(1),
		\end{equation}
		uniformly over $P\in\mathcal{P}_{0,n}$. For $\hat{q}_{ij,s}$ in \cref{eq:transferability_test_influence_weights}, the variance estimators for the first embedding function satisfy $n_s\sum_{i\in\mathcal{I}_s}\hat{q}_{i1,s}^2\xrightarrow{P}V$ for $s\in\{1,2\}$, uniformly over $P\in\mathcal{P}_{0,n}$. Differences in $\hat{q}_{ij,s}$ across embedding functions satisfy $\sqrt{\log R}\allowbreak\max_{1\leq j\leq R,s\in\{1,2\}}\allowbreak\sqrt{n_s\sum_{i\in\mathcal{I}_s}\left(\hat{q}_{ij,s}-\hat{q}_{i1,s}\right)^2}=o_P(1)$ uniformly over $P\in\mathcal{P}_{0,n}$.

		\item
		\begin{samepage}
			\textup{(Uniformity of the Transferability Bound)}
			For every fixed $M>0$, there exists $N_M<\infty$, independent of $P_{\ta}$, such that, for all $n\geq N_M$ and each $\diamond\in\{\gamma,\alpha\}$,
			\begin{equation}
				\label{eq:transferability_test_uniform_approximation}
				\sup_{h\in\mathcal{H}_{\so,m}(M\delta_{\so,m})}
				\min_{f\in\mathcal{F}_{\diamond,n}^{\sub}}
				\left\|f\circ h-f_{\diamond,n}\circ h_m\right\|_{P_{\ta},2}
				\leq M\delta_{\so,m}/\nu_{\diamond,n},
			\end{equation}
			uniformly over $P_{\ta}\in\mathcal{P}_{0,\ta,n}$.
			\par
		\end{samepage}
	\end{enumerate}
\end{assumption}

\cref{asm:transferability_test}(iii) adapts the conditions for nonlinear scores in \cite{chernozhukov2018double}: the moment condition in Assumption 3.3(a), smoothness in 3.3(b), the Jacobian bounds in 3.3(c), and orthogonality in 3.3(d). From Assumption 3.4, it uses (a) jointly over embedding functions, only the moment requirement of (b) at the true score, some components of (c), and (d). The score error bound accommodates growing $R$.
\cref{asm:transferability_test}(v) requires the estimator of the derivative of the score to converge to $J$ uniformly over $j,s$ and $\theta\in \mathcal{U}_n(P)$. The mean value theorem then yields uniform linearization of the sample moment functions. Evaluating the derivative bound at the estimators also gives uniform consistency of $\hat{J}_{j,s}$ in \cref{eq:transferability_test_derivative_estimate}. The intervals need not shrink with $n$. Part (v) also requires consistency of the variance estimators for the first embedding function and a bound on the differences in $\hat{q}_{ij,s}$ across embedding functions. Together, these variance and weight conditions imply consistency of the variance estimators uniformly over embedding functions, estimation samples, and $P\in\mathcal{P}_{0,n}$.
\cref{asm:transferability_test}(vi) requires the bound in \cref{eq:transferability_test_uniform_approximation} to hold beyond a common threshold $N_M$ for all target distributions. For any fixed target distribution satisfying the null, this bound holds for all sufficiently large $n$ because $M\delta_{\so,m}\asymp\delta_{\so,m}$ for every fixed $M>0$. The common threshold may depend on $M$ and the fixed source problem, but not on $P_{\ta}$.

In the following theorem, we write $\Pr^{(\xi)}$ for probability conditional on all observed data, the partition, and pre-training randomness. For asymptotic statements, we set $\mathcal{T}$, $c_{1-\textsf{a}}^{(\xi)}$, and all $\mathcal{T}_b^{(\xi)}$ to zero if the algorithm stops without a test decision.

\begin{theorem}[Uniform Asymptotic Size of the Transferability Test]
	\label{thm:transferability_test_size}
	For the fixed source distribution, suppose that $H_{0,n}$ in \cref{eq:transferability_test_null_nu} and \cref{asm:transferability_test} hold for the sequence of classes $\mathcal{P}_{0,n}$. Fix $\textsf{a}\in(0,1)$ and let $B$ be any deterministic sequence of positive integers tending to infinity as $n\to\infty$. Compute $\mathcal{T}$, $\{\mathcal{T}_b^{(\xi)}\}_{b=1}^B$, and $c_{1-\textsf{a}}^{(\xi)}$ using \cref{alg:bootstrap_transferability}, with the zero convention stated above. The standard normal multipliers are independent across observations and bootstrap draws, and independent of the data, partition, and pre-training randomness. Then the following properties hold.
	\begin{enumerate}[label=(\roman*)]
		\item The probability that the algorithm stops without a decision tends to zero uniformly over $P\in\mathcal{P}_{0,n}$.
		\item For each $b\in\{1,\ldots,B\}$, the conditional bootstrap distribution satisfies
		\begin{equation*}
			\sup_{t\in\mathbb{R}}\left|\Pr^{(\xi)}(\mathcal{T}_b^{(\xi)}\leq t)-\Pr_P(\mathcal{T}\leq t)\right|=o_P(1),
		\end{equation*}
		uniformly over $P\in\mathcal{P}_{0,n}$.
		\item The rejection probability satisfies
		\begin{equation*}
			\sup_{P\in\mathcal{P}_{0,n}}\left|\Pr_P\left(\mathcal{T}>c_{1-\textsf{a}}^{(\xi)}\right)-\textsf{a}\right|\longrightarrow0
			\qquad\text{as }n,B\to\infty,
		\end{equation*}
		where the probability includes the randomness of the source and target samples, sample partition, training procedures, and bootstrap draws.
	\end{enumerate}
\end{theorem}

\subsection{Uniform Consistency under Alternatives}
The test compares the estimated $\theta$ induced by the fitted nuisance functions for the observed embeddings. Failure of transferability need not make these estimators differ since the available embedding functions may omit the same relevant information or induce the same bias in the parameter of interest. We establish uniform consistency on a subclass of alternatives under which at least two embedding functions yield sufficiently separated solutions of their population moment equations.
That is, under such alternatives, where the population moment solutions associated with different embedding functions are well separated, the test detects the lack of transferability with probability approaching one uniformly over this subclass.

For each embedding function $j$, let $\theta_{j,n}^{\dagger}\in\Theta$ be a measurable solution of
\begin{equation}
	\label{eq:transferability_test_alternative_root}
	\mathbb{E}_{\ta}[\psi(W;\theta_{j,n}^{\dagger},\hat{\gamma}_j,\hat{\alpha}_j)]=0.
\end{equation}
The expectation holds the fitted nuisance functions fixed; thus, these roots may depend on the source sample, the nuisance-training sample, the partition, and pre-training randomness. Under an alternative, these solutions need not converge in probability to $\theta^*$, and the derivative estimators in \cref{alg:bootstrap_transferability} need not converge in probability to the common derivative $J$.

Let $\mathcal{P}_{\mathrm{sep},n}\subseteq\mathcal{P}_{0,n}^{c}$ be a subclass of joint laws $P=P_{\so}\times P_{\ta}$ under alternatives to \cref{eq:transferability_test_null_nu}, with the same fixed source problem as above. The complement is taken within the class of joint laws with the fixed source problem. For the uniform consistency results, we focus on this subclass of alternatives satisfying certain separation conditions.

\begin{assumption}[Conditions for Uniform Consistency of the Transferability Test]
	\label{asm:transferability_test_consistency}
	For some deterministic estimation error rate $r_{\theta,n}>0$ and separation sequence $\Delta_n>0$, common across $\mathcal{P}_{\mathrm{sep},n}$, the following conditions hold:
	\begin{enumerate}[label=(\roman*)]
		\item There exist constants $c_J,c_{\psi}>0$ such that
		\begin{equation}
			\label{eq:transferability_test_alternative_nondegeneracy}
			\inf_{P\in\mathcal{P}_{\mathrm{sep},n}}\Pr_P\left(
			\min_{\substack{1\leq j\leq R\\s\in\{1,2\}}}|\hat{J}_{j,s}|\geq c_J,\quad
			\min_{\substack{1\leq j\leq R\\s\in\{1,2\}}}\frac{1}{n_s}\sum_{i\in\mathcal{I}_s}\psi(W_i;\hat{\theta}_{j,s},\hat{\gamma}_j,\hat{\alpha}_j)^2\geq c_{\psi}
			\right)\longrightarrow1,
		\end{equation}
		where the derivative estimators are defined in \cref{eq:transferability_test_derivative_estimate}.
		\item Uniformly over $P\in\mathcal{P}_{\mathrm{sep},n}$,
		\begin{equation}
			\label{eq:transferability_test_alternative_estimation}
			\max_{1 \leq j \leq R, s\in\{1,2\}}|\hat{\theta}_{j,s}-\theta_{j,n}^{\dagger}|=O_P(r_{\theta,n}),
			\qquad
			\sqrt{n}\max_{j<k}\hat{\sigma}_{jk}=O_P(1),
		\end{equation}
		where $\hat{\sigma}_{jk}^2$ is defined in \cref{eq:transferability_test_standard_error}.
		\item The population solutions satisfy
		\begin{equation}
			\label{eq:transferability_test_alternative_separation}
			\inf_{P\in\mathcal{P}_{\mathrm{sep},n}}\Pr_P\left(\max_{j<k}|\theta_{j,n}^{\dagger}-\theta_{k,n}^{\dagger}|\geq\Delta_n\right)\longrightarrow1,
		\end{equation}
		with $r_{\theta,n}=o(\Delta_n)$ and $\sqrt{\log R/n}=o(\Delta_n)$.
	\end{enumerate}
\end{assumption}

\cref{asm:transferability_test_consistency}(ii) and (iii) require that the separation between the population solutions is sufficiently larger than the estimation error of the DML estimators. The next result establishes the uniform consistency of the transferability test under these conditions.

\begin{theorem}[Uniform Consistency of the Transferability Test]
	\label{thm:transferability_test_consistency}
	Suppose that the sampling conditions in \cref{asm:transferability_test}(i) hold uniformly over $\mathcal{P}_{\mathrm{sep},n}$, the source convergence bound in \cref{eq:transferability_test_source_rate} holds under $P_{\so}$, and the solutions in \cref{eq:transferability_test_alternative_root} exist. Suppose also that \cref{asm:transferability_test_consistency} holds. Fix $\textsf{a}\in(0,1)$ and let $B\geq1$ be any deterministic sequence of integers. Compute the statistic and empirical bootstrap critical value using \cref{alg:bootstrap_transferability}, with the same zero convention as in \cref{thm:transferability_test_size}. Then,
	\begin{equation}
		\label{eq:transferability_test_consistency}
		\inf_{P\in\mathcal{P}_{\mathrm{sep},n}}\Pr_P\left(\mathcal{T}>c_{1-\textsf{a}}^{(\xi)}\right)\longrightarrow1.
	\end{equation}
\end{theorem}

\section{Simulations}
\label{sec:simulations}
We study the finite-sample performance of DML with pre-trained embeddings in two designs.
The linear design compares ordinary least squares, cross-validated LASSO, and cross-validated ridge as target nuisance learners over $5{,}000$ Monte Carlo replications.
The nonlinear design compares nine nonlinear nuisance learners and evaluates the transferability test over $1{,}000$ Monte Carlo replications per scenario.
Both designs have $m=20{,}000$ source observations and a true target coefficient $\theta^*=0.5$.
Each replication draws independent source and target samples and estimates the source embeddings.
The oracle benchmark supplies the population embedding $h^*$ to the target learner, yielding the performance of the ordinary DML procedure with known embeddings.
We write $[h^*(z)]_j$ for the $j$th coordinate of $h^*(z)$. We use the last hidden layer after the activation and the layer normalization as embeddings, as we do for the BERT embedding in the empirical application.

All specifications estimate the partially linear score in \cref{ex:partially_linear}, using three-fold cross-fitting repeated five times to reduce the variability of the cross-fitted estimates due to random fold assignments.
For each embedding and learner, we report the median coefficient across the five repetitions and estimate its variance by the median of the repetition-specific variances augmented by the squared deviation of each repetition's coefficient from the median coefficient.
The median aggregation is recommended over the mean aggregation and theoretically justified in \citet[Corollary 3.3]{chernozhukov2018double}.
Confidence intervals use normal critical values at the $95\%$ level.
All target learners' tuning and validation take place within the corresponding outer training fold.
We report bias and coverage for each target learner.

\subsection{Linear Design with High-Dimensional Embeddings}
\label{subsec:simulation_linear}

The population embedding has $200$ coordinates and is given by\footnote{For $v\in\mathbb{R}^K$, let $\bar v=K^{-1}\sum_{k=1}^K v_k$. For $k=1,\ldots,K$, we use
\begin{align*}
	[\operatorname{ReLU}(v)]_k &= \max\{v_k,0\},\\
	[\operatorname{LayerNorm}(v)]_k &= \frac{v_k-\bar v}{\sqrt{K^{-1}\sum_{k=1}^K(v_k-\bar v)^2+10^{-5}}}.
\end{align*}
Here, layer normalization centers and rescales the coordinates within each observation.}
\begin{equation*}
	h^*(z)=s_0^{-1}\operatorname{LayerNorm}(\operatorname{ReLU}(z+\mathbf{1})),
\end{equation*}
where $s_0\approx1.0025$ is a scaling factor explained below.
Source coordinates are independent uniform draws on $Z\sim[-\sqrt{3},\sqrt{3}]$, whereas target coordinates are independent uniform draws on $Z\sim[-\sqrt{2},\sqrt{2}]$ and they are mutually independent. We stack these independent $Z$ and construct the $200$-dimensional vectors to form source and target inputs $Z_j^{\so}\in\mathbb{R}^{200}$ and $Z_i^{\ta}\in\mathbb{R}^{200}$, respectively.
We set $s_0$ so that $([h^*(Z)]_1-[h^*(Z)]_2)/\sqrt{2}$ and $([h^*(Z)]_3-[h^*(Z)]_4)/\sqrt{2}$ each have variance approximately one under the source distribution.

The source response is $S_j=B'h^*(Z_j^{\so})+\eta_j$, where $\eta_j\sim N(0,I_{200})$ is independent noise.
The fixed source coefficient matrix $B\in\mathbb{R}^{200\times200}$ has rank $198$.
The columns of $B$ are scaled so that the first column has Euclidean norm $10$ and each of the other $199$ columns is orthogonal to the first and has Euclidean norm $\sqrt{3}$.

The target sample has $n=2{,}000$ observations and is generated via
\begin{align*}
	X_i &= 0.3\left\{\frac{[h^*(Z_i)]_1-[h^*(Z_i)]_2}{\sqrt{2}}+\frac{[h^*(Z_i)]_3-[h^*(Z_i)]_4}{\sqrt{2}}\right\}+V_i,\\
	Y_i &= 0.5X_i+2\left\{\frac{[h^*(Z_i)]_3-[h^*(Z_i)]_4}{\sqrt{2}}-\frac{[h^*(Z_i)]_1-[h^*(Z_i)]_2}{\sqrt{2}}\right\}+U_i,
\end{align*}
where $V_i\sim N(0,1)$ and $U_i\sim N(0,0.25)$ are independent of each other and of all covariates and source observations.
The coefficient vectors of both target contrasts lie in the column space of $B$, so the reported design satisfies the linear span condition for transferability. They are also sparse, involving only a few coordinates of the source representation.

The learned source model has one dense hidden layer of width $200$, followed by ReLU, LayerNorm, and a linear source head.
Source training uses an $80/10/10$ training, validation, and test split.
We compare three initializations and select the epoch and initialization with the highest average validation $R^2$ across source outcomes.
We then refit the selected specification on the combined training and validation observations.
The target nuisance regressions use either this learned embedding or $h^*$.

\cref{tab:simulation-linear} reports the results.
With the population embedding, coverage is $93.66\%$ for OLS, $92.14\%$ for LASSO, and $95.42\%$ for ridge. With learned embeddings, their coverages are $93.42\%$, $83.86\%$, and $94.78\%$, respectively.
Thus, the coverage loss from replacing the population embedding by the learned embedding is substantially larger for LASSO in this design.
LASSO also has the largest absolute bias among the three learners under either representation.
This pattern is consistent with the sensitivity of coordinatewise sparsity to the representation, as discussed in \cref{sec:lasso}.
\begin{table}[htbp]
\centering
\footnotesize
\caption{Linear design: bias and interval coverage of DML estimates}\label{tab:simulation-linear}
\begin{tabular}{lrrrr}
\toprule
& \multicolumn{2}{c}{Oracle} & \multicolumn{2}{c}{Learned} \\
\cmidrule(lr){2-3}\cmidrule(lr){4-5}
Learner & Bias & Coverage & Bias & Coverage \\
\midrule
OLS & \num{0.0000} & 93.66 & \num{-0.0005} & 93.42 \\
CV-LASSO & \num{-0.0062} & 92.14 & \num{-0.0149} & 83.86 \\
CV-Ridge & \num{-0.0005} & 95.42 & \num{-0.0025} & 94.78 \\
\bottomrule
\end{tabular}
\par
\vspace{0.5em}
\begin{minipage}{\linewidth}
\footnotesize
\textit{Notes:} Each cell uses 5,000 Monte Carlo replications. The target and source sample sizes are $n=2{,}000$ and $m=20{,}000$. Bias is relative to $\theta^*=0.5$. Coverage is for nominal 95\% normal intervals using HC1 standard errors. Oracle uses the population representation.
\end{minipage}
\end{table}

\subsection{Nonlinear Design}
Next, we consider a nonlinear design under varying levels of the transferability strength. That is, the source signals on the embedding components relevant to the target task vary in strength.
We examine bias and coverage across these scenarios and then evaluate the rejection frequencies of the transferability test.
\label{subsec:simulation_ej18}
\subsubsection{Coverage}
Source and target inputs are independent draws from $Z\sim N(0,I_{30})$.
The population embedding is $h^*(z)=\operatorname{LayerNorm}(\operatorname{GELU}(z))$, where $\operatorname{GELU}(z_j)=z_j\Phi(z_j)$ and $\Phi$ is the standard normal distribution function.
Using the map $\Lambda$ defined in \cref{sec:dml}, the target equations are
\begin{align*}
	X_i &= [h^*(Z_i)]_1+\frac{1}{4}\Lambda([h^*(Z_i)]_3)+V_i,\\
	Y_i &= 0.5X_i+\Lambda([h^*(Z_i)]_1)+\frac{1}{4}[h^*(Z_i)]_3+U_i,
\end{align*}
where $V_i$ and $U_i$ are independent standard normal innovations, independent of the covariates and source sample.
The target equations adapt the nonlinear functional form in \cite{chernozhukov2018double} to the population embedding.
For $c(z)=([h^*(z)]_1,[h^*(z)]_3,[h^*(z)]_2,[h^*(z)]_4,[h^*(z)]_5,[h^*(z)]_6)'$, the source response is
\begin{equation*}
	S_j=Q\operatorname{diag}(1-\rho,1-\rho,1,1,1,1)c(Z_j^{\so})+\eta_j,
	\qquad \eta_j\sim N(0,I_{20}),
\end{equation*}
where $Q$ is a fixed $20\times6$ matrix with orthonormal columns, generated by a reduced QR decomposition of a standard Gaussian matrix, and source noise is independent of all other draws.
We consider $\rho\in\{0,1-n^{-1/4},1-n^{-1/2},1\}$ and $n\in\{5{,}000,10{,}000\}$.
At $\rho=0$, the direct source coefficients on the two target-relevant components equal one.
$\rho=1$ removes the direct source coefficients on the target-relevant components entirely, which corresponds to non-transferabile case.
The two intermediate scenarios reduce these coefficients to $n^{-1/4}$ and $n^{-1/2}$, respectively, and they correspond to weakly transferable cases.
These scenarios vary how strongly the source task informs the embedding components relevant for the target nuisance functions.

The trained source network has one dense hidden layer of width $30$, followed by GELU (Gaussian Error Linear Unit) activation, affine LayerNorm, and a linear head with $20$ outputs.
Training minimizes unstandardized source mean squared error on an $80\%$ training split, selects the best checkpoint using a $10\%$ validation split, and reserves the remaining $10\%$ for diagnostics; there is no refit after checkpoint selection.

The target learners are three random forests (RF 1--3), three gradient-boosted tree models (XGB 1--3), and three neural networks (NN 1--3).
The forests use $200$ trees, search all features, and have minimum leaf sizes $100$, $30$, and $1$.
The boosted trees have maximum depth three and at most $200$ rounds; their minimum child weights are $10\%$ and $3\%$ of the inner fitting sample size, rounded up, and $1$, respectively.
The neural networks use ReLU hidden layers of widths $(16)$, $(32,32)$, and $(32,32,32)$, with at most $600$ epochs and early stopping after $60$ epochs without validation improvement.

\begin{table}[!htb]
\centering
\scriptsize
\caption{Nonlinear design: bias and interval coverage of DML estimates}\label{tab:simulation-ej18-coverage}
\begin{tabular}{lrrrrrrrrrr}
\toprule
& \multicolumn{10}{c}{$n=5{,}000$} \\
\cmidrule(lr){2-11}
& \multicolumn{2}{c}{Oracle} & \multicolumn{2}{c}{Learned} & \multicolumn{2}{c}{Learned} & \multicolumn{2}{c}{Learned} & \multicolumn{2}{c}{Learned} \\
& \multicolumn{2}{c}{$\rho=0$} & \multicolumn{2}{c}{$\rho=0$} & \multicolumn{2}{c}{$\rho=1-n^{-1/4}$} & \multicolumn{2}{c}{$\rho=1-n^{-1/2}$} & \multicolumn{2}{c}{$\rho=1$} \\
\cmidrule(lr){2-3}\cmidrule(lr){4-5}\cmidrule(lr){6-7}\cmidrule(lr){8-9}\cmidrule(lr){10-11}
Learner & Bias & Coverage & Bias & Coverage & Bias & Coverage & Bias & Coverage & Bias & Coverage \\
\midrule
RF 1 & \num{0.0002} & 94.40 & \num{0.0072} & 91.70 & \num{0.0462} & 12.20 & \num{0.0718} & 0.00 & \num{0.0719} & 0.00 \\
RF 2 & \num{-0.0028} & 95.10 & \num{0.0033} & 93.30 & \num{0.0385} & 23.00 & \num{0.0671} & 0.20 & \num{0.0672} & 0.10 \\
RF 3 & \num{-0.0075} & 92.20 & \num{-0.0020} & 94.60 & \num{0.0314} & 37.90 & \num{0.0615} & 0.60 & \num{0.0617} & 0.70 \\
XGB 1 & \num{0.0011} & 95.40 & \num{0.0045} & 93.60 & \num{0.0303} & 39.10 & \num{0.0577} & 1.90 & \num{0.0577} & 1.40 \\
XGB 2 & \num{-0.0025} & 95.00 & \num{0.0015} & 94.90 & \num{0.0248} & 54.30 & \num{0.0532} & 3.90 & \num{0.0534} & 4.40 \\
XGB 3 & \num{-0.0034} & 94.90 & \num{0.0005} & 94.80 & \num{0.0234} & 59.10 & \num{0.0518} & 4.60 & \num{0.0522} & 5.10 \\
NN 1 & \num{-0.0038} & 95.70 & \num{0.0030} & 94.20 & \num{0.0338} & 31.40 & \num{0.0628} & 1.00 & \num{0.0631} & 1.00 \\
NN 2 & \num{-0.0037} & 95.00 & \num{0.0035} & 93.90 & \num{0.0340} & 30.70 & \num{0.0634} & 0.80 & \num{0.0636} & 0.80 \\
NN 3 & \num{-0.0029} & 95.60 & \num{0.0044} & 93.10 & \num{0.0343} & 29.50 & \num{0.0636} & 0.70 & \num{0.0638} & 0.60 \\
\midrule
& \multicolumn{10}{c}{$n=10{,}000$} \\
\cmidrule(lr){2-11}
& \multicolumn{2}{c}{Oracle} & \multicolumn{2}{c}{Learned} & \multicolumn{2}{c}{Learned} & \multicolumn{2}{c}{Learned} & \multicolumn{2}{c}{Learned} \\
& \multicolumn{2}{c}{$\rho=0$} & \multicolumn{2}{c}{$\rho=0$} & \multicolumn{2}{c}{$\rho=1-n^{-1/4}$} & \multicolumn{2}{c}{$\rho=1-n^{-1/2}$} & \multicolumn{2}{c}{$\rho=1$} \\
\cmidrule(lr){2-3}\cmidrule(lr){4-5}\cmidrule(lr){6-7}\cmidrule(lr){8-9}\cmidrule(lr){10-11}
Learner & Bias & Coverage & Bias & Coverage & Bias & Coverage & Bias & Coverage & Bias & Coverage \\
\midrule
RF 1 & \num{-0.0010} & 96.00 & \num{0.0056} & 92.60 & \num{0.0505} & 2.00 & \num{0.0701} & 0.00 & \num{0.0702} & 0.00 \\
RF 2 & \num{-0.0028} & 95.20 & \num{0.0029} & 94.70 & \num{0.0445} & 3.00 & \num{0.0662} & 0.00 & \num{0.0662} & 0.00 \\
RF 3 & \num{-0.0067} & 92.30 & \num{-0.0015} & 95.70 & \num{0.0385} & 7.30 & \num{0.0612} & 0.00 & \num{0.0612} & 0.00 \\
XGB 1 & \num{0.0016} & 95.90 & \num{0.0050} & 92.30 & \num{0.0377} & 5.80 & \num{0.0591} & 0.10 & \num{0.0591} & 0.00 \\
XGB 2 & \num{-0.0015} & 96.10 & \num{0.0026} & 95.00 & \num{0.0322} & 13.10 & \num{0.0548} & 0.40 & \num{0.0548} & 0.20 \\
XGB 3 & \num{-0.0020} & 95.40 & \num{0.0020} & 95.30 & \num{0.0305} & 16.70 & \num{0.0534} & 0.40 & \num{0.0534} & 0.30 \\
NN 1 & \num{-0.0017} & 95.80 & \num{0.0045} & 93.10 & \num{0.0426} & 3.20 & \num{0.0651} & 0.00 & \num{0.0651} & 0.00 \\
NN 2 & \num{-0.0007} & 96.00 & \num{0.0053} & 92.40 & \num{0.0435} & 2.80 & \num{0.0664} & 0.00 & \num{0.0665} & 0.00 \\
NN 3 & \num{0.0012} & 95.70 & \num{0.0069} & 90.10 & \num{0.0446} & 2.70 & \num{0.0676} & 0.00 & \num{0.0677} & 0.00 \\
\bottomrule
\end{tabular}
\par
\vspace{0.5em}
\begin{minipage}{\linewidth}
\scriptsize
\textit{Notes:} Each cell uses 1,000 Monte Carlo replications. The source sample size is $m=20{,}000$. Bias is the mean estimation error relative to $\theta^*=0.5$. Coverage is for nominal 95\% normal intervals. Estimates are medians across five repetitions of three-fold cross-fitting. Variances are medians of repetition-specific HC1 variances plus squared deviations from the median estimate. Oracle uses the population representation.
\end{minipage}
\end{table}

\cref{tab:simulation-ej18-coverage} reports bias and coverage.
The coverage experiment trains one source embedding per replication using a single initialization.
Each replication contributes one estimate and interval per learner, using the repeated cross-fitting procedure described above.
With the population embedding, coverage ranges from $92.2\%$ to $95.7\%$ at $n=5{,}000$ and from $92.3\%$ to $96.1\%$ at $n=10{,}000$.
With learned embeddings at $\rho=0$, the corresponding ranges are $91.7\%$--$94.9\%$ and $90.1\%$--$95.7\%$.
Coverage deteriorates sharply as the source signals on the target-relevant components weaken.
At $\rho=1-n^{-1/4}$, coverage ranges from $12.2\%$ to $59.1\%$ at the smaller sample size and from $2.0\%$ to $16.7\%$ at the larger sample size.
At $\rho=1-n^{-1/2}$, the corresponding ranges are $0\%$--$4.6\%$ and $0\%$--$0.4\%$, close to those at $\rho=1$, where coverage ranges from $0\%$ to $5.1\%$ and from $0\%$ to $0.3\%$.
All learners have positive Monte Carlo bias in these three scenarios with weak or zero direct source coefficients on the target-relevant components.

\subsubsection{Transferability Test}
\label{subsec:simulation_transferability_test}

The transferability test uses a separate experiment that trains $25$ source embeddings on the same source sample for each $\rho$ and replication, using different initialization and training-order seeds.
We use a nominal $5\%$ rejection threshold.
Within each held-out fold, observations are split into two estimation blocks that share nuisance fits from the training complement.
After pooling held-out observations within each block across folds, we form the maximum statistic over all $300$ pairs of embeddings (all possible pairs of the $25$ source embeddings) in each of the two block orientations, using $500$ multiplier bootstrap draws per repetition and orientation.
While \cref{thm:transferability_test_size} assumes a single nuisance-training block independent of both estimation blocks, we encourage implementation of the cross-fitting procedure for the efficient use of the data, even though its theoretical size validity requires a separate argument.
The five cross-fitting repetitions and two test sample orientations yield ten bootstrap $p$-values.
We compare the grid-harmonic merger with the corrected upper-median rule to aggregate the ten bootstrap $p$-values into a single $p$-value for each replication.

For any positive integer $J$, let $p_1,\ldots,p_J\in[0,1]$ denote the input $p$-values and write $\ell_J=\sum_{j=1}^{J}1/j$.
For positive inputs, the grid-harmonic $p$-value is the smallest $t\in(0,1]$ satisfying
\begin{equation*}
	\sum_{i=1}^{J}\frac{\mathbf{1}\{\ell_Jp_i\leq t\}}{\lceil J\ell_Jp_i/t\rceil}\geq1,
\end{equation*}
and is one if no such threshold exists.
We set it to zero if any input $p$-value is zero.
\citet[Theorem 7.1]{vovk2022admissible} show that the grid-harmonic merging function is admissible among symmetric $p$-merging functions and, when $J$ is not prime, among all $p$-merging functions.
Here admissibility means that no other rule in the relevant class returns weakly smaller $p$-values for every input vector and strictly smaller $p$-values for at least one input vector.

For the corrected upper-median rule, order the inputs as $p_{(1)}\leq\cdots\leq p_{(J)}$ and set $k_J=\lfloor J/2\rfloor+1$.
The corrected upper-median $p$-value justified by \citet[Theorem 3.1]{Meinshausen2009p} is
\begin{equation*}
	p_{\mathrm{med}}=\min\{1,(J/k_J)p_{(k_J)}\}.
\end{equation*}
Here $p_{(k_J)}$ is the middle order statistic for odd $J$ and the larger of the two middle order statistics for even $J$.
Rejection at level $\alpha\in(0,1)$ requires at least $k_J$ inputs to be at most $\alpha k_J/J$.
For ten inputs, the corrected upper-median rule is $\min\{1,(10/6)p_{(6)}\}$, where $p_{(6)}$ is the sixth order statistic.

\begin{table}[!htb]
\centering
\scriptsize
\caption{Nonlinear design: rejection rates of transferability tests of DML estimates}\label{tab:simulation-ej18-tests}
\begin{tabular}{lrrrrrrrr}
\toprule
& \multicolumn{8}{c}{$n=5{,}000$} \\
\cmidrule(lr){2-9}
& \multicolumn{2}{c}{$\rho=0$} & \multicolumn{2}{c}{$\rho=1-n^{-1/4}$} & \multicolumn{2}{c}{$\rho=1-n^{-1/2}$} & \multicolumn{2}{c}{$\rho=1$} \\
\cmidrule(lr){2-3}\cmidrule(lr){4-5}\cmidrule(lr){6-7}\cmidrule(lr){8-9}
Learner & Median & Harmonic & Median & Harmonic & Median & Harmonic & Median & Harmonic \\
\midrule
RF 1 & 0.00 & 2.20 & 12.50 & 31.50 & 3.70 & 15.90 & 3.80 & 16.20 \\
RF 2 & 0.00 & 2.40 & 7.10 & 24.60 & 5.00 & 20.90 & 4.90 & 21.50 \\
RF 3 & 0.00 & 2.90 & 4.50 & 19.20 & 5.90 & 22.30 & 5.50 & 21.70 \\
XGB 1 & 0.00 & 2.80 & 1.10 & 9.10 & 7.00 & 26.60 & 7.70 & 27.80 \\
XGB 2 & 0.10 & 2.60 & 1.40 & 10.80 & 8.80 & 28.90 & 10.30 & 33.20 \\
XGB 3 & 0.10 & 3.80 & 0.80 & 11.00 & 10.50 & 32.50 & 11.40 & 34.20 \\
NN 1 & 0.00 & 2.90 & 1.60 & 12.50 & 6.10 & 22.20 & 5.70 & 23.80 \\
NN 2 & 0.10 & 2.80 & 1.70 & 13.30 & 6.60 & 23.70 & 6.50 & 24.60 \\
NN 3 & 0.00 & 3.50 & 1.60 & 12.80 & 6.30 & 24.20 & 6.60 & 24.50 \\
\midrule
& \multicolumn{8}{c}{$n=10{,}000$} \\
\cmidrule(lr){2-9}
& \multicolumn{2}{c}{$\rho=0$} & \multicolumn{2}{c}{$\rho=1-n^{-1/4}$} & \multicolumn{2}{c}{$\rho=1-n^{-1/2}$} & \multicolumn{2}{c}{$\rho=1$} \\
\cmidrule(lr){2-3}\cmidrule(lr){4-5}\cmidrule(lr){6-7}\cmidrule(lr){8-9}
Learner & Median & Harmonic & Median & Harmonic & Median & Harmonic & Median & Harmonic \\
\midrule
RF 1 & 0.00 & 2.10 & 45.40 & 61.00 & 23.60 & 44.70 & 25.30 & 44.80 \\
RF 2 & 0.10 & 1.80 & 35.80 & 55.00 & 31.90 & 52.10 & 33.70 & 52.00 \\
RF 3 & 0.10 & 1.70 & 29.10 & 47.70 & 32.70 & 52.30 & 34.10 & 52.80 \\
XGB 1 & 0.10 & 2.40 & 7.70 & 26.20 & 35.00 & 56.20 & 36.70 & 57.30 \\
XGB 2 & 0.10 & 2.40 & 9.10 & 26.80 & 42.10 & 61.20 & 44.00 & 64.90 \\
XGB 3 & 0.10 & 2.80 & 8.70 & 27.90 & 46.20 & 65.30 & 47.70 & 65.80 \\
NN 1 & 0.00 & 1.80 & 15.90 & 37.50 & 37.40 & 55.80 & 37.90 & 57.20 \\
NN 2 & 0.10 & 2.10 & 17.80 & 38.60 & 40.30 & 58.10 & 39.20 & 59.70 \\
NN 3 & 0.10 & 2.20 & 18.10 & 39.10 & 41.40 & 60.50 & 42.50 & 60.40 \\
\bottomrule
\end{tabular}
\par
\vspace{0.5em}
\begin{minipage}{\linewidth}
\scriptsize
\textit{Notes:} Entries are rejection percentages at the 5\% level, based on 1,000 Monte Carlo replications per cell. The source sample size is $m=20{,}000$, and each replication uses 25 trained source embeddings. Median and harmonic denote the corrected upper-median rule and grid-harmonic merger, respectively. Both combine the same ten $p$-values from five repetitions of three-fold cross-fitting and two swap test orientations. Each repetition and orientation uses 500 multiplier bootstrap draws.
\end{minipage}
\end{table}

\cref{tab:simulation-ej18-tests} reports rejection frequencies across the four scenarios.
At $\rho=0$, the grid-harmonic merger rejects in $2.2\%$--$3.8\%$ of replications at $n=5{,}000$ and $1.7\%$--$2.8\%$ at $n=10{,}000$.
Rejection frequencies for the corrected upper-median rule range from $0\%$ to $0.1\%$ at both sample sizes.
Both rules have rejection frequencies below the nominal $5\%$ level, indicating valid size control at the null hypothesis.

At $\rho=1-n^{-1/4}$, grid-harmonic rejection frequencies range from $9.1\%$ to $31.5\%$ at the smaller sample size and from $26.2\%$ to $61.0\%$ at the larger sample size.
At $\rho=1-n^{-1/2}$, the corresponding ranges are $15.9\%$--$32.5\%$ and $44.7\%$--$65.3\%$, close to the $16.2\%$--$34.2\%$ and $44.8\%$--$65.8\%$ ranges at $\rho=1$.
For both merging rules, rejection frequencies are higher at the larger target sample size in each scenario with weak or zero direct source coefficients on the target-relevant components.
The results suggest nontrivial power in scenarios where transferability may fail.

\section{Empirical Application: Monopsony on an Online Labor Market Platform}
\label{sec:empirical_monopsony}
We study the use of pre-trained text embeddings as controls in the online labor-market application of \cite{dube2020monopsony}.
The empirical question is how the speed at which workers accept a task responds to its posted reward, holding task characteristics fixed.
Task descriptions contain information about these characteristics, but converting that information into suitable controls requires choices about the text representation.
The dataset is revisited by \cite{ahrens2026introduction} with another text embedding function; however, our focus is on illustrating our transferability theory, comparing different pre-trained text embeddings as well as different machine learners in the downstream analysis, applying the test for transferability in \cref{sec:test_transferability} to pre-trained BERT representations, and suggesting a new aggregation robust to the pre-trained embeddings.

\subsection{Background and Text as Controls}
The analysis uses the sample of Amazon Mechanical Turk postings, comprising $n=258{,}352$ task groups posted by $24{,}923$ requesters. We treat each requester as a cluster in the analysis to account for potential within-requester correlation.
Let $Y_i$ be log posting duration in minutes, $X_i$ be log posted reward in cents, and $Z_i$ collect the job description text and structured controls.
The identifying specification is the partially linear model
\begin{equation}
	\label{eq:empirical_text_plm}
	Y_i=X_i\theta^*+\kappa^*(Z_i)+U_i,
	\qquad
	\mathbb{E}[U_i\mid X_i,Z_i]=0,
\end{equation}
where $\kappa^*$ is an unknown nonparametric function that captures the relationship between task characteristics and duration.
Since the outcome is log duration and the main regressor is log posted reward, the labor supply elasticity of filling speed is $-\theta^*$ under the monopsony model considered in \cite{dube2020monopsony}. To control for unstructured heterogeneity among the job postings, they include both structured controls and unstructured text-derived controls in the specification.

The first control specification, labeled Dube, contains $611$ numerically distinct variables: $12$ non-text controls, $110$ hand-engineered text controls manually constructed by the researchers, and $489$ learned text controls consisting of supervised n-grams, topic-model features, and Doc2Vec coordinates.\footnote{
	The n-gram features are selected by \cite{dube2020monopsony}. They first select the n-gram features and then treat them as fixed in the DML estimation. We collect all selected n-gram features and include them as fixed controls in our analysis after eliminating numerical duplicates. 
	The hand-engineered text controls include counts of specific keywords (e.g., ``easy'' and ``only''), counts of terms associated with task categories (e.g., ``survey'' and ``read''), and other measures that the researchers extract from the posting text. Importantly, the researchers specify the rules for constructing these controls, and the counts record each matching term regardless of the surrounding words.
	}
The remaining six specifications replace all text-derived Dube controls with frozen BERT embeddings, while retaining the same $12$ non-text controls.
We use Small, Base, and Large models with embedding dimensions $512$, $768$, and $1{,}024$, respectively, and two pooling rules for each size.
The CLS rule uses the final encoder state at the \texttt{[CLS]} position.
The mean rule first applies the pre-trained masked-language-modeling transformation to each token in sentences and then averages these transformed vectors in the last hidden layer, excluding padding and special tokens.
Both rules use the task title followed by its description, truncated to at most $160$ WordPiece tokens, and the BERT parameters are not fine-tuned on the application sample.
The resulting BERT specifications contain $524$, $780$, or $1{,}036$ controls.
To visualize the embeddings, we provide figures in \cref{sec:visualization_embeddings} illustrating the relationships between the embeddings and task characteristics.

We emphasise that while Dube controls depend on an application-specific selection of text features by researchers, the BERT-based controls provide a flexible representation of the textual information in job postings without requiring manual feature engineering. The BERT embeddings also capture the context and semantic relationships between words.

\subsection{Estimation and Comparison of Specifications}
We implement the orthogonal score in \cref{ex:partially_linear} by predicting $Y_i$ and $X_i$ from the chosen controls and regressing the out-of-fold outcome residual on the corresponding reward residual.
For each control specification, we use twelve nuisance learners: ordinary least squares, cross-validated lasso and ridge, three random forests (RF 1--3), three gradient-boosted tree models (XGB 1--3), and three neural networks (NN 1--3).
The same learner specification is used for both nuisance regressions, with separate fitted models.
RF 1--3 use $600$ trees and minimum leaf sizes $100$, $30$, and $1$.
XGB 1--3 use trees of maximum depth six and minimum child weights $2{,}000$, $500$, and $1$, with up to $800$ boosting rounds and early stopping.
The neural networks have three hidden layers of width ten, three of width thirty, or five of width fifty, respectively, with ReLU activations. Thus, learners labeled with 1 correspond to the least flexible models, and those labeled with 3 correspond to the most flexible models.

The results use three-fold cross-fitting with five repetitions clustered at the requester level.\footnote{\cref{sec:individual_cross_fitting} reports the Dube specification with individual-level cross-fitting. While \cite{dube2020monopsony} use the cluster-robust variance for requester-level clustering as we do here, they use individual-level folds for cross-fitting. Potentially, the very small p-values in \cref{sec:individual_cross_fitting} can be explained by the over-fitting of the learners due to within-cluster dependence. To mitigate this concern, we use cross-fitting with cluster-level folds in our analysis.}
All postings by a requester belong to the same fold, and learner-specific tuning and validation splits within a fold also preserve requester groups.
Each repetition produces a coefficient using all out-of-fold residuals and a requester-clustered variance estimator.
We report the median coefficient across repetitions and take the median of the repetition variances augmented by the squared deviation of each repetition's coefficient from that median.
Reported standard errors are the square roots of these aggregated variances.
Prediction performance is summarized by out-of-fold $R_Y^2$ and $R_X^2$ for duration and reward, averaged across repetitions.

\cref{tab:dml-cluster-summary} reports all combinations of controls and nuisance learners.
The XGBoost specifications attain the highest outcome out-of-fold $R_Y^2$ among the reported learners for every control specification, with $R_Y^2$ ranging from $0.8421$ to $0.8621$.
Their reward $R_X^2$ values are lower and vary more across control specifications, with $R_X^2$ ranging from $0.2944$ to $0.4109$.
XGBoost also attains the highest $R_X^2$ for six of the seven control specifications; the exception is BERT Large mean, for which ridge performs best.
The choice of text embeddings matters more for reward prediction than for duration prediction within the XGBoost specifications.
Mean pooling consistently yields higher $R_Y^2$ and generally higher $R_X^2$ than CLS pooling, while increasing BERT size from Base to Large lowers $R_X^2$ for all three XGBoost learners under both pooling rules.
BERT Base mean embeddings, used in the subsequent robustness analysis, attain the highest $R_X^2$ within each XGBoost learner: $0.4011$, $0.4109$, and $0.3847$ for XGB 1--3, respectively, compared with $0.3562$, $0.3348$, and $0.3434$ using the Dube controls.
\begin{table}[htbp]
\centering
\footnotesize
\renewcommand{\arraystretch}{0.95}
\caption{DML estimates by embedding and nuisance learner}\label{tab:dml-cluster-summary}
\makebox[\linewidth][c]{%
\resizebox{1.1\linewidth}{!}{%
\begin{tabular}{lrrrrrr}
\toprule
\multicolumn{7}{l}{\textbf{Panel A: Linear and random-forest learners}} \\
 & OLS & CV-Lasso & CV-Ridge & RF 1 & RF 2 & RF 3 \\
\midrule
\multicolumn{7}{l}{\textit{Dube (611 features)}} \\
Estimate (s.e.) & \num{0.0235} (\num{1.1945}) & \num{-0.2648} (\num{0.1751}) & \num{-0.2087} (\num{1.4814}) & \num{0.0379} (\num{0.2028}) & \num{0.0285} (\num{0.2101}) & \num{0.1085} (\num{0.2571}) \\
$p$-value & \num{0.9843} & \num{0.1304} & \num{0.8880} & \num{0.8519} & \num{0.8922} & \num{0.6731} \\
$R_Y^2/R_D^2$ & \num{-16.5183} / \num{-15.3010} & \num{-15.9878} / \num{0.3218} & \num{-21.8118} / \num{-18.9989} & \num{0.0355} / \num{0.1083} & \num{0.0470} / \num{0.1311} & \num{0.0883} / \num{0.2178} \\
\midrule
\multicolumn{7}{l}{\textit{BERT Small Mean (524 features)}} \\
Estimate (s.e.) & \num{0.0109} (\num{0.2010}) & \num{-0.1795} (\num{0.1425}) & \num{-0.1500} (\num{0.1539}) & \num{0.0403} (\num{0.1980}) & \num{0.0583} (\num{0.2047}) & \num{0.0721} (\num{0.2017}) \\
$p$-value & \num{0.9568} & \num{0.2079} & \num{0.3297} & \num{0.8386} & \num{0.7757} & \num{0.7208} \\
$R_Y^2/R_D^2$ & \num{0.0064} / \num{0.2589} & \num{0.1191} / \num{0.2723} & \num{0.1177} / \num{0.3500} & \num{0.0913} / \num{0.1505} & \num{0.0948} / \num{0.1681} & \num{0.1403} / \num{0.2073} \\
\midrule
\multicolumn{7}{l}{\textit{BERT Small CLS (524 features)}} \\
Estimate (s.e.) & \num{-0.1545} (\num{0.1784}) & \num{-0.1606} (\num{0.1386}) & \num{-0.1799} (\num{0.1437}) & \num{0.0978} (\num{0.2254}) & \num{0.1018} (\num{0.2289}) & \num{0.1175} (\num{0.2355}) \\
$p$-value & \num{0.3863} & \num{0.2464} & \num{0.2104} & \num{0.6643} & \num{0.6566} & \num{0.6178} \\
$R_Y^2/R_D^2$ & \num{0.1214} / \num{0.2449} & \num{0.1000} / \num{0.2800} & \num{0.1393} / \num{0.3520} & \num{0.0332} / \num{0.1562} & \num{0.0370} / \num{0.1783} & \num{0.0952} / \num{0.2154} \\
\midrule
\multicolumn{7}{l}{\textit{BERT Base Mean (780 features)}} \\
Estimate (s.e.) & \num{-0.0386} (\num{0.1801}) & \num{-0.1403} (\num{0.1657}) & \num{-0.1774} (\num{0.1589}) & \num{0.1134} (\num{0.2231}) & \num{0.1170} (\num{0.2281}) & \num{0.1309} (\num{0.2234}) \\
$p$-value & \num{0.8303} & \num{0.3972} & \num{0.2642} & \num{0.6112} & \num{0.6081} & \num{0.5579} \\
$R_Y^2/R_D^2$ & \num{0.0622} / \num{0.1775} & \num{0.0764} / \num{0.2465} & \num{0.1376} / \num{0.3093} & \num{0.0493} / \num{0.1513} & \num{0.0512} / \num{0.1752} & \num{0.1020} / \num{0.2099} \\
\midrule
\multicolumn{7}{l}{\textit{BERT Base CLS (780 features)}} \\
Estimate (s.e.) & \num{-0.1058} (\num{0.1383}) & \num{-0.1254} (\num{0.1815}) & \num{-0.1141} (\num{0.1525}) & \num{0.1325} (\num{0.2485}) & \num{0.1430} (\num{0.2548}) & \num{0.1736} (\num{0.2536}) \\
$p$-value & \num{0.4446} & \num{0.4897} & \num{0.4543} & \num{0.5939} & \num{0.5745} & \num{0.4936} \\
$R_Y^2/R_D^2$ & \num{0.1082} / \num{0.2657} & \num{0.0821} / \num{0.2648} & \num{0.1266} / \num{0.3052} & \num{-0.0121} / \num{0.1258} & \num{-0.0099} / \num{0.1438} & \num{0.0433} / \num{0.1756} \\
\midrule
\multicolumn{7}{l}{\textit{BERT Large Mean (1,036 features)}} \\
Estimate (s.e.) & \num{-0.0848} (\num{0.1379}) & \num{-0.1790} (\num{0.1687}) & \num{-0.1956} (\num{0.1455}) & \num{0.0636} (\num{0.2195}) & \num{0.0710} (\num{0.2272}) & \num{0.0946} (\num{0.2206}) \\
$p$-value & \num{0.5388} & \num{0.2885} & \num{0.1789} & \num{0.7720} & \num{0.7547} & \num{0.6680} \\
$R_Y^2/R_D^2$ & \num{-0.0239} / \num{0.2993} & \num{0.0871} / \num{0.2754} & \num{0.1231} / \num{0.3376} & \num{0.0637} / \num{0.1432} & \num{0.0662} / \num{0.1607} & \num{0.1137} / \num{0.1951} \\
\midrule
\multicolumn{7}{l}{\textit{BERT Large CLS (1,036 features)}} \\
Estimate (s.e.) & \num{0.0248} (\num{0.1866}) & \num{-0.1665} (\num{0.1598}) & \num{0.0137} (\num{0.2146}) & \num{0.1672} (\num{0.2629}) & \num{0.1842} (\num{0.2674}) & \num{0.2022} (\num{0.2639}) \\
$p$-value & \num{0.8943} & \num{0.2974} & \num{0.9492} & \num{0.5247} & \num{0.4909} & \num{0.4435} \\
$R_Y^2/R_D^2$ & \num{0.0272} / \num{0.2288} & \num{0.0097} / \num{0.2469} & \num{0.0541} / \num{0.2446} & \num{-0.0069} / \num{0.0689} & \num{-0.0030} / \num{0.0843} & \num{0.0516} / \num{0.1160} \\
\midrule
\multicolumn{7}{l}{\textbf{Panel B: XGBoost and neural-network learners}} \\
 & XGB 1 & XGB 2 & XGB 3 & NN 1 & NN 2 & NN 3 \\
\midrule
\multicolumn{7}{l}{\textit{Dube (611 features)}} \\
Estimate (s.e.) & \num{-0.0695} (\num{0.0460}) & \num{-0.0422} (\num{0.0338}) & \num{-0.0262} (\num{0.0361}) & \num{1.0597} (\num{0.7457}) & \num{-0.0764} (\num{1.3586}) & \num{0.7116} (\num{0.6657}) \\
$p$-value & \num{0.1307} & \num{0.2116} & \num{0.4687} & \num{0.1553} & \num{0.9551} & \num{0.2851} \\
$R_Y^2/R_D^2$ & \num{0.8441} / \num{0.3562} & \num{0.8539} / \num{0.3348} & \num{0.8566} / \num{0.3434} & \num{-157.8101} / \num{-68.2209} & \num{-138.9729} / \num{-85.0758} & \num{-103.6923} / \num{-88.4491} \\
\midrule
\multicolumn{7}{l}{\textit{BERT Small Mean (524 features)}} \\
Estimate (s.e.) & \num{-0.0263} (\num{0.0217}) & \num{-0.0285} (\num{0.0204}) & \num{-0.0263} (\num{0.0135}) & \num{-0.2543} (\num{0.1082}) & \num{-0.0920} (\num{0.1658}) & \num{-0.1024} (\num{0.1391}) \\
$p$-value & \num{0.2257} & \num{0.1627} & \num{0.0520} & \num{0.0187} & \num{0.5791} & \num{0.4613} \\
$R_Y^2/R_D^2$ & \num{0.8508} / \num{0.3509} & \num{0.8596} / \num{0.3683} & \num{0.8621} / \num{0.3677} & \num{0.3904} / \num{0.3335} & \num{0.2605} / \num{0.2858} & \num{0.2299} / \num{0.3165} \\
\midrule
\multicolumn{7}{l}{\textit{BERT Small CLS (524 features)}} \\
Estimate (s.e.) & \num{-0.0422} (\num{0.0271}) & \num{-0.0376} (\num{0.0250}) & \num{-0.0430} (\num{0.0212}) & \num{-0.1729} (\num{0.1320}) & \num{-0.1670} (\num{0.1710}) & \num{-0.1181} (\num{0.1450}) \\
$p$-value & \num{0.1198} & \num{0.1331} & \num{0.0424} & \num{0.1904} & \num{0.3286} & \num{0.4155} \\
$R_Y^2/R_D^2$ & \num{0.8421} / \num{0.3475} & \num{0.8567} / \num{0.3602} & \num{0.8572} / \num{0.3457} & \num{0.2968} / \num{0.3145} & \num{0.1840} / \num{0.3201} & \num{0.1768} / \num{0.3372} \\
\midrule
\multicolumn{7}{l}{\textit{BERT Base Mean (780 features)}} \\
Estimate (s.e.) & \num{-0.0369} (\num{0.0199}) & \num{-0.0374} (\num{0.0175}) & \num{-0.0334} (\num{0.0216}) & \num{-0.0571} (\num{0.1384}) & \num{-0.0396} (\num{0.1887}) & \num{-0.0287} (\num{0.1537}) \\
$p$-value & \num{0.0635} & \num{0.0326} & \num{0.1225} & \num{0.6800} & \num{0.8339} & \num{0.8519} \\
$R_Y^2/R_D^2$ & \num{0.8467} / \num{0.4011} & \num{0.8566} / \num{0.4109} & \num{0.8605} / \num{0.3847} & \num{0.2608} / \num{0.2938} & \num{0.2423} / \num{0.2819} & \num{0.2185} / \num{0.3107} \\
\midrule
\multicolumn{7}{l}{\textit{BERT Base CLS (780 features)}} \\
Estimate (s.e.) & \num{-0.0444} (\num{0.0330}) & \num{-0.0422} (\num{0.0214}) & \num{-0.0355} (\num{0.0221}) & \num{0.0138} (\num{0.1690}) & \num{0.0981} (\num{0.2110}) & \num{0.1057} (\num{0.2120}) \\
$p$-value & \num{0.1780} & \num{0.0492} & \num{0.1077} & \num{0.9348} & \num{0.6419} & \num{0.6181} \\
$R_Y^2/R_D^2$ & \num{0.8430} / \num{0.3522} & \num{0.8530} / \num{0.3897} & \num{0.8575} / \num{0.3466} & \num{0.2778} / \num{0.1560} & \num{0.1966} / \num{0.2298} & \num{0.1613} / \num{0.1600} \\
\midrule
\multicolumn{7}{l}{\textit{BERT Large Mean (1,036 features)}} \\
Estimate (s.e.) & \num{-0.0376} (\num{0.0317}) & \num{-0.0419} (\num{0.0253}) & \num{-0.0255} (\num{0.0184}) & \num{-0.1140} (\num{0.1142}) & \num{-0.1645} (\num{0.2040}) & \num{-0.0450} (\num{0.1447}) \\
$p$-value & \num{0.2360} & \num{0.0978} & \num{0.1661} & \num{0.3183} & \num{0.4201} & \num{0.7559} \\
$R_Y^2/R_D^2$ & \num{0.8459} / \num{0.3293} & \num{0.8589} / \num{0.3281} & \num{0.8615} / \num{0.3175} & \num{0.3068} / \num{0.3094} & \num{0.2174} / \num{0.2702} & \num{0.2111} / \num{0.3017} \\
\midrule
\multicolumn{7}{l}{\textit{BERT Large CLS (1,036 features)}} \\
Estimate (s.e.) & \num{-0.0451} (\num{0.0399}) & \num{-0.0274} (\num{0.0254}) & \num{-0.0321} (\num{0.0312}) & \num{0.0471} (\num{0.2022}) & \num{0.1126} (\num{0.2679}) & \num{0.1571} (\num{0.2075}) \\
$p$-value & \num{0.2584} & \num{0.2814} & \num{0.3028} & \num{0.8160} & \num{0.6743} & \num{0.4489} \\
$R_Y^2/R_D^2$ & \num{0.8422} / \num{0.2944} & \num{0.8539} / \num{0.3145} & \num{0.8555} / \num{0.3205} & \num{0.1998} / \num{0.2629} & \num{0.1360} / \num{0.1884} & \num{0.0906} / \num{0.2166} \\
\bottomrule
\end{tabular}%
}%
}
\par
\begin{minipage}{\linewidth}
\vspace{0.5em}
\footnotesize
\textit{Notes:} Requester-level cross-fitting uses 3 folds and 5 repetitions; 258,352 observations and 24,923 requesters. Estimates are median coefficients on log reward in the log-duration equation. Parentheses report requester-clustered standard errors with split-deviation inflation. Two-sided normal $p$-values test a zero coefficient. The final row lists out-of-sample outcome and treatment $R^2$ scores, averaged across repetitions. Dube uses all 611 numerically unique controls. BERT embeddings include 12 non-text controls.
\end{minipage}
\end{table}

These patterns highlight the importance of both the choice of text embeddings and the selection of the nuisance learner for accurate reward prediction. As our theory suggests the tradeoff of the embedding size between the approximation errors and model complexity, the choice of text embeddings should balance these two considerations to optimize prediction performance. In practice, we recommend comparing multiple embedding sizes and pooling strategies based on the out-of-fold prediction performance measured by $R^2$.
Moreover, BERT embeddings with mean pooling tend to perform better than those with CLS pooling for reward prediction, which is also consistent with our intuition discussed in \cref{sec:transferability_transformer}.
Finally, we emphasize that the BERT embeddings do not require the hand-engineered features used in the Dube controls for accurate reward prediction. By using the modern text representations provided by BERT, we can achieve high prediction performance without relying on additional manually crafted features.

The XGBoost coefficients are comparatively stable across text representations.
Across the six BERT specifications and three XGBoost learners, the duration coefficient ranges from $-0.0451$ to $-0.0255$, corresponding to labor supply elasticities between $0.0255$ and $0.0451$.
With BERT Base mean embeddings, XGB 1--3 give coefficients of $-0.0369$, $-0.0374$, and $-0.0334$, with standard errors $0.0199$, $0.0175$, and $0.0216$.
Their coefficient $p$-values are $0.0635$, $0.0326$, and $0.1225$, respectively. In the subsequent analysis, we focus on the XGBoost specifications with BERT Base mean embeddings.\footnote{Formally, we also need to consider the model uncertainty among both learners and text representations, but for simplicity, we simply concentrate on this particular combination.}

\subsection{Robustness to Pre-Trained Embeddings}
Next, we use a transferability diagnostic to assess the sensitivity of the estimated duration coefficient to different pre-trained embeddings.
We compare the original BERT Base model with the $25$ MultiBERT reproductions of \cite{sellam2022the}, which use different training random seeds governing the initial values and the order of observations used in the optimization process.
All $26$ embeddings, including the original BERT embeddings, use the same architecture, mean pooling, and $12$ non-text controls.
We apply the sensitivity diagnostic separately to XGB 1--3.

We adapt the bootstrap in \cref{sec:test_transferability} to requester-clustered cross-fitting and pool contributions across folds. This implementation is not covered by \cref{thm:transferability_test_size}, which assumes i.i.d. target observations and a single nuisance-training block independent of both estimation blocks. Extensions to pooled cross-fitting and to clustered observations require separate justification and are left to future work.
Within each held-out fold, requesters are split into estimation blocks $\mathcal{I}_1$ and $\mathcal{I}_2$, sharing nuisance models trained on the complement $\mathcal{I}_3$.
After pooling block residuals across folds, each repetition tests the maximum absolute studentized contrast between embedding $j$ on $\mathcal{I}_1$ and $k$ on $\mathcal{I}_2$ for $j<k$, and repeats the test with the block assignments reversed.
Influence contributions are summed within requesters, with one Gaussian multiplier per requester shared across embeddings, to preserve dependence among contributions from the same requester in the bootstrap.
We use $2{,}000$ bootstrap draws per DML cross-fitting repetition and testing orientation, yielding ten $p$-values from five cross-fitting repetitions and two testing orientations.
For each learner, we aggregate these same ten $p$-values using the grid-harmonic and corrected upper-median rules defined in \cref{subsec:simulation_transferability_test}.

\cref{tab:transferability-xgb} reports the resulting aggregated $p$-values for XGB 1--3.
The grid-harmonic merger yields larger $p$-values for XGB 1 and XGB 2 but a smaller value for XGB 3.
Neither diagnostic rejects at the nominal $5\%$ or $10\%$ threshold for any learner. \begin{table}[htbp]
\centering
\begin{minipage}{\linewidth}
\centering
\small
\caption{Transferability-test $p$-values}
\label{tab:transferability-xgb}
\begin{tabular}{lrr}
\toprule
Learner & Median & Harmonic \\
\midrule
XGB 1 & \num{0.6817} & \num{1.0000} \\
XGB 2 & \num{0.4875} & \num{0.8567} \\
XGB 3 & \num{0.6550} & \num{0.6078} \\
\bottomrule
\end{tabular}
\par
\raggedright
\footnotesize
\vspace{0.5em}
\textit{Notes:} Test for transferability compares BERT Original and 25 MultiBERT seeds. Both mergers use the same 10 bootstrap $p$-values from 5 cross-fitting repetitions and both testing orientations with swapped blocks. The corrected upper median is $\min\{1,(10/6)p_{(6)}\}$, where $p_{(6)}$ is the order statistic of rank 6. The grid-harmonic merger follows \cite{vovk2022admissible}.
\end{minipage}
\end{table}

Given non-rejection under both aggregation methods, we do not detect strong evidence for violations of the transferability condition.
In that case, to reduce sensitivity to the choice of any one embedding, we suggest reporting the aggregated coefficient and its corresponding confidence interval. If the test for transferability had rejected, such aggregation is misleading, and we suggest reporting the range of coefficients across embeddings instead.
\cref{fig:xgb-2-confidence-intervals} displays the XGB 2 coefficients and pointwise $90\%$ and $95\%$ normal confidence intervals for each embedding.
The corresponding figures for XGB 1 and XGB 3 appear in \cref{sec:additional_embedding_sensitivity}.
The estimates use the full-sample cross-fitting procedure described above.
Under transferability and the regularity conditions for requester-clustered DML, the  estimates have the same first-order expansion across embeddings and cross-fitting repetitions.
Their median therefore has the same first-order distribution as each individual estimate.
We report the median of all $130$ estimates and apply the variance aggregation rule described above to all $130$ fits.
\cref{fig:xgb-2-confidence-intervals} shows that BERT Original's coefficient is generally close to the median of all embeddings; however, its confidence intervals are on the lower side compared to those of other embeddings. Consequently, if the researcher reports BERT Original's coefficient and confidence intervals only, they might overstate the evidence for its effect. By considering the aggregated median coefficient and its confidence intervals, the researcher can provide a more robust summary of the embedding-specific variations.

The negative median coefficient in \cref{fig:xgb-2-confidence-intervals} corresponds to a small positive estimate of the labor supply elasticity, $-\theta^*$.
Both 90\% and 95\% confidence intervals for the aggregate include zero.
Under the monopsony model, this point estimate suggests inelastic labor supply with respect to the posted reward after controlling for job characteristics, including text, and hence potential employer monopsony power in this online labor market.

\begin{figure}[p]
	\centering
	\includegraphics[width=\linewidth,height=0.84\textheight,keepaspectratio]{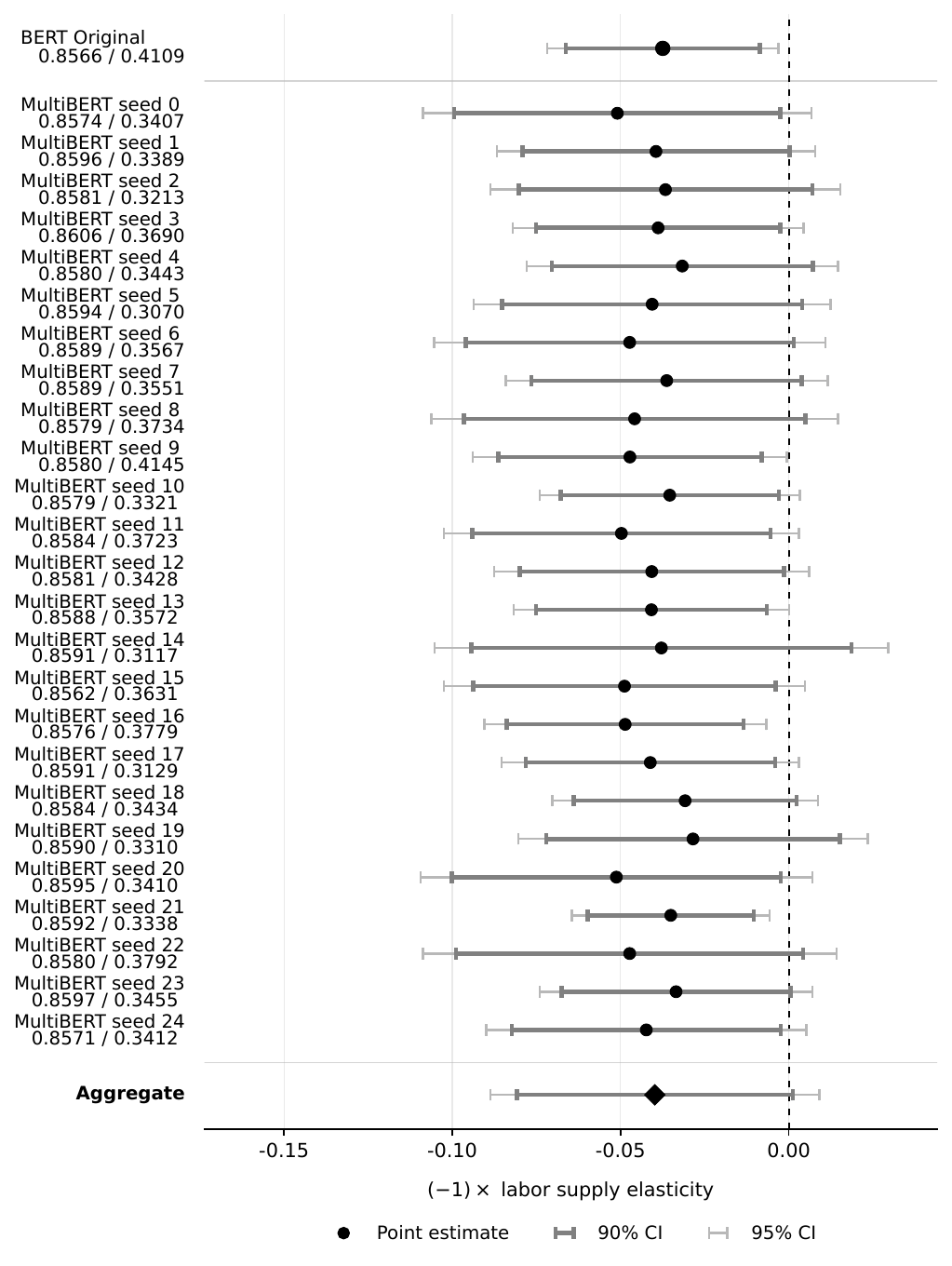}
	\caption{Coefficient robustness to 26 BERT embeddings (Base BERT with Mean pooling) with XGB 2 learners}
	\label{fig:xgb-2-confidence-intervals}
	\captionsetup{font=footnotesize,justification=justified,singlelinecheck=false}
	\caption*{\textit{Notes:} Round points show duration coefficients by DML; gray and light-gray bars show pointwise $90\%$ and $95\%$ confidence intervals using requester-clustered standard errors, respectively. Two numbers under the label report duration/reward out-of-fold $R^2$ averaged across cross-fitting repetitions. The diamond shows the pooled median estimate, with $90\%$ and $95\%$ confidence intervals for the common duration coefficient.}
\end{figure}

\section{Conclusion}
\label{sec:conclusion}
This paper establishes sufficient conditions for valid DML inference with pre-trained embeddings, accounting for source estimation error and partial identification of the embedding function.
The derived convergence bounds depend on three components: target estimation error with the pre-trained embeddings treated as standard covariates, source estimation error adjusted for the strength of transferability, and approximation error of the target model class using embeddings as inputs. We also develop a computationally feasible multiplier bootstrap test to detect violations of transferability by comparing DML estimators across embeddings trained for the same source task with different random seeds. The empirical application motivates two empirical practices: jointly selecting model classes for the pre-trained embedding function and downstream nuisance learners based on out-of-sample performance in the target task, and assessing and reporting the robustness of DML estimates across embeddings pre-trained for the same source task using different random seeds.

\bibliographystyle{ecta}
\bibliography{listb}

@article{tripuraneni2020theory,
  title={On the theory of transfer learning: The importance of task diversity},
  author={Tripuraneni, Nilesh and Jordan, Michael and Jin, Chi},
  journal={Advances in Neural Information Processing Systems},
  volume={33},
  pages={7852--7862},
  year={2020}
}

@article{watkins2023optimistic,
  title={Optimistic rates for multi-task representation learning},
  author={Watkins, Austin and Ullah, Enayat and Nguyen-Tang, Thanh and Arora, Raman},
  journal={Advances in Neural Information Processing Systems},
  volume={36},
  pages={2207--2251},
  year={2023}
}

@article{angelopoulos2023prediction,
  title={Prediction-powered inference},
  author={Angelopoulos, Anastasios N and Bates, Stephen and Fannjiang, Clara and Jordan, Michael I and Zrnic, Tijana},
  journal={Science},
  volume={382},
  number={6671},
  pages={669--674},
  year={2023},
  publisher={American Association for the Advancement of Science}
}

@article{carlson2025unifying,
  title={A Unifying Framework for Robust and Efficient Inference with Unstructured Data},
  author={Carlson, Jacob and Dell, Melissa},
  journal={arXiv preprint arXiv:2505.00282},
  year={2025}
}

@article{battaglia2025inference,
  title={Inference for Regression with Variables Generated by {AI} or Machine Learning},
  author={Battaglia, Laura and Christensen, Timothy and Hansen, Stephen and Sacher, Szymon},
  journal={arXiv preprint arXiv:2402.15585},
  year={2025}
}

@article{zhang2026debiasing,
  title={Debiasing {ML}-or {AI}-generated regressors in partially linear models},
  author={Zhang, Jingwen and Xue, Wendao and Yu, Yifan and Tan, Yong},
  journal={Information Systems Research},
  year={2026},
  publisher={INFORMS}
}

@article{rambachan2024program,
  title={Program Evaluation with Remotely Sensed Outcomes},
  author={Rambachan, Ashesh and Singh, Rahul and Viviano, Davide},
  journal={arXiv preprint arXiv:2411.10959},
  year={2024}
}

@article{chernozhukov2018double,
  title={Double/debiased machine learning for treatment and structural parameters},
  author={Chernozhukov, Victor and Chetverikov, Denis and Demirer, Mert and Duflo, Esther and Hansen, Christian and Newey, Whitney and Robins, James},
  journal={The Econometrics Journal},
  pages={C1--C68},
  year={2018},
  publisher={JSTOR}
}

@article{jean2016combining,
  title={Combining satellite imagery and machine learning to predict poverty},
  author={Jean, Neal and Burke, Marshall and Xie, Michael and Alampay Davis, W Matthew and Lobell, David B and Ermon, Stefano},
  journal={Science},
  volume={353},
  number={6301},
  pages={790--794},
  year={2016},
  publisher={American Association for the Advancement of Science}
}

@article{compiani2025demand,
  title={Demand estimation with text and image data},
  author={Compiani, Giovanni and Morozov, Ilya and Seiler, Stephan},
  journal={arXiv preprint arXiv:2503.20711},
  year={2025}
}

@article{bajari2025hedonic,
  title = {Hedonic prices and quality adjusted price indices powered by AI},
  author = {P. Bajari and Z. Cen and V. Chernozhukov and M. Manukonda and S. Vijaykumar and J. Wang and R. Huerta and J. Li and L. Leng and G. Monokroussos and S. Wang},
  journal = {Journal of Econometrics},
  volume = {251},
  pages = {106052},
  year = {2025},
  issn = {0304-4076},
  publisher = {Elsevier}
}

@article{farrell2021deep,
  title={Deep neural networks for estimation and inference},
  author={Farrell, Max H and Liang, Tengyuan and Misra, Sanjog},
  journal={Econometrica},
  volume={89},
  number={1},
  pages={181--213},
  year={2021},
  publisher={Wiley Online Library}
}

@book{van2023weak,
  title={Weak Convergence and Empirical Processes: With Applications to Statistics},
  author={van der Vaart, AW and Wellner, Jon A},
  year={2023},
  publisher={Springer Nature}
}

@book{van1998asymptotic,
  title={Asymptotic statistics},
  author={Van der Vaart, Aad W},
  volume={3},
  year={1998},
  publisher={Cambridge University Press}
}

@book{wainwright2019high,
  title={High-dimensional statistics: A non-asymptotic viewpoint},
  author={Wainwright, Martin J},
  volume={48},
  year={2019},
  publisher={Cambridge University Press}
}

@article{bartlett2019nearly,
  title={Nearly-tight VC-dimension and pseudodimension bounds for piecewise linear neural networks},
  author={Bartlett, Peter L and Harvey, Nick and Liaw, Christopher and Mehrabian, Abbas},
  journal={Journal of Machine Learning Research},
  volume={20},
  number={63},
  pages={1--17},
  year={2019}
}

@article{schulte2025adjustment,
  title={Adjustment for confounding using pre-trained representations},
  author={Schulte, Rickmer and R{\"u}gamer, David and Nagler, Thomas},
  journal={arXiv preprint arXiv:2506.14329},
  year={2025}
}

@article{klaassen2024doublemldeep,
  title={Doublemldeep: Estimation of causal effects with multimodal data},
  author={Klaassen, Sven and Teichert-Kluge, Jan and Bach, Philipp and Chernozhukov, Victor and Spindler, Martin and Vijaykumar, Suhas},
  journal={arXiv preprint arXiv:2402.01785},
  year={2024}
}

@article{avivi2025patent,
  title={Are Patent Examiners Gender Neutral},
  author={Avivi, Hadar},
  journal={Unpublished manuscript},
  year={2025}
}

@article{vafa2025estimating,
  title={Estimating wage disparities using foundation models},
  author={Vafa, Keyon and Athey, Susan and Blei, David M},
  journal={Proceedings of the National Academy of Sciences},
  volume={122},
  number={22},
  pages={e2427298122},
  year={2025},
  publisher={National Academy of Sciences}
}

@article{berry1994estimating,
  title={Estimating discrete-choice models of product differentiation},
  author={Berry, Steven T},
  journal={The RAND Journal of Economics},
  pages={242--262},
  year={1994},
  publisher={JSTOR}
}

@inproceedings{devlin2019bert,
  title={Bert: Pre-training of deep bidirectional transformers for language understanding},
  author={Devlin, Jacob and Chang, Ming-Wei and Lee, Kenton and Toutanova, Kristina},
  booktitle={Proceedings of the 2019 conference of the North American chapter of the association for computational linguistics: human language technologies, volume 1 (long and short papers)},
  pages={4171--4186},
  year={2019}
}

@inproceedings{simonyan2015very,
  title = {Very deep convolutional networks for large-scale image recognition},
  author = {Simonyan, Karen and Zisserman, Andrew},
  booktitle = {3rd International Conference on Learning Representations, {ICLR} 2015},
  year = {2015}
}

@inproceedings{he2016deep,
  title={Deep residual learning for image recognition},
  author={He, Kaiming and Zhang, Xiangyu and Ren, Shaoqing and Sun, Jian},
  booktitle={Proceedings of the IEEE Conference on Computer Vision and Pattern Recognition},
  pages={770--778},
  year={2016}
}

@article{xu2021representation,
  title={Representation learning beyond linear prediction functions},
  author={Xu, Ziping and Tewari, Ambuj},
  journal={Advances in Neural Information Processing Systems},
  volume={34},
  pages={4792--4804},
  year={2021}
}

@inproceedings{grigsby2023hidden,
  title={Hidden symmetries of ReLU networks},
  author={Grigsby, Elisenda and Lindsey, Kathryn and Rolnick, David},
  booktitle={International Conference on Machine Learning},
  pages={11734--11760},
  year={2023},
  organization={PMLR}
}

@book{dudley2002real,
  title={Real analysis and probability},
  author={Dudley, Richard M},
  year={2002},
  publisher={Cambridge University Press}
}

@book{folland1999real,
  title={Real Analysis: Modern Techniques and Their Applications},
  author={Folland, Gerald B},
  edition={2},
  year={1999},
  publisher={John Wiley \& Sons}
}

@book{harville2008matrix,
  title={Matrix Algebra From a Statistician's Perspective},
  author={Harville, David A},
  year={2008},
  publisher={Springer Science \& Business Media}
}

@article{bach2024adventures,
  title={Adventures in demand analysis using AI},
  author={Bach, Philipp and Chernozhukov, Victor and Klaassen, Sven and Spindler, Martin and Teichert-Kluge, Jan and Vijaykumar, Suhas},
  journal={arXiv preprint arXiv:2501.00382},
  year={2024}
}

@article{kolesar2026fragility,
  title={The fragility of sparsity},
  author={Koles{\'a}r, Michal and M{\"u}ller, Ulrich K and Roelsgaard, Sebastian T},
  journal={arXiv preprint arXiv:2311.02299},
  year={2026}
}

@article{yarotsky2017error,
  title={Error bounds for approximations with deep ReLU networks},
  author={Yarotsky, Dmitry},
  journal={Neural Networks},
  volume={94},
  pages={103--114},
  year={2017},
  publisher={Elsevier}
}

@article{christensen2026unstructured,
  title={From Unstructured Data to Demand Counterfactuals: Theory and Practice},
  author={Christensen, Timothy and Compiani, Giovanni},
  journal={arXiv preprint arXiv:2601.05374},
  year={2026}
}

@article{escanciano2026automatic,
  title={Automatic Locally Robust {GMM} with Machine-Learning-Generated Regressors},
  author={Escanciano, Juan Carlos and P{\'e}rez-Izquierdo, Telmo},
  journal={arXiv preprint arXiv:2301.10643v4},
  year={2026}
}

@article{han2025copyright,
  title={Copyright and Competition: Estimating Supply and Demand with Unstructured Data},
  author={Han, Sukjin and Lee, Kyungho},
  journal={arXiv preprint arXiv:2501.16120},
  year={2025}
}

@article{kumar2022fine,
  title={Fine-tuning can distort pretrained features and underperform out-of-distribution},
  author={Kumar, Ananya and Raghunathan, Aditi and Jones, Robbie and Ma, Tengyu and Liang, Percy},
  journal={arXiv preprint arXiv:2202.10054},
  year={2022}
}

@article{chen2022debiased,
  title={Debiased machine learning without sample-splitting for stable estimators},
  author={Chen, Qizhao and Syrgkanis, Vasilis and Austern, Morgane},
  journal={Advances in Neural Information Processing Systems},
  volume={35},
  pages={3096--3109},
  year={2022}
}

@article{hastie2022surprises,
  title={Surprises in high-dimensional ridgeless least squares interpolation},
  author={Hastie, Trevor and Montanari, Andrea and Rosset, Saharon and Tibshirani, Ryan J},
  journal={Annals of Statistics},
  volume={50},
  number={2},
  pages={949--986},
  year={2022}
}

@inproceedings{johansson2019support,
  title={Support and invertibility in domain-invariant representations},
  author={Johansson, Fredrik D and Sontag, David and Ranganath, Rajesh},
  booktitle={The 22nd International Conference on Artificial Intelligence and Statistics},
  pages={527--536},
  year={2019},
  organization={PMLR}
}

@book{stewart1990matrix,
  address = {Boston},
  author = {Stewart, Gilbert W. and Sun, Ji-guang},
  publisher = {Academic Press},
  series = {Computer Science and Scientific Computing},
  title = {Matrix Perturbation Theory},
  year = {1990}
}

@book{horn1994topics,
  title={Topics in matrix analysis},
  author={Horn, Roger and Johnson, Charles R.},
  year={1994},
  publisher={Cambridge University Press Cambridge, UK}
}

@book{ledoux1991probability,
  title={Probability in Banach Spaces: isoperimetry and processes},
  author={Ledoux, Michel and Talagrand, Michel},
  volume={23},
  year={1991},
  publisher={Springer}
}

@book{kosorok2008introduction,
  title={Introduction to empirical processes and semiparametric inference},
  author={Kosorok, Michael R},
  year={2008},
  publisher={Springer}
}

@article{fava2025training,
  title={Training and testing with multiple splits: A central limit theorem for split-sample estimators},
  author={Fava, Bruno},
  journal={arXiv preprint arXiv:2511.04957},
  year={2025}
}

@article{ludwig2026large,
  title={Large language models: An applied econometric framework},
  author={Ludwig, Jens and Mullainathan, Sendhil and Rambachan, Ashesh},
  journal={Annual Review of Economics},
  volume={18},
  year={2026}
}

@article{pagan1984econometric,
  title={Econometric issues in the analysis of regressions with generated regressors},
  author={Pagan, Adrian},
  journal={International Economic Review},
  pages={221--247},
  year={1984},
  publisher={JSTOR}
}

@article{hahn2013asymptotic,
  title={Asymptotic variance of semiparametric estimators with generated regressors},
  author={Hahn, Jinyong and Ridder, Geert},
  journal={Econometrica},
  volume={81},
  number={1},
  pages={315--340},
  year={2013},
  publisher={Wiley Online Library}
}

@article{bhattacharya2024deep,
  title={Deep neural networks for nonparametric interaction models with diverging dimension},
  author={Bhattacharya, Sohom and Fan, Jianqing and Mukherjee, Debarghya},
  journal={Annals of Statistics},
  volume={52},
  number={6},
  pages={2738--2766},
  year={2024},
  publisher={Institute of Mathematical Statistics}
}

@article{srivastava2014dropout,
  title={Dropout: a simple way to prevent neural networks from overfitting},
  author={Srivastava, Nitish and Hinton, Geoffrey and Krizhevsky, Alex and Sutskever, Ilya and Salakhutdinov, Ruslan},
  journal={Journal of Machine Learning Research},
  volume={15},
  number={1},
  pages={1929--1958},
  year={2014},
  publisher={JMLR. org}
}

@inproceedings{wei2020implicit,
  title={The implicit and explicit regularization effects of dropout},
  author={Wei, Colin and Kakade, Sham and Ma, Tengyu},
  booktitle={International Conference on Machine Learning},
  pages={10181--10192},
  year={2020},
  organization={PMLR}
}

@article{kohler2021rate,
  title={On the rate of convergence of fully connected deep neural network regression estimates},
  author={Kohler, Michael and Langer, Sophie},
  journal={Annals of Statistics},
  volume={49},
  number={4},
  pages={2231--2249},
  year={2021},
  publisher={JSTOR}
}

@article{jiao2023deep,
  title={Deep nonparametric regression on approximate manifolds: Nonasymptotic error bounds with polynomial prefactors},
  author={Jiao, Yuling and Shen, Guohao and Lin, Yuanyuan and Huang, Jian},
  journal={Annals of Statistics},
  volume={51},
  number={2},
  pages={691--716},
  year={2023},
  publisher={Institute of Mathematical Statistics}
}

@article{modarressi2025causal,
  title={Causal inference on outcomes learned from text},
  author={Modarressi, Iman and Spiess, Jann and Venugopal, Amar},
  journal={arXiv preprint arXiv:2503.00725},
  year={2025}
}

@article{fong2021machine,
  title={Machine learning predictions as regression covariates},
  author={Fong, Christian and Tyler, Matthew},
  journal={Political Analysis},
  volume={29},
  number={4},
  pages={467--484},
  year={2021},
  publisher={Cambridge University Press}
}

@article{egami2023using,
  title={Using imperfect surrogates for downstream inference: Design-based supervised learning for social science applications of large language models},
  author={Egami, Naoki and Hinck, Musashi and Stewart, Brandon and Wei, Hanying},
  journal={Advances in Neural Information Processing Systems},
  volume={36},
  pages={68589--68601},
  year={2023}
}

@article{carlson2025making,
  title={Making Interpretable Discoveries from Unstructured Data: A High-Dimensional Multiple Hypothesis Testing Approach},
  author={Carlson, Jacob},
  journal={arXiv preprint arXiv:2511.01680},
  year={2025}
}

@article{chen2008semiparametric,
  title={Semiparametric Efficiency in GMM Models with Auxiliary Data},
  author={Chen, Xiaohong and Hong, Han and Tarozzi, Alessandro},
  journal={Annals of Statistics},
  pages={808--843},
  year={2008},
  publisher={JSTOR}
}

@article{robinson1988root,
  title={Root-N-consistent semiparametric regression},
  author={Robinson, Peter M},
  journal={Econometrica},
  pages={931--954},
  year={1988},
  publisher={JSTOR}
}

@article{rubin1976inference,
  title={Inference and missing data},
  author={Rubin, Donald B},
  journal={Biometrika},
  volume={63},
  number={3},
  pages={581--592},
  year={1976},
  publisher={Oxford University Press}
}

@article{rosenbaum1983central,
  title={The central role of the propensity score in observational studies for causal effects},
  author={Rosenbaum, Paul R and Rubin, Donald B},
  journal={Biometrika},
  volume={70},
  number={1},
  pages={41--55},
  year={1983},
  publisher={Oxford University Press}
}

@article{pan2009survey,
  title={A survey on transfer learning},
  author={Pan, Sinno Jialin and Yang, Qiang},
  journal={IEEE Transactions on Knowledge and Data Engineering},
  volume={22},
  number={10},
  pages={1345--1359},
  year={2010},
  publisher={IEEE}
}

@article{hinton2006reducing,
  title={Reducing the dimensionality of data with neural networks},
  author={Hinton, Geoffrey E and Salakhutdinov, Ruslan R},
  journal={Science},
  volume={313},
  number={5786},
  pages={504--507},
  year={2006},
  publisher={American Association for the Advancement of Science}
}

@article{mammen2016semiparametric,
  title={Semiparametric estimation with generated covariates},
  author={Mammen, Enno and Rothe, Christoph and Schienle, Melanie},
  journal={Econometric Theory},
  volume={32},
  number={5},
  pages={1140--1177},
  year={2016},
  publisher={Cambridge University Press}
}

@article{murphy1985estimation,
  title={Estimation and Inference in Two-Step Econometric Models},
  author={Murphy, Kevin M and Topel, Robert H},
  journal={Journal of Business \& Economic Statistics},
  volume={3},
  number={4},
  pages={370--379},
  year={1985},
  publisher={Taylor \& Francis}
}

@article{dube2020monopsony,
  title={Monopsony in online labor markets},
  author={Dube, Arindrajit and Jacobs, Jeff and Naidu, Suresh and Suri, Siddharth},
  journal={American Economic Review: Insights},
  volume={2},
  number={1},
  pages={33--46},
  year={2020},
  publisher={American Economic Association 2014 Broadway, Suite 305, Nashville, TN 37203}
}

@article{magnolfi2025triplet,
  title={Triplet embeddings for demand estimation},
  author={Magnolfi, Lorenzo and McClure, Jonathon and Sorensen, Alan},
  journal={American Economic Journal: Microeconomics},
  volume={17},
  number={1},
  pages={282--307},
  year={2025},
  publisher={American Economic Association 2014 Broadway, Suite 305, Nashville, TN 37203-2425}
}

@article{newey1997convergence,
  title={Convergence rates and asymptotic normality for series estimators},
  author={Newey, Whitney K},
  journal={Journal of econometrics},
  volume={79},
  number={1},
  pages={147--168},
  year={1997},
  publisher={Elsevier}
}

@techreport{duan2026inference,
  title={Inference with AI-Generated Covariates},
  author={Duan, Junting and Pelger, Markus},
  year={2026},
  institution={National Bureau of Economic Research}
}

@techreport{casella2026structural,
  title={Structural Estimation with Unstructured Data},
  author={Casella, Sara and Fern{\'a}ndez-Villaverde, Jes{\'u}s and Hansen, Stephen and Oishi, Ryohei and Shin, Minchul},
  year={2026},
  institution={National Bureau of Economic Research}
}

@article{ahrens2026introduction,
  title={An Introduction to Double/Debiased Machine Learning},
  author={Ahrens, Achim and Chernozhukov, Victor and Hansen, Christian and Kozbur, Damian and Schaffer, Mark E. and Wiemann, Thomas},
  journal={Journal of Economic Literature},
  year={2026},
  doi={10.1257/jel.20261758}
}

@article{mcinnes2018umap,
  title={{UMAP}: Uniform manifold approximation and projection for dimension reduction},
  author={McInnes, Leland and Healy, John and Melville, James},
  journal={arXiv preprint arXiv:1802.03426},
  year={2018}
}

@inproceedings{
sellam2022the,
title={The Multi{BERT}s: {BERT} Reproductions for Robustness Analysis},
author={Thibault Sellam and Steve Yadlowsky and Ian Tenney and Jason Wei and Naomi Saphra and Alexander D'Amour and Tal Linzen and Jasmijn Bastings and Iulia Raluca Turc and Jacob Eisenstein and Dipanjan Das and Ellie Pavlick},
booktitle={International Conference on Learning Representations},
year={2022}
}

@article{Meinshausen2009p,
author = {Nicolai Meinshausen and Lukas Meier and Peter Bühlmann},
title = {p-Values for High-Dimensional Regression},
journal = {Journal of the American Statistical Association},
volume = {104},
number = {488},
pages = {1671--1681},
year = {2009},
publisher = {Taylor \& Francis},
doi = {10.1198/jasa.2009.tm08647}
}

@article{vovk2022admissible,
  title={Admissible ways of merging p-values under arbitrary dependence},
  author={Vovk, Vladimir and Wang, Bin and Wang, Ruodu},
  journal={The Annals of Statistics},
  volume={50},
  number={1},
  pages={351--375},
  year={2022},
  publisher={Institute of Mathematical Statistics}
}

@article{fernandez2026pragmatic,
  title = {Pragmatic {DML} with {AI}-Learned Representations},
  author = {Fern{\'a}ndez, Andr{\'e}s Aradillas and Chernozhukov, Victor and Cinelli, Carlos and Klaassen, Sven and Newey, Whitney and Spindler, Martin and Teichert-Kluge, Jan and Vijaykumar, Suhas},
  year          = {2026},
  eprint        = {2610.01935},
  archivePrefix = {arXiv}
}

@article{chernozhukov2026long,
  title={Long story short: Omitted variable bias in causal machine learning},
  author={Chernozhukov, Victor and Cinelli, Carlos and Newey, Whitney K and Sharma, Amit and Syrgkanis, Vasilis},
  journal={The Review of Economics and Statistics},
  pages={1--45},
  year={2026}
}

@article{li2026econometric,
  title={Econometric Inference with Machine-Learned Proxies: Partial Identification via Data Combination},
  author={Li, Lixiong},
  journal={arXiv preprint arXiv:2604.10770},
  year={2026}
}

@article{dell2025deep,
  title={Deep learning for economists},
  author={Dell, Melissa},
  journal={Journal of Economic Literature},
  volume={63},
  number={1},
  pages={5--58},
  year={2025},
  publisher={American Economic Association}
}


\newpage
\appendix
\etocdepthtag.toc{appendix}

\begin{center}
{\LARGE Supplement to ``Econometrics with Pre-trained Embeddings for~Unstructured~Data''\par}
\vspace{0.5em}
{\large Yuya Shimizu\par}
{\large University of Wisconsin-Madison\par}
\end{center}

\begingroup
  \small
  \setlength{\parskip}{0pt}
  \setlength{\parindent}{0pt}
  \makeatletter
  \renewcommand{\@dotsep}{1.5}
  \renewcommand{\l@section}{\@dottedtocline{1}{0em}{2.2em}}
  \renewcommand{\l@subsection}{\@dottedtocline{2}{2.2em}{3.0em}}
  \makeatother
  \etocsettagdepth{main}{none}
  \etocsettagdepth{appendix}{2}
  \etocsettocstyle{%
    \subsection*{Contents}%
    \par\noindent\hrule\par
    \vspace{0.4\baselineskip}%
  }{%
    \vspace{0.4\baselineskip}%
    \par\noindent\hrule\par
  }%
  \tableofcontents
\endgroup

\bigskip

This supplementary appendix contains proofs of the results in the main text as well as auxiliary results.

\paragraph{Additional Notation}
In the supplementary appendix, we use the following notation in addition to the notation introduced in the main text.
For deterministic sequences $a_n$ and $b_n$ and for a scalar $\varepsilon$, we write $a_n \lesssim_\varepsilon b_n$ if there exists a constant $C_\varepsilon>0$, depending only on $\varepsilon$, such that $a_n \leq C_\varepsilon b_n$ for large $n$. 
In the proofs, $C_{\varepsilon,1}, C_{\varepsilon,2}, \dots$ denote positive constants that depend only on $\varepsilon$, while $C_1, C_2, \dots$ denote positive absolute constants. These constants may denote different values in different proof sections.
$P^*$ and $\mathbb{E}^*$ denote the outer probability and outer expectation, respectively.
For two symmetric matrices $A\in\mathbb{R}^{K\times K}$ and $\tilde{A}\in\mathbb{R}^{K\times K}$, we write $A\preceq\tilde{A}$ if $\tilde{A}-A$ is positive semi-definite.

\section{Related Results in the Literature}
\label{sec:related_lit}
\subsection{Transfer Learning in the Machine Learning Literature}
\label{sec:diversity_lit}

While a version of the diversity condition in \cref{def:diversity} is considered in the machine learning literature on the excess-risk scale (e.g., \citealp{tripuraneni2020theory}; \citealp{xu2021representation}), we characterize the diversity condition by a novel relationship between the diversity condition and the set inclusion condition on the near-identified sets of the embedding function.
Although the machine learning literature provides some sufficient conditions for task diversity and rate bounds for transfer learning, this paper provides results that are more general and tailored to economic applications in the following three aspects.

First, our sufficient condition in \cref{thm:suf_transfer} is more flexible than the existing results. The existing results focus on the full-rank coefficient matrix $B_{\so,m}$ for the linear source head, but our sufficient condition does not require the full-rank condition on $B_{\so,m}$ and allows for rank deficiency of $B_{\so,m}$ as long as the row span of $B_{\so,m}$ includes the target coefficient $\beta_{\ta,n}$.
In addition, our sufficient condition in \cref{thm:suf_transfer} is more flexible than the existing results as it allows different source and target covariate distributions and more general relationships between source and target models, including Lipschitz transformations.

Second, the necessity results in \cref{thm:nec_transfer_linear_pre,thm:nec_transfer} are new to the best of our knowledge. For linear heads, \cref{thm:nec_transfer_linear_pre} establishes the necessity of the row-span condition for $(0,0)$-diversity under its stated class and regularity assumptions. For general heads, \cref{thm:nec_transfer} establishes necessary measurable factorization through the source prediction for $(0,0)$-diversity under its assumptions, including the sigma-field inclusion.

Finally, for general source and target heads, \cref{thm:target_convergence} gives conditions for the $o_P(n^{-1/4})$ rate of overall target prediction error used in our DML framework.
Theorem 3 of \cite{tripuraneni2020theory} gives a target excess-risk bound that contains a term of order $n^{-1/2}$ at fixed confidence.
If the remaining terms are also $O(n^{-1/2})$, the curvature condition in \cref{asm:loss_ta}(i), applied to the same target-risk benchmark, gives an $O_P(n^{-1/4})$ bound on the overall $L^2(P_{\ta})$ prediction error. Their bound alone therefore does not establish the strict rate used in our DML analysis.
Theorems 1 and 2 of \cite{watkins2023optimistic} derive excess-risk bounds for the target model that depend on the optimal target risk and the minimum average source risk, with Theorem 2 adding smoothness and regularity conditions to bound the complexity terms in Theorem 1.
Compare the overall bounds at the same target and total source sample sizes, holding the number of source tasks and the embedding dimension fixed.
Under several conditions, our overall target excess-risk bound from \cref{thm:target_convergence,thm:source_convergence} and \cref{asm:loss_ta}(i) matches their fast rate up to logarithmic factors when both bounds are expressed in terms of the same VC dimensions.\footnote{For the comparison, assume the source heads have a common Lipschitz bound and that source loss curvature, our transferability strength, and their task-diversity coefficient are bounded away from zero.
Suppose also that the squared approximation errors in our bound and the additive diversity error in theirs are of no larger order than the combined estimation terms.
}
Their bound attains this rate when each optimal risk is at most of the same order as its corresponding complexity term. Our bound allows both optimal risks to be bounded away from zero.
Vanishing optimal risk can be restrictive in economic applications because, under squared loss, it requires the expected conditional variance of the outcome to converge to zero.

\subsection{Double Machine Learning with Generated Regressors}
\label{sec:escanciano}
\cite{escanciano2026automatic} study locally robust estimation with generated regressors. While their theory suggests that the generated regressors affect inference on the parameter of interest in general, this effect of generated regressors does not arise in our setup.
Return to the setup in \cref{sec:dml}, where we have a target parameter $\theta^*$ and nuisance functions $\eta^*=(\gamma^*,\alpha^*)$. The goal is to estimate $\theta^*$ using the estimated nuisance functions $\hat{\eta} = (\hat{\gamma},\hat{\alpha})$ while accounting for the estimation error in the embeddings.

To clarify the relationship, we first consider the case with no approximation error in the target task $\epsilon_{\ta,n} = 0$, and we next consider the case with approximation error.
For exposition, we also focus on a partially linear model (\cref{ex:partially_linear}) with squared-loss nuisance training, so that $\Lambda_{\gamma}$ and $\Lambda_{\alpha}$ are identity maps. The same logic applies to the other applications with closed-form Neyman orthogonal scores and nuisance functions taking conditional expectation forms, including \cref{ex:demand,ex:imputation,ex:ate}.

With no approximation error $\epsilon_{\ta,n} = 0$, the population nuisance function in the target task is $\eta^*=f_{\eta,n}\circ h_m$.
The partially linear model in \cref{ex:partially_linear} has $\eta^*(Z_i)=(\gamma^*(Z_i), \alpha^*(Z_i))=(\mathbb{E}_{\ta}[Y_i|Z_i], \mathbb{E}_{\ta}[X_i|Z_i])$ and $f_{\eta,n}\circ h_m(Z_i)=(f_{\gamma,n}\circ h_m(Z_i), f_{\alpha,n}\circ h_m(Z_i)) = (\mathbb{E}_{\ta}[Y_i|h_m(Z_i)], \mathbb{E}_{\ta}[X_i|h_m(Z_i)])$.
Thus, our setup considers the case where $\mathbb{E}_{\ta}[Y_i|Z_i]=\mathbb{E}_{\ta}[Y_i|h_m(Z_i)]$ and $\mathbb{E}_{\ta}[X_i|Z_i]=\mathbb{E}_{\ta}[X_i|h_m(Z_i)]$.
On the other hand, \cite{escanciano2026automatic} consider a more general setup allowing $\mathbb{E}_{\ta}[Y_i|Z_i]\neq\mathbb{E}_{\ta}[Y_i|h_m(Z_i)]=f_{\gamma,n}\circ h_m(Z_i)$ and $\mathbb{E}_{\ta}[X_i|Z_i]\neq\mathbb{E}_{\ta}[X_i|h_m(Z_i)]=f_{\alpha,n}\circ h_m(Z_i)$.
In our setting, $h_m(Z_i)$ captures sufficient statistics of $Z_i$ for the target task, making the Neyman-orthogonal score insensitive to small perturbations in both $f_{\eta}$ and $h$.
In the more general setting considered by \cite{escanciano2026automatic}, the same Neyman-orthogonal score can be sensitive to perturbations in generated regressors $h(Z_i)$.

More concretely, recall that the Neyman orthogonality condition is defined as follows:
\begin{definition}[Neyman Orthogonality]
    A moment function $\psi(W; \theta, \eta)$ is {\it Neyman orthogonal} at $(\theta^*, \eta^*)$ if
    \begin{equation*}
        \left.\frac{\partial}{\partial r} \mathbb{E}_{\ta}[\psi(W; \theta^*, \eta^* + r(\eta - \eta^*))]\right|_{r=0} = 0
    \end{equation*}
    for all $\eta$ in some neighborhood of $\eta^*$.
\end{definition}
Consider a Neyman-orthogonal score for the partially linear model:
$\psi(W_i;\theta,\gamma,\alpha)=\{(Y_i-\gamma(Z_i))-(X_i-\alpha(Z_i))\theta\}(X_i-\alpha(Z_i))$.
Under the regularity conditions for differentiability,
\begin{align*}
    &\left.\frac{\partial}{\partial r} \mathbb{E}_{\ta}[\psi(W; \theta^*, f_{\eta,n} \circ h_m + r(f_{\eta} \circ h - f_{\eta,n} \circ h_m))]\right|_{r=0}\\
    & \quad = \mathbb{E}_{\ta}[\{(f_{\gamma,n} \circ h_m(Z) - f_{\gamma}\circ h(Z)) - (f_{\alpha,n} \circ h_m(Z) - f_{\alpha}\circ h(Z))\theta^*\}(X-f_{\alpha,n} \circ h_m(Z))]\\
    & \qquad + \mathbb{E}_{\ta}[\{(Y - f_{\gamma,n} \circ h_m(Z))-(X-f_{\alpha,n} \circ h_m(Z))\theta^*\}(f_{\alpha,n} \circ h_m(Z) - f_{\alpha} \circ h(Z))].
\end{align*}
By the law of iterated expectations, the derivative equals zero in our setup by taking conditional expectation given $Z$.
Without these conditional-mean sufficiency restrictions, orthogonality with respect to the second-step regressions does not generally eliminate the direct evaluation effect induced by perturbing the generated regressor. This is the effect addressed by the first-step correction in \cite{escanciano2026automatic}.\footnote{
    While \cite{escanciano2026automatic} propose debiasing estimators for lower-dimensional generated regressors, debiasing for high-dimensional generated regressors is not studied in the literature.
}
Our setup is the DML analog of the index restriction in the literature on semiparametric models (\citealp{hahn2013asymptotic,mammen2016semiparametric}).

When the embedding sufficiency is approximate ($\epsilon_{\ta,n}>0$), Neyman orthogonality remains defined at the true nuisance vector $\eta^*$, rather than at its approximation $f_{\eta,n}\circ h_m$. Approximation and source-task estimation errors enter through the total errors of the composite nuisance estimators. Standard DML inference therefore remains valid when these total errors satisfy the required product-rate and sample-splitting conditions. In this sense, the two approaches are complementary. Our analysis derives source-to-target convergence rates for independently pre-trained embeddings, whereas \cite{escanciano2026automatic} construct influence-function corrections for direct first-step effects that are not absorbed by orthogonality of the composite nuisance.

\section{Pre-trained Deep Learning}
\label{sec:pretrained_dl}
\subsection{Fully Connected Feedforward Neural Networks (FNNs)}
Let us consider a fully connected feedforward neural network with $L_{\mathrm{NN}}$ layers. The input layer is denoted as layer $0$, and the output layer is layer $L_{\mathrm{NN}}$. Each layer $\ell$ has $H_{\ell}$ hidden units (or neurons), and the output dimension is $T$. The weights connecting layer $\ell-1$ to layer $\ell$ are represented by a weight matrix $W^{(\ell)} \in \mathbb{R}^{H_{\ell} \times H_{\ell-1}}$, where $H_0 = d$ and $H_{L_{\mathrm{NN}}} = T$. The constant terms (``biases'') for layer $\ell$ are represented by a vector $b^{(\ell)} \in \mathbb{R}^{H_{\ell}}$.
Let $\sigma(\cdot)$ denote an element-wise activation function, such as the ReLU $\sigma(a)=\max\{a,0\}$. To avoid confusion with the target functions $a_{\so}^*$ and $a_{\ta}^*$ used elsewhere in the paper, write the layer states as
\begin{equation*}
	u^{(0)}(Z)=Z
\end{equation*}
and, for $\ell=1,\ldots,L_{\mathrm{NN}}-1$,
\begin{equation*}
	u^{(\ell)}(Z)
	=
	\sigma\left(W^{(\ell)}u^{(\ell-1)}(Z)+b^{(\ell)}\right).
\end{equation*}
If the economist uses the last hidden layer as the embedding, the embedding and source head in the notation of the main text are
\begin{equation*}
	h_m(Z):=u^{(L_{\mathrm{NN}}-1)}(Z)\in\mathbb{R}^{H_{L_{\mathrm{NN}}-1}}
\end{equation*}
and
\begin{equation*}
	g_m(v):=W^{(L_{\mathrm{NN}})}v+b^{(L_{\mathrm{NN}})}\in\mathbb{R}^{T}.
\end{equation*}
Our theory is also applicable to other choices of the embeddings.
For example, the economist may use the second-to-last layer as the embedding, in which case
\begin{equation*}
	h_m(Z):=u^{(L_{\mathrm{NN}}-2)}(Z)\in\mathbb{R}^{H_{L_{\mathrm{NN}}-2}}
\end{equation*}
and
\begin{equation*}
    g_m(v):=W^{(L_{\mathrm{NN}})}\sigma\left(W^{(L_{\mathrm{NN}}-1)}v+b^{(L_{\mathrm{NN}}-1)}\right)+b^{(L_{\mathrm{NN}})}\in\mathbb{R}^{T}.
\end{equation*}

\subsection{Convolutional Neural Networks (CNNs)}
CNNs are designed for grid-structured inputs such as images. Let $Z\in\mathbb{R}^{H_{\mathrm{img}}\times W_{\mathrm{img}}\times C_{\mathrm{img}}}$ denote an image with height $H_{\mathrm{img}}$, width $W_{\mathrm{img}}$, and $C_{\mathrm{img}}$ channels (e.g., RGB color scale). A CNN replaces dense linear maps by local convolutional filters and pooling operations, so that early layers detect local patterns and deeper layers aggregate them into more abstract features. Denote by
\begin{equation*}
	h_m(Z)\in\mathbb{R}^{K}
\end{equation*}
the input to the final affine classification layer, with dropout disabled when extracting embeddings. For VGG19 (\citealp{simonyan2015very}), we can use the final hidden fully connected layer after activation, with $K=4{,}096$. For ResNet50 (\citealp{he2016deep}), the flattened final average-pooling output has $K=2{,}048$ coordinates and feeds directly into the final affine classifier.
The source head $g_m$ consists of the baseline logit contrasts of the final affine classifier, so both models have the form $g_m\circ h_m$ with an affine source head. Economists can use the frozen embeddings $\check{h}(Z_i)$ as regressors in the target task.

\subsection{Transformer (BERT)}
BERT is a bidirectional transformer encoder pre-trained on large text corpora (\citealp{devlin2019bert}). Its original pre-training combines masked language modeling (MLM) and next sentence prediction (NSP). This structure is useful here because it makes the shared-embedding structure explicit. A packed BERT input has the form
\begin{equation*}
	Z=
	(\texttt{[CLS]},\ \text{span A},\ \texttt{[SEP]},\ \text{span B},\ \texttt{[SEP]}),
\end{equation*}
where the two spans may represent either a genuine neighboring pair or an artificially constructed pair. BERT tokenizes text into WordPiece units with a vocabulary size of $T+1\approx 30{,}000$, so there are $T$ log-odds relative to a baseline token, as in \cref{sec:transfer_setup}. If $p$ denotes a token position in the packed sequence, the input to the encoder is
\begin{equation*}
	u_m^{(0)}(Z,p)
	=
	e_{\mathrm{tok}}(z_p)+e_{\mathrm{seg}}(s_p)+e_{\mathrm{pos}}(p),
\end{equation*}
where $e_{\mathrm{tok}}$, $e_{\mathrm{seg}}$, and $e_{\mathrm{pos}}$ denote token, segment, and 
position embeddings, respectively. 
Let $L_{\mathrm{tr}}$ denote the number of transformer blocks. Let $Z=(z_1,\ldots,z_M)$ denote a tokenized sequence of length $M$. For each block $\ell$ and token position $p$, let $u_m^{(\ell)}(Z,p)\in\mathbb{R}^{K}$ denote the hidden state at position $p$ after block $\ell$ in the fitted source encoder.
The transformer encoder maps $\{u_m^{(0)}(Z,p)\}_{p=1}^{M}$ into contextualized states $\{u_m^{(L_{\mathrm{tr}})}(Z,p)\}_{p=1}^{M}$.

For MLM, let $\mathcal{M}_{\mathrm{pred}}(Z)\subseteq\{1,\ldots,M\}$ be the positions eligible for token prediction, excluding \texttt{[CLS]}, \texttt{[SEP]}, and \texttt{[PAD]}. Let $\mathcal{M}\subseteq\mathcal{M}_{\mathrm{pred}}(Z)$ be the random subset selected for prediction and let $\tilde{Z}$ be the corrupted sequence used as input, where we randomly replace the tokens at positions in $\mathcal{M}$ with a special \texttt{[MASK]} token or random tokens. Selected positions contribute to the loss even when their tokens remain unchanged. The MLM head produces raw class logits through a shallow neural network
\begin{equation*}
	\tilde{g}_{\mathrm{MLM},m}:\mathbb{R}^{K}\to\mathbb{R}^{T+1},
\end{equation*}
applied to each contextualized token state $u_m^{(L_{\mathrm{tr}})}(\tilde{Z},p)$ for $p\in\mathcal{M}$. This head applies a dense layer with an activation function, followed by layer normalization, before the final affine vocabulary projection. Let
\begin{equation*}
	S^{\mathrm{BERT}}
	=
	\left(\{S_p\}_{p\in\mathcal{M}},S_{\mathrm{NSP}}\right)
\end{equation*}
collect the MLM and NSP labels for one sequence. If $S_p$ denotes the original token at a selected position $p$, then, for sequences with a fixed number $|\mathcal{M}|=k\geq1$ of selected positions, the per-sequence averaged MLM loss is
\begin{equation*}
	\ell_{\mathrm{MLM}}(\tilde{Z},S^{\mathrm{BERT}})
	=
	-\frac{1}{k}\sum_{p\in\mathcal{M}}\log\Pr(S_p\mid\tilde{Z},p).
\end{equation*}
For fixed $k$, averaging this loss across sequences gives the same weighting as normalizing by the total number of selected positions in a batch, as in the original implementation.
For NSP, the head produces two raw logits
\begin{equation*}
	\tilde{g}_{\mathrm{NSP},m}:\mathbb{R}^{K}\to\mathbb{R}^{2},
\end{equation*}
from the \texttt{[CLS]} state $u_m^{(L_{\mathrm{tr}})}(\tilde{Z},\texttt{[CLS]})$, including a dense layer with a hyperbolic tangent activation (the pooler) before the final affine classifier. If $S_{\mathrm{NSP}}\in\{0,1\}$ indicates whether span B truly follows span A, and
\begin{equation*}
	\pi_{\mathrm{NSP}}(\tilde{Z})
	=
	\texttt{softmax}(\tilde{g}_{\mathrm{NSP},m}(u_m^{(L_{\mathrm{tr}})}(\tilde{Z},\texttt{[CLS]}))),
\end{equation*}
then the NSP loss is
\begin{equation*}
	\ell_{\mathrm{NSP}}(\tilde{Z},S^{\mathrm{BERT}})
	=
	-S_{\mathrm{NSP}}\log \pi_{\mathrm{NSP},1}(\tilde{Z})
	-(1-S_{\mathrm{NSP}})\log \pi_{\mathrm{NSP},0}(\tilde{Z}).
\end{equation*}
With this normalization, the combined pre-training loss is
\begin{equation*}
	\ell_{\mathrm{BERT}}(\tilde{Z},S^{\mathrm{BERT}})
	=
	\ell_{\mathrm{MLM}}(\tilde{Z},S^{\mathrm{BERT}})+\ell_{\mathrm{NSP}}(\tilde{Z},S^{\mathrm{BERT}}).
\end{equation*}
The source outputs $g_m$ in our theory are obtained by subtracting a baseline logit within each classification task, leaving $T$ MLM contrasts at each selected position and one NSP contrast. To apply the results for affine source heads, we absorb the task-specific transformations into the embeddings. Let $u_{\mathrm{NSP},m}$ denote the dense layer with a hyperbolic tangent activation in the NSP pooler, and let $u_{\mathrm{MLM},m}$ denote the dense layer, activation, and layer normalization before the MLM vocabulary projection. With dropout disabled, one text-level embedding is the input to the final NSP classifier,
\begin{equation*}
	\check{h}(Z)=u_{\mathrm{NSP},m}\circ u_m^{(L_{\mathrm{tr}})}(Z,\texttt{[CLS]}).
\end{equation*}
For the MLM-based analysis, we instead use the mean of the transformed states over eligible positions,
\begin{equation*}
	\check{h}(Z)=\frac{1}{|\mathcal{M}_{\mathrm{pred}}(Z)|}\sum_{p\in\mathcal{M}_{\mathrm{pred}}(Z)}u_{\mathrm{MLM},m}\circ u_m^{(L_{\mathrm{tr}})}(Z,p),
\end{equation*}
where $|\mathcal{M}_{\mathrm{pred}}(Z)|\geq1$. This average includes all eligible positions, whether or not they were selected in a particular pre-training draw. Its relation to the selected-position source objective is established in \cref{sec:transferability_transformer}. The economist treats $\check{h}(Z)$ as the generated regressor in the target sample and estimates the downstream model $f\circ\check{h}$.\footnote{The original BERT paper fine-tunes the full network for the target task. In this paper, unless stated otherwise, we treat the pre-trained encoder as fixed and analyze the target estimator conditional on using the learned embedding $\check{h}$.}

Across fully connected networks, CNNs, and transformers, a shared encoder is learned in the source task, task-specific heads predict the source labels, and an estimated embedding $\check{h}$ is carried to the target sample. Applying the transfer-learning framework requires matching the chosen embedding to the source outputs and verifying that source prediction error controls the outputs used in the transferability argument.

\section{Transferability Condition for Deep Learning Models}
\label{sec:transferability_examples}
In this section, we examine the transferability condition for CNNs and transformers.
This appendix explains how the sufficient conditions in \cref{thm:suf_transfer,cor:transferability_full_rank} translate into common pre-trained deep learning models. The main message is that transferability depends on how informative the source-task head is about the learned embedding. When the last source head is linear, its matrix of slope contrasts has enough rank, and the target class satisfies the required functional class condition, the sufficient conditions follow from \cref{cor:transferability_full_rank}.

Throughout this appendix, source outputs are log-odds relative to a baseline category, as in \cref{sec:transfer_setup}. To make the normalization explicit, consider raw logits $\tilde{g}(v)=\tilde{W}v+\tilde{b}$ for $T+1$ classes, where $\tilde{W}\in\mathbb{R}^{(T+1)\times K}$ and $\tilde{b}\in\mathbb{R}^{T+1}$. With class $T+1$ as the baseline, define $W\in\mathbb{R}^{T\times K}$ and $b\in\mathbb{R}^{T}$ by
\begin{align*}
	W_{t,\cdot} &= \tilde{W}_{t,\cdot}-\tilde{W}_{T+1,\cdot},
	\qquad b_t=\tilde{b}_t-\tilde{b}_{T+1},\\
	g_t(v) &= \tilde{g}_t(v)-\tilde{g}_{T+1}(v)=W_{t,\cdot}v+b_t,
\end{align*}
for $t=1,\ldots,T$, where $W_{t,\cdot}$ denotes the $t$-th row of $W$. Softmax probabilities identify the logit differences $g_t(v)$ and are invariant to adding a common affine function to all raw logits. All rank, singular-value, row-span, and pseudo-inverse conditions below therefore concern $W$, and all intercept adjustments concern $b$. Candidate source heads use the same baseline normalization. Full column rank of $\tilde{W}$ alone is insufficient.

\subsection{Transferability Condition for Pre-trained CNN}
When the CNN embedding is the input to the final affine classifier, the normalized source head has the form
\begin{equation*}
	g_m(v)=W_{\mathrm{cnn},m}v+b_{\mathrm{cnn},m},
\end{equation*}
where $v\in\mathbb{R}^{K}$ is the chosen embedding, $W_{\mathrm{cnn},m}\in\mathbb{R}^{T\times K}$, and $b_{\mathrm{cnn},m}\in\mathbb{R}^{T}$. Here $W_{\mathrm{cnn},m}$ and $b_{\mathrm{cnn},m}$ are the slope and intercept contrasts of the raw classifier with $T+1$ classes, defined relative to the baseline class as above. This decomposition applies to VGG19's final hidden fully connected layer after activation and ResNet50's flattened final average-pooling output, as specified in \cref{sec:pretrained_dl}.
For their 1{,}000-class ImageNet versions, $T=999$ is smaller than both embedding dimensions, so the full-column-rank condition below cannot hold for these embeddings.

Suppose that the target model class $\mathcal{F}_n=\mathcal{F}_n^{\sub}$ is Lipschitz with constant $L_{\mathcal{F}}$ and that \cref{asm:support} holds. If $W_{\mathrm{cnn},m}$ has full column rank, then
\begin{equation*}
	W_{\mathrm{cnn},m}^{+}
	=
	(W_{\mathrm{cnn},m}'W_{\mathrm{cnn},m})^{-1}W_{\mathrm{cnn},m}'
\end{equation*}
is well defined, and for every target model $f_n$ we can define
\begin{equation*}
	\Gamma(x)
	=
	f_n\left(W_{\mathrm{cnn},m}^{+}(x-b_{\mathrm{cnn},m})\right).
\end{equation*}
Then
\begin{equation*}
	\Gamma\circ g_m\circ h_m(Z)
	=
	f_n\left(W_{\mathrm{cnn},m}^{+}(W_{\mathrm{cnn},m}h_m(Z)+b_{\mathrm{cnn},m}-b_{\mathrm{cnn},m})\right)
	=
	f_n\circ h_m(Z),
\end{equation*}
and the Lipschitz constant of $\Gamma$ is bounded by
\begin{equation*}
	L_{\Gamma}
	\leq
	L_{\mathcal{F}}\|W_{\mathrm{cnn},m}^{+}\|_{\mathrm{sp}}.
\end{equation*}
Therefore, by \cref{cor:transferability_full_rank}, the CNN source task is $(\rho_{\so,m},\rho_{\ta,n})$-diverse over $f_n$ with
\begin{equation*}
	\rho_{\ta,n}
	=
	\frac{L_{\mathcal{F}}}{\sigma_{\min}^+(W_{\mathrm{cnn},m})\underline{w}_n}\rho_{\so,m},
\end{equation*}
provided that the functional class requirement in \cref{cor:transferability_full_rank} is satisfied. Specifically, for every embedding function $h$ and corresponding best normalized linear source head $b_{\mathrm{cnn},m}^{\bullet}+W_{\mathrm{cnn},m}^{\bullet}(\cdot)$, the target class must contain the function
\begin{equation*}
    u\mapsto
    f_n\left(W_{\mathrm{cnn},m}^{+}
    (b_{\mathrm{cnn},m}^{\bullet}-b_{\mathrm{cnn},m}
    +W_{\mathrm{cnn},m}^{\bullet}u)\right).
\end{equation*}
The bias difference accounts for translations of the embedding function that can be absorbed by the source intercept. This functional class condition holds, for example, when the target class is closed under the displayed linear transformations with intercept shifts. Transferability is stronger when the smallest singular value of $W_{\mathrm{cnn},m}$ is bounded away from zero, because then $\|W_{\mathrm{cnn},m}^{+}\|_{\mathrm{sp}}$ is small.

\subsection{Transferability Condition for Pre-trained Transformer}
\label{sec:transferability_transformer}
For BERT-style transformers, the relevant source tasks depend on how the downstream embedding is constructed and on the fact that pre-training is carried out on a corrupted sequence $\tilde{Z}=C(Z,R)$ rather than on the original sequence $Z$. Two common choices are the \texttt{[CLS]} embedding and the mean-pooled token embedding described in \cref{sec:pretrained_dl}. The main distinction is that the next sentence prediction task is a single sequence-level task, whereas masked language modeling generates many token-level source tasks.
For the population analysis in this subsection, we change the notation of \cref{sec:pretrained_dl}: $u_m^{(\ell)}$, $u_{\mathrm{NSP},m}$, $u_{\mathrm{MLM},m}$, and the corresponding source heads now denote population reference maps, with $h_m$ denoting the resulting reference embedding. Fitted maps carry check accents.

\paragraph{\texttt{[CLS]} embedding.}
Suppose that the economist uses
\begin{equation*}
	h_m(Z)=u_m^{(L_{\mathrm{tr}})}(Z,\texttt{[CLS]}).
\end{equation*}
BERT's NSP head includes the nonlinear pooler described in \cref{sec:pretrained_dl}, so it is nonlinear in this raw \texttt{[CLS]} state. To apply \cref{cor:transferability_full_rank}, the pooler must be absorbed into the embedding, or a head linear in the chosen embedding must be assumed. For a $K$-dimensional input $v$ to the final linear NSP classifier, let $\tilde{W}_{\mathrm{NSP},m}\in\mathbb{R}^{2\times K}$ and $\tilde{b}_{\mathrm{NSP},m}\in\mathbb{R}^{2}$ be its raw slope matrix and intercept vector, with rows indexed by the labels $0$ and $1$. Taking class $0$ as the baseline gives
\begin{align*}
	W_{\mathrm{NSP},m} &= \tilde{W}_{\mathrm{NSP},m,1,\cdot}-\tilde{W}_{\mathrm{NSP},m,0,\cdot},\\
	b_{\mathrm{NSP},m} &= \tilde{b}_{\mathrm{NSP},m,1}-\tilde{b}_{\mathrm{NSP},m,0},
\end{align*}
and the normalized source head is
\begin{equation*}
	g_{\mathrm{NSP},m}(v)=W_{\mathrm{NSP},m}v+b_{\mathrm{NSP},m},
\end{equation*}
where $W_{\mathrm{NSP},m}\in\mathbb{R}^{1\times K}$ and $b_{\mathrm{NSP},m}\in\mathbb{R}$. The full-rank sufficient condition in \cref{cor:transferability_full_rank} can then hold only if $\operatorname{rank}(W_{\mathrm{NSP},m})=K$, which requires $K\leq1$. For empirically relevant transformer embeddings such as BERT, $K$ is much larger than $1$. Therefore, next sentence prediction alone cannot generically identify a high-dimensional representation. Under the conditions of \cref{thm:nec_transfer}, $(0,0)$-diversity based only on this binary source task requires the target function to depend on the chosen representation through the scalar NSP log-odds, including the pooler when the raw \texttt{[CLS]} state is used.

\paragraph{Mean-pooled token embedding.}
For MLM, we pool only positions eligible for prediction and include the MLM-specific transformation in the token states. Let
\begin{equation*}
	u_m(Z,p)
	=
	u_{\mathrm{MLM},m}\circ u_m^{(L_{\mathrm{tr}})}(Z,p)\in\mathbb{R}^{K}
\end{equation*}
denote the state after BERT's dense activation and layer normalization, before the final vocabulary projection. Recall that $\mathcal{M}_{\mathrm{pred}}(Z)$ excludes the positions of \texttt{[CLS]}, \texttt{[SEP]}, and padding. The embedding is
\begin{equation*}
	h_m(Z)
	=
	\frac{1}{|\mathcal{M}_{\mathrm{pred}}(Z)|}\sum_{p\in\mathcal{M}_{\mathrm{pred}}(Z)}u_m(Z,p).
\end{equation*}
Thus, every eligible position enters the pool, including positions not selected in a particular corruption draw. Special tokens remain in the encoder input but their states do not enter this average. The normalized MLM head is
\begin{equation*}
	g_{\mathrm{MLM},m}(v)
	=
	W_{\mathrm{MLM},m}v+b_{\mathrm{MLM},m},
\end{equation*}
where $W_{\mathrm{MLM},m}\in\mathbb{R}^{T\times K}$ and $b_{\mathrm{MLM},m}\in\mathbb{R}^{T}$ are the baseline slope and intercept contrasts of the vocabulary projection. In BERT, the raw vocabulary projection is tied to the input token-embedding matrix. Full column rank of $W_{\mathrm{MLM},m}$ requires $K\leq T$. Since the same linear head is applied at every position,
\begin{equation*}
	g_{\mathrm{MLM},m}\circ h_m(Z)
	=
	\frac{1}{|\mathcal{M}_{\mathrm{pred}}(Z)|}\sum_{p\in\mathcal{M}_{\mathrm{pred}}(Z)}g_{\mathrm{MLM},m}(u_m(Z,p)).
\end{equation*}
This identity does not by itself relate the MLM objective to the pooled output, because the objective evaluates only selected positions. We establish that relation before applying the full-rank result.

\paragraph{From selected-token predictions to raw pooled predictions.}
Let $P_{\mathrm{raw}}$ denote the law of the raw source sequence before corruption. We consider ordinary WordPiece masking, for which we select a prescribed number of positions without replacement, subject to a prediction cap.
For the following result, assume that this number is a fixed integer $k\geq1$ and $|\mathcal{M}_{\mathrm{pred}}(Z)|\geq k$, $P_{\mathrm{raw}}$-a.s. Conditional on $Z$, let $\mathcal{M}$ be a uniformly drawn $k$-element subset of $\mathcal{M}_{\mathrm{pred}}(Z)$. Independently across selected positions, replace the token by \texttt{[MASK]} with probability $0.8$, draw a uniform vocabulary replacement with probability $0.1$, and retain the original token with probability $0.1$. Positions outside $\mathcal{M}$ are unchanged. With
\begin{equation*}
	q=0.1+\frac{0.1}{T+1},
\end{equation*}
the resulting corrupted input satisfies
\begin{equation}
	\Pr(\tilde{Z}=Z\mid Z,\mathcal{M})=q^k.
	\label{eq:mlm_raw_probability}
\end{equation}
The second term in $q$ accounts for a random replacement matching the original token.

For a candidate token-state map $u$ and shared affine head $g(v)=Wv+b$, define its pooled embedding and its token-level difference from the reference model by
\begin{align*}
	h(Z) &= \frac{1}{|\mathcal{M}_{\mathrm{pred}}(Z)|}\sum_{p\in\mathcal{M}_{\mathrm{pred}}(Z)}u(Z,p),\\
	\Delta_p(Z) &= g(u(Z,p))-g_{\mathrm{MLM},m}(u_m(Z,p)),
\end{align*}
and define the population discrepancy at selected prediction positions by
\begin{equation}
	D_{\mathrm{MLM}}(g,u)^2
	=
	\mathbb{E}\left[\frac{1}{k}\sum_{p\in\mathcal{M}}\|\Delta_p(\tilde{Z})\|_2^2\right],
	\label{eq:mlm_selected_discrepancy}
\end{equation}
where the expectation integrates over $Z\sim P_{\mathrm{raw}}$, the selected subset, and the corruption draws.

\begin{lemma}[MLM Prediction Bound for Eligible-Token Pooling]
	\label{lem:mlm_pooling_bound}
	Under the selection and corruption conditions above, every candidate pair $(g,u)$ with a shared affine head and $D_{\mathrm{MLM}}(g,u)<\infty$ satisfies
	\begin{equation}
		\|g\circ h-g_{\mathrm{MLM},m}\circ h_m\|_{P_{\mathrm{raw}},2}
		\leq q^{-k/2}D_{\mathrm{MLM}}(g,u).
		\label{eq:mlm_pooling_bound}
	\end{equation}
\end{lemma}

\begin{proof}
	Conditional on $Z$ and $\mathcal{M}$, the prediction maps on the event $\{\tilde{Z}=Z\}$ equal their values at the raw input. Hence \cref{eq:mlm_raw_probability} gives
	\begin{align*}
		&\mathbb{E}\left[\mathds{1}\{\tilde{Z}=Z\}\frac{1}{k}\sum_{p\in\mathcal{M}}\|\Delta_p(\tilde{Z})\|_2^2\mid Z,\mathcal{M}\right]\\
		&=q^k\frac{1}{k}\sum_{p\in\mathcal{M}}\|\Delta_p(Z)\|_2^2.
	\end{align*}
	Since $\Pr(p\in\mathcal{M}\mid Z)=k/|\mathcal{M}_{\mathrm{pred}}(Z)|$ for every $p\in\mathcal{M}_{\mathrm{pred}}(Z)$,
	\begin{equation*}
		\mathbb{E}\left[\frac{1}{k}\sum_{p\in\mathcal{M}}\|\Delta_p(Z)\|_2^2\mid Z\right]
		=\frac{1}{|\mathcal{M}_{\mathrm{pred}}(Z)|}\sum_{p\in\mathcal{M}_{\mathrm{pred}}(Z)}\|\Delta_p(Z)\|_2^2.
	\end{equation*}
	Dropping the nonnegative contribution from $\{\tilde{Z}\neq Z\}$ in \cref{eq:mlm_selected_discrepancy} and taking expectations therefore yields
	\begin{equation*}
		D_{\mathrm{MLM}}(g,u)^2
		\geq q^k\mathbb{E}_{\mathrm{raw}}\left[\frac{1}{|\mathcal{M}_{\mathrm{pred}}(Z)|}\sum_{p\in\mathcal{M}_{\mathrm{pred}}(Z)}\|\Delta_p(Z)\|_2^2\right].
	\end{equation*}
	The shared affine heads imply
	\begin{equation*}
		g\circ h(Z)-g_{\mathrm{MLM},m}\circ h_m(Z)
		=\frac{1}{|\mathcal{M}_{\mathrm{pred}}(Z)|}\sum_{p\in\mathcal{M}_{\mathrm{pred}}(Z)}\Delta_p(Z).
	\end{equation*}
	Jensen's inequality now gives
	\begin{align*}
		\|g\circ h-g_{\mathrm{MLM},m}\circ h_m\|_{P_{\mathrm{raw}},2}^2
		&\leq\mathbb{E}_{\mathrm{raw}}\left[\frac{1}{|\mathcal{M}_{\mathrm{pred}}(Z)|}\sum_{p\in\mathcal{M}_{\mathrm{pred}}(Z)}\|\Delta_p(Z)\|_2^2\right]\\
		&\leq q^{-k}D_{\mathrm{MLM}}(g,u)^2.
	\end{align*}
	Taking square roots proves \cref{eq:mlm_pooling_bound}.
\end{proof}

\paragraph{Implication for transferability.}
Suppose that $\mathcal{F}_n^{\sub}$ consists of $L_{\mathcal{F}}$-Lipschitz functions and fix $f_n\in\mathcal{F}_n^{\sub}(h_m)$. Assume that $W_{\mathrm{MLM},m}$ has full column rank and that \cref{asm:support} holds with $P_{\so}=P_{\mathrm{raw}}$, so the density of the raw source sequence relative to the target density is at least $\underline{w}_n^2$. Suppose also that the class of pooled embeddings and the affine source-head class satisfy the target functional class condition in \cref{cor:transferability_full_rank}, with source prediction distances measured under $P_{\mathrm{raw}}$. In particular, for every pooled embedding $h$ and its best affine head $v\mapsto W^{\bullet}v+b^{\bullet}$ approximating $g_{\mathrm{MLM},m}\circ h_m$ under $P_{\mathrm{raw}}$, the target class must contain a function agreeing $P_{\ta}$-a.s. on $h(Z)$ with
\begin{equation*}
	v\mapsto f_n\left(W_{\mathrm{MLM},m}^{+}(b^{\bullet}-b_{\mathrm{MLM},m}+W^{\bullet}v)\right).
\end{equation*}
For any candidate pair with $D_{\mathrm{MLM}}(g,u)\leq\rho_{\so,m}$, \cref{lem:mlm_pooling_bound} places its pooled embedding in the raw-source near-identified set at radius $q^{-k/2}\rho_{\so,m}$. Applying \cref{cor:transferability_full_rank} at this enlarged radius gives
\begin{equation}
	\min_{f\in\mathcal{F}_n^{\sub}}\|f\circ h-f_n\circ h_m\|_{P_{\ta},2}
	\leq
	\frac{L_{\mathcal{F}}}{\underline{w}_n\sigma_{\min}(W_{\mathrm{MLM},m})q^{k/2}}\rho_{\so,m}.
	\label{eq:mlm_transfer_bound}
\end{equation}
Here $\rho_{\so,m}$ bounds selected-token prediction error, and $\underline{w}_n$ compares raw source and target sequences. The factor $q^{-k/2}$ already accounts for corruption. Let $\check{u}$ denote the fitted counterpart of the transformed token-state map $u_m$, let $\check{g}$ denote its fitted affine MLM head, and let $\check{h}$ denote the mean of $\check{u}(Z,p)$ over eligible positions, as in \cref{sec:pretrained_dl}. If these fitted maps satisfy $D_{\mathrm{MLM}}(\check{g},\check{u})=O_P(\delta_{\so,m})$, their pooled embedding satisfies
\begin{equation*}
	\min_{f\in\mathcal{F}_n^{\sub}}\|f\circ\check{h}-f_n\circ h_m\|_{P_{\ta},2}
	=O_P\left(\frac{L_{\mathcal{F}}\delta_{\so,m}}{\underline{w}_n\sigma_{\min}(W_{\mathrm{MLM},m})q^{k/2}}\right).
\end{equation*}
This selected-token convergence condition must be established separately. Relating cross-entropy regret to \cref{eq:mlm_selected_discrepancy} requires curvature and approximation conditions for the selected-token logit contrasts, as in \cref{asm:loss_so,asm:convergence}, with sequences as the sampling units.

This bound is very conservative because it uses only the event that no effective corruption occurs. The mean-pooling structure suggests that a larger constant may be available under additional conditions on the transformer states, since many tokens remain unchanged and the average embedding can still preserve the target-relevant variation even when some positions are corrupted. A sharper comparison requires additional restrictions on how corruption changes the token-state functions.

\section{Source Head Transformation for Ridge Regression}
\label{sec:source_head_transformation}
For any $B\in\mathbb{R}^{T\times K}$, define
\begin{align}
	R_B
	&:=
	(B'B)^{1/2}+I_K-B^+B,\nonumber\\
	B^{\mathrm{N}}
	&:=
	BR_B^{-1},\nonumber\\
	h_B^{\mathrm{N}}
	&:=
	R_Bh.
	\label{eq:source_transformation}
\end{align}
Under these linear transformations of the source head and embedding function, we obtain the same source output $g_{\alpha,B}\circ h = g_{\alpha,B^{\mathrm{N}}}\circ h_B^{\mathrm{N}}$ and the desirable properties $B^{\mathrm{N}}(B^{\mathrm{N}})'B^{\mathrm{N}}=B^{\mathrm{N}}$ and $(B^{\mathrm{N}})^+= (B^{\mathrm{N}})'$.\footnote{
If $\operatorname{rank}(B)=r\geq1$ and $B=U_BD_BV_B'$ is the compact singular-value decomposition, where $D_B\in\mathbb{R}^{r\times r}$ contains the positive singular values, then
\begin{align*}
	R_B
	&=
	V_BD_BV_B'+I_K-V_BV_B',\\
	R_B^{-1}
	&=
	V_BD_B^{-1}V_B'+I_K-V_BV_B',\\
	B^{\mathrm{N}}
	&=
	U_BV_B'.
\end{align*}
Thus, $R_B$ is invertible even when $B$ is rank deficient. The transformation preserves the composite source prediction and satisfies
\begin{align*}
	g_{\alpha,B^{\mathrm{N}}}\circ h_B^{\mathrm{N}}
	&=
	g_{\alpha,B}\circ h,\\
	B^{\mathrm{N}}(B^{\mathrm{N}})'B^{\mathrm{N}}
	&=
	B^{\mathrm{N}},\\
	(B^{\mathrm{N}})^+
	&=
	(B^{\mathrm{N}})'.
\end{align*}
Every positive singular value of $B^{\mathrm{N}}$ therefore equals one. 
}
Note that these transformations are simple and only require the source head matrix as well as the estimated embedding function $\check{h}$ from the source task. The source head matrix is typically available to the economist in the pre-trained source model.
After this transformation, if $\rank(B)\geq1$, then the transformed source head $B^{\mathrm{N}}$ satisfies $\|B^{\mathrm{N}}\|_{\mathrm{sp}}\leq\bar{\sigma}$ and $\sigma_{\min}^+(B^{\mathrm{N}})\geq\underline{\sigma}$ with $\bar{\sigma}=\underline{\sigma}=1$.

After the transformation, we can take the projected estimated embedding as $\check{h}^{\mathrm{proj}} := (\check{B}_{\so}^{\mathrm{N}})^+ \check{B}_{\so}^{\mathrm{N}} \check{h}^{\mathrm{N}}$. Thus, every occurrence of $\check{h}^{\mathrm{proj}}$ in the ridge estimator refers to this projection. We also form $\check{q}$ and $\check{Q}$ after transformation. If the transformation is used, the population and pre-trained composite source predictions, and hence the convergence rate in \cref{asm:convergence} (i), remain unchanged.

The transformation also preserves target-slope compatibility. Because $R_B$ is symmetric, if $\beta'=b'B$ before transformation, then $\beta^{\mathrm{N}}:=R_B^{-1}\beta$ satisfies
\begin{align*}
	\beta^{\mathrm{N}\prime}h_B^{\mathrm{N}}
	&=
	\beta'h,\\
	\beta^{\mathrm{N}\prime}
	&=
	b'B^{\mathrm{N}}.
\end{align*}
Thus, the transformation changes neither the target score nor the norm of the linking coefficient $b$.
When the transformation is applied to the population source matrix, we reuse $\beta_{\ta,n}$ for this transformed target slope.

\section{Technical Lemmas}
\begin{lemma}[Transferability Decomposition]
	\label{lem:transfer_var}
	Let $P$ denote the distribution of $Z$, and suppose that $\mathcal{F}_n^{\sub}$ and $\mathcal{H}_m$ are classes of measurable functions. Fix measurable functions $f_n$ and $h_m$ such that $f_n\circ h_m\in L^2(P)$. Then, for every $h\in\mathcal{H}_m$ such that $f\circ h\in L^2(P)$ for every $f\in\mathcal{F}_n^{\sub}$,
	\begin{align*}
		&\inf_{f\in\mathcal{F}_n^{\sub}}\|f \circ h - f_n \circ h_m\|_{P,2}^2\\
		&\quad=\mathbb{E}\left[\operatorname{Var}\left(f_n \circ h_m(Z)\mid h(Z)\right)\right]
		+\inf_{f\in\mathcal{F}_n^{\sub}}\mathbb{E}\left[\left(\mathbb{E}\left[f_n \circ h_m(Z)\mid h(Z)\right]-f\circ h(Z)\right)^2\right].
	\end{align*}
\end{lemma}
\begin{remark}
    The first term on the right-hand side represents the variation of $f_n \circ h_m(Z)$ that cannot be explained by $h(Z)$, and the second term represents the best approximation of the conditional expectation of $f_n \circ h_m(Z)$ given $h(Z)$ by functions in $\mathcal{F}_n^{\sub}$.
	Our sufficient condition for transferability in \cref{thm:suf_transfer} ensures that both terms are small when $h\in\mathcal{H}_{\so,m}(\rho_{\so,m})$.
	Under $(0,0)$-diversity, both terms vanish for every $h\in\mathcal{H}_{\so,m}(0)$, with $P=P_{\ta}$.
	\cref{ex:counter} illustrates the importance of the second term: it is bounded away from zero even when the first term is zero, ruling out $\nu_n$-transferability whenever $\delta_m/\nu_n=o(1)$.
\end{remark}

\begin{lemma}[Doob-Dynkin Factorization under Completion]
	\label{lem:doob_dynkin_completion}
	Let $(\Omega,\mathcal{A},P)$ be a probability space, let $U:\Omega\to\mathcal{U}$ be a random element taking values in a measurable space, and let $V:\Omega\to\mathbb{R}^{q}$ be a $q$-dimensional random vector. Let $\overline{\sigma(U)}^{P}$ be the $P$-completion of $\sigma(U)$, which is defined by
    $\overline{\sigma(U)}^{P}=\{A\cup N:A\in\sigma(U),N\subset B \text{ for some } B\in\sigma(U) \text{ with } P(B)=0\}$.
    If
	\begin{equation*}
		\sigma(V)\subseteq\overline{\sigma(U)}^{P},
	\end{equation*}
	then there exists a measurable function $\Gamma:\mathcal{U}\to\mathbb{R}^{q}$ such that $V=\Gamma\circ U$, $P$-a.s.
\end{lemma}

\begin{lemma}[Properties of Loss Functions, \citealp{farrell2021deep}, Lemma 8]
	\label{lem:loss_property}
	Let $P$ denote the distribution of $Z$. In each setup below, suppose that
	$\|f\|_{\infty},\|a\|_{\infty},\|f^*\|_{\infty}\leq M$ for some constant $M>0$ and all candidate functions $f$ and $a$.
	For each setup, there exist constants $c_L>0$, $c_U>0$, and $C_\ell>0$ such that
	\begin{equation*}
		c_L\|f-f^*\|_{P,2}^2
		\leq
		\mathbb{E}[\ell(f(Z),Y)]-\mathbb{E}[\ell(f^*(Z),Y)]
		\leq
		c_U\|f-f^*\|_{P,2}^2
	\end{equation*}
	and
	\begin{equation*}
		|\ell(f(z),y)-\ell(a(z),y)|
		\leq
		C_\ell|f(z)-a(z)|
	\end{equation*}
	for all candidate functions $f$ and $a$, every $z\in\mathcal{Z}$, and every $y$ in the range of $Y$.
	\begin{enumerate}
		\renewcommand{\labelenumi}{(\alph{enumi})}
		\item \textbf{Least squares.}
		Suppose that $|Y|\leq M$ almost surely, $f^*(Z)=\mathbb{E}[Y\mid Z]$, and $\ell(u,y)=(1/2)(y-u)^2$.
		Then, we can take $c_L=c_U=1/2$ and $C_\ell=2M$.

		\item \textbf{Binary logistic regression.}
		Suppose that $Y\in\{0,1\}$ almost surely, $P(Y=1\mid Z)>0$ almost surely, and $P(Y=0\mid Z)>0$ almost surely. Let
		\begin{equation*}
			f^*(Z)=\log\left(\frac{P(Y=1\mid Z)}{P(Y=0\mid Z)}\right)
		\end{equation*}
		and $\ell(u,y)=-yu+\log(1+\exp(u))$.
		Then, we can take $c_L=1/(2(\exp(M)+\exp(-M)+2))$, $c_U=1/8$, and $C_\ell=1$.
	\end{enumerate}
\end{lemma}

\begin{lemma}[Properties of Loss Functions for Multinomial Logit]
	\label{lem:loss_property_multilogit}
	Consider the $T+1$-class classification problem with the multinomial-logit negative log-likelihood loss.
	Let $P$ denote the distribution of $Z$. Suppose that $Y\in\{1,\ldots,T+1\}$ almost surely and that $P(Y=t\mid Z)>0$ almost surely for every $t=1,\ldots,T+1$.
	Take class $T+1$ as the baseline category and let $f^*:\mathcal{Z}\to\mathbb{R}^T$ be a measurable version of the population log-odds vector $f^*=(f_1^*,\ldots,f_T^*)'$ defined by
	\begin{equation*}
		f_t^*(Z)
		=
		\log\left(\frac{P(Y=t\mid Z)}{P(Y=T+1\mid Z)}\right),
		\qquad t=1,\ldots,T,\quad P\text{-a.s.},
	\end{equation*}
	and define the multinomial logistic loss by
	\begin{equation*}
		\ell(u,y)
		=
		-\sum_{t=1}^T\mathds{1}\{y=t\}u_t
		+\log\left(1+\sum_{t=1}^T\exp(u_t)\right),
		\qquad u\in\mathbb{R}^T.
	\end{equation*}
	Suppose that, for some constant $M>0$ and all candidate functions $f$ and $a$,
	\begin{equation*}
		\sup_{z\in\mathcal{Z}}\|f(z)\|_2,
		\sup_{z\in\mathcal{Z}}\|a(z)\|_2,
		\sup_{z\in\mathcal{Z}}\|f^*(z)\|_2
		\leq
		M.
	\end{equation*}
	Then,
	\begin{equation*}
		c_L
		=
		\frac{1}{2(1+T\exp(M))^2},
		\qquad
		c_U
		=
		\frac{\exp(M)}{2(1+(T-1)\exp(-M)+\exp(M))},
		\qquad
		C_\ell
		=
		\sqrt{2}
	\end{equation*}
	satisfy
	\begin{equation*}
		c_L\|f-f^*\|_{P,2}^2
		\leq
		\mathbb{E}[\ell(f(Z),Y)]-\mathbb{E}[\ell(f^*(Z),Y)]
		\leq
		c_U\|f-f^*\|_{P,2}^2
	\end{equation*}
	and
	\begin{equation*}
		|\ell(f(z),y)-\ell(a(z),y)|
		\leq
		C_\ell\|f(z)-a(z)\|_2
	\end{equation*}
	for all candidate functions $f$ and $a$, every $z\in\mathcal{Z}$, and every $y\in\{1,\ldots,T+1\}$.
\end{lemma}

\begin{remark}
	For fixed $M>0$, $c_L\asymp T^{-2}$ as $T\to\infty$, and this order cannot generally be improved. To see this, take $Z$ to be deterministic and $f^*(z)=\mathbf{0}_T\in\mathbb{R}^T$, so that $P(Y=t\mid Z)=1/(T+1)$ for $t=1,\ldots,T+1$. Let $\delta_T=MT^{-1/2}$ and choose $f(z)=\delta_T\mathbf{1}_T$, which satisfies $\|f-f^*\|_{P,2}^2=M^2$ and $\sup_{z\in\mathcal{Z}}\|f(z)\|_2=M$. Then,
	\begin{align*}
		\mathbb{E}[\ell(f(Z),Y)]-\mathbb{E}[\ell(f^*(Z),Y)]
		&=
		-\frac{T\delta_T}{T+1}
		+\log(1+T\exp(\delta_T))
		-\log(1+T)\\
		&=
		\frac{\delta_T}{T+1}
		+\log\left(1+\frac{\exp(-\delta_T)-1}{T+1}\right)\\
		&=
		\frac{\delta_T+\exp(-\delta_T)-1}{T+1}
		+o(T^{-2})\\
		&=
		\frac{M^2}{2T(T+1)}
		+o(T^{-2})
		\asymp
		T^{-2},
	\end{align*}
	where the third equality uses $\log(1+x)=x+O(x^2)$ and the fourth equality uses $\exp(-x)=1-x+x^2/2+o(x^2)$ as $x\to0$.
\end{remark}

For a norm $\|\cdot\|$ and $\varepsilon>0$, let $N\left(\varepsilon, \mathcal{F}, \|\cdot\|\right)$ denote the covering number of $\mathcal{F}$, i.e., the minimum number of balls of radius $\varepsilon$ with respect to the norm $\|\cdot\|$ needed to cover the set $\mathcal{F}$.\footnote{
	The centers of the balls are not required to belong to $\mathcal{F}$, but they must have finite norms.
}
The following two lemmas are also proved as part of the proof of Theorem 7 in \cite{tripuraneni2020theory}, but we state them here for completeness and clarity.

\begin{lemma}[Properties of Products of Covering Numbers]
    \label{lem:entropy_property}
	Let $\mathcal{H}$ be a measurable class of functions with $K$-dimensional outputs. Suppose that $\mathcal{H}_k$ is the $k$-th coordinate projection of $\mathcal{H}$, i.e., $\mathcal{H}_k=\{h_k: h=(h_1,\ldots,h_K)\in\mathcal{H}\}$. Then, we have
    $$
    N\left(\varepsilon, \mathcal{H}, L^2(Q)\right) \leq \prod_{k=1}^K N\left(\varepsilon / K, \mathcal{H}_k, L^2(Q)\right).
    $$
\end{lemma}

\begin{lemma}[Properties of Composition of Covering Numbers]
	\label{lem:entropy_property_composite}
	Let $\mathcal{F}$ and $\mathcal{H}$ be measurable classes of functions. Suppose that $\mathcal{F}$ is Lipschitz with constant $L_\mathcal{F}>0$. Then, we have
	\begin{equation*}
		N\left(\varepsilon, \mathcal{F} \circ \mathcal{H}, L^2(Q)\right) \leq N\left(\varepsilon / (4L_\mathcal{F}), \mathcal{H}, L^2(Q)\right) \cdot \sup_{h'\in\mathcal{H}}N\left(\varepsilon /2, \mathcal{F}, L^2(Q_{h'})\right),
	\end{equation*}
	where $Q_{h'}$ is the measure such that $Q_{h'}(A)=Q(h'^{-1}(A))$ for any measurable set $A$.
\end{lemma}

\begin{definition}[Uniform Entropy Integral]
    \label{def:entropy}
	For a class $\mathcal{F}$ of measurable functions with measurable envelope function $F$, define the {\it uniform entropy integral} by
    $$
	J\left(\delta, \mathcal{F} | F, L_2\right)=\sup _Q \int_0^\delta \sqrt{1+\log N\left(\varepsilon\|F\|_{Q, 2}, \mathcal{F}, L_2(Q)\right)} d \varepsilon,
    $$
    where the supremum is taken over all discrete probability measures $Q$ with $\|F\|_{Q, 2}>0$.
\end{definition}

Importantly, the uniform entropy integral depends only on the function class and its domain, not on a particular probability measure.

\begin{lemma}[Localized Maximal Inequality for Lipschitz Loss]
	\label{lem:maximal_lipschitz}
	Let $\mathcal{F}_n$ be a pointwise measurable class of measurable functions mapping into $\mathbb{R}^T$ with $\sup_v\|f(v)\|_2\leq M <\infty$ for every $f\in\mathcal{F}_n$, let $\mathcal{H}_m$ be a class of measurable functions, and let $h\mapsto f_{n,h}^* \in \mathcal{F}_n$ be a mapping.
	Let $m_{f,h}(z,y)=-\ell(f\circ h(z),y)$,
	$\mathbb{M}_n(f,h)=\mathbb{P}_n m_{f,h}$,
	$M_n(f,h)=P m_{f,h}$,
	$\mathbb{G}_n g=\sqrt{n}(\mathbb{P}_n-P)g$,
	and, for every $h\in\mathcal{H}_m$ and $\delta>0$, define
	$d_{n,h}(f,f')=\|f\circ h-f'\circ h\|_{P,2}$,
	$B_{n,h}(\delta)=\left\{f\in\mathcal{F}_n: d_{n,h}\left(f,f_{n,h}^*\right) < \delta\right\}$,
	and
	$\mathcal{M}_{n,h}(\delta)=\left\{m_{f,h}-m_{f_{n,h}^*,h}: f\in B_{n,h}(\delta)\right\}$.

	Suppose that $\ell:\mathbb{R}^T\times\mathcal{Y}\mapsto\mathbb{R}$ is Lipschitz continuous in the first argument with respect to $\|\cdot\|_2$ with constant $C_{\ell}$.
	Then,
	\begin{align*}
		&\sup_{h\in\mathcal{H}_m}\mathbb{E}_{P}\left[\sup_{f\in B_{n,h}(\delta)} \sqrt{n}\left|\left(\mathbb{M}_n-M_n\right)(f,h)-\left(\mathbb{M}_n-M_n\right)\left(f_{n,h}^*,h\right)\right|\right]\\
		&\quad \lesssim C_{\ell} M \left(J\left(\delta/M, \mathcal{F}_n | M, L_2\right)
		+\frac{M^2 J^2\left(\delta/M, \mathcal{F}_n | M, L_2\right)}{\delta^2 \sqrt{n}}\right).
	\end{align*}
\end{lemma}

\begin{lemma}[Localized Pointwise Measurable Class]
    \label{lem:pointwise_measurable}
    Let $T\geq1$ and let $\mathcal{F}$ be a pointwise measurable class of measurable functions mapping into $\mathbb{R}^T$ with measurable envelope function $F$ such that $\|f(x)\|_2\leq F(x)$ for every $f\in\mathcal{F}$ and $P F^2<\infty$.
    Then, the localized class $B(\delta)=\{f\in\mathcal{F}: \|f-f^*\|_{P,2} < \delta\}$ is also pointwise measurable for every $\delta>0$ and every measurable function $f^*$ mapping into $\mathbb{R}^T$ such that $\|f^*(x)\|_2\leq F(x)$.
\end{lemma}

\begin{lemma}[Covering Number for VC Function Classes, \citealp{van2023weak}, Theorem 2.6.7]
	\label{lem:vc_covering}
	For a nonempty VC-class of functions with measurable envelope function $F$ and $r\geq1$, every probability measure $Q$ with $\|F\|_{Q,r}>0$ satisfies
	\begin{equation*}
		N\left(\varepsilon\|F\|_{Q,r},\mathcal{F},L_r(Q)\right)
		\lesssim (1\vee\mathsf{V}(\mathcal{F}))(16e)^{\mathsf{V}(\mathcal{F})}\left(\frac{1}{\varepsilon}\right)^{r\mathsf{V}(\mathcal{F})},
	\end{equation*}
	for $0<\varepsilon<1$.
\end{lemma}

\begin{lemma}[Rate of Convergence] 
    \label{lem:rate_of_conv} 
    For each $n$, let $\mathcal{F}_n$ be a set of measurable functions, let $\mathcal{F}_n^*$ be a set of measurable functions, and let $\mathcal{H}_m$ be a set of measurable functions $h$.
    Assume that for every fixed $h\in\mathcal{H}_m$, there exist mappings
    $$
    h \mapsto f_{n, h} \in \mathcal{F}_n, \quad h \mapsto f_{n, h}^* \in \mathcal{F}_n^*.
    $$
    For each $n$, let $\mathbb{M}_n$ be a stochastic process and let $M_n$ be a deterministic function indexed by a set $(\mathcal{F}_n \cup\mathcal{F}_n^*)\times\mathcal{H}_m$.
    For every fixed $h\in \mathcal{H}_m$, let $f \mapsto d_{n,h}(f, f_{n,h}^*)$ be a deterministic function from $\mathcal{F}_n$ to $[0, \infty)$, and let $\underline{\delta}_n \geq 0$ and $\tau_n\in(0,1]$ be some sequences.
    We also assume that for every fixed $h\in\mathcal{H}_m$ and $\delta>0$, the set
    $\{f\in\mathcal{F}_n: d_{n,h}(f, f_{n,h}^*) < \delta\}$ is pointwise measurable.
	Let $\check{h} \in \mathcal{H}_m$ be some possibly random function estimated using a sample of size $m$ independent of $\mathbb{M}_n$, $M_n$, and $\hat{f}_{n,h}$ defined below.
    For every $\varepsilon>0$, let $\mathcal{H}_{m,\varepsilon}\subseteq\mathcal{H}_m$ be a subset of $\mathcal{H}_m$ such that $\limsup_{m\rightarrow\infty}P(\check{h}\notin\mathcal{H}_{m,\varepsilon})\leq\varepsilon$.
    Suppose that, for every $n$, every $\varepsilon>0$, and every positive $\delta\geq \tau_n^{-1}C_{\varepsilon,0}\underline{\delta}_n$ for some constant $C_{\varepsilon,0}>0$ depending on $\varepsilon$,
    \begin{equation}
        \sup_{h\in\mathcal{H}_{m,\varepsilon}}
        \sup_{f \in \mathcal{F}_n: \delta / 2<d_{n,h}\left(f, f_{n,h}^*\right) \leq \delta} M_n(f,h)-M_n\left(f_{n,h}^*,h\right)
        \lesssim_\varepsilon -\tau_n^2\delta^2,\label{eq:rate_of_conv_curvature}
    \end{equation}
     and
    \begin{equation}
        \sup_{h\in\mathcal{H}_{m,\varepsilon}}\mathbb{E} \left[\sup_{f \in \mathcal{F}_n: d_{n,h}\left(f, f_{n,h}^*\right) < \delta} \sqrt{n}\left|\left(\mathbb{M}_n-M_n\right)(f,h)-\left(\mathbb{M}_n-M_n\right)\left(f_{n,h}^*,h\right)\right|\right] 
        \lesssim_\varepsilon \phi_n(\delta)\label{eq:rate_of_conv_modulus}
    \end{equation}
	for some increasing functions $\phi_n:\left[\underline{\delta}_n, \infty\right) \rightarrow \mathbb{R}$ such that $\delta \mapsto \phi_n(\delta) / \delta^\xi$ is decreasing for some $\xi<2$. Let $\{\delta_n\}_{n=1}^{\infty}$ be a positive deterministic sequence satisfying
    \begin{equation}
        \label{eq:rate_of_conv_delta_n}
        \phi_n\left(\delta_n/\tau_n\right) \lesssim_\varepsilon \sqrt{n} \delta_n^2, \quad \delta_n^2 \gtrsim_\varepsilon \sup_{h\in\mathcal{H}_{m,\varepsilon}}M_n\left(f_{n,h}^*,h\right)-M_n\left(f_{n,h},h\right), \quad \delta_n \gtrsim_\varepsilon \underline{\delta}_n.
    \end{equation}
    Let the sequence $\hat{f}_{n,h}$ take values in $\mathcal{F}_n$ for every fixed $h\in\mathcal{H}_m$.
    If 
    $\mathbb{M}_n\left(\hat{f}_{n,\check{h}},\check{h}\right) \geq \mathbb{M}_n\left(f_{n,\check{h}},\check{h}\right)-O_P\left(\delta_n^2\right)$ is satisfied, then $d_{n,\check{h}}\left(\hat{f}_{n,\check{h}}, f_{n,\check{h}}^*\right)=O_P\left(\delta_n/\tau_n\right)$.
\end{lemma}

\begin{remark}
    The pointwise-measurability assumptions are required to ensure the measurability of the supremum within the expectations. We can eliminate this requirement by using outer expectations and probabilities instead (see \citealp{van2023weak}, Sections 1.2 and 2.3).
\end{remark}

\begin{lemma}[Rate of Convergence for Regularized Estimator] 
    \label{lem:rate_of_conv_regu}
    For each $n$, let $\mathcal{F}_n$ be a set of measurable functions, let $\mathcal{F}_n^*$ be a set of measurable functions, and let $\mathcal{H}_m$ be a set of measurable functions $h$.
    Assume that for every fixed $h\in\mathcal{H}_m$, there exist mappings
    $$
    h \mapsto f_{n, h} \in \mathcal{F}_n, \quad h \mapsto f_{n, h}^* \in \mathcal{F}_n^*.
    $$
    For each $n$, let $\mathbb{M}_n$ be a stochastic process and let $M_n$ be a deterministic function indexed by a set $(\mathcal{F}_n \cup\mathcal{F}_n^*)\times\mathcal{H}_m$.
    For every fixed $h\in \mathcal{H}_m$, let $f \mapsto d_{n,h}(f, f_{n,h}^*)$ be a deterministic function from $\mathcal{F}_n$ to $[0, \infty)$, and let $\underline{\delta}_n \geq 0$, $\tau_n\in(0,1]$, and $\lambda_{n}>0$ be some sequences.
    Let $\mathcal{J}_n:\mathcal{F}_n\rightarrow[0,\infty)$ be some complexity measure of $f$.
    We also assume that for every fixed $h\in\mathcal{H}_m$ and $\delta>0$, the set
    $\{f\in\mathcal{F}_n: d_{n,h}(f, f_{n,h}^*) < \delta, \mathcal{J}_n(f)<\tau_n \delta / \lambda_{n}\}$ is pointwise measurable.
	Let $\check{h} \in \mathcal{H}_m$ be some possibly random function estimated using a sample of size $m$ that is independent of $\mathbb{M}_n$, $M_n$, $\hat{f}_{n,h}$, and $\hat{\lambda}_{n,h}$ defined below.
    For every $\varepsilon>0$, let $\mathcal{H}_{m,\varepsilon}\subseteq\mathcal{H}_m$ be a subset of $\mathcal{H}_m$ such that $\limsup_{m}P(\check{h}\notin\mathcal{H}_{m,\varepsilon})\leq\varepsilon$.
    Suppose that, for every $n$, every $\varepsilon>0$, and every positive $\delta\geq \tau_n^{-1}C_{\varepsilon,0}\underline{\delta}_n$ for some constant $C_{\varepsilon,0}>0$ depending on $\varepsilon$,
    \begin{equation}
        \sup_{h\in\mathcal{H}_{m,\varepsilon}} \sup_{f \in \mathcal{F}_n: \delta / 2<d_{n,h}\left(f, f_{n,h}^*\right) \leq \delta} M_n(f,h)-M_n\left(f_{n,h}^*,h\right) \lesssim_\varepsilon -\tau_n^2\delta^2,
        \label{eq:rate_of_conv_regu_curvature}
    \end{equation}
     and
    \begin{equation}
        \sup_{h\in\mathcal{H}_{m,\varepsilon}}\mathbb{E} \left[\sup_{\substack{f \in \mathcal{F}_n: d_{n,h}\left(f, f_{n,h}^*\right) < \delta\\ \mathcal{J}_n(f)<\tau_n \delta / \lambda_{n}}} \sqrt{n}\left|\left(\mathbb{M}_n-M_n\right)(f,h)-\left(\mathbb{M}_n-M_n\right)\left(f_{n,h}^*,h\right)\right|\right] 
        \lesssim_\varepsilon \phi_n(\delta)\label{eq:rate_of_conv_regu_modulus}
    \end{equation}
    for some increasing functions $\phi_n:\left[\underline{\delta}_n, \infty\right) \rightarrow \mathbb{R}$ such that $\delta \mapsto \phi_n(\delta) / \delta^\xi$ is decreasing for some $\xi<2$. 
	Let $\{\delta_n\}_{n=1}^{\infty}$ be a positive deterministic sequence satisfying
    \begin{align}
        &\phi_n\left(\delta_n/\tau_n\right) \lesssim_\varepsilon \sqrt{n} \delta_n^2, \quad
        \delta_n^2 \gtrsim_\varepsilon \sup_{h\in\mathcal{H}_{m,\varepsilon}} M_n\left(f_{n,h}^*,h\right)-M_n\left(f_{n,h},h\right), \quad
        \delta_n \gtrsim_\varepsilon \underline{\delta}_n,\nonumber\\
		& \qquad \sup_{h\in\mathcal{H}_{m,\varepsilon}} \lambda_{n} \mathcal{J}_n\left(f_{n,h}\right) < \delta_n, \quad
        \sup_{h\in\mathcal{H}_{m,\varepsilon}} d_{n,h}\left(f_{n,h}, f_{n,h}^*\right)<\delta_n/\tau_n.\label{eq:rate_of_conv_regu_delta_n}
    \end{align}
    Let the sequence $\hat{f}_{n,h}$ take values in $\mathcal{F}_n$ for every fixed $h\in\mathcal{H}_m$ and $\hat{\lambda}_{n,h}$ be a sequence of nonnegative random variables for every fixed $h\in\mathcal{H}_m$.
    If
    \begin{equation}
        \label{eq:rate_of_conv_regu_approx_max}
        \mathbb{M}_n\left(\hat{f}_{n,\check{h}},\check{h}\right)-\hat{\lambda}_{n,\check{h}}^2 \mathcal{J}_n^2\left(\hat{f}_{n,\check{h}}\right) \geq \mathbb{M}_n\left(f_{n,\check{h}},\check{h}\right)-\hat{\lambda}_{n,\check{h}}^2 \mathcal{J}_n^2\left(f_{n,\check{h}}\right)-O_P\left(\delta_{n}^2\right),
    \end{equation}
	then, on the event $\{\hat{\lambda}_{n,\check{h}}\geq\lambda_{n}\}$, we have
    $d_{n,\check{h}}\left(\hat{f}_{n,\check{h}}, f_{n,\check{h}}^*\right) = O_P(1)\times\left(\delta_n/\tau_n + \hat{\lambda}_{n,\check{h}} \mathcal{J}_n\left(f_{n,\check{h}}\right)/\tau_n\right)$.
\end{lemma}

The next lemma is useful when we want to derive the convergence rate of $\hat{f}_{n,h}$ around $f_{n,h}$ instead of $f_{n,h}^*$.

\begin{lemma}[Curvature Condition for Approximate Maximizer]
	\label{lem:curvature_approx}
	Let $d_{n,h}(\cdot,\cdot)$ be a pseudometric on a set containing $\mathcal{F}_n\cup\mathcal{F}_n^*$ for every fixed $h\in\mathcal{H}_m$. Distinct heads may have zero distance when their compositions with $h$ agree almost surely.
	Let $\tau_n\in(0,1]$ be some sequence.
	For every $\varepsilon>0$, let $\mathcal{H}_{m,\varepsilon}\subseteq\mathcal{H}_m$ be a subset of $\mathcal{H}_m$.
	We assume that there exists some sequence $\underline{\delta}_n \geq0 $ such that $\sup_{h\in\mathcal{H}_{m,\varepsilon}} d_{n,h}\left(f_{n,h}, f_{n,h}^*\right) \lesssim_\varepsilon \underline{\delta}_n$.
	Suppose that for every $n$, every $\varepsilon>0$, and every $\delta>0$,
	\begin{equation}
		\sup_{h\in\mathcal{H}_{m,\varepsilon}} \sup_{f \in \mathcal{F}_n: \delta / 2 < d_{n,h}\left(f, f_{n,h}^*\right) \leq \delta} M_n(f,h)-M_n\left(f_{n,h}^*,h\right) \lesssim_\varepsilon -\tau_n^2 \delta^2
	\end{equation}
	and
	\begin{equation}
		d_{n,h}\left(f_{n,h}, f_{n,h}^*\right)^2 \gtrsim_\varepsilon M_n\left(f_{n,h}^*,h\right)-M_n\left(f_{n,h},h\right)
	\end{equation}
	for every $h\in\mathcal{H}_{m,\varepsilon}$.
	Then, for every $n$, every $\varepsilon>0$, and every positive $\delta\geq \tau_n^{-1}C_{\varepsilon,0}\underline{\delta}_n$ for some constant $C_{\varepsilon,0}>0$ depending on $\varepsilon$,
	\begin{equation}
		\sup_{h\in\mathcal{H}_{m,\varepsilon}} \sup_{f \in \mathcal{F}_n: \delta / 2 < d_{n,h}\left(f, f_{n,h}\right) \leq \delta} M_n(f,h)-M_n\left(f_{n,h},h\right) \lesssim_\varepsilon -\tau_n^2 \delta^2.
	\end{equation}
\end{lemma}

\begin{lemma}[Inequality for Rademacher Complexity]
    \label{lem:contraction}
	Let $\xi_1,\ldots,\xi_n$ be i.i.d. Rademacher random variables, i.e., $P(\xi_i=1)=P(\xi_i=-1)=1/2$ for every $i=1,\ldots,n$, and suppose that they are independent of the data.
    Let $\mathcal{U}$ be a bounded subset of $\mathbb{R}^n$ and for every $i=1,\dots,n$, let $\phi_i:\mathbb{R}\rightarrow\mathbb{R}$ satisfy $\phi_i(0)=0$ and be Lipschitz continuous with constant $L_\phi$, i.e., $|\phi_i(u)-\phi_i(v)|\leq L_\phi |u-v|$ for every $u,v\in\mathbb{R}$.
    Then, for $u=(u_1,\ldots,u_n)'\in\mathcal{U}$, we have
    \begin{equation*}
        \mathbb{E}\left[\sup_{u\in\mathcal{U}} \left|\sum_{i=1}^n \xi_i \phi_i(u_i)\right|\right] \leq 2 L_\phi \cdot \mathbb{E}\left[\sup_{u\in\mathcal{U}} \left|\sum_{i=1}^n \xi_i u_i\right|\right].
    \end{equation*}
\end{lemma}

\begin{lemma}[Matrix Norm Inequalities, \citealp{stewart1990matrix}, Theorem II.3.9]
    \label{lem:unitary_ineq}
    Let $\|\cdot\|$ be a unitarily invariant norm such as the Frobenius norm $\|\cdot\|_F$, the spectral norm $\|\cdot\|_{\mathrm{sp}}$, and the nuclear norm $\|\cdot\|_*$. For any matrices $A\in\mathbb{R}^{T\times K}$, $\tilde{A}\in\mathbb{R}^{K\times T'}$, we have
    \begin{equation*}
        \|A\tilde{A}\| \leq \|A\|_{\mathrm{sp}} \cdot \|\tilde{A}\| \quad \text{and} \quad \|A\tilde{A}\| \leq \|A\| \cdot \|\tilde{A}\|_{\mathrm{sp}}.
    \end{equation*}
\end{lemma}

\begin{lemma}
    \label{lem:rank_one}
    For every matrix $A\in\mathbb{R}^{T\times K}$ with rank at most one, we have $\|A\|_* = \|A\|_F = \|A\|_{\mathrm{sp}} = \sqrt{\operatorname{tr}(A'A)}$.
    If $A=uv'$ for some vectors $u$ and $v$, then $\|A\|_* = \|A\|_F = \|A\|_{\mathrm{sp}} = \|u\|_2 \|v\|_2$.
\end{lemma}

\begin{lemma}
    \label{lem:l2_ineq}
	For every function $a:\mathcal{V}\rightarrow\mathbb{R}^K$, every matrix $A\in\mathbb{R}^{T\times K}$, and every probability measure $P$ on $\mathcal{V}$, we have
    \begin{equation*}
        \|Aa\|_{P,2} \leq \|A\|_{\mathrm{sp}} \cdot \|a\|_{P,2}.
    \end{equation*}
\end{lemma}

\begin{lemma}
	\label{lem:spectral_inverse}
	For every matrix $A\in\mathbb{R}^{T\times K}$ with $\operatorname{rank}(A)\geq1$, the spectral norm of its Moore-Penrose inverse satisfies
	\begin{equation*}
		\|A^+\|_{\mathrm{sp}}
		=
		1/\sigma_{\min}^+(A).
	\end{equation*}
\end{lemma}

\begin{lemma}[Bound on Effective Degrees of Freedom]
    \label{lem:effective_dof}
    Let $\mathcal{N}_A(\lambda)$ be the effective degrees of freedom defined as $\mathcal{N}_A(\lambda)=\operatorname{tr}\left(A(A+\lambda^2 I)^{-1}\right)$ for a positive semi-definite matrix $A\in\mathbb{R}^{K\times K}$ and $\lambda>0$.
    Let $\tilde{A}\in\mathbb{R}^{K\times K}$ be another positive semi-definite matrix.
    Then, 
    \begin{equation*}
        \left|\mathcal{N}_A(\lambda)-\mathcal{N}_{\tilde{A}}(\lambda)\right| \leq \|A - \tilde{A}\|_* / \lambda^2.
    \end{equation*}
\end{lemma}

\begin{lemma}[Transformation on Effective Degrees of Freedom]
    \label{lem:transformation_dof}
    Let $\mathcal{N}_A(\lambda)$ be the effective degrees of freedom defined as in \cref{lem:effective_dof}.
    Let $Q\in\mathbb{R}^{K\times K}$ be another matrix.
    Then, 
    \begin{equation*}
        \mathcal{N}_{QAQ'}(\lambda)\leq \max\{1,\|Q\|_{\mathrm{sp}}^2\} \mathcal{N}_A(\lambda).
    \end{equation*}
\end{lemma}

\section{Proofs for Main Results}

\subsection{Proof for \cref{thm:transfer_learning_general}}
\begin{proof}
    By the triangle inequality, we have
    \begin{align}
        \left\|\hat{f}\circ \check{h}-a^*_{\ta}\right\|_{P_{\ta},2}
        &\leq \left\|\hat{f}\circ \check{h}-f_n^{\bullet}\circ \check{h}\right\|_{P_{\ta},2}
        +\left\|f_n^{\bullet}\circ \check{h}-f_n\circ h_m\right\|_{P_{\ta},2}
        + \left\|f_n\circ h_m - a_{\ta}^*\right\|_{P_{\ta},2}\nonumber\\
        &=: \text{(I)} + \text{(II)} + \text{(III)}.
        \label{eq:rate_decomposition}
    \end{align}
    We bound each term separately.

    By \cref{eq:target_convergence} in \cref{asm:convergence} (ii),
    \begin{equation*}
        \text{(I)}=O_{P}(r_{\ta,n}).
    \end{equation*}

    Next, we bound the second term (II). By \cref{eq:source_convergence} in \cref{asm:convergence} (i),
    \begin{equation*}
        \left\|\check{g}\circ \check{h}-g_m\circ h_m\right\|_{P_{\so},2}=O_{P_{\so}}(\delta_{\so,m}).
    \end{equation*}
    Since $\check{g}\in\mathcal{G}_m$, we have
    \begin{equation*}
		\min_{g\in\mathcal{G}_m}\left\|g\circ \check{h}-g_m\circ h_m\right\|_{P_{\so},2}=O_{P_{\so}}(\delta_{\so,m}).
    \end{equation*}
	Hence, for every $\varepsilon>0$, there exists a constant $1\leq M_{\varepsilon}<\infty$ such that
    \begin{equation*}
        \liminf_{m\rightarrow\infty}
        P_{\so}\left(
        \min_{g\in\mathcal{G}_m}\left\|g\circ \check{h}-g_m\circ h_m\right\|_{P_{\so},2}\leq M_{\varepsilon}\delta_{\so,m}
        \right)
        \geq
        1-\varepsilon.
    \end{equation*}
    Thus, we have
    \begin{equation*}
        \liminf_{m\rightarrow\infty}
        P_{\so}\left(
        \check{h}\in\mathcal{H}_{\so,m}(\rho_{\so,m})
        \right)
        \geq
		1-\varepsilon,
    \end{equation*}
	with $\rho_{\so,m}=M_{\varepsilon}\delta_{\so,m}\asymp\delta_{\so,m}$.
	Since the event $\{\check{h}\in\mathcal{H}_{\so,m}(\rho_{\so,m})\}$ depends only on the source sample, the above lower bound also holds under the joint law $P$. For all sufficiently large $n$, \cref{asm:transferability} and \cref{eq:diversity} imply on this event that $\check{h}\in\mathcal{H}_{\ta,m}(\rho_{\so,m}/\nu_n)$, i.e.,
    \begin{equation*}
        \min_{f\in\mathcal{F}_n^{\sub}}\left\|f\circ \check{h}-f_n\circ h_m\right\|_{P_{\ta},2}
        \leq
        \rho_{\so,m}/\nu_n = M_{\varepsilon}\delta_{\so,m}/\nu_n.
    \end{equation*}
    Therefore,
    \begin{equation*}
        \min_{f\in\mathcal{F}_n^{\sub}}\left\|f\circ \check{h}-f_n\circ h_m\right\|_{P_{\ta},2}
        =
        O_{P_{\so}}(\delta_{\so,m}/\nu_n).
    \end{equation*}
    Hence, the second term satisfies
    \begin{equation*}
        \text{(II)}=O_{P}(\delta_{\so,m}/\nu_n).
    \end{equation*}

    By \cref{def:approx_embedding},
    \begin{equation*}
        \text{(III)}=O(\epsilon_{\ta,n}).
    \end{equation*}
    Combining the three bounds yields
    \begin{equation*}
        \left\|\hat{f}\circ \check{h}-a_{\ta}^*\right\|_{P_{\ta},2}
        =
        O_{P}\left(r_{\ta,n}+\delta_{\so,m}/\nu_n+\epsilon_{\ta,n}\right),
    \end{equation*}
    which proves the claim.
\end{proof}

\subsection{Proof for \cref{thm:conv_nec}}
\begin{proof}
	For every $h\in\mathcal{H}_{\so,m}(0)$, define
	\begin{equation*}
		A_n(h)
		:=
		\min_{f\in\mathcal{F}_n^{\sub}}
		\|f\circ h-f_n\circ h_m\|_{P_{\ta},2}.
	\end{equation*}
	First, consider the deterministic sequence $h_m^\dagger=h_m$. By assumption, $f_n\circ h_m = f_{n,h_m}^{\bullet}\circ h_m$ $P_{\ta}$-almost surely.
	The assumptions and the triangle inequality then imply
	\begin{align*}
		\epsilon_{\ta,n}
		&=
		\|f_n\circ h_m-a_{\ta}^*\|_{P_{\ta},2}\\
		&\leq
		\|f_n\circ h_m-\widehat f_{n,h_m}\circ h_m\|_{P_{\ta},2}
		+
		\|\widehat f_{n,h_m}\circ h_m-a_{\ta}^*\|_{P_{\ta},2}\\
		&=
		o_{P_{\ta}}(1).
	\end{align*}
	Since $\epsilon_{\ta,n}$ is deterministic, this implies
	$\epsilon_{\ta,n}=o(1)$.

	Now fix any deterministic sequence
	$h_m^\dagger\in\mathcal{H}_{\so,m}(0)$. By the triangle inequality,
	\begin{align*}
		A_n(h_m^\dagger)
		&=
		\|f_{n,h_m^\dagger}^{\bullet}\circ h_m^\dagger
		-f_n\circ h_m\|_{P_{\ta},2}\\
		&\leq
		\|f_{n,h_m^\dagger}^{\bullet}\circ h_m^\dagger
		-\widehat f_{n,h_m^\dagger}\circ h_m^\dagger\|_{P_{\ta},2}\\
		&\quad+
		\|\widehat f_{n,h_m^\dagger}\circ h_m^\dagger
		-a_{\ta}^*\|_{P_{\ta},2}
		+\epsilon_{\ta,n}\\
		&=
		o_{P_{\ta}}(1).
	\end{align*}
	Because $A_n(h_m^\dagger)$ is deterministic,
	$A_n(h_m^\dagger)=o(1)$ for every such deterministic sequence.

	We claim that $\sup_{h\in\mathcal{H}_{\so,m}(0)}A_n(h) = o(1)$.
	Otherwise, there would exist $\varepsilon>0$, a subsequence	$\{n_j\}$, and deterministic functions
	$h_{m_{n_j}}^\dagger\in\mathcal{H}_{\so,m_{n_j}}(0)$ such that $A_{n_j}(h_{m_{n_j}}^\dagger) \geq \varepsilon$ for every $j$. Setting $h_m^\dagger=h_m$ outside this subsequence would produce a deterministic sequence contradicting $A_n(h_m^\dagger)=o(1)$. Finally, define $\rho_{\ta,n} := \sup_{h\in\mathcal{H}_{\so,m}(0)}A_n(h)$.
	Then $\rho_{\ta,n}=o(1)$ and, by definition, $\mathcal{H}_{\so,m}(0) \subseteq \mathcal{H}_{\ta,m}(\rho_{\ta,n})$.
\end{proof}

\subsection{Proof for \cref{thm:suf_transfer}}
\begin{proof}
    Pick any $h\in\mathcal{H}_{\so,m}(\rho_{\so,m})$. 
	Then, by the triangle inequality and $\|f_n \circ h_m(Z) - \Gamma \circ g_m \circ h_m(Z)\|_{P_{\ta},2}\leq\varrho_n$, we have, for every $f\in\mathcal{F}_n^{\sub}$,
    \begin{align*}
        & \left\|f \circ h - f_n \circ h_m\right\|_{P_{\ta},2}\\
        &\quad \leq \left\|f \circ h - \Gamma \circ g_m \circ h_m\right\|_{P_{\ta},2}
         + \left\|\Gamma \circ g_m \circ h_m - f_n \circ h_m\right\|_{P_{\ta},2}\\
        &\quad \leq \left\|f \circ h - \Gamma \circ g_m \circ h_m\right\|_{P_{\ta},2}
         + \varrho_n\\
        &\quad \leq \left\|f \circ h - \Gamma \circ g^{\bullet}_m \circ h\right\|_{P_{\ta},2}
         + \left\|\Gamma \circ g^{\bullet}_m \circ h - \Gamma \circ g_m \circ h_m\right\|_{P_{\ta},2}
		 + \varrho_n.
    \end{align*}
    
    Taking the minimum over $f\in\mathcal{F}_n^{\sub}$ on both sides, we have
    \begin{align*}
        & \min_{f\in\mathcal{F}_n^{\sub}} \left\|f \circ h - f_n \circ h_m\right\|_{P_{\ta},2}\\
        &\quad \leq \min_{f\in\mathcal{F}_n^{\sub}} \left\|f \circ h - \Gamma \circ g^{\bullet}_m \circ h\right\|_{P_{\ta},2}
         + \left\|\Gamma \circ g^{\bullet}_m \circ h - \Gamma \circ g_m \circ h_m\right\|_{P_{\ta},2}
         + \varrho_n\\
        &\quad \leq \left\|\Gamma \circ g^{\bullet}_m \circ h - \Gamma \circ g_m \circ h_m\right\|_{P_{\ta},2}
         + 2\varrho_n\\
        &\quad \leq L_{\Gamma} \left\|g^{\bullet}_m \circ h - g_m \circ h_m\right\|_{P_{\ta},2}
         + 2\varrho_n\\
        &\quad \leq L_{\Gamma} \rho_{\so,m}/\underline{w}_n + 2\varrho_n,
    \end{align*}
	where the second inequality follows from the existence of $f\in\mathcal{F}_n^{\sub}$ such that $\left\|f \circ h - \Gamma \circ g^{\bullet}_m \circ h\right\|_{P_{\ta},2} \leq\varrho_n$. The third inequality follows from the Lipschitz continuity of $\Gamma$, and the last inequality follows from \cref{asm:support} and the definition of $\mathcal{H}_{\so,m}(\rho_{\so,m})$.
\end{proof}

\subsection{Proof for \cref{cor:transferability_full_rank}}
\begin{proof}
    Define the function $\Gamma:\mathbb{R}^T\to\mathbb{R}$ by
    \begin{equation*}
        \Gamma(x)=f_n\left(B_{\so,m}^{+}(x-\alpha_{\so,m})\right).
    \end{equation*}
    Since $B_{\so,m}$ has full column rank, $B_{\so,m}^{+}B_{\so,m}=I_K$, and hence
    \begin{equation*}
        \Gamma\circ g_m\circ h_m
        =f_n\left(B_{\so,m}^{+}B_{\so,m}h_m\right)
        =f_n\circ h_m.
    \end{equation*}
    Moreover, $\Gamma$ has Lipschitz constant $L_{\Gamma} \leq L_{\mathcal{F}}\|B_{\so,m}^{+}\|_{\mathrm{sp}}$ since
	    \begin{align*}
	        \left|\Gamma(x_1)-\Gamma(x_2)\right|
            &\leq L_{\mathcal{F}}\left\|B_{\so,m}^{+}(x_1-x_2)\right\|_2 \\
	        &\leq L_{\mathcal{F}}\left\|B_{\so,m}^{+}\right\|_{\mathrm{sp}}\|x_1-x_2\|_2
	    \end{align*}
	    for every $x_1,x_2\in\mathbb{R}^{T}$, where the last inequality follows from \cref{lem:l2_ineq}.
    
    Fix $\rho_{\so,m}\geq0$ and $h\in\mathcal{H}_{\so,m}(\rho_{\so,m})$, and let $(\alpha_{\so,m}^{\bullet},B_{\so,m}^{\bullet})$ be the minimizer in the statement of the corollary. Define $g_m^{\bullet}(v)=\alpha_{\so,m}^{\bullet}+B_{\so,m}^{\bullet}v$. By the assumption, there exists $f\in\mathcal{F}_n^{\sub}$ such that
    \begin{align*}
        f\circ h
        &=f_n\left(B_{\so,m}^{+}
        (\alpha_{\so,m}^{\bullet}-\alpha_{\so,m}+B_{\so,m}^{\bullet}h)\right)\\
        &=\Gamma\circ g_m^{\bullet}\circ h,
    \end{align*}
    $P_{\ta}$-a.s. Therefore, by \cref{thm:suf_transfer}, $g_m$ is
    \begin{equation*}
        \left(\rho_{\so,m},
        \frac{L_{\mathcal{F}}\|B_{\so,m}^{+}\|_{\mathrm{sp}}}
        {\underline{w}_n}\rho_{\so,m}\right)
    \end{equation*}
    diverse over $f_n$ with respect to $h_m$. Finally, \cref{lem:spectral_inverse} gives $\|B_{\so,m}^{+}\|_{\mathrm{sp}}=1/\sigma_{\min}(B_{\so,m})$ because $B_{\so,m}$ has full column rank.
\end{proof}

\subsection{Proof for \cref{thm:nec_transfer_linear_pre}}
\begin{proof}
    Suppose that there is no such function $b$.
    This implies that $\beta_{\ta,n} \notin \lin\{B_{\so,m}'\}:=S$.
    Thus, there exists a unit vector $v\in S^{\perp}$, where $S^{\perp}$ is the orthogonal complement of $S$, such that $\beta_{\ta,n}' v \neq 0$.
	Define $P_v = I_{K} - v v'$ as the projection matrix onto the linear subspace orthogonal to $v$.
    Let $h(Z) = P_v h_m(Z) + c$, where $c\in\mathbb{R}^{K}$ is the vector such that $h\in\mathcal{H}_m$, and let $g(u) = \alpha_{\so,m} - B_{\so,m}c + B_{\so,m} u$.
    The head-class assumption in the theorem ensures that $g\in\mathcal{G}_m$.
    Then, we have
    \begin{align*}
        \|g \circ h - g_m \circ h_m\|_{P_{\so},2}^2
        &= \mathbb{E}_{\so}\left[\|\alpha_{\so,m} + B_{\so,m} (P_v h_m(Z) + c - c) - (\alpha_{\so,m} + B_{\so,m} h_m(Z))\|_2^2\right]\\
        &= \mathbb{E}_{\so}\left[\|B_{\so,m} (P_v h_m(Z) - h_m(Z))\|_2^2\right]\\
        &= \mathbb{E}_{\so}\left[\|B_{\so,m} v v' h_m(Z)\|_2^2\right]
        = 0,
    \end{align*}
    where the last equality follows from $v\in S^{\perp}$.
    Thus, $h\in\mathcal{H}_{\so,m}(0)$.
    On the other hand, every $f\in\mathcal{F}_n^{\sub}$ can be written as $f(u)=\alpha+\beta'u$ for some $\alpha\in\mathbb{R}$ and $\beta\in\mathbb{R}^{K}$. Hence,
    \begin{align*}
        & \min_{f\in\mathcal{F}_n^{\sub}}\|f \circ h - f_n \circ h_m\|_{P_{\ta},2}^2\\
        & \quad \geq \min_{\alpha\in\mathbb{R},\beta\in\mathbb{R}^{K}}\mathbb{E}_{\ta}\left[\left(
            \begin{pmatrix}
                \alpha + \beta' c - \alpha_{\ta,n}\\
                P_v \beta - \beta_{\ta,n}
            \end{pmatrix}'
            \begin{pmatrix}
                    1\\
                    h_m(Z)
            \end{pmatrix}\right)^2\right]\\
        & \quad \geq \mu_{\min}\left(\mathbb{E}_{\ta}\left[\begin{pmatrix}
                1\\
                h_m(Z)
            \end{pmatrix}
            \begin{pmatrix}
                1\\
                h_m(Z)
            \end{pmatrix}'\right]\right)
            \min_{\alpha\in\mathbb{R},\beta\in\mathbb{R}^{K}}\left\|\begin{pmatrix}
                \alpha + \beta' c - \alpha_{\ta,n}\\
                P_v \beta - \beta_{\ta,n}
            \end{pmatrix}\right\|_2^2\\
        & \quad \geq \mu_{\min}\left(\mathbb{E}_{\ta}\left[\begin{pmatrix}
                1\\
                h_m(Z)
            \end{pmatrix}
            \begin{pmatrix}
                1\\
                h_m(Z)
            \end{pmatrix}'\right]\right)(v' \beta_{\ta,n})^2
        > 0,
    \end{align*}
    where the first inequality follows because the target subclass contains only linear heads, and the second inequality follows from the property of the minimum eigenvalue $\mu_{\min}(\cdot)$, and the third inequality follows from
    \begin{align*}
        \min_{\alpha\in\mathbb{R},\beta\in\mathbb{R}^{K}}\left\|\begin{pmatrix}
                \alpha + \beta' c - \alpha_{\ta,n}\\
                P_v \beta - \beta_{\ta,n}
            \end{pmatrix}\right\|_2^2
        & \geq \min_{\beta\in\mathbb{R}^{K}}\|P_v \beta - \beta_{\ta,n}\|_2^2\\
        & = \min_{\beta\in\mathbb{R}^{K}}\|P_v (\beta - \beta_{\ta,n}) - v v' \beta_{\ta,n}\|_2^2\\
        & = \min_{\beta\in\mathbb{R}^{K}}\|P_v (\beta - \beta_{\ta,n})\|_2^2 + (v' \beta_{\ta,n})^2\\
        & \geq (v' \beta_{\ta,n})^2 > 0,
    \end{align*}
    and the positive definiteness of the second moment matrix of $(1,h_m(Z))$ under $P_{\ta}$.
    Thus, $h\notin\mathcal{H}_{\ta,m}(0)$.
\end{proof}

\subsection{Proof for \cref{thm:nec_transfer}}
\begin{proof}
	Suppose that there is no such function $\Gamma_0$.
	Thus, for every measurable function $\Gamma_0:\mathbb{R}^{K}\to\mathbb{R}$, we have $P_{\ta}\left(f_n \circ h_m(Z) \neq \Gamma_0 \circ h_0(Z)\right)>0$.
	The assumed square integrability of $f_n\circ h_m$ and $f\circ h_0$ under $P_{\ta}$ allows us to apply \cref{lem:transfer_var} with $P=P_{\ta}$ and $h=h_0$:
	\begin{align*}
		&\inf_{f\in\mathcal{F}_n^{\sub}}\|f \circ h_0 - f_n \circ h_m\|_{P_{\ta},2}^2\\
		&\quad = \mathbb{E}_{\ta}\left[\Var_{\ta}\left(f_n \circ h_m(Z) \mid h_0(Z)\right)\right]
		+ \inf_{f\in\mathcal{F}_n^{\sub}}\mathbb{E}_{\ta}\left[\left(\mathbb{E}_{\ta}\left[f_n \circ h_m(Z) \mid h_0(Z)\right]-f \circ h_0(Z)\right)^2\right]\\
		&\quad \geq \mathbb{E}_{\ta}\left[\Var_{\ta}\left(f_n \circ h_m(Z) \mid h_0(Z)\right)\right]\\
		&\quad = \mathbb{E}_{\ta}\left[(f_n \circ h_m(Z) - \mathbb{E}_{\ta}\left[f_n \circ h_m(Z) \mid h_0(Z)\right])^2\right].
	\end{align*}
	The last term is strictly positive since $P_{\ta}\left(f_n \circ h_m(Z) \neq \Gamma_0 \circ h_0(Z)\right)>0$ for every measurable function $\Gamma_0$.
	Thus, $h_0\notin\mathcal{H}_{\ta,m}(0)$, while we have $h_0\in\mathcal{H}_{\so,m}(0)$ by assumption, which contradicts the $(0,0)$-diversity of $g_m$ over $f_n$ with respect to $h_m$.
	This proves the first claim.
	For the second claim, suppose that the additional sigma-field inclusion condition holds.
	Since $h_0(Z)$ and $g_m\circ h_m(Z)$ take values in the standard Borel spaces $\mathbb{R}^{K}$ and $\mathbb{R}^{T}$, respectively, \cref{lem:doob_dynkin_completion}, applied with $U=g_m\circ h_m(Z)$, $V=h_0(Z)$, and $P=P_{\ta}$, implies that there exists a measurable function $\Gamma_1:\mathbb{R}^{T}\to\mathbb{R}^{K}$ such that $h_0(Z)=\Gamma_1\circ g_m\circ h_m(Z)$, $P_{\ta}$-a.s.
	The second claim then follows by defining $\Gamma := \Gamma_0 \circ \Gamma_1$.
\end{proof}

\subsection{Proof for \cref{thm:target_convergence}}
\begin{proof}
	By \cref{thm:transfer_learning_general}, it suffices to verify \cref{eq:target_convergence} with $r_{\ta,n}$ there replaced by $r_{\ta,n}+\delta_{\so,m}/\nu_n+\epsilon_{\ta,n}$, since adding the transfer and approximation terms again does not change the resulting stochastic order.
	For every $h\in\mathcal{H}_m$, write
	\begin{equation*}
		f_{n,h}^{\bullet}
		\in\argmin_{f\in\mathcal{F}_n^{\sub}}\|f\circ h-f_n\circ h_m\|_{P_{\ta},2},
	\end{equation*}
	and retain the shorthand $f_n^{\bullet}=f_{n,\check{h}}^{\bullet}$.
	By definition and the inclusion $\mathcal{F}_n^{\sub}\subseteq\mathcal{F}_n$, we have
	\begin{equation*}
		f_{n,h}^{\bullet}\in\mathcal{F}_n^{\sub}\subseteq\mathcal{F}_n.
	\end{equation*}
	Moreover, $\hat{f}\in\mathcal{F}_n$ by the statement of the theorem.
	If $\mathsf{V}(\mathcal{F}_n)=0$, the nonempty class $\mathcal{F}_n$ contains only one pointwise function. Then $\hat{f}=f_{n,h}^{\bullet}$ for every $h$, so the target estimation error is zero and the general theorem gives the result. We therefore consider $\mathsf{V}(\mathcal{F}_n)\geq1$ and enlarge $M_F$ to at least one if necessary.
	By \cref{eq:source_convergence} in \cref{asm:convergence} (i), for every $\varepsilon>0$, there exists a constant $1\leq M_{\varepsilon}<\infty$ such that
	\begin{equation*}
		\limsup_{m\rightarrow\infty}
		P_{\so}\left(
		\check{h}\notin\mathcal{H}_{\so,m}(M_{\varepsilon}\delta_{\so,m})
		\right)
		\leq \varepsilon.
	\end{equation*}

	We will apply \cref{lem:rate_of_conv} with $\mathcal{H}_{m,\varepsilon}=\mathcal{H}_{\so,m}(M_{\varepsilon}\delta_{\so,m})$, $\mathcal{F}_n$ as given in the theorem, $f_{n,h}^* = f_{n,h} = f_{n,h}^{\bullet}$, $\hat{f}_{n,h} = \hat{f}$, $P = P_{\ta}$, $M_n(f,h) = - \mathbb{E}_{\ta}[\ell_{\ta}(f\circ h(Z),Y)]$, $\mathbb{M}_n(f,h) = -\frac{1}{n}\sum_{i=1}^{n} \ell_{\ta}(f\circ h(Z_i),Y_i)$, $d_{n,h}(f,f') = \|f\circ h - f'\circ h\|_{P_{\ta},2}$, and $\tau_n=1$, where the notation on the left-hand side is that of \cref{lem:rate_of_conv}.\footnote{
		To apply \cref{lem:curvature_approx} below, extend these definitions to the best target predictor $a_{\ta}^*$ by setting, for every $f\in\mathcal{F}_n$ and $h\in\mathcal{H}_m$, $M_n(a_{\ta}^*,h) = -\mathbb{E}_{\ta}[\ell_{\ta}(a_{\ta}^*(Z),Y)]$, $\mathbb{M}_n(a_{\ta}^*,h) = -\frac{1}{n}\sum_{i=1}^n \ell_{\ta}(a_{\ta}^*(Z_i),Y_i)$, $d_{n,h}(f,a_{\ta}^*) = d_{n,h}(a_{\ta}^*,f) = \|f\circ h-a_{\ta}^*\|_{P_{\ta},2}$, and $d_{n,h}(a_{\ta}^*,a_{\ta}^*) = 0$.
	}

	By \cref{lem:pointwise_measurable}, the pointwise measurability condition in \cref{lem:rate_of_conv} is satisfied.
	Note that $f_{n,h}^{\bullet}$ and $\hat{f}$ depend on $h$ through the minimization problem given $h$.

	Since $f_n^{\bullet}\circ \check{h}$ need not equal $a_{\ta}^*$, we cannot directly apply the curvature condition in \cref{asm:loss_ta} to $f_n^{\bullet}\circ \check{h}$ to verify \cref{eq:rate_of_conv_curvature}.
	For this purpose, we first apply \cref{lem:curvature_approx} with $\mathcal{H}_{m,\varepsilon}=\mathcal{H}_{\so,m}(M_{\varepsilon}\delta_{\so,m})$, $\mathcal{F}_n$ as given in the theorem, $f_{n,h} = f_{n,h}^{\bullet}$ and $f_{n,h}^* = a_{\ta}^*$.
	We can bound
	\begin{align}
		& \sup_{h\in\mathcal{H}_{m,\varepsilon}}\|f_{n,h}^{\bullet}\circ h - a_{\ta}^*\|_{P_{\ta},2}\nonumber\\
		& \quad \leq \sup_{h\in\mathcal{H}_{\ta,m}(M_\varepsilon\delta_{\so,m}/\nu_n)} \|f_{n,h}^{\bullet}\circ h - f_n \circ h_m\|_{P_{\ta},2} + \|f_n \circ h_m - a_{\ta}^*\|_{P_{\ta},2}\nonumber\\
		& \quad \leq M_{\varepsilon}\delta_{\so,m}/\nu_n + \epsilon_{\ta,n},
		\label{eq:curvature_approx_conv_bound}
	\end{align}
	where the first inequality follows from the triangle inequality and the inclusion $\mathcal{H}_{m,\varepsilon}\subseteq\mathcal{H}_{\ta,m}(M_{\varepsilon}\delta_{\so,m}/\nu_n)$, which holds for all sufficiently large $n$ by \cref{asm:transferability} since $M_{\varepsilon}\delta_{\so,m}\asymp\delta_{\so,m}$. The second inequality follows from $f_{n,h}^{\bullet} \in \argmin_{f\in\mathcal{F}_n^{\sub}} \|f\circ h - f_n \circ h_m\|_{P_{\ta},2}$, the definition of $\mathcal{H}_{\ta,m}(\cdot)$, and \cref{def:approx_embedding}.
	If we set $\underline{\delta}_n = \delta_{\so,m}/\nu_n + \epsilon_{\ta,n}$, we obtain the conditions in \cref{lem:curvature_approx} with $\tau_n=1$ by the curvature condition in \cref{asm:loss_ta}.
	Thus, \cref{lem:curvature_approx} implies that there exists a constant $C_{\varepsilon,0}>0$ such that for every $\delta > C_{\varepsilon,0}(\delta_{\so,m}/\nu_n + \epsilon_{\ta,n})$, \cref{eq:rate_of_conv_curvature} in \cref{lem:rate_of_conv} holds with $f_{n,h} = f_{n,h}^* = f_{n,h}^{\bullet}$ and $\tau_n=1$.

	Next, we verify \cref{eq:rate_of_conv_modulus}.
	By \cref{lem:maximal_lipschitz} with $P=P_{\ta}$, $\ell=\ell_{\ta}$, $C_{\ell}=C_{\ell,\ta}$, and $M=M_F$,
	\begin{align*}
		&\sup_{h\in\mathcal{H}_m}\mathbb{E} \left[\sup_{f \in \mathcal{F}_n: d_{n,h}\left(f, f_{n,h}^*\right) < \delta} \sqrt{n}\left|\left(\mathbb{M}_n-M_n\right)(f,h)-\left(\mathbb{M}_n-M_n\right)\left(f_{n,h}^*,h\right)\right|\right] \\
		&\quad \lesssim C_{\ell,\ta}M_F \left(J\left(\delta/M_F,\mathcal{F}_n | M_F, L_2\right)
		+ \frac{M_F^2 J^2\left(\delta/M_F,\mathcal{F}_n | M_F, L_2\right)}{\delta^2 \sqrt{n}}\right).
	\end{align*}
	For $0<\delta\leq M_F$, \cref{lem:vc_covering} gives
	\begin{align*}
		J\left(\delta/M_F,\mathcal{F}_n | M_F,L_2\right)
		&\lesssim \int_0^{\delta/M_F}\sqrt{\mathsf{V}(\mathcal{F}_n)\log(e/\varepsilon)}d\varepsilon\\
		&\lesssim \delta M_F^{-1}\sqrt{\mathsf{V}(\mathcal{F}_n)\log(eM_F/\delta)}.
	\end{align*}
	For $\delta>M_F$, split the entropy integral at one and use monotonicity of covering numbers to obtain $J(\delta/M_F,\mathcal{F}_n | M_F,L_2)\lesssim(\delta/M_F)\sqrt{\mathsf{V}(\mathcal{F}_n)}$.
	For $n\geq2$, define the following modulus on every nonnegative radius:
	\begin{equation*}
		\phi_n(\delta):=\delta\sqrt{n}r_{\ta,n}+M_F\sqrt{n}r_{\ta,n}^2.
	\end{equation*}
	For $\delta\geq M_F/\sqrt{n}$, the entropy bounds above and $\log(en)\lesssim\log(n)$ show that the preceding empirical-process expectation is bounded by a constant multiple of $\phi_n(\delta)$, uniformly in $h$. For $0<\delta<M_F/\sqrt{n}$, nesting of the localized classes bounds that expectation by its value at $M_F/\sqrt{n}$. Since $\mathsf{V}(\mathcal{F}_n)\geq1$,
	\begin{equation*}
		\phi_n(M_F/\sqrt{n})
		\lesssim M_F\sqrt{n}r_{\ta,n}^2
		\leq\phi_n(\delta).
	\end{equation*}
	Thus \cref{eq:rate_of_conv_modulus} holds for every positive radius. The function $\phi_n$ is increasing, and
	\begin{equation*}
		\frac{\phi_n(\delta)}{\delta}
		=\sqrt{n}r_{\ta,n}+\frac{M_F\sqrt{n}r_{\ta,n}^2}{\delta}
	\end{equation*}
	is nonincreasing for $\delta>0$, so we may take $\xi=1$ in \cref{lem:rate_of_conv}. Moreover, $\phi_n(\delta)\lesssim\sqrt{n}\delta^2$ whenever $\delta\geq r_{\ta,n}$. Taking $\delta_n=r_{\ta,n}+\delta_{\so,m}/\nu_n+\epsilon_{\ta,n}$ also satisfies the approximation and optimization requirements in \cref{eq:rate_of_conv_delta_n} and the approximate maximization condition, by \cref{eq:target_app_minimizer} and $f_{n,h}=f_{n,h}^*$.
	By \cref{lem:rate_of_conv}, we obtain the following convergence rate using the condition \cref{eq:rate_of_conv_delta_n}:
	\begin{equation}
		\left\|\hat{f}\circ \check{h} - f_n^{\bullet} \circ \check{h}\right\|_{P_{\ta},2}=O_P\left(\sqrt{\frac{\mathsf{V}(\mathcal{F}_n) \log(n)}{n}} +  \delta_{\so,m}/\nu_n + \epsilon_{\ta,n}\right).
	\end{equation}
	Thus, \cref{asm:convergence} (ii) holds with rate $r_{\ta,n}+\delta_{\so,m}/\nu_n+\epsilon_{\ta,n}$ rather than $r_{\ta,n}$ alone. Applying \cref{thm:transfer_learning_general} yields \cref{eq:target_convergence_rate}, since adding the transfer and approximation terms once more does not change the resulting stochastic order. This completes the proof.
\end{proof}

\subsection{Proof for \cref{thm:source_convergence}}
\begin{proof}
	If all coordinate classes $\mathcal{H}_{m,k}$ and $\mathcal{G}_{m,t}$ have VC dimension zero, the composite class $\mathcal{G}_m\circ\mathcal{H}_m$ contains only one pointwise function. Then $\check{g}\circ\check{h}=g_m\circ h_m$, and the error relative to $a_{\so}^*$ equals $\epsilon_{\so,m}$. Both conclusions follow since $\tau_m\leq1$. We therefore consider the case in which at least one coordinate VC dimension is positive.
	As in the proof of \cref{thm:target_convergence}, we will apply \cref{lem:rate_of_conv}, but here the argument is unconditional on $\mathcal{I}$.
	Thus, we set $f_{m}^* = f_{m} = g_m \circ h_m$, $\hat{f}_{m} = \check{g} \circ \check{h}$, and $P = P_{\so}$. We define the objective functions for joint optimization over $g$ and $h$, rather than optimization over $f$ for fixed $h$.
	For every $a,\widetilde a\in(\mathcal{G}_m\circ\mathcal{H}_m)\cup\{a_{\so}^*\}$,
	define $M_m(a) = -\mathbb{E}_{\so}[\ell_{\so}(a(Z),S)]$,
	$\mathbb{M}_m(a) = -\frac{1}{m}\sum_{j=1}^m\ell_{\so}(a(Z_j),S_j)$, and
	$d_m(a,\widetilde a) = \|a-\widetilde a\|_{P_{\so},2}$.

	The curvature condition \cref{eq:rate_of_conv_curvature} is verified as in the proof of \cref{thm:target_convergence}, by applying \cref{lem:curvature_approx} with $f_m=g_m\circ h_m$, $f_m^*=a_{\so}^*$, and $\underline{\delta}_m=\epsilon_{\so,m}$. That is, there exists a constant $C>0$ such that for every $\delta>\tau_m^{-1}C\epsilon_{\so,m}$, \cref{eq:rate_of_conv_curvature} holds.

	The constant function $M_G$ is an envelope for $\mathcal{G}_m\circ\mathcal{H}_m$ (enlarge it so that $M_G\geq\max\{1,M_H\}$ if necessary). By \cref{lem:entropy_property,lem:entropy_property_composite,lem:vc_covering} and the uniform Lipschitz condition in \cref{asm:loss_so} (iii), for $0<\varepsilon\leq1$ and every discrete probability measure $Q$,
	\begin{align*}
		1+\log N\left(\varepsilon M_G,\mathcal{G}_m\circ\mathcal{H}_m,\|\cdot\|_{Q,2}\right)
		&\lesssim \sum_{k=1}^K\mathsf{V}(\mathcal{H}_{m,k})\log\left(\frac{eL_{\mathcal{G},m}K}{\varepsilon}\right)\\
		&\quad+\sum_{t=1}^T\mathsf{V}(\mathcal{G}_{m,t})\log\left(\frac{eT}{\varepsilon}\right).
	\end{align*}
	The positive coordinate dimension ensures that the right-hand side also bounds the constant one. Integrating as in the target proof, and using monotonicity of covering numbers beyond radius one, gives, for every $\delta>0$,
	\begin{align*}
		&J\left(\delta,\mathcal{G}_m\circ\mathcal{H}_m | M_G,L_2\right)\\
		&\quad\lesssim\delta\left[\sum_{k=1}^K\mathsf{V}(\mathcal{H}_{m,k})\log\left(\frac{eL_{\mathcal{G},m}K}{\delta\wedge1}\right)
		+\sum_{t=1}^T\mathsf{V}(\mathcal{G}_{m,t})\log\left(\frac{eT}{\delta\wedge1}\right)\right]^{1/2}.
	\end{align*}
	For $m\geq2$, define
	\begin{equation*}
		\phi_m(\delta):=\delta\tau_m\sqrt{m}r_{\so,m}+M_G\tau_m^2\sqrt{m}r_{\so,m}^2,\qquad \delta\geq0.
	\end{equation*}
	For $\delta\geq M_G/\sqrt{m}$, apply \cref{lem:maximal_lipschitz} with $\mathcal{F}_n=\mathcal{G}_m\circ\mathcal{H}_m$, $C_\ell=C_{\ell,\so}$, and $M=M_G$. The entropy bound and \cref{eq:source_convergence_r} imply \cref{eq:rate_of_conv_modulus} with this $\phi_m$, since the logarithms are bounded by constant multiples of $\log(mL_{\mathcal{G},m}K)$ and $\log(mT)$. For a smaller positive radius, nesting bounds the empirical-process expectation by its value at $M_G/\sqrt{m}$. Since at least one coordinate VC dimension is positive,
	\begin{equation*}
		\phi_m(M_G/\sqrt{m})
		\lesssim M_G\tau_m^2\sqrt{m}r_{\so,m}^2
		\leq\phi_m(\delta).
	\end{equation*}
	Hence the modulus bound holds at every positive radius. The function $\phi_m$ is increasing and $\phi_m(\delta)/\delta$ is nonincreasing, so the rate lemma applies with $\xi=1$. For $\delta\geq r_{\so,m}$,
	\begin{equation*}
		\phi_m(\delta/\tau_m)
		=\delta\sqrt{m}r_{\so,m}+M_G\tau_m^2\sqrt{m}r_{\so,m}^2
		\lesssim\sqrt{m}\delta^2.
	\end{equation*}
	Thus $\delta_m=r_{\so,m}+\epsilon_{\so,m}$ satisfies the rate and approximation conditions in \cref{eq:rate_of_conv_delta_n}; the approximate maximization condition follows from \cref{eq:source_app_minimizer}.
	By \cref{lem:rate_of_conv}, we obtain the following convergence rate using the condition \cref{eq:rate_of_conv_delta_n}:
	\begin{equation}
		\left\|\check{g}\circ \check{h} - g_m \circ h_m\right\|_{P_{\so},2}=O_{P_{\so}}\left(r_{\so,m}/\tau_m+\epsilon_{\so,m}/\tau_m\right).
	\end{equation}
	Since $\|g_m\circ h_m-a_{\so}^*\|_{P_{\so},2}\leq\epsilon_{\so,m}$ and $\tau_m\leq1$, the triangle inequality gives \cref{eq:source_convergence_rate}.
	This completes the proof.
\end{proof}

\subsection{Proof for \cref{lem:convergence_h}}
\begin{proof}
	By the definitions of $\check{h}^{\mathrm{proj}}$, $\check{q}$, and $\check{Q}$,
	\begin{align}
		\left\|\check{h}^{\mathrm{proj}}-(\check{q}+\check{Q}h_m)\right\|_{P_{\so},2}
		&=
		\left\|\check{B}_{\so}^+\left(\check{\alpha}_{\so}+\check{B}_{\so}\check{h}-(\alpha_{\so,m}+B_{\so,m}h_m)\right)\right\|_{P_{\so},2}\nonumber\\
		&\leq
		\|\check{B}_{\so}^+\|_{\mathrm{sp}}
		\left\|\check{\alpha}_{\so}+\check{B}_{\so}\check{h}-(\alpha_{\so,m}+B_{\so,m}h_m)\right\|_{P_{\so},2}.
		\label{eq:convergence_h_bound}
	\end{align}
	The inequality follows from \cref{lem:l2_ineq}. By \cref{asm:convergence} (i), the second factor on the right-hand side of \cref{eq:convergence_h_bound} is $O_{P_{\so}}(\delta_{\so,m})$. Combining this rate with $\|\check{B}_{\so}^{+}\|_{\mathrm{sp}}=O_{P_{\so}}(\underline{\sigma}_m^{-1})$ proves the stated result. Finally, on the event in the sufficient condition, \cref{lem:spectral_inverse} gives $\|\check{B}_{\so}^{+}\|_{\mathrm{sp}}\leq\underline{\sigma}_m^{-1}$, which verifies the condition in \cref{eq:convergence_h_condition}.
\end{proof}

\subsection{Proof for \cref{thm:ridge_target_convergence}}
\begin{proof}
	For each $h\in\mathcal{H}_m^{\mathrm{proj}}$, select a source minimizer
	\begin{equation*}
		(\alpha_{\so,m,h}^{\bullet},B_{\so,m,h}^{\bullet})
		\in
		\argmin_{(\alpha,B)\in\Theta_{\so,m}}
		\left\|\alpha+Bh-(\alpha_{\so,m}+B_{\so,m}h_m)\right\|_{P_{\so},2}.
	\end{equation*}
	For every $h\in\mathcal{H}_m^{\mathrm{proj}}$, define the target-head parameters through the selected source minimizer by
	\begin{align*}
		\alpha_{\ta,n,h}^{\bullet}
		&:=
		\alpha_{\ta,n}
		+
		\beta_{\ta,n}'B_{\so,m}^+
		(\alpha_{\so,m,h}^{\bullet}-\alpha_{\so,m}),\\
		\beta_{\ta,n,h}^{\bullet\prime}
		&:=
		\beta_{\ta,n}'B_{\so,m}^+B_{\so,m,h}^{\bullet},\\
		f_{n,h}^{\bullet}(v)
		&:=
		\alpha_{\ta,n,h}^{\bullet}
		+
		\beta_{\ta,n,h}^{\bullet\prime}v.
	\end{align*}
	By definition, $(\alpha_{\so,m,h}^{\bullet},B_{\so,m,h}^{\bullet})\in\Theta_{\so,m}$. Indeed, \cref{asm:ridge_target} implies, uniformly over $h\in\mathcal{H}_m^{\mathrm{proj}}$,
	\begin{align}
		|\alpha_{\ta,n,h}^{\bullet}|
		&\leq
		|\alpha_{\ta,n}|
		+\|b_n'B_{\so,m}B_{\so,m}^+\|_2
		\|\alpha_{\so,m,h}^{\bullet}-\alpha_{\so,m}\|_2
		\leq
		|\alpha_{\ta,n}|+2C_bC_{\alpha,\so}
		\leq
		C_{\alpha},\nonumber\\
		\|\beta_{\ta,n,h}^{\bullet}\|_2
		&\leq
		\|b_n'B_{\so,m}B_{\so,m}^+\|_2
		\|B_{\so,m,h}^{\bullet}\|_{\mathrm{sp}}
		\leq
		C_b\bar{\sigma}
		\leq
		C_{\beta}.
		\label{eq:transported_target_parameter_bounds}
	\end{align}
	Hence, $f_{n,h}^{\bullet}\in\mathcal{F}_n$ for every $h\in\mathcal{H}_m^{\mathrm{proj}}$.
	For each fixed $h\in\mathcal{H}_m^{\mathrm{proj}}$, also define the estimator family
	\begin{equation*}
		\hat{f}_{n,h}
		\in
		\argmin_{f\in\mathcal{F}_n}
		\left\{
			\frac{1}{n}\sum_{i=1}^n\ell_{\ta}(f\circ h(Z_i),Y_i)
			+
			\lambda_n^2\|\beta\|_2^2
		\right\},
	\end{equation*}
	such that $\hat{f}_{n,\check{h}^{\mathrm{proj}}}=\hat{f}_n^{\mathrm{ridge}}$, and let $(\hat{\alpha}_{\ta,n,h},\hat{\beta}_{\ta,n,h})$ denote the corresponding parameters.
	
	As in \cref{eq:rate_decomposition}, we decompose the estimation error of the target model into three terms.
	\begin{align*}
		&\|\hat{f}_n^{\mathrm{ridge}}\circ\check{h}^{\mathrm{proj}}-a_{\ta}^*\|_{P_{\ta},2}\\
		&\quad\leq
		\|\hat{f}_n^{\mathrm{ridge}}\circ\check{h}^{\mathrm{proj}}-f_{n,\check{h}^{\mathrm{proj}}}^{\bullet}\circ\check{h}^{\mathrm{proj}}\|_{P_{\ta},2}\\
		&\qquad+
		\|f_{n,\check{h}^{\mathrm{proj}}}^{\bullet}\circ\check{h}^{\mathrm{proj}}-f_n\circ h_m\|_{P_{\ta},2}
		+
		\|f_n\circ h_m-a_{\ta}^*\|_{P_{\ta},2}\\
		&\quad=: \text{(I)}+\text{(II)}+\text{(III)}.
	\end{align*}

	By \cref{def:approx_embedding}, $\text{(III)}=\epsilon_{\ta,n}$. To bound $\text{(II)}$, first observe that
	\begin{equation*}
		\beta_{\ta,n}'B_{\so,m}^+
		=
		b_n'B_{\so,m}B_{\so,m}^+,
		\qquad
		\|\beta_{\ta,n}'B_{\so,m}^+\|_2
		\leq
		C_b.
	\end{equation*}
	The same identity implies $\beta_{\ta,n}'B_{\so,m}^+B_{\so,m}=\beta_{\ta,n}'$. Hence, for every $h\in\mathcal{H}_m^{\mathrm{proj}}$,
	\begin{align*}
		&\|\alpha_{\ta,n,h}^{\bullet}+\beta_{\ta,n,h}^{\bullet\prime}h-(\alpha_{\ta,n}+\beta_{\ta,n}'h_m)\|_{P_{\ta},2}\\
		&\quad=
		\|\beta_{\ta,n}'B_{\so,m}^+(\alpha_{\so,m,h}^{\bullet}+B_{\so,m,h}^{\bullet}h-(\alpha_{\so,m}+B_{\so,m}h_m))\|_{P_{\ta},2}\\
		&\quad\leq
		\frac{C_b}{\underline{w}_n}
		\|\alpha_{\so,m,h}^{\bullet}+B_{\so,m,h}^{\bullet}h-(\alpha_{\so,m}+B_{\so,m}h_m)\|_{P_{\so},2},
	\end{align*}
	where the inequality follows from \cref{asm:support,lem:l2_ineq}. Moreover,
	\begin{equation*}
		\check{B}_{\so}\check{h}^{\mathrm{proj}}
		=
		\check{B}_{\so}\check{B}_{\so}^+\check{B}_{\so}\check{h}
		=
		\check{B}_{\so}\check{h}.
	\end{equation*}
	By \cref{asm:ridge_target} (i), $(\check{\alpha}_{\so},\check{B}_{\so})\in\Theta_{\so,m}$ on the event under consideration and is therefore feasible in the source minimization defining $(\alpha_{\so,m,\check{h}^{\mathrm{proj}}}^{\bullet},B_{\so,m,\check{h}^{\mathrm{proj}}}^{\bullet})$. Consequently, \cref{asm:convergence} (i) gives
	\begin{align}
		\|g_{\alpha_{\so,m,\check{h}^{\mathrm{proj}}}^{\bullet},B_{\so,m,\check{h}^{\mathrm{proj}}}^{\bullet}}\circ\check{h}^{\mathrm{proj}}-g_m\circ h_m\|_{P_{\so},2}
		&\leq
		\|\check{g}\circ\check{h}^{\mathrm{proj}}-g_m\circ h_m\|_{P_{\so},2}\nonumber\\
		&=
		\|\check{g}\circ\check{h}-g_m\circ h_m\|_{P_{\so},2}\nonumber\\
		&=
		O_{P_{\so}}(\delta_{\so,m}).
		\label{eq:source_convergence_ridge}
	\end{align}
	It follows that $\text{(II)}=O_{P_{\so}}(\delta_{\so,m}/\underline{w}_n)$.

	It remains to bound $\text{(I)}$ uniformly over a high-probability source-side set. For each $h\in\mathcal{H}_m^{\mathrm{proj}}$, select
	\begin{equation*}
		(q_h,Q_h)
		\in
		\argmin_{\substack{q\in\mathbb{R}^{K},\ Q\in\mathbb{R}^{K\times K}:\\
		\|Q\|_{\mathrm{sp}}\leq\bar{\sigma}/\underline{\sigma}}}
		\|h-(q+Qh_m)\|_{P_{\so},2}.
	\end{equation*}
	The source parameter restriction implies $\|\check{Q}\|_{\mathrm{sp}}\leq\bar{\sigma}/\underline{\sigma}$, so $(\check{q},\check{Q})$ is feasible. It also implies $\|\check{B}_{\so}^+\|_{\mathrm{sp}}\leq\underline{\sigma}^{-1}$. The minimizing property of $(q_{\check{h}^{\mathrm{proj}}},Q_{\check{h}^{\mathrm{proj}}})$ and \cref{lem:convergence_h}, applied with $\underline{\sigma}_m=\underline{\sigma}$, give
	\begin{align}
		\|\check{h}^{\mathrm{proj}}-q_{\check{h}^{\mathrm{proj}}}-Q_{\check{h}^{\mathrm{proj}}}h_m\|_{P_{\so},2}
		&\leq
		\|\check{h}^{\mathrm{proj}}-(\check{q}+\check{Q}h_m)\|_{P_{\so},2}\nonumber\\
		&=
		O_{P_{\so}}(\delta_{\so,m}/\underline{\sigma})\nonumber\\
		&=
		O_{P_{\so}}(\delta_{\so,m}),
		\label{eq:h_convergence_ridge}
	\end{align}
	where the last equality uses that $\underline{\sigma}$ is fixed.
	For every $\varepsilon>0$, \cref{eq:source_convergence_ridge,eq:h_convergence_ridge} therefore give a constant $1\leq M_{\varepsilon}<\infty$ such that
	\begin{equation*}
		\limsup_{m\rightarrow\infty}
		P_{\so}(\check{h}^{\mathrm{proj}}\notin\mathcal{H}_{m,\varepsilon})
		\leq
		\varepsilon,
	\end{equation*}
	where
	\begin{align}
		\mathcal{H}_{m,\varepsilon}
		:=
		\left\{h\in\mathcal{H}_m^{\mathrm{proj}}:
		\right.&
		\|g_{\alpha_{\so,m,h}^{\bullet},B_{\so,m,h}^{\bullet}}\circ h-g_m\circ h_m\|_{P_{\so},2}
		\leq
		M_{\varepsilon}\delta_{\so,m},\nonumber\\
		&\left.
		\|h-(q_h+Q_hh_m)\|_{P_{\so},2}
		\leq
		M_{\varepsilon}\delta_{\so,m}
		\right\}.
		\label{eq:H_m_epsilon_ridge}
	\end{align}
	For every $h\in\mathcal{H}_{m,\varepsilon}$, the selected source head belongs to $\mathcal{G}_m$. Therefore, the first restriction defining $\mathcal{H}_{m,\varepsilon}$ implies
	\begin{equation*}
		\min_{g\in\mathcal{G}_m}
		\|g\circ h-g_m\circ h_m\|_{P_{\so},2}
		\leq
		M_{\varepsilon}\delta_{\so,m}.
	\end{equation*}
	Consequently, $\mathcal{H}_{m,\varepsilon}\subseteq\mathcal{H}_{\so,m}^{\mathrm{proj}}(M_{\varepsilon}\delta_{\so,m})$. Since $M_{\varepsilon}\delta_{\so,m}\asymp\delta_{\so,m}$, the target-loss conditions assumed in the theorem apply on this set for all sufficiently large $n$.

	For each fixed $h$, the finite-dimensional parameterization of $\mathcal{F}_n$ is pointwise continuous, and every subset of $\Theta_n$ is separable. Hence, the localized classes required by \cref{lem:rate_of_conv_regu} are pointwise measurable.
	We apply \cref{lem:rate_of_conv_regu} with $\mathcal{H}_m=\mathcal{H}_m^{\mathrm{proj}}$, $\check{h}=\check{h}^{\mathrm{proj}}$, $\mathcal{F}_n^*=\mathcal{F}_n$, $f_{n,h}=f_{n,h}^*=f_{n,h}^{\bullet}$, $\tau_n=1$, the estimator family $\hat{f}_{n,h}$ defined above, $\mathcal{J}_n(f_{\alpha,\beta})=\|\beta\|_2$, and $\hat{\lambda}_{n,h}=\lambda_n$.\footnote{
		The event in \cref{asm:ridge_target} (i) has probability approaching one. For the formal application of \cref{lem:rate_of_conv_regu}, redefine $\check{h}^{\mathrm{proj}}$ on the complement of this event as $B_{\so,m}^+B_{\so,m}h_m\in\mathcal{H}_m^{\mathrm{proj}}$ and define the corresponding ridge estimator by the same criterion. These modified objects agree with the original objects with probability approaching one, so this extension does not affect the claimed stochastic order.
		}
	Set $M_n(f,h) := -\mathbb{E}_{\ta}[\ell_{\ta}(f\circ h(Z),Y)]$, $\mathbb{M}_n(f,h) := -\frac{1}{n}\sum_{i=1}^n\ell_{\ta}(f\circ h(Z_i),Y_i)$, and $d_{n,h}(f,f') := \|f\circ h-f'\circ h\|_{P_{\ta},2}$.
	We next verify curvature directly around $f_{n,h}^{\bullet}$. The transported-head bound above and \cref{def:approx_embedding} give
	\begin{equation*}
		\sup_{h\in\mathcal{H}_{m,\varepsilon}}
		\|f_{n,h}^{\bullet}\circ h-a_{\ta}^*\|_{P_{\ta},2}
		\leq
		C_bM_{\varepsilon}\frac{\delta_{\so,m}}{\underline{w}_n}+\epsilon_{\ta,n}.
	\end{equation*}
	Set
	\begin{equation*}
		\underline{\delta}_n
		:=
		\frac{\delta_{\so,m}}{\underline{w}_n}+\epsilon_{\ta,n}.
	\end{equation*}
	For $\delta\geq C_{\varepsilon,0}\underline{\delta}_n$ with a sufficiently large $C_{\varepsilon,0}$ and any $f\in\mathcal{F}_n$ satisfying
	\begin{equation*}
		\delta/2
		<
		\|f\circ h-f_{n,h}^{\bullet}\circ h\|_{P_{\ta},2}
		\leq
		\delta,
	\end{equation*}
	the reverse triangle inequality gives $\|f\circ h-a_{\ta}^*\|_{P_{\ta},2}\geq\delta/4$. Let $c_L,c_U>0$ denote uniform constants in the lower and upper excess-risk bounds in the first condition of \cref{asm:loss_ta}. Then,
	\begin{align*}
		M_n(f,h)-M_n(f_{n,h}^{\bullet},h)
		&\leq
		-c_L\|f\circ h-a_{\ta}^*\|_{P_{\ta},2}^2
		+c_U\|f_{n,h}^{\bullet}\circ h-a_{\ta}^*\|_{P_{\ta},2}^2\\
		&\lesssim_{\varepsilon}
		-\delta^2,
	\end{align*}
	where the last inequality follows by increasing $C_{\varepsilon,0}$ if necessary. This proves \cref{eq:rate_of_conv_regu_curvature}.

	The sample ridge estimator satisfies
	\begin{equation*}
		\mathbb{M}_n(\hat{f}_{n,h},h)-\lambda_n^2\mathcal{J}_n^2(\hat{f}_{n,h})
		\geq
		\mathbb{M}_n(f_{n,h}^{\bullet},h)-\lambda_n^2\mathcal{J}_n^2(f_{n,h}^{\bullet}),
	\end{equation*}
	so \cref{eq:rate_of_conv_regu_approx_max} holds with zero approximation error.

	We next verify the empirical-process modulus. For each $h$, let
	\begin{equation*}
		A_{h,\lambda,n}
		:=
		\Sigma_{h,n}+\lambda_n^2I_K.
	\end{equation*}
	Let $\xi_1,\ldots,\xi_n$ be i.i.d. Rademacher random variables independent of the target and source samples, and define
	\begin{equation*}
		\mathfrak{R}_n(a)
		:=
		\frac{1}{\sqrt{n}}\sum_{i=1}^n\xi_i a(Z_i)
	\end{equation*}
	for every measurable function $a$.
	For $f_{\alpha,\beta}\in\mathcal{F}_n$, write $\Delta\alpha:=\alpha-\alpha_{\ta,n,h}^{\bullet}$ and $\Delta\beta:=\beta-\beta_{\ta,n,h}^{\bullet}$. For each fixed $i$ and $h$, the loss-increment map evaluated between $f_{\alpha,\beta}\circ h(Z_i)$ and $f_{n,h}^{\bullet}\circ h(Z_i)$ vanishes at zero and is $C_{\ell,\ta}$-Lipschitz by the Lipschitz condition in \cref{asm:loss_ta}. Symmetrization and \cref{lem:contraction}, applied before enlarging the parameter set, yield
	\begin{align}
		&\sup_{h\in\mathcal{H}_{m,\varepsilon}}
		\mathbb{E}_{\ta}\left[
		\sup_{\substack{f\in\mathcal{F}_n:
		d_{n,h}(f,f_{n,h}^{\bullet})<\delta\\
		\mathcal{J}_n(f)<\delta/\lambda_n}}
		\sqrt{n}\left|(\mathbb{M}_n-M_n)(f,h)-(\mathbb{M}_n-M_n)(f_{n,h}^{\bullet},h)\right|
		\right]\nonumber\\
		&\quad\leq
		4C_{\ell,\ta}
		\sup_{h\in\mathcal{H}_{m,\varepsilon}}
		\mathbb{E}_{\xi,\ta}\left[
		\sup_{\substack{f_{\alpha,\beta}\in\mathcal{F}_n:
		\|\Delta\alpha+\Delta\beta'h\|_{P_{\ta},2}<\delta\\
		\lambda_n\|\beta\|_2<\delta}}
		|\mathfrak{R}_n(\Delta\alpha+\Delta\beta'h)|
		\right]\nonumber\\
		&\quad\leq
		4C_{\ell,\ta}
		\sup_{h\in\mathcal{H}_{m,\varepsilon}}
		\mathbb{E}_{\xi,\ta}\left[
		\sup_{\substack{(\Delta\alpha,\Delta\beta):
		\|\Delta\alpha+\Delta\beta'h\|_{P_{\ta},2}<\delta\\
		\lambda_n\|\Delta\beta+\beta_{\ta,n,h}^{\bullet}\|_2<\delta}}
		|\mathfrak{R}_n(\Delta\alpha+\Delta\beta'h)|
		\right].
		\label{eq:ridge_multiplier_modulus_bound}
	\end{align}
	The last inequality is the only step that enlarges the compact coefficient class. For every increment in the supremum,
	\begin{equation*}
		\mathfrak{R}_n(\Delta\alpha+\Delta\beta'h)
		=
		\mathbb{E}_{\ta}[\Delta\alpha+\Delta\beta'h]\mathfrak{R}_n(1)
		+
		\Delta\beta'\mathfrak{R}_n(h-\mathbb{E}_{\ta}[h]).
	\end{equation*}
	The first coefficient is bounded by $\delta$ by the Cauchy-Schwarz inequality, and $\mathbb{E}_{\xi,\ta}[|\mathfrak{R}_n(1)|]\leq1$. For the second term,
	\begin{equation*}
		|\Delta\beta'\mathfrak{R}_n(h-\mathbb{E}_{\ta}[h])|
		\leq
		\|A_{h,\lambda,n}^{1/2}\Delta\beta\|_2
		\|A_{h,\lambda,n}^{-1/2}\mathfrak{R}_n(h-\mathbb{E}_{\ta}[h])\|_2.
	\end{equation*}
	The variance identity gives
	\begin{equation*}
		\Delta\beta'\Sigma_{h,n}\Delta\beta
		=
		\operatorname{Var}_{\ta}(\Delta\beta'h(Z))
		\leq
		\|\Delta\alpha+\Delta\beta'h\|_{P_{\ta},2}^2
		<
		\delta^2.
	\end{equation*}
	In addition, \cref{eq:transported_target_parameter_bounds} implies
	\begin{equation*}
		\lambda_n\|\Delta\beta\|_2
		\leq
		\lambda_n\|\Delta\beta+\beta_{\ta,n,h}^{\bullet}\|_2
		+
		\lambda_n\|\beta_{\ta,n,h}^{\bullet}\|_2
		<
		\delta+C_{\beta}\lambda_n.
	\end{equation*}
	Consequently,
	\begin{equation*}
		\|A_{h,\lambda,n}^{1/2}\Delta\beta\|_2
		\lesssim
		\delta+\lambda_n.
	\end{equation*}
	Moreover, Jensen's inequality and the Rademacher second-moment identity imply
	\begin{align*}
		&\mathbb{E}_{\xi,\ta}\left[
		\|A_{h,\lambda,n}^{-1/2}\mathfrak{R}_n(h-\mathbb{E}_{\ta}[h])\|_2
		\right]\\
		&\quad\leq
		\operatorname{tr}\left(A_{h,\lambda,n}^{-1/2}\Sigma_{h,n}A_{h,\lambda,n}^{-1/2}\right)^{1/2}\\
		&\quad=
		\mathcal{N}_{h,n}(\lambda_n)^{1/2}.
	\end{align*}
	It remains to control this effective dimension uniformly over $h\in\mathcal{H}_{m,\varepsilon}$.

	Define $r_h:=h-q_h-Q_hh_m$ and $u_h:=r_h-\mathbb{E}_{\ta}[r_h]$. Conditional on the source sample, $q_h$ and $Q_h$ are constant under $P_{\ta}$, and
	\begin{align*}
		\Sigma_{h,n}
		&=
		Q_h\Sigma_{h_m,n}Q_h'+\Delta_{h,n},\\
		\Delta_{h,n}
		&=
		\mathbb{E}_{\ta}[u_hu_h']
		+Q_h\mathbb{E}_{\ta}[(h_m-\mathbb{E}_{\ta}[h_m])u_h']
		+\mathbb{E}_{\ta}[u_h(h_m-\mathbb{E}_{\ta}[h_m])']Q_h'.
	\end{align*}
	By \cref{asm:support,eq:H_m_epsilon_ridge},
	\begin{equation*}
		\sup_{h\in\mathcal{H}_{m,\varepsilon}}\|r_h\|_{P_{\ta},2}
		\leq
		M_{\varepsilon}\frac{\delta_{\so,m}}{\underline{w}_n}.
	\end{equation*}
	Because $\mathbb{E}_{\ta}\|u_h\|_2^2\leq\mathbb{E}_{\ta}\|r_h\|_2^2$, $\|Q_h\|_{\mathrm{sp}}\leq\bar{\sigma}/\underline{\sigma}$, and $\|h_m-\mathbb{E}_{\ta}[h_m]\|_{P_{\ta},2}=\operatorname{tr}(\Sigma_{h_m,n})^{1/2}$, the nuclear-norm bounds in \cref{lem:unitary_ineq,lem:rank_one} and the Cauchy-Schwarz inequality give
	\begin{equation*}
		\sup_{h\in\mathcal{H}_{m,\varepsilon}}\|\Delta_{h,n}\|_*
		\lesssim_{\varepsilon}
		\left(\frac{\delta_{\so,m}}{\underline{w}_n}\right)^2
		+
		\operatorname{tr}(\Sigma_{h_m,n})^{1/2}\frac{\delta_{\so,m}}{\underline{w}_n}.
	\end{equation*}
	Applying \cref{lem:effective_dof,lem:transformation_dof} and absorbing the constants depending on $\bar{\sigma},\underline{\sigma}$ now yields
	\begin{align*}
		\sup_{h\in\mathcal{H}_{m,\varepsilon}}\mathcal{N}_{h,n}(\lambda_n)
		&\lesssim_{\varepsilon}
		\mathcal{N}_{h_m,n}(\lambda_n)
		+
		\frac{(\delta_{\so,m}/\underline{w}_n)^2+\operatorname{tr}(\Sigma_{h_m,n})^{1/2}\delta_{\so,m}/\underline{w}_n}{\lambda_n^2},\\
		\sup_{h\in\mathcal{H}_{m,\varepsilon}}\mathcal{N}_{h,n}(\lambda_n)^{1/2}
		&\lesssim_{\varepsilon}
		\mathcal{N}_{h_m,n}(\lambda_n)^{1/2}
		+
		\frac{\delta_{\so,m}/\underline{w}_n+(\delta_{\so,m}/\underline{w}_n)^{1/2}\operatorname{tr}(\Sigma_{h_m,n})^{1/4}}{\lambda_n}.
	\end{align*}
	Combining these bounds with \cref{eq:ridge_multiplier_modulus_bound} verifies \cref{eq:rate_of_conv_regu_modulus} with
	\begin{equation*}
		\phi_n(\delta)
		:=
		(\delta+\lambda_n)
		\left(
		1
		+\mathcal{N}_{h_m,n}(\lambda_n)^{1/2}
		+\frac{\delta_{\so,m}/\underline{w}_n+(\delta_{\so,m}/\underline{w}_n)^{1/2}\operatorname{tr}(\Sigma_{h_m,n})^{1/4}}{\lambda_n}
		\right).
	\end{equation*}
	This function is increasing, and $\phi_n(\delta)/\delta$ is decreasing, so the modulus requirement in \cref{lem:rate_of_conv_regu} holds with $\xi=1<2$. Choose $\delta_n$ to be a sufficiently large constant multiple of
	\begin{equation*}
		\sqrt{\frac{1+\mathcal{N}_{h_m,n}(\lambda_n)}{n}}
		+\lambda_n
		+\frac{\delta_{\so,m}/\underline{w}_n+(\delta_{\so,m}/\underline{w}_n)^{1/2}\operatorname{tr}(\Sigma_{h_m,n})^{1/4}}{\sqrt{n}\lambda_n}
		+\frac{\delta_{\so,m}}{\underline{w}_n}
		+\epsilon_{\ta,n}.
	\end{equation*}
	Choose the multiplicative constant large enough that $C_{\beta}\lambda_n<\delta_n$. Then $\phi_n(\delta_n)\lesssim_{\varepsilon}\sqrt{n}\delta_n^2$, $\delta_n\gtrsim_{\varepsilon}\underline{\delta}_n$, and
	\begin{equation*}
		\sup_{h\in\mathcal{H}_{m,\varepsilon}}
		\lambda_n\mathcal{J}_n(f_{n,h}^{\bullet})
		\leq
		C_{\beta}\lambda_n
		<
		\delta_n.
	\end{equation*}
	Thus, the remaining requirements in \cref{eq:rate_of_conv_regu_delta_n} hold, and \cref{lem:rate_of_conv_regu} gives
	\begin{align*}
		\text{(I)}
		&=
		\|\hat{f}_{n,\check{h}^{\mathrm{proj}}}\circ\check{h}^{\mathrm{proj}}-f_{n,\check{h}^{\mathrm{proj}}}^{\bullet}\circ\check{h}^{\mathrm{proj}}\|_{P_{\ta},2}\\
		&=
		O_P\left(\delta_n+\lambda_n\mathcal{J}_n(f_{n,\check{h}^{\mathrm{proj}}}^{\bullet})\right)\\
		&=
		O_P(\delta_n).
	\end{align*}
	Combining this result with $\text{(II)}=O_P(\delta_{\so,m}/\underline{w}_n)$ and $\text{(III)}=\epsilon_{\ta,n}$ proves the theorem.
\end{proof}

\subsection{Proof for \cref{thm:transferability_test_size}}
\label{subsec:proof_transferability_test_size}
\begin{proof}
	\textbf{Step 1: Joint asymptotic linearity.}
	We first establish \cref{eq:transferability_test_expansion}. By the source bound in \cref{eq:transferability_test_source_rate}, for every $\varepsilon>0$ there is a constant $M>0$, depending on $\varepsilon$ but not on $n$, such that
	\begin{equation*}
		\Pr_P\left(\max_{1\leq j\leq R}\left\|\check{g}_j\circ\check{h}_j-g_m\circ h_m\right\|_{P_{\so},2}\leq M\delta_{\so,m}\right)\geq1-\varepsilon
	\end{equation*}
	for all sufficiently large $n$. Since $\check{g}_j\in\mathcal{G}_m$, the minimum source approximation error for $\check{h}_j$ satisfies
	\begin{equation*}
		\min_{g\in\mathcal{G}_m}\left\|g\circ\check{h}_j-g_m\circ h_m\right\|_{P_{\so},2}
		\leq\left\|\check{g}_j\circ\check{h}_j-g_m\circ h_m\right\|_{P_{\so},2}.
	\end{equation*}
	Thus, on the event in the probability bound, every $\check{h}_j$ belongs to $\mathcal{H}_{\so,m}(M\delta_{\so,m})$ by the definition of the source near-identified set. The maximum in the source bound controls this membership simultaneously over all embeddings. This event depends only on the source sample and pre-training randomness, whose law is the same for every $P\in\mathcal{P}_{0,n}$. Hence, the probability bound holds uniformly over $P\in\mathcal{P}_{0,n}$. On this event, the triangle inequality and the definition of $f_{\diamond,j,n}^{\bullet}$ give
	\begin{align*}
		\max_{1\leq j\leq R}\max_{\diamond\in\{\gamma,\alpha\}}\left\|\hat{f}_{\diamond,j}\circ\check{h}_j-a_{\diamond}^*\right\|_{P_{\ta},2}
		&\leq\max_{1\leq j\leq R}\max_{\diamond\in\{\gamma,\alpha\}}\left\|\hat{f}_{\diamond,j}\circ\check{h}_j-f_{\diamond,j,n}^{\bullet}\circ\check{h}_j\right\|_{P_{\ta},2}\\
		&\quad+\max_{\diamond\in\{\gamma,\alpha\}}\sup_{h\in\mathcal{H}_{\so,m}(M\delta_{\so,m})}\min_{f\in\mathcal{F}_{\diamond,n}^{\sub}}
		\left\|f\circ h-f_{\diamond,n}\circ h_m\right\|_{P_{\ta},2}\\
		&\quad+\max_{\diamond\in\{\gamma,\alpha\}}\epsilon_{\diamond,n},
	\end{align*}
	where the maxima range over $j=1,\ldots,R$ and $\diamond\in\{\gamma,\alpha\}$. For each fixed $M$ and all $n\geq N_M$, \cref{eq:transferability_test_uniform_approximation} bounds the middle term by $\max_{\diamond\in\{\gamma,\alpha\}}(M\delta_{\so,m}/\nu_{\diamond,n})$, uniformly over $P_{\ta}\in\mathcal{P}_{0,\ta,n}$. By \cref{eq:transferability_test_target_rate,eq:transferability_test_rate}, the right-hand side is therefore $o_P(n^{-1/4})$ uniformly over $P\in\mathcal{P}_{0,n}$. The Lipschitz constants of the prediction maps $\Lambda_{\diamond}$ are uniformly bounded and $\diamond^*=\Lambda_{\diamond}\circ a_{\diamond}^*$. Since $\varepsilon$ can be arbitrarily small, applying these prediction maps yields the nuisance-scale bound
	\begin{equation}
		\label{eq:transferability_test_nuisance_rate}
		\max_{1\leq j\leq R}\max_{\diamond\in\{\gamma,\alpha\}}
		\left\|\Lambda_{\diamond}\circ\hat{f}_{\diamond,j}\circ\check{h}_j-\diamond^*\right\|_{P_{\ta},2}
		=o_P(n^{-1/4}),
	\end{equation}
	uniformly over $P\in\mathcal{P}_{0,n}$.

	We next bound the effect of estimation error of the nuisance functions on the sample moment at $\theta^*$. We first bound the population mean of the score error $\Delta\psi_j(W)$ and then the deviation of its sample average from the population mean. For the former, we use \cref{asm:transferability_test}(iii). By \cref{eq:transferability_test_nuisance_set}, all fitted nuisance pairs belong to $\mathcal{N}_n$ with probability approaching one uniformly over $P\in\mathcal{P}_{0,n}$.
	On this event, for each $j$, a second-order Taylor expansion around zero, together with the population moment condition and Neyman orthogonality in \cref{eq:transferability_test_orthogonality}, gives, for some $r\in(0,1)$,
	\begin{align*}
		&\mathbb{E}_{\ta}[\psi(W;\theta^*,\hat{\gamma}_j,\hat{\alpha}_j)]\\
		&=\frac{1}{2}\frac{\partial^2}{\partial r^2}\mathbb{E}_{\ta}\left[\psi\left(W;\theta^*,\gamma^*+r(\hat{\gamma}_j-\gamma^*),\alpha^*+r(\hat{\alpha}_j-\alpha^*)\right)\right].
	\end{align*}
	By \cref{eq:transferability_test_second_order}, each fitted population moment is therefore bounded in absolute value by $o(n^{-1/2})$, with the bound uniform over $P_{\ta}\in\mathcal{P}_{0,\ta,n}$ and all nuisance pairs in $\mathcal{N}_n$. Hence, this bound applies simultaneously to all fitted population moments on the event in \cref{eq:transferability_test_nuisance_set}. Since that event has probability approaching one uniformly over $P\in\mathcal{P}_{0,n}$, and the true score has population mean zero, we obtain
	\begin{equation}
		\label{eq:transferability_test_score_bias}
		\max_j\left|\mathbb{E}_{\ta}[\Delta\psi_j(W)]\right|
		=\max_j\left|\mathbb{E}_{\ta}[\psi(W;\theta^*,\hat{\gamma}_j,\hat{\alpha}_j)]\right|
		=o_P(n^{-1/2}),
	\end{equation}
	uniformly over $P\in\mathcal{P}_{0,n}$.

	To analyze the sampling variation in the score errors, we first condition on everything used to construct the nuisance estimators, including the pre-trained models, the partition, the observations in $\mathcal{I}_3$, and the auxiliary pre-training randomness.
	By \cref{asm:transferability_test}(i), the functions $\Delta\psi_j$ and sample sizes are then fixed, while observations in each estimation sample remain i.i.d. from $P_{\ta}$. For each $s\in\{1,2\}$, the symmetrization inequality in Lemma 2.3.1 and the maximal inequality in Lemma 2.2.2 of \cite{van2023weak}, applied with the Orlicz norm given by $x\mapsto e^{x^2}-1$, bound the conditional expectation of
	\begin{equation}
		\label{eq:transferability_test_centered_score}
		\max_{1\leq j\leq R}\left|\frac{1}{\sqrt{n_s}}\sum_{i\in\mathcal{I}_s}
		\left\{\Delta\psi_j(W_i)-\mathbb{E}_{\ta}[\Delta\psi_j(W)]\right\}\right|
	\end{equation}
	by a universal constant times
	\begin{equation*}
		\sqrt{\log R}\left\|\max_{1\leq j\leq R}|\Delta\psi_j(W)|\right\|_{P_{\ta},2}.
	\end{equation*}
	To obtain this bound, let $\xi_i$, $i\in\mathcal{I}_s$, be independent Rademacher variables, independent of all other randomness. All expectations below remain conditional on the quantities used to construct the nuisance estimators. Write $\mathbb{E}_{\xi}$ for expectation over the Rademacher variables, $\mathbb{E}_{\ta}$ for expectation over the target observations with the nuisance functions fixed, and $\mathbb{E}_{\xi,\ta}$ for expectation over both. Symmetrization gives
	\begin{align}
		&\mathbb{E}_{\ta}\left[\max_{1\leq j\leq R}\left|\frac{1}{\sqrt{n_s}}\sum_{i\in\mathcal{I}_s}\left\{\Delta\psi_j(W_i)-\mathbb{E}_{\ta}[\Delta\psi_j(W)]\right\}\right|\right]\notag\\
		\label{eq:transferability_test_score_symmetrization}
		&\leq2\mathbb{E}_{\xi,\ta}\left[\max_{1\leq j\leq R}\left|\frac{1}{\sqrt{n_s}}\sum_{i\in\mathcal{I}_s}\xi_i\Delta\psi_j(W_i)\right|\right].
	\end{align}
	Denote the Orlicz norm given by $x\mapsto e^{x^2}-1$ under the conditional law of the signs, with the estimation observations fixed, by $\|\cdot\|_{e^{x^2}-1}$. Hoeffding's inequality (\citealp{van2023weak}, Lemma 2.2.8) implies
	\begin{equation}
		\label{eq:transferability_test_hoeffding_orlicz}
		\left\|\frac{1}{\sqrt{n_s}}\sum_{i\in\mathcal{I}_s}\xi_i\Delta\psi_j(W_i)\right\|_{e^{x^2}-1}
		\leq \sqrt{6}\left(\frac{1}{n_s}\sum_{i\in\mathcal{I}_s}|\Delta\psi_j(W_i)|^2\right)^{1/2}.
	\end{equation}
	Since the Orlicz norm bounds the expected absolute value up to a universal constant, we obtain
	\begin{align}
		&\mathbb{E}_{\xi}\left[\max_{1\leq j\leq R}\left|\frac{1}{\sqrt{n_s}}\sum_{i\in\mathcal{I}_s}\xi_i\Delta\psi_j(W_i)\right|\right]\notag\\
		&\lesssim \left\|\max_{1\leq j\leq R}\left|\frac{1}{\sqrt{n_s}}\sum_{i\in\mathcal{I}_s}\xi_i\Delta\psi_j(W_i)\right|\right\|_{e^{x^2}-1}\notag\\
		&\lesssim \sqrt{\log(1+2R)}\max_{1\leq j\leq R}\left\|\frac{1}{\sqrt{n_s}}\sum_{i\in\mathcal{I}_s}\xi_i\Delta\psi_j(W_i)\right\|_{e^{x^2}-1}\notag\\
		\label{eq:transferability_test_rademacher_maximum}
		&\lesssim \sqrt{\log R}\left(\frac{1}{n_s}\sum_{i\in\mathcal{I}_s}\max_{1\leq j\leq R}|\Delta\psi_j(W_i)|^2\right)^{1/2},
	\end{align}
	where the second inequality follows from the maximal inequality applied to the $2R$ signed sums and the last inequality follows from \cref{eq:transferability_test_hoeffding_orlicz}, $R\geq2$, and bounding each squared score error by the pointwise maximum across embeddings. Taking expectations over the estimation observations and applying Jensen's inequality gives
	\begin{align}
		&\mathbb{E}_{\ta}\left[\left(\frac{1}{n_s}\sum_{i\in\mathcal{I}_s}\max_{1\leq j\leq R}|\Delta\psi_j(W_i)|^2\right)^{1/2}\right]\notag\\
		&\leq\left(\mathbb{E}_{\ta}\left[\frac{1}{n_s}\sum_{i\in\mathcal{I}_s}\max_{1\leq j\leq R}|\Delta\psi_j(W_i)|^2\right]\right)^{1/2}\notag\\
		\label{eq:transferability_test_score_jensen}
		&=\left\|\max_{1\leq j\leq R}|\Delta\psi_j(W)|\right\|_{P_{\ta},2}.
	\end{align}
	Combining \cref{eq:transferability_test_score_symmetrization,eq:transferability_test_rademacher_maximum,eq:transferability_test_score_jensen} gives the stated conditional expectation bound:
	\begin{align}
		&\mathbb{E}_{\ta}\left[\max_{1\leq j\leq R}\left|\frac{1}{\sqrt{n_s}}\sum_{i\in\mathcal{I}_s}\left\{\Delta\psi_j(W_i)-\mathbb{E}_{\ta}[\Delta\psi_j(W)]\right\}\right|\right]\notag\\
		\label{eq:transferability_test_conditional_expectation_bound}
		&\lesssim\sqrt{\log R}\left\|\max_{1\leq j\leq R}|\Delta\psi_j(W)|\right\|_{P_{\ta},2}.
	\end{align}
	By \cref{eq:transferability_test_score_envelope}, the right-hand side of \cref{eq:transferability_test_conditional_expectation_bound} is $o_P(1)$ uniformly over $P\in\mathcal{P}_{0,n}$.

	Conditional Markov's inequality and Lemma 6.1(a) of \cite{chernozhukov2018double} therefore imply that \cref{eq:transferability_test_centered_score} is $o_P(1)$ unconditionally. This conclusion is uniform over $P\in\mathcal{P}_{0,n}$ since the conditional expectation bound holds along every sequence of laws drawn from these classes, and the lemma applies along every such sequence. Since there are only two estimation samples and $n_s/n\xrightarrow{P}\pi_s>0$ uniformly, the largest centered score average is $o_P(n^{-1/2})$ uniformly over $j,s$ and $P\in\mathcal{P}_{0,n}$. Adding the population bias in \cref{eq:transferability_test_score_bias} gives
	\begin{equation}
		\label{eq:transferability_test_moment_expansion}
		\hat{M}_{j,s}(\theta^*):=n_s^{-1}\sum_{i\in\mathcal{I}_s}\psi(W_i;\theta^*,\hat{\gamma}_j,\hat{\alpha}_j)
		=\frac{1}{n_s}\sum_{i\in\mathcal{I}_s}\psi(W_i;\theta^*,\gamma^*,\alpha^*)+o_P(n^{-1/2}),
	\end{equation}
	uniformly over $j,s$ and $P\in\mathcal{P}_{0,n}$.
	Conditional on the partition, the leading term in \cref{eq:transferability_test_moment_expansion} has mean zero and variance $J^2V/n_s\leq\overline{J}^2\overline{V}/n_s$ by \cref{asm:transferability_test}(i) and (iii). Chebyshev's inequality and the sample proportion condition therefore imply that this average is $O_P(n^{-1/2})$ uniformly over $s$ and $P\in\mathcal{P}_{0,n}$.
	By \cref{eq:transferability_test_parameter_localization}, all estimators lie in $\mathcal{U}_n(P)$ with probability approaching one uniformly over $P\in\mathcal{P}_{0,n}$. On this event, the segment between $\theta^*$ and each $\hat{\theta}_{j,s}$ lies in $\mathcal{U}_n(P)$. The sample estimating equations in \cref{asm:transferability_test}(iv), the mean value theorem, and \cref{eq:transferability_test_local_derivative} give
	\begin{equation*}
		-\hat{M}_{j,s}(\theta^*)=(J+o_P(1))(\hat{\theta}_{j,s}-\theta^*),
	\end{equation*}
	uniformly over $j,s$ and $P\in\mathcal{P}_{0,n}$. Evaluating the derivative bound in \cref{eq:transferability_test_local_derivative} at each $\hat{\theta}_{j,s}$ gives $\max_{j,s}|\hat{J}_{j,s}-J|=o_P(1)$ uniformly over $P\in\mathcal{P}_{0,n}$. The preceding moment expansion therefore yields \cref{eq:transferability_test_expansion}, uniformly over $j,s$ and $P\in\mathcal{P}_{0,n}$:
	\begin{align*}
		\hat{\theta}_{j,s}-\theta^*
		&=-(J+o_P(1))^{-1}\hat{M}_{j,s}(\theta^*)\\
		&=\frac{1}{n_s}\sum_{i\in\mathcal{I}_s}\phi(W_i)+o_P(n^{-1/2}).
	\end{align*}

	\textbf{Step 2: Limiting distribution of the statistic.}
	Define
	\begin{equation}
		\label{eq:transferability_test_weight_difference}
		a_n:=\max_{j,s}\left\{n_s\sum_{i\in\mathcal{I}_s}
		\left(\hat{q}_{ij,s}-\hat{q}_{i1,s}\right)^2\right\}^{1/2}.
	\end{equation}
	By \cref{asm:transferability_test}(v), $a_n\sqrt{\log R}=o_P(1)$ and hence $a_n=o_P(1)$, both uniformly over $P\in\mathcal{P}_{0,n}$. The reverse triangle inequality gives, uniformly over $P\in\mathcal{P}_{0,n}$,
	\begin{equation}
		\label{eq:transferability_test_block_sd_bound}
		\max_{j,s}\left|
		\sqrt{n_s\sum_{i\in\mathcal{I}_s}\hat{q}_{ij,s}^2}
		-\sqrt{n_s\sum_{i\in\mathcal{I}_s}\hat{q}_{i1,s}^2}
		\right|
		\leq a_n=o_P(1).
	\end{equation}
	Combining \cref{eq:transferability_test_block_sd_bound}, consistency of the variance estimators for the first embedding function in \cref{asm:transferability_test}(v), and the upper bound $V\leq\overline{V}$ gives, uniformly over $P\in\mathcal{P}_{0,n}$,
	\begin{equation}
		\label{eq:transferability_test_block_variance}
		\max_{j,s}\left|n_s\sum_{i\in\mathcal{I}_s}\hat{q}_{ij,s}^2-V\right|=o_P(1).
	\end{equation}
	By \cref{eq:transferability_test_block_variance}, the definition of $\hat{\sigma}_{jk}$ in \cref{eq:transferability_test_standard_error}, and the common lower bound $V\geq\underline{V}>0$, we obtain, uniformly over $P\in\mathcal{P}_{0,n}$,
	\begin{equation}
		\label{eq:transferability_test_contrast_variance}
		\max_{j<k}\left|\frac{\hat{\sigma}_{jk}^2}{V(n_1^{-1}+n_2^{-1})}-1\right|=o_P(1).
	\end{equation}
	The bounds $V\geq\underline{V}>0$ and $|J|\geq\underline{J}>0$, \cref{eq:transferability_test_contrast_variance}, and consistency of the derivative estimators imply that the probability that any $\hat{\sigma}_{jk}$ or $\hat{J}_{j,s}$ equals zero tends to zero uniformly over $P\in\mathcal{P}_{0,n}$. These are the stopping conditions in \cref{alg:bootstrap_transferability}, and therefore the probability that the algorithm stops tends to zero uniformly over $P\in\mathcal{P}_{0,n}$. We have shown \cref{thm:transferability_test_size}(i).

	Combining \cref{eq:transferability_test_expansion,eq:transferability_test_contrast_variance} with the sample proportion condition in \cref{asm:transferability_test}(i) gives, uniformly over $j<k$ and $P\in\mathcal{P}_{0,n}$,
	\begin{equation}
		\label{eq:transferability_test_contrast_expansion}
		\frac{\hat{\theta}_{j,1}-\hat{\theta}_{k,2}}{\hat{\sigma}_{jk}}
		=\frac{n_1^{-1}\sum_{i\in\mathcal{I}_1}\phi(W_i)-n_2^{-1}\sum_{i\in\mathcal{I}_2}\phi(W_i)}{\sqrt{V(n_1^{-1}+n_2^{-1})}}+o_P(1).
	\end{equation}
	By \cref{asm:transferability_test}(i), conditional on the partition, the common leading term in \cref{eq:transferability_test_contrast_expansion} is a sum of independent random variables with population mean zero and total variance one. To establish asymptotic normality, we verify Lyapunov's condition. On the event that $n_s/n\geq\pi_s/2$ for $s\in\{1,2\}$, we have
	\begin{align*}
		&\sum_{s=1}^2\sum_{i\in\mathcal{I}_s}
		\mathbb{E}\left[
		\left|\frac{\phi(W_i)}{n_s\sqrt{V(n_1^{-1}+n_2^{-1})}}\right|^{2+\zeta}
		\mid\mathcal{I}_1,\mathcal{I}_2,\mathcal{I}_3
		\right]\\
		&=\frac{n_1^{-1-\zeta}+n_2^{-1-\zeta}}{(n_1^{-1}+n_2^{-1})^{1+\zeta/2}}
		\mathbb{E}_{\ta}\left[\left|\frac{\phi(W)}{\sqrt{V}}\right|^{2+\zeta}\right]
		=O(n^{-\zeta/2}),
	\end{align*}
	uniformly over $P_{\ta}\in\mathcal{P}_{0,\ta,n}$ by \cref{eq:transferability_test_moment_bound}. The sample proportion condition implies that this event has probability approaching one uniformly over $P\in\mathcal{P}_{0,n}$. Note that the bound for the Lyapunov condition is uniform over partitions satisfying these sample-size bounds and over sequences of target distributions. Applying Lyapunov's central limit theorem conditionally on the partition and then averaging over the partition therefore gives convergence of the leading term to a standard normal variable uniformly over $P\in\mathcal{P}_{0,n}$. 
	
	On the event that the algorithm does not stop, $\mathcal{T}$ defined in \cref{eq:transferability_test_statistic} differs from the absolute value of this leading term by at most $o_P(1)$ uniformly over $P\in\mathcal{P}_{0,n}$, as shown in \cref{eq:transferability_test_contrast_expansion}. 
	By the continuous mapping theorem and Slutsky's lemma, together with the uniformly vanishing stopping probability, $\mathcal{T}$ converges in distribution to the absolute value of a standard normal random variable: $\sup_{P\in\mathcal{P}_{0,n}}\left|\Pr_P(\mathcal{T}\leq t)-(2\Phi(t)-1)\right|\longrightarrow0$ for every fixed $t\geq0$, where $\Phi$ denotes the standard normal distribution function.
	To obtain uniformity in $t$, we take a finite grid of points starting at zero, independent of $n$ and $P$, such that each interval between consecutive points and the region above the last point has small probability under the limiting distribution of $\mathcal{T}$. By monotonicity of distribution functions, the supremum error over $t\geq0$ is bounded by the largest error at the grid points plus an arbitrarily small error. Since convergence at each grid point is uniform over $P$ and the grid is finite, we obtain
	\begin{equation}
		\label{eq:transferability_test_statistic_limit}
		\sup_{P\in\mathcal{P}_{0,n}}\sup_{t\geq0}
		\left|\Pr_P(\mathcal{T}\leq t)-(2\Phi(t)-1)\right|\longrightarrow0.
	\end{equation}

	\textbf{Step 3: Conditional bootstrap distribution.}
	Define the stacked vectors
	\begin{equation*}
		v_{jk}:=\begin{pmatrix}
			(\hat{q}_{ij,1})_{i\in\mathcal{I}_1}\\
			(-\hat{q}_{ik,2})_{i\in\mathcal{I}_2}
		\end{pmatrix},
		\qquad
		v_0:=v_{11}.
	\end{equation*}
	For the Euclidean norm $\|\cdot\|_2$, we have $\|v_{jk}\|_2=\hat{\sigma}_{jk}$ for $j<k$ by the definition in \cref{eq:transferability_test_standard_error}. By \cref{eq:transferability_test_block_variance,eq:transferability_test_contrast_variance}, the squared norms of all $v_{jk}$, $j<k$, and $v_0$, divided by $V(n_1^{-1}+n_2^{-1})$, converge to one uniformly in probability. Thus, with probability approaching one uniformly, the algorithm does not stop and these norms are positive. On this event, the reverse triangle inequality and \cref{eq:transferability_test_weight_difference} give uniformly over $P\in\mathcal{P}_{0,n}$,
	\begin{align}
		\max_{j<k}\left\|\frac{v_{jk}}{\|v_{jk}\|_2}-\frac{v_0}{\|v_0\|_2}\right\|_2
		&\leq2\max_{j<k}\frac{\|v_{jk}-v_0\|_2}{\|v_0\|_2}\notag\\
		&\leq\frac{2a_n\sqrt{n_1^{-1}+n_2^{-1}}}{\|v_0\|_2}
		=O_P(a_n).\label{eq:transferability_test_unit_vector_distance}
	\end{align}
	Fix $b\in\{1,\ldots,B\}$ and define
	\begin{equation*}
		\xi^{(b)}:=\begin{pmatrix}
			(\xi_i^{(b)})_{i\in\mathcal{I}_1}\\
			(\xi_i^{(b)})_{i\in\mathcal{I}_2}
		\end{pmatrix}.
	\end{equation*}
	Its coordinates are independent standard normal variables, independent of the data, partition, pre-training randomness, and other bootstrap draws. Define
	\begin{equation*}
		Z_b^{(\xi)}:=\frac{v_0^\top\xi^{(b)}}{\|v_0\|_2},
	\end{equation*}
	setting $Z_b^{(\xi)}=0$ when $\|v_0\|_2=0$. On the event above, $Z_b^{(\xi)}$ is conditionally standard normal and \cref{eq:transferability_test_bootstrap_statistic} gives $\mathcal{T}_b^{(\xi)}=\max_{j<k}|v_{jk}^\top\xi^{(b)}|/\|v_{jk}\|_2$. Thus,
	\begin{equation}
		\label{eq:transferability_test_bootstrap_coupling}
		\left|\mathcal{T}_b^{(\xi)}-|Z_b^{(\xi)}|\right|
		\leq\max_{j<k}\left|
		\left(\frac{v_{jk}}{\|v_{jk}\|_2}-\frac{v_0}{\|v_0\|_2}\right)^\top\xi^{(b)}
		\right|.
	\end{equation}
	Conditional on all observed data, the partition, and pre-training randomness, the vectors $v_{jk}$ and $v_0$ are fixed, while the coordinates of $\xi^{(b)}$ remain independent standard normal variables. All expectations in the following calculations are taken under this conditional law. Each variable $(v_{jk}/\|v_{jk}\|_2-v_0/\|v_0\|_2)^\top\xi^{(b)}$ on the right-hand side of \cref{eq:transferability_test_bootstrap_coupling} is therefore mean-zero Gaussian, with variance
	\begin{equation*}
		\mathbb{E}^{(\xi)}\left[\left\{
		\left(\frac{v_{jk}}{\|v_{jk}\|_2}-\frac{v_0}{\|v_0\|_2}\right)^\top\xi^{(b)}
		\right\}^2\right]
		=\left\|\frac{v_{jk}}{\|v_{jk}\|_2}-\frac{v_0}{\|v_0\|_2}\right\|_2^2.
	\end{equation*}
	Since $|x|=\max\{x,-x\}$, the maximum of the absolute values of the $R(R-1)/2$ linear forms equals the maximum over $R(R-1)$ signed forms. For a mean-zero Gaussian variable, the Orlicz norm associated with $x\mapsto e^{x^2}-1$ is a universal constant times its standard deviation. This norm also bounds the expected absolute value up to a universal constant. Applying Lemma 2.2.2 of \cite{van2023weak} to the signed forms therefore gives
	\begin{align*}
		&\mathbb{E}^{(\xi)}\left[\max_{j<k}\left|
		\left(\frac{v_{jk}}{\|v_{jk}\|_2}-\frac{v_0}{\|v_0\|_2}\right)^\top\xi^{(b)}
		\right|\right]\\
		&\lesssim\sqrt{\log(1+R(R-1))}
		\max_{j<k}\left\|\frac{v_{jk}}{\|v_{jk}\|_2}-\frac{v_0}{\|v_0\|_2}\right\|_2\\
		&\lesssim\sqrt{\log R}
		\max_{j<k}\left\|\frac{v_{jk}}{\|v_{jk}\|_2}-\frac{v_0}{\|v_0\|_2}\right\|_2.
	\end{align*}
	Conditional Markov's inequality, \cref{eq:transferability_test_bootstrap_coupling}, and the preceding expectation bound now give, for every fixed $\varepsilon>0$,
	\begin{align}
		\Pr^{(\xi)}\left(\left|\mathcal{T}_b^{(\xi)}-|Z_b^{(\xi)}|\right|>\varepsilon\right)
		&\leq\frac{1}{\varepsilon}\mathbb{E}^{(\xi)}\left[\left|\mathcal{T}_b^{(\xi)}-|Z_b^{(\xi)}|\right|\right]\notag\\
		&\leq\frac{1}{\varepsilon}\mathbb{E}^{(\xi)}\left[\max_{j<k}\left|
		\left(\frac{v_{jk}}{\|v_{jk}\|_2}-\frac{v_0}{\|v_0\|_2}\right)^\top\xi^{(b)}
		\right|\right]\notag\\
		&\lesssim\frac{\sqrt{\log R}}{\varepsilon}
		\max_{j<k}\left\|\frac{v_{jk}}{\|v_{jk}\|_2}-\frac{v_0}{\|v_0\|_2}\right\|_2\notag\\
		&=O_P\left(\frac{a_n\sqrt{\log R}}{\varepsilon}\right)
		=o_P(1),\label{eq:transferability_test_bootstrap_coupling_probability}
	\end{align}
	where the $O_P$ bound follows from \cref{eq:transferability_test_unit_vector_distance}, and the final equality uses $a_n\sqrt{\log R}=o_P(1)$ from \cref{asm:transferability_test}(v) and the fact that $\varepsilon$ is fixed. Both statements hold uniformly over $P\in\mathcal{P}_{0,n}$. These calculations apply on the event above, but its complement has uniformly vanishing probability and hence does not affect convergence in probability. Thus, the conditional probability on the left-hand side of \cref{eq:transferability_test_bootstrap_coupling_probability} converges to zero in probability unconditionally and uniformly over $P\in\mathcal{P}_{0,n}$.
	
	Along any sequence of laws whose $n$th member belongs to $\mathcal{P}_{0,n}$, \cref{eq:transferability_test_bootstrap_coupling_probability}, the conditional standard normality of $Z_b^{(\xi)}$ on the event above, and the uniformly vanishing probability of its complement imply convergence in probability of the conditional expectations of all bounded Lipschitz functions of $\mathcal{T}_b^{(\xi)}$ to their expectations under the absolute standard normal law.
	Since the limiting distribution function is continuous, Lemma 10.11(i) of \cite{kosorok2008introduction} gives
	\begin{equation}
		\label{eq:transferability_test_bootstrap_limit}
		\sup_{t\geq0}\left|\Pr^{(\xi)}(\mathcal{T}_b^{(\xi)}\leq t)-(2\Phi(t)-1)\right|=o_P(1).
	\end{equation}
	This convergence is uniform over $P\in\mathcal{P}_{0,n}$ because the sequence of laws is arbitrary. Combining \cref{eq:transferability_test_bootstrap_limit,eq:transferability_test_statistic_limit} by the triangle inequality yields
	\begin{equation*}
		\sup_{t\in\mathbb{R}}\left|\Pr^{(\xi)}(\mathcal{T}_b^{(\xi)}\leq t)-\Pr_P(\mathcal{T}\leq t)\right|=o_P(1),
	\end{equation*}
	uniformly over $P\in\mathcal{P}_{0,n}$. We have shown \cref{thm:transferability_test_size}(ii).

	\textbf{Step 4: Empirical critical value and rejection probability.}
	Finally, we show \cref{thm:transferability_test_size}(iii).
	Define the empirical bootstrap distribution function by
	\begin{equation}
		\label{eq:transferability_test_empirical_cdf}
		\widehat{F}_B^{(\xi)}(t):=\frac{1}{B}\sum_{b=1}^B\mathds{1}\{\mathcal{T}_b^{(\xi)}\leq t\}.
	\end{equation}
	Conditional on the data, partition, and pre-training randomness, the bootstrap statistics are i.i.d. by \cref{alg:bootstrap_transferability}. The Dvoretzky-Kiefer-Wolfowitz inequality (Theorem 11.6 of \cite{kosorok2008introduction}) therefore gives, for every $\varepsilon>0$,
	\begin{equation}
		\label{eq:transferability_test_empirical_cdf_bound}
		\Pr^{(\xi)}\left(\sup_{t\in\mathbb{R}}\left|\widehat{F}_B^{(\xi)}(t)-\Pr^{(\xi)}(\mathcal{T}_b^{(\xi)}\leq t)\right|>\varepsilon\right)
		\leq2\exp(-2B\varepsilon^2).
	\end{equation}
	By \cref{eq:transferability_test_bootstrap_limit,eq:transferability_test_empirical_cdf_bound} and $B\to\infty$,
	\begin{equation*}
		\sup_{t\geq0}\left|\widehat{F}_B^{(\xi)}(t)-(2\Phi(t)-1)\right|=o_P(1),
	\end{equation*}
	uniformly over $P\in\mathcal{P}_{0,n}$, with no restriction on the relative growth of $n$ and $B$. Since the distribution function $2\Phi(t)-1$ is strictly increasing, applying Lemma 21.2 of \cite{van1998asymptotic} along every sequence of laws whose $n$th member belongs to $\mathcal{P}_{0,n}$ therefore gives, uniformly over $P\in\mathcal{P}_{0,n}$,
	\begin{equation}
		\label{eq:transferability_test_critical_value}
		c_{1-\textsf{a}}^{(\xi)}\xrightarrow{P}\Phi^{-1}(1-\textsf{a}/2).
	\end{equation}
	Since the distribution function $2\Phi(t)-1$ is continuous, combining \cref{eq:transferability_test_statistic_limit,eq:transferability_test_critical_value} by Slutsky's lemma, as in the proof of Lemma 23.3 of \cite{van1998asymptotic}, yields the uniform size conclusion by the same sequence argument:
	\begin{equation*}
		\sup_{P\in\mathcal{P}_{0,n}}\left|\Pr_P\left(\mathcal{T}>c_{1-\textsf{a}}^{(\xi)}\right)-\textsf{a}\right|
		\longrightarrow0.
	\end{equation*}
	This completes the proof of \cref{thm:transferability_test_size}.
\end{proof}

\subsection{Proof for \cref{thm:transferability_test_consistency}}
\label{subsec:proof_transferability_test_consistency}
\begin{proof}
	On the event in \cref{eq:transferability_test_alternative_nondegeneracy}, every $\hat{J}_{j,s}$ is nonzero, and \cref{eq:transferability_test_influence_weights} gives
	\begin{equation*}
		\sum_{i\in\mathcal{I}_s}\hat{q}_{ij,s}^2
		=\frac{\sum_{i\in\mathcal{I}_s}\psi(W_i;\hat{\theta}_{j,s},\hat{\gamma}_j,\hat{\alpha}_j)^2}{n_s(n_s-1)\hat{J}_{j,s}^2}
		\geq\frac{c_{\psi}}{(n_s-1)\hat{J}_{j,s}^2}>0.
	\end{equation*}
	Thus, every $\hat{\sigma}_{jk}$ is positive by \cref{eq:transferability_test_standard_error}. \cref{asm:transferability_test_consistency}(i) therefore implies that the probability that the algorithm stops without a test decision tends to zero uniformly over $P\in\mathcal{P}_{\mathrm{sep},n}$.

	Conditional on the data, partition, and pre-training randomness, and on the event that the algorithm does not stop, each standardized difference
	\begin{equation*}
		\frac{\sum_{i\in\mathcal{I}_1}\xi_i^{(b)}\hat{q}_{ij,1}-\sum_{i\in\mathcal{I}_2}\xi_i^{(b)}\hat{q}_{ik,2}}{\hat{\sigma}_{jk}}
	\end{equation*}
	in \cref{eq:transferability_test_bootstrap_statistic} has a standard normal distribution for every $1\leq j<k\leq R$ and $b=1,\ldots,B$. The Gaussian tail bound (\citealp{wainwright2019high}, Example 2.1) and the union bound over $R(R-1)/2$ terms therefore give, for every $t\geq0$,
	\begin{equation}
		\label{eq:transferability_test_alternative_bootstrap_tail}
		\Pr^{(\xi)}(\mathcal{T}_b^{(\xi)}>t)\leq R(R-1)\exp(-t^2/2).
	\end{equation}
	If the empirical $(1-\textsf{a})$-quantile exceeds $t$, then more than a fraction $\textsf{a}$ of the bootstrap draws exceed $t$. Thus, conditional Markov's inequality and the bound in \cref{eq:transferability_test_alternative_bootstrap_tail} give
	\begin{align}
		\Pr^{(\xi)}(c_{1-\textsf{a}}^{(\xi)}>t)
		&\leq\frac{1}{\textsf{a}}\Pr^{(\xi)}(\mathcal{T}_b^{(\xi)}>t)\nonumber\\
		&\leq\frac{R(R-1)}{\textsf{a}}\exp(-t^2/2).\label{eq:transferability_test_alternative_critical_tail}
	\end{align}
	This bound also holds on the stopping event by the zero convention. Taking $t=K\sqrt{\log R}$ for $K>2$, the right-hand side is bounded by $\textsf{a}^{-1}2^{2-K^2/2}$, since $R\geq2$. Taking unconditional probabilities and then letting $K\to\infty$, we have $c_{1-\textsf{a}}^{(\xi)}=O_P(\sqrt{\log R})$ uniformly over $P\in\mathcal{P}_{\mathrm{sep},n}$ and over all $B\geq1$.

	The triangle inequality gives
	\begin{equation*}
		\max_{j<k}|\hat{\theta}_{j,1}-\hat{\theta}_{k,2}|
		\geq\max_{j<k}|\theta_{j,n}^{\dagger}-\theta_{k,n}^{\dagger}|
		-2\max_{j,s}|\hat{\theta}_{j,s}-\theta_{j,n}^{\dagger}|.
	\end{equation*}
	By \cref{eq:transferability_test_alternative_estimation,eq:transferability_test_alternative_separation} and $r_{\theta,n}=o(\Delta_n)$, the right-hand side is at least $\Delta_n/2$ with probability approaching one uniformly over $P\in\mathcal{P}_{\mathrm{sep},n}$. On this event, if the algorithm does not stop,
	\begin{equation*}
		\frac{\mathcal{T}}{\sqrt{\log R}}
		\geq\frac{\Delta_n}{2\sqrt{\log R}\max_{j<k}\hat{\sigma}_{jk}}
		=\frac{\sqrt{n}\Delta_n}{\sqrt{\log R}}\frac{1}{2\sqrt{n}\max_{j<k}\hat{\sigma}_{jk}}.
	\end{equation*}
	By \cref{asm:transferability_test_consistency}(ii) and (iii), the first factor on the right-hand side diverges, and the denominator of the second is bounded in probability uniformly over $P\in\mathcal{P}_{\mathrm{sep},n}$. Since the algorithm does not stop with probability approaching one uniformly, the inequality above shows that $\mathcal{T}/\sqrt{\log R}$ diverges in probability uniformly over $P\in\mathcal{P}_{\mathrm{sep},n}$. Consequently, for every fixed $K>0$, $\sup_{P\in\mathcal{P}_{\mathrm{sep},n}}\Pr_P(\mathcal{T}\leq K\sqrt{\log R})\to0$. Finally,
	\begin{equation*}
		\Pr_P(\mathcal{T}\leq c_{1-\textsf{a}}^{(\xi)})
		\leq\Pr_P(\mathcal{T}\leq K\sqrt{\log R})
		+\Pr_P(c_{1-\textsf{a}}^{(\xi)}>K\sqrt{\log R}).
	\end{equation*}
	Taking the supremum over $P\in\mathcal{P}_{\mathrm{sep},n}$, letting $n\to\infty$, and then letting $K\to\infty$ completes the proof.
\end{proof}

\section{Proofs for the Appendix}
\subsection{Proof for \cref{lem:transfer_var}}
\begin{proof}
	The assumed square integrability of the composed predictors implies that their conditional expectations and the products in the following expansion are integrable:
	\begin{align*}
		& \inf_{f\in\mathcal{F}_n^{\sub}}\mathbb{E}\left[\left(f\circ h(Z)-f_n \circ h_m(Z)\right)^2\right]\\
		&\quad = \inf_{f\in\mathcal{F}_n^{\sub}}\mathbb{E}\left[\left(\mathbb{E}\left[f_n \circ h_m(Z)\mid h(Z)\right]-f_n \circ h_m(Z)
		+ f\circ h(Z) - \mathbb{E}\left[f_n \circ h_m(Z)\mid h(Z)\right]\right)^2\right]\\
		&\quad =\mathbb{E}\left[\operatorname{Var}\left(f_n \circ h_m(Z)\mid h(Z)\right)\right]
		+ \inf_{f\in\mathcal{F}_n^{\sub}}\mathbb{E}\left[\left(\mathbb{E}\left[f_n \circ h_m(Z)\mid h(Z)\right]-f\circ h(Z)\right)^2\right],
	\end{align*}
	where the last equality follows because the conditional mean of $f_n\circ h_m(Z)-\mathbb{E}[f_n\circ h_m(Z)\mid h(Z)]$ given $h(Z)$ is zero, so the cross term vanishes by the law of iterated expectation.
\end{proof}

\subsection{Proof for \cref{lem:doob_dynkin_completion}}
\begin{proof}
	The sigma-field inclusion implies that each coordinate $V_j$ is measurable with respect to $\overline{\sigma(U)}^{P}$. By Proposition 2.12 of \cite{folland1999real}, there exists a $\sigma(U)$-measurable random variable $\widetilde{V}_j$ such that $\widetilde{V}_j = V_j$, $P$-a.s.
	For each $j=1,\ldots,q$, Theorem 4.2.8 of \cite{dudley2002real} yields a measurable function $\Gamma_j:\mathcal{U}\to\mathbb{R}$ such that $\widetilde{V}_j=\Gamma_j\circ U$. Hence, $\Gamma=(\Gamma_1,\ldots,\Gamma_q)':\mathcal{U}\to\mathbb{R}^{q}$ is measurable and satisfies $\widetilde{V}=\Gamma\circ U$.
	Since $V=\widetilde{V}$, $P$-a.s., the result follows.
\end{proof}

\subsection{Proof for \cref{lem:loss_property}}
\begin{proof}
    See \cite{farrell2021deep}, Lemma 8.
\end{proof}

\subsection{Proof for \cref{lem:loss_property_multilogit}}
\begin{proof}
	The assumed Euclidean-norm bounds imply that
	\begin{equation*}
		\sup_{z\in\mathcal{Z}}|f_t(z)|,
		\sup_{z\in\mathcal{Z}}|f_t^*(z)|
		\leq
		M,
		\qquad t=1,\ldots,T.
	\end{equation*}
	The constants $c_L$ and $c_U$ can therefore be derived as in Lemma 9 of \cite{farrell2021deep}.

	It remains to verify the Lipschitz constant. For $u\in\mathbb{R}^T$, define
	\begin{equation*}
		p_t(u)
		:=
		\frac{\exp(u_t)}{1+\sum_{j=1}^T\exp(u_j)},
		\qquad t=1,\ldots,T.
	\end{equation*}
	Write $p(u)=(p_1(u),\ldots,p_T(u))'$. For $y=1,\ldots,T$, let $e_y\in\mathbb{R}^T$ be the $y$-th standard basis vector, and let $e_{T+1}=\mathbf{0}_T$. Then,
	\begin{equation*}
		\nabla_1\ell(u,y)
		=
		p(u)-e_y.
	\end{equation*}
	For $y=1,\ldots,T$,
	\begin{align*}
		\|\nabla_1\ell(u,y)\|_2^2
		&=
		(1-p_y(u))^2+\sum_{j\neq y}p_j(u)^2\\
		&\leq
		(1-p_y(u))^2+\left(\sum_{j\neq y}p_j(u)\right)^2\\
		&\leq
		2(1-p_y(u))^2\\
		&\leq
		2,
	\end{align*}
	where the first inequality follows from $p_j(u)\geq0$ and the second inequality follows from $\sum_{j\neq y}p_j(u)\leq1-p_y(u)$.
	We also have $\|\nabla_1\ell(u,T+1)\|_2=\|p(u)\|_2\leq1$. The fundamental theorem of calculus and the Cauchy-Schwarz inequality now give
	\begin{align*}
		|\ell(f(z),y)-\ell(a(z),y)|
		&\leq
		\sup_{s\in[0,1]}\|\nabla_1\ell(a(z)+s(f(z)-a(z)),y)\|_2\|f(z)-a(z)\|_2\\
		&\leq
		\sqrt{2}\|f(z)-a(z)\|_2.
	\end{align*}
\end{proof}

\subsection{Proof for \cref{lem:entropy_property}}
\begin{proof}
    For each $k=1,\ldots,K$, let $\{h_{k,j}: j=1,\ldots,N_k\}$ be an $\varepsilon/K$-covering set of $\mathcal{H}_k$ with respect to the norm $L^2(Q)$, where $N_k=N\left(\varepsilon / K, \mathcal{H}_k, L^2(Q)\right)$. Then, for every $h\in\mathcal{H}$, there exists $j_k\in\{1,\ldots,N_k\}$ such that $\|h_k - h_{k,j_k}\|_{2,Q} < \varepsilon/K$ for each $k=1,\ldots,K$. Thus, we have
    $$
    \|h - (h_{1,j_1},\ldots,h_{K,j_K})\|_{2,Q} \leq \sum_{k=1}^K \|h_k - h_{k,j_k}\|_{2,Q} < K \cdot (\varepsilon/K) = \varepsilon.
    $$
    Therefore, the set $\{(h_{1,j_1},\ldots,h_{K,j_K}): j_k=1,\ldots,N_k, k=1,\ldots,K\}$ forms an $\varepsilon$-covering set of $\mathcal{H}$ with respect to the norm $\|\cdot\|_{2,Q}$, and its cardinality is $\prod_{k=1}^K N_k$. This completes the proof.
\end{proof}

\subsection{Proof for \cref{lem:entropy_property_composite}}
\begin{proof}
	Cover $\mathcal H$ by $N_H=N(\varepsilon/(4L_{\mathcal F}),\mathcal H,L^2(Q))$ balls of radius $\varepsilon/(4L_{\mathcal F})$. Since the original centers may not belong to $\mathcal H$, choose a point $h_j\in\mathcal H$ in each ball as its new center. By the triangle inequality, the ball of radius $\varepsilon/(2L_{\mathcal F})$ centered at $h_j$ covers every element of $\mathcal H$ in the original ball.
	For each $j$, let $\{f_{j,k}:k=1,\ldots,N_j\}$ be a minimum $\varepsilon/2$-cover of $\mathcal F$ under $L^2(Q_{h_j})$.
	For every $f\circ h\in\mathcal F\circ\mathcal H$, choose $h_j$ and $f_{j,k}$ from these covers. The triangle inequality and the Lipschitz property give
	\begin{equation*}
		\|f\circ h-f_{j,k}\circ h_j\|_{2,Q}
		\leq
		L_{\mathcal F}\|h-h_j\|_{2,Q}
		+\|f-f_{j,k}\|_{2,Q_{h_j}}
		\leq
		\varepsilon.
	\end{equation*}
	Thus, $\{f_{j,k}\circ h_j:j=1,\ldots,N_H,\ k=1,\ldots,N_j\}$ is an $\varepsilon$-cover of $\mathcal F\circ\mathcal H$. Therefore,
	\begin{equation*}
		N\left(\varepsilon,\mathcal F\circ\mathcal H,L^2(Q)\right)
		\leq
		\sum_{j=1}^{N_H}N_j
		\leq
		N\left(\frac{\varepsilon}{4L_{\mathcal F}},\mathcal H,L^2(Q)\right)
		\sup_{h'\in\mathcal H}
		N\left(\frac{\varepsilon}{2},\mathcal F,L^2(Q_{h'})\right).
	\end{equation*}
\end{proof}

\subsection{Proof for \cref{lem:maximal_lipschitz}}
\begin{proof}
	For every discrete probability measure $Q$ and every $h\in\mathcal{H}_m$, define
	\begin{equation*}
		d_{Q,h}(f,g)=\|f\circ h-g\circ h\|_{Q,2}.
	\end{equation*}
	For every $h\in\mathcal{H}_m$ and $f\in B_{n,h}(\delta)$,
	\begin{equation*}
		\sqrt{n}\left[\left(\mathbb{M}_n-M_n\right)(f,h)-\left(\mathbb{M}_n-M_n\right)\left(f_{n,h}^*,h\right)\right]=\mathbb{G}_n\left(m_{f,h}-m_{f_{n,h}^*,h}\right).
	\end{equation*}
	Hence,
	\begin{equation*}
		\sup_{f\in B_{n,h}(\delta)} \sqrt{n}\left|\left(\mathbb{M}_n-M_n\right)(f,h)-\left(\mathbb{M}_n-M_n\right)\left(f_{n,h}^*,h\right)\right|=\left\|\mathbb{G}_n\right\|_{\mathcal{M}_{n,h}(\delta)}.
	\end{equation*}
	Fix $h\in\mathcal{H}_m$ and let $Q$ be any discrete probability measure.
	Let $Q_h$ denote the image of $Q$ under $h$.
	For every $f,\tilde{f}\in B_{n,h}(\delta)$, the Lipschitz continuity of $\ell$ implies that, for every $(z,y)$,
	\begin{equation*}
		\left|m_{f,h}(z,y)-m_{\tilde{f},h}(z,y)\right|
		=\left|\ell\left(f\circ h(z),y\right)-\ell\left(\tilde{f}\circ h(z),y\right)\right|
		\leq C_{\ell}\left\|f\circ h(z)-\tilde{f}\circ h(z)\right\|_2.
	\end{equation*}
	Hence,
	\begin{align}
		\left\|\left(m_{f,h}-m_{f_{n,h}^*,h}\right)-\left(m_{\tilde{f},h}-m_{f_{n,h}^*,h}\right)\right\|_{Q,2}
		& =\left\|m_{f,h}-m_{\tilde{f},h}\right\|_{Q,2}\nonumber\\
		& \leq C_{\ell}\left\|f\circ h-\tilde{f}\circ h\right\|_{Q,2}\nonumber\\
		& =C_{\ell} d_{Q,h}(f,\tilde{f}).
        \label{eq:lipschitz_diff_bound}
	\end{align}
	Define the normalized class
	\begin{equation*}
		\widetilde{\mathcal{M}}_{n,h}(\delta)=\left\{\frac{\ell_{\mathrm{diff}}}{2C_{\ell}M}: \ell_{\mathrm{diff}}\in\mathcal{M}_{n,h}(\delta)\right\}.
	\end{equation*}
    We have
	\begin{align}
        N\left(\varepsilon/(2C_{\ell}M),\widetilde{\mathcal{M}}_{n,h}(\delta),L_2(Q)\right)
		& = N\left(\varepsilon,\mathcal{M}_{n,h}(\delta),L_2(Q)\right)\nonumber\\
        & \leq N\left(\varepsilon/C_{\ell},B_{n,h}(\delta),d_{Q,h}\right)\nonumber\\
        & \leq N\left(\varepsilon/C_{\ell},\mathcal{F}_n,d_{Q,h}\right),
        \label{eq:covering_bound_lipschitz}
	\end{align}
    where the first inequality follows from \cref{eq:lipschitz_diff_bound} and the second inequality follows from the fact that $B_{n,h}(\delta)\subseteq\mathcal{F}_n$.
	By the definition of $Q_h$,
	\begin{equation*}
		d_{Q,h}(f,g)
		=\|f\circ h-g\circ h\|_{Q,2}
		=\|f-g\|_{Q_h,2}.
	\end{equation*}
    Also, for every $\ell_{\mathrm{diff}}=m_{f,h}-m_{f_{n,h}^*,h}\in\mathcal{M}_{n,h}(\delta)$ with $f\in B_{n,h}(\delta)$,
	\begin{align*}
		|\ell_{\mathrm{diff}}(z,y)|\leq C_{\ell}\left\|f\circ h(z)-f_{n,h}^*\circ h(z)\right\|_2
        & \leq C_{\ell}\left\|f\circ h(z)\right\|_2+C_{\ell}\left\|f_{n,h}^*\circ h(z)\right\|_2\\
        & \leq 2C_{\ell}M.
	\end{align*}
	Hence, $\widetilde{\mathcal{M}}_{n,h}(\delta)$ has an envelope function equal to $1$, and, for every $\tilde{\ell}_{\mathrm{diff}}=\left(m_{f,h}-m_{f_{n,h}^*,h}\right)/(2C_{\ell}M)\in\widetilde{\mathcal{M}}_{n,h}(\delta)$, the Lipschitz continuity of $\ell$ implies that
	\begin{align*}
		P\tilde{\ell}_{\mathrm{diff}}^2
		& =\frac{1}{4C_{\ell}^2M^2}P\left(m_{f,h}-m_{f_{n,h}^*,h}\right)^2\\
		& \leq \frac{1}{4M^2}\left\|f\circ h-f_{n,h}^*\circ h\right\|_{P,2}^2\\
		& < \frac{\delta^2}{4M^2}.
	\end{align*}
	Because $B_{n,h}(\delta)$ is pointwise measurable by \cref{lem:pointwise_measurable}, there exists a countable subset $B_{n,h}^0(\delta)\subseteq B_{n,h}(\delta)$ such that every $f\in B_{n,h}(\delta)$ is the pointwise limit of a sequence in $B_{n,h}^0(\delta)$. Since $h$ is fixed and $\ell(\cdot,y)$ is continuous on $\mathbb{R}^T$ for every $y$ by Lipschitz continuity, the class $\mathcal{M}_{n,h}(\delta)$ is pointwise measurable. Multiplication by $(2C_{\ell}M)^{-1}$ preserves pointwise measurability; thus, $\widetilde{\mathcal{M}}_{n,h}(\delta)$ is pointwise measurable. Because the map $x\mapsto x^2$ is continuous, the classes $\mathcal{M}_{n,h}^2(\delta)$ and $\widetilde{\mathcal{M}}_{n,h}^2(\delta)$ are also pointwise measurable. Therefore, these classes are $P$-measurable. Theorem 2.14.2 in \cite{van2023weak}, applied to $\widetilde{\mathcal{M}}_{n,h}(\delta)$ with local radius $\delta/(2 M)$, implies
	\begin{align*}
		\mathbb{E}_{P}\left\|\mathbb{G}_n\right\|_{\widetilde{\mathcal{M}}_{n,h}(\delta)}
		&\lesssim J\left(\delta/(2 M),\widetilde{\mathcal{M}}_{n,h}(\delta) | 1,L_2\right)\\
		&\quad
		+\frac{J^2\left(\delta/(2 M),\widetilde{\mathcal{M}}_{n,h}(\delta) | 1,L_2\right)}{(\delta/(2 M))^2 \sqrt{n}}.
	\end{align*}
	Next, for any discrete probability measure $Q'$ on the domain of $\mathcal{F}_n$, we have
	\begin{align*}
		&J\left(\delta/(2 M),\widetilde{\mathcal{M}}_{n,h}(\delta) | 1,L_2\right)\\
		&\quad =\sup_Q \int_0^{\delta/(2 M)}\sqrt{1+\log N\left(\varepsilon,\widetilde{\mathcal{M}}_{n,h}(\delta),L_2(Q)\right)} d \varepsilon\\
		&\quad \leq \sup_Q \int_0^{\delta/(2 M)}\sqrt{1+\log N\left(2 M\varepsilon,\mathcal{F}_n,d_{Q,h}\right)} d \varepsilon\\
		&\quad \leq \sup_{Q'} \int_0^{\delta/(2 M)}\sqrt{1+\log N\left(2\varepsilon M,\mathcal{F}_n,L_2(Q')\right)} d \varepsilon\\
		&\quad =\frac{1}{2}J\left(\delta/M,\mathcal{F}_n | M,L_2\right),
	\end{align*}
	where the first equality follows from the definition of the uniform entropy integral, the first inequality follows from \cref{eq:covering_bound_lipschitz},
	the second inequality uses the fact that every image measure $Q_h$ is a discrete probability measure on the domain of $\mathcal{F}_n$, and the last equality uses the change of variables $u=2 \varepsilon$.
	Moreover,
	\begin{equation*}
		\mathbb{E}_{P}\left\|\mathbb{G}_n\right\|_{\mathcal{M}_{n,h}(\delta)}=2C_{\ell}M\mathbb{E}_{P}\left\|\mathbb{G}_n\right\|_{\widetilde{\mathcal{M}}_{n,h}(\delta)}.
	\end{equation*}
	Combining these bounds yields
	\begin{align*}
		\mathbb{E}_{P}\left\|\mathbb{G}_n\right\|_{\mathcal{M}_{n,h}(\delta)}
		&\lesssim M C_\ell \left(J\left(\delta/M,\mathcal{F}_n | M,L_2\right)
		+\frac{M^2 J^2\left(\delta/M,\mathcal{F}_n | M,L_2\right)}{\delta^2 \sqrt{n}}\right).
	\end{align*}
	Taking the supremum over $h\in\mathcal{H}_m$ completes the proof.
\end{proof}

\subsection{Proof for \cref{lem:pointwise_measurable}}
\begin{proof}
	Let $\mathcal{F}^0 \subset \mathcal{F}$ be a countable subset such that for every $f\in\mathcal{F}$, there exists a sequence $\{f_j\}_{j=1}^\infty$ in $\mathcal{F}^0$ such that $\|f_j(x)-f(x)\|_2\rightarrow0$ for every $x$.
	Pick an arbitrary $f\in B(\delta)$. We have $P\|f-f^*\|_2^2<\delta^2$.
	The continuity of the squared Euclidean norm implies that $\|f_j(x)-f^*(x)\|_2^2\rightarrow\|f(x)-f^*(x)\|_2^2$ for every $x$. Moreover, since $\|f_j(x)\|_2\leq F(x)$ and $\|f^*(x)\|_2\leq F(x)$, we have $\|f_j(x)-f^*(x)\|_2^2\leq4F(x)^2$. Since $PF^2<\infty$, applying the dominated convergence theorem to the scalar functions $\|f_j(x)-f^*(x)\|_2^2$ yields
	\begin{equation*}
		P\|f_j-f^*\|_2^2
		\rightarrow
		P\|f-f^*\|_2^2
		<\delta^2.
	\end{equation*}
	Therefore, $f_j\in B(\delta)$ eventually. Hence, every $f\in B(\delta)$ is the pointwise limit of a sequence in the countable set $\mathcal{F}^0\cap B(\delta)$, which proves that $B(\delta)$ is pointwise measurable.
\end{proof}

\subsection{Proof for \cref{lem:vc_covering}}
\begin{proof}
	For $\mathsf{V}(\mathcal{F})\geq1$, apply Theorem 2.6.7 in \cite{van2023weak}. If $\mathsf{V}(\mathcal{F})=0$, all functions in the nonempty class agree pointwise: two different values at any point would allow their subgraphs to shatter a singleton. Hence the covering number is one, and the displayed bound also holds.
\end{proof}

\subsection{Proof for \cref{lem:rate_of_conv}}
\begin{proof}
	The proof follows the structure of Theorems 3.2.5 and 3.4.1 in \cite{van2023weak}, with additional care required for the pre-trained estimator $\check{h}$ and local curvature $\tau_n$.
    Suppose that $\mathbb{M}_n\left(\hat{f}_{n,\check{h}},\check{h}\right) \geq \mathbb{M}_n\left(f_{n,\check{h}},\check{h}\right)-R_n$ for some $R_n = O_P\left(\delta_n^2\right)$.
	For every $h\in\mathcal{H}_m$ and $j\in\mathbb{N}$, define the following sets, called ``shells'':
    $$
    S_{n,h,j}=\left\{f \in \mathcal{F}_n: 2^{j-1} \tau_n^{-1}\delta_n<d_{n,h}\left(f, f_{n,h}^*\right) \leq 2^j \tau_n^{-1}\delta_n\right\}.
    $$
    Because $d_{n,h}\left(\cdot, f_{n,h}^*\right) \geq 0$, it suffices to show that for every $\varepsilon>0$, there exists an integer $J$, possibly depending on $\varepsilon$ but not on $\check{h}$, $n$, and $m$, such that
    \begin{align*}
        &\limsup_{n\rightarrow\infty} P\left(d_{n,\check{h}}\left(\hat{f}_{n,\check{h}}, f_{n,\check{h}}^*\right) > 2^J \tau_n^{-1}\delta_n\right)\\
        &\quad \leq \limsup_{n\rightarrow\infty} P\left(\hat{f}_{n,\check{h}} \in \bigcup_{j>J}S_{n,\check{h},j} \right)\\
        &\quad \leq \limsup_{n\rightarrow\infty}\sum_{j>J} P\left(\hat{f}_{n,\check{h}} \in S_{n,\check{h},j}, R_n \leq 2^J\delta_n^2, \check{h}\in\mathcal{H}_{m,\varepsilon} \right)\\
        &\qquad
        +\limsup_{n\rightarrow\infty}P\left(R_n > 2^J\delta_n^2 \right)
        +\limsup_{m\rightarrow\infty}P\left(\check{h}\notin\mathcal{H}_{m,\varepsilon} \right)\\
        &\quad\leq 3\varepsilon.
    \end{align*}
    The second term on the right-hand side can be made arbitrarily small by choosing $J$ large enough since $R_n = O_{P}\left(\delta_n^2\right)$.
    The third term on the right-hand side is less than or equal to $\varepsilon$ by the assumption on $\mathcal{H}_{m,\varepsilon}$.
    We will show that the first term on the right-hand side can also be made arbitrarily small by choosing $J$ large enough.
    
    Take $h\in\mathcal{H}_{m,\varepsilon}$ and $f \in S_{n,h,j}$. By \cref{eq:rate_of_conv_curvature}, if $2^j\tau_n^{-1}\delta_n \geq \tau_n^{-1}C_{\varepsilon,0}\underline{\delta}_n$, or equivalently $2^j\delta_n \geq C_{\varepsilon,0}\underline{\delta}_n$, then
    $$
    M_n(f,h)-M_n\left(f_{n,h}^*,h\right) \lesssim_\varepsilon -2^{2 j} \delta_n^2.
    $$
    Now consider $M_n(f,h)-M_n\left(f_{n,h},h\right)$. Add and subtract $M_n\left(f_{n,h}^*,h\right)$ to get
    \begin{align*}
        & M_n(f,h)-M_n\left(f_{n,h},h\right)\\
        &\quad =M_n(f,h)-M_n\left(f_{n,h}^*,h\right)+M_n\left(f_{n,h}^*,h\right)-M_n\left(f_{n,h},h\right)\\
        &\quad \leq -2^{2 j} C_{\varepsilon,1}\delta_n^2+M_n\left(f_{n,h}^*,h\right)-M_n\left(f_{n,h},h\right).
    \end{align*}
	By \cref{eq:rate_of_conv_delta_n}, the second term is bounded above by a constant multiple of $\delta_n^2$. Thus,
    $$
	M_n(f,h)-M_n\left(f_{n,h},h\right) \leq -2^{2 j} C_{\varepsilon,1}\delta_n^2+C_{\varepsilon,2}\delta_n^2.
    $$
    Thus, for every $j>J$ with sufficiently large $J$, we can bound
    \begin{equation}
        \label{eq:rate_of_conv_1step}
        M_n(f,h)-M_n\left(f_{n,h},h\right) \lesssim_\varepsilon -2^{2 j} \delta_n^2 \quad \text { for every } h\in\mathcal{H}_{m,\varepsilon} \text{ and } f \in S_{n,h,j}.
    \end{equation}

    Define the empirical process $\mathbb{G}_n(f,h)=\sqrt{n}\left(\mathbb{M}_n-M_n\right)(f,h)$.
    Then, 
    \begin{align*}
        &\mathbb{G}_n\left(\hat{f}_{n,\check{h}},\check{h}\right)-\mathbb{G}_n\left(f_{n,\check{h}},\check{h}\right)\\
        &\quad =\sqrt{n}\left[\mathbb{M}_n\left(\hat{f}_{n,\check{h}},\check{h}\right)-\mathbb{M}_n\left(f_{n,\check{h}},\check{h}\right)\right]-\sqrt{n}\left[M_n\left(\hat{f}_{n,\check{h}},\check{h}\right)-M_n\left(f_{n,\check{h}},\check{h}\right)\right].
    \end{align*}
    The first term on the right-hand side is bounded below by the approximate maximization property of $\hat{f}_{n,h}$:
    $$
    \sqrt{n}\left[\mathbb{M}_n\left(\hat{f}_{n,\check{h}},\check{h}\right)-\mathbb{M}_n\left(f_{n,\check{h}},\check{h}\right)\right] \geq - \sqrt{n}R_n.
    $$
    Suppose that the event $\check{h} \in \mathcal{H}_{m,\varepsilon}$ and $\hat{f}_{n,\check{h}} \in S_{n,\check{h},j}$ occurs. The second term on the right-hand side is bounded below by \cref{eq:rate_of_conv_1step},
    $$
    -\sqrt{n}\left[M_n\left(\hat{f}_{n,\check{h}},\check{h}\right)-M_n\left(f_{n,\check{h}},\check{h}\right)\right] \gtrsim_\varepsilon \sqrt{n} 2^{2 j} \delta_n^2.
    $$
    Combining these two bounds, we have
    $$
	\mathbb{G}_n\left(\hat{f}_{n,\check{h}},\check{h}\right)-\mathbb{G}_n\left(f_{n,\check{h}},\check{h}\right) \geq \sqrt{n}(C_{\varepsilon,3} 2^{2 j} \delta_n^2 - R_n) \quad \text { if } \check{h}\in\mathcal{H}_{m,\varepsilon} \text{ and } \hat{f}_{n,\check{h}} \in S_{n,\check{h},j}.
    $$
    Thus, by the set inclusion, we have for every $j>J$ with sufficiently large $J$,
    \begin{align*}
        &P\left(\hat{f}_{n,\check{h}} \in S_{n,\check{h},j}, R_n \leq 2^J\delta_n^2, \check{h} \in \mathcal{H}_{m,\varepsilon}\right)\\
        &\quad \leq P^*\left(\sup_{f\in S_{n,\check{h},j}}\mathbb{G}_n\left(f,\check{h}\right)-\mathbb{G}_n\left(f_{n,\check{h}},\check{h}\right) \gtrsim_\varepsilon \sqrt{n} 2^{2 j} \delta_n^2, \check{h}\in\mathcal{H}_{m,\varepsilon} \right)\\
        &\quad \leq \sup_{h\in\mathcal{H}_{m,\varepsilon}} P\left(\sup_{f\in S_{n,h,j}}\mathbb{G}_n\left(f,h\right)-\mathbb{G}_n\left(f_{n,h},h\right) \gtrsim_\varepsilon \sqrt{n} 2^{2 j} \delta_n^2 \right),
    \end{align*}
	where the last inequality follows because $\check{h}$ is independent of $\mathbb{G}_n\left(f,h\right)-\mathbb{G}_n\left(f_{n,h},h\right)$.

	By Markov's inequality, the main term on the right-hand side is bounded above by a constant multiple of
    \begin{align}
        &\sup_{h\in\mathcal{H}_{m,\varepsilon}}\mathbb{E}\left[\sup _{f \in S_{n,h,j}}\left|\mathbb{G}_n(f,h)-\mathbb{G}_n\left(f_{n,h},h\right)\right|\right]/(\sqrt{n} 2^{2 j} \delta_n^2)\nonumber\\
        &\quad \lesssim_\varepsilon \sup_{h\in\mathcal{H}_{m,\varepsilon}}\mathbb{E}\left[\sup_{f\in\mathcal{F}_n:d_{n,h}\left(f, f_{n,h}^*\right) \leq 2^j \tau_n^{-1}\delta_n}\left|\mathbb{G}_n(f,h)-\mathbb{G}_n\left(f_{n,h}^*,h\right)\right|\right]/(\sqrt{n} 2^{2 j} \delta_n^2)\label{eq:rate_of_conv_decomp1}\\
        &\qquad + \sup_{h\in\mathcal{H}_{m,\varepsilon}}\mathbb{E}\left[\left|\mathbb{G}_n\left(f_{n,h}^*,h\right)-\mathbb{G}_n\left(f_{n,h},h\right)\right|\right]/(\sqrt{n} 2^{2 j} \delta_n^2)\label{eq:rate_of_conv_decomp2},
    \end{align}
    where the inequality follows from the triangle inequality and the definition of $S_{n,h,j}$.
    By the modulus bound in \cref{eq:rate_of_conv_modulus}, the first term \cref{eq:rate_of_conv_decomp1} is proportionally bounded above by
    $\phi_n\left(2^{j+1} \tau_n^{-1}\delta_n\right) /(\sqrt{n} 2^{2 j} \delta_n^2)$.
    To bound the second term \cref{eq:rate_of_conv_decomp2}, note that $d_{n,h}\left(f_{n,h}, f_{n,h}^*\right) \lesssim_\varepsilon \tau_n^{-1}\delta_n$ since either $d_{n,h}\left(f_{n,h}, f_{n,h}^*\right) \leq \tau_n^{-1}C_{\varepsilon,0} \underline{\delta}_n$ or $d_{n,h}\left(f_{n,h}, f_{n,h}^*\right) > \tau_n^{-1}C_{\varepsilon,0}\underline{\delta}_n$ holds.
    For the former case, we have $d_{n,h}\left(f_{n,h}, f_{n,h}^*\right) \lesssim_\varepsilon \tau_n^{-1}\delta_n$ directly from $\delta_n \gtrsim_\varepsilon \underline{\delta}_n$.
    For the latter case, apply \cref{eq:rate_of_conv_curvature} with $\delta=d_{n,h}\left(f_{n,h}, f_{n,h}^*\right)$ to get
    $$M_n\left(f_{n,h},h\right)-M_n\left(f_{n,h}^*,h\right) \lesssim_\varepsilon -\tau_n^2 d_{n,h}\left(f_{n,h}, f_{n,h}^*\right)^2.$$
	By \cref{eq:rate_of_conv_delta_n}, the left-hand side is bounded below by a constant multiple of $-\delta_n^2$. Thus, $-\delta_n^2 \lesssim_\varepsilon -\tau_n^2 d_{n,h}\left(f_{n,h}, f_{n,h}^*\right)^2$, which implies $d_{n,h}\left(f_{n,h}, f_{n,h}^*\right) \lesssim_\varepsilon \tau_n^{-1}\delta_n$.
    We can apply \cref{eq:rate_of_conv_modulus} to bound the second term \cref{eq:rate_of_conv_decomp2} proportionally by $\phi_n\left(\tau_n^{-1}\delta_n\right) /(\sqrt{n} 2^{2 j} \delta_n^2)$ since $\phi_n(\delta) / \delta^\xi$ is decreasing.\footnote{If we have $d_{n,h}\left(f_{n,h}, f_{n,h}^*\right) \leq C_{\varepsilon,4} \tau_n^{-1}\delta_n$ for some constant $C_{\varepsilon,4}\leq1$, the bound follows directly from \cref{eq:rate_of_conv_modulus}. If $C_{\varepsilon,4}>1$, apply the decreasing property and get $\phi_n\left(C_{\varepsilon,4}\tau_n^{-1}\delta_n\right)\leq C_{\varepsilon,4}^\xi \phi_n\left(\tau_n^{-1}\delta_n\right)$.}
    Combining these bounds, we have
    $$
    P\left(\hat{f}_{n,\check{h}} \in S_{n,\check{h},j}, R_n \leq 2^J\delta_n^2, \check{h} \in \mathcal{H}_{m,\varepsilon}\right) \lesssim_\varepsilon \frac{\phi_n\left(2^{j+1} \tau_n^{-1}\delta_n\right)}{\sqrt{n} 2^{2 j} \delta_n^2}.
    $$
    
	Using the properties of $\phi_n$ in \cref{eq:rate_of_conv_delta_n}, we have
    $$
    P\left(\hat{f}_{n,\check{h}} \in S_{n,\check{h},j}, R_n \leq 2^J\delta_n^2, \check{h} \in \mathcal{H}_{m,\varepsilon}\right) \lesssim_\varepsilon \frac{\phi_n\left(\tau_n^{-1}\delta_n\right)}{\sqrt{n} \delta_n^2} \cdot \frac{2^\xi}{2^{j(2-\xi)}}
    \lesssim_\varepsilon \frac{1}{2^{j(2-\xi)}},
    $$
    where the first inequality follows from the decreasing property of $\delta \mapsto \phi_n(\delta) / \delta^\xi$ and the second inequality follows from $\phi_n\left(\delta_n/\tau_n\right) \lesssim_\varepsilon \sqrt{n} \delta_n^2$.
    Recall that $\lesssim_\varepsilon$ means that the constant in the bound may depend on $\varepsilon$ but not on $\check{h}$, $n$, $m$, and $j$.
    Therefore, by summing over $j>J$, we can take $J$ not depending on $\check{h}$, $n$, and $m$ such that the right-hand side of
    $$
    \limsup_{n\rightarrow\infty}\sum_{j>J} P\left(\hat{f}_{n,\check{h}} \in S_{n,\check{h},j}, R_n \leq 2^J\delta_n^2, \check{h} \in \mathcal{H}_{m,\varepsilon}\right)\lesssim_\varepsilon \sum_{j>J} \frac{1}{2^{j(2-\xi)}}
    $$
    is arbitrarily small. This completes the proof.
\end{proof}

\subsection{Proof for \cref{lem:rate_of_conv_regu}}
\begin{proof}
	The proof follows the structure of Addendum 3.4.7 in \cite{van2023weak}, with additional care required for the pre-trained estimator $\check{h}$ and local curvature $\tau_n$.
    For every $h\in\mathcal{H}_m$ and $j\in\mathbb{N}$, define
    \begin{align*}
        S_{n,h,j}=\left\{(f,\lambda)\in \mathcal{F}_n \times [\lambda_n, \infty): \right.&2^{2j-2} \delta_n^2<\tau_n^2 d_{n,h}\left(f, f_{n,h}^*\right)^2 + \lambda^2 \mathcal{J}_n^2(f) \leq 2^{2j} \delta_n^2,\\
        &\left.2^{2J}\lambda^2 \mathcal{J}_n^2(f_{n,h}) < \tau_n^2 d_{n,h}\left(f, f_{n,h}^*\right)^2 + \lambda^2 \mathcal{J}_n^2(f)\right\}
    \end{align*}
    for some large enough integer $J$ to be specified later.
    We show that for every $\varepsilon>0$, there exists an integer $J$ such that
    \begin{align}
        & \limsup_{n\rightarrow\infty} P\left(\tau_n^2 d_{n,\check{h}}\left(\hat{f}_{n,\check{h}}, f_{n,\check{h}}^*\right)^2 + \hat{\lambda}_{n,\check{h}}^2 \mathcal{J}_n^2\left(\hat{f}_{n,\check{h}}\right) > 2^{2J}\delta_n^2\right.,\nonumber\\
        & \qquad\qquad\qquad
        \tau_n^2 d_{n,\check{h}}\left(\hat{f}_{n,\check{h}}, f_{n,\check{h}}^*\right)^2 + \hat{\lambda}_{n,\check{h}}^2 \mathcal{J}_n^2\left(\hat{f}_{n,\check{h}}\right) > 2^{2J} \hat{\lambda}_{n,\check{h}}^2 \mathcal{J}_n^2\left(f_{n,\check{h}}\right),\nonumber\\
        & \qquad\qquad\qquad\qquad\qquad\qquad\qquad\qquad\left.
            d_{n,\check{h}}\left(\hat{f}_{n,\check{h}}, f_{n,\check{h}}^*\right) \geq \tau_n^{-1}C_{\varepsilon,0}\underline{\delta}_n, \hat{\lambda}_{n,\check{h}}\geq\lambda_{n}\right)\nonumber\\
        & \quad \leq \limsup_{n\rightarrow\infty} P\left((\hat{f}_{n,\check{h}},\hat{\lambda}_{n,\check{h}}) \in \bigcup_{j>J} S_{n,\check{h},j}, d_{n,\check{h}}\left(\hat{f}_{n,\check{h}}, f_{n,\check{h}}^*\right) \geq \tau_n^{-1}C_{\varepsilon,0}\underline{\delta}_n, \hat{\lambda}_{n,\check{h}}\geq\lambda_{n} \right)\nonumber\\
        & \quad \leq\varepsilon.\label{eq:rate_of_conv_regu_goal}
    \end{align}
    For simplicity, we assume that $\mathbb{M}_n\left(\hat{f}_{n,\check{h}},\check{h}\right) - \mathbb{M}_n\left(f_{n,\check{h}},\check{h}\right) \geq \hat{\lambda}_{n,\check{h}}^2 \mathcal{J}_n^2\left(\hat{f}_{n,\check{h}}\right) - \hat{\lambda}_{n,\check{h}}^2 \mathcal{J}_n^2\left(f_{n,\check{h}}\right) -2^J\delta_n^2$ for some large enough integer $J$, which holds with arbitrarily high probability by $\mathbb{M}_n\left(\hat{f}_{n,\check{h}},\check{h}\right) - \mathbb{M}_n\left(f_{n,\check{h}},\check{h}\right) \geq \hat{\lambda}_{n,\check{h}}^2 \mathcal{J}_n^2\left(\hat{f}_{n,\check{h}}\right) - \hat{\lambda}_{n,\check{h}}^2 \mathcal{J}_n^2\left(f_{n,\check{h}}\right) -O_P\left(\delta_{n}^2\right)$, and that $\check{h}\in\mathcal{H}_{m,\varepsilon}$, which holds with probability at least $1-\varepsilon$ by the assumption on $\mathcal{H}_{m,\varepsilon}$.\footnote{
		Let $R_n\geq0$ denote the optimization-error remainder, with $R_n=O_P(\delta_n^2)$, as in the proof of \cref{lem:rate_of_conv}. Take $J$ large enough that $\limsup_{n\rightarrow\infty}P\left(R_n > 2^J\delta_n^2 \right)$ can be made arbitrarily small and the proportional bounds hold. The remaining probabilities in the proof can be evaluated after intersecting with the event $\{R_n \leq 2^J\delta_n^2,\check{h}\in\mathcal{H}_{m,\varepsilon}\}$.
    }

    Take $(f,\lambda) \in S_{n,h,j}$ with $d_{n,h}\left(f, f_{n,h}^*\right) \geq \tau_n^{-1}C_{\varepsilon,0}\underline{\delta}_n$.
    Apply \cref{eq:rate_of_conv_regu_curvature} with $\delta=d_{n,h}\left(f, f_{n,h}^*\right)$ to get
    $$
    M_n(f,h)-M_n\left(f_{n,h}^*,h\right) \lesssim_\varepsilon -\tau_n^2 d_{n,h}\left(f, f_{n,h}^*\right)^2.
    $$
    Then, for every $j>J$ with sufficiently large $J$, we can bound
    \begin{align}
        &M_n(f,h)-M_n\left(f_{n,h},h\right) - \lambda^2 \mathcal{J}_n^2(f) + \lambda^2 \mathcal{J}_n^2\left(f_{n,h}\right) \nonumber\\
        &\quad = \left(M_n(f,h)-M_n\left(f_{n,h}^*,h\right)\right)+\left(M_n\left(f_{n,h}^*,h\right)-M_n\left(f_{n,h},h\right)\right) - \lambda^2 \mathcal{J}_n^2(f) + \lambda^2 \mathcal{J}_n^2\left(f_{n,h}\right) \nonumber\\
        &\quad \leq -\tau_n^2 C_{\varepsilon,1}d_{n,h}\left(f, f_{n,h}^*\right)^2 +C_{\varepsilon,2}\delta_{n}^2 - \lambda^2 \mathcal{J}_n^2(f) + \lambda^2 \mathcal{J}_n^2\left(f_{n,h}\right) \nonumber\\
        &\quad \leq -(\tau_n^2 d_{n,h}\left(f, f_{n,h}^*\right)^2 + \lambda^2 \mathcal{J}_n^2(f))(\min\{C_{\varepsilon,1},1\} - 2^{-2J}) + C_{\varepsilon,2}\delta_n^2 \nonumber\\
        &\quad \leq -(2^{2j-2}(\min\{C_{\varepsilon,1},1\} - 2^{-2J}) - C_{\varepsilon,2})\delta_n^2\nonumber\\
        &\quad \lesssim_\varepsilon -2^{2j}\delta_n^2, \label{eq:rate_of_conv_regu_1step}
    \end{align}
    where the first equality follows from adding and subtracting $M_n\left(f_{n,h}^*,h\right)$, the first inequality follows from the above display and $M_n\left(f_{n,h}^*,h\right)-M_n(f_{n,h},h) \lesssim_\varepsilon \delta_n^2$ in \cref{eq:rate_of_conv_regu_delta_n}, the second inequality follows from $2^{2J}\lambda^2 \mathcal{J}_n^2(f_{n,h}) < \tau_n^2 d_{n,h}\left(f, f_{n,h}^*\right)^2 + \lambda^2 \mathcal{J}_n^2(f)$ in the definition of $S_{n,h,j}$, and the third inequality follows from $2^{2j-2} \delta_n^2 \leq \tau_n^2 d_{n,h}\left(f, f_{n,h}^*\right)^2 + \lambda^2 \mathcal{J}_n^2(f)$ in the definition of $S_{n,h,j}$.

    Define the empirical process $\mathbb{G}_n(f,h)=\sqrt{n}\left(\mathbb{M}_n-M_n\right)(f,h)$.
    Suppose that $(\hat{f}_{n,\check{h}},\hat{\lambda}_{n,\check{h}}) \in S_{n,\check{h},j}$ with $d_{n,\check{h}}\left(\hat{f}_{n,\check{h}}, f_{n,\check{h}}^*\right) \geq \tau_n^{-1}C_{\varepsilon,0} \underline{\delta}_n$ and $\hat{\lambda}_{n,\check{h}}\geq\lambda_{n}$.
    Then,
    \begin{align*}
        &\mathbb{G}_n\left(\hat{f}_{n,\check{h}},\check{h}\right)-\mathbb{G}_n\left(f_{n,\check{h}},\check{h}\right) \\
        &\quad =\sqrt{n}\left[\mathbb{M}_n\left(\hat{f}_{n,\check{h}},\check{h}\right)-\mathbb{M}_n\left(f_{n,\check{h}},\check{h}\right)\right]-\sqrt{n}\left[M_n\left(\hat{f}_{n,\check{h}},\check{h}\right)-M_n\left(f_{n,\check{h}},\check{h}\right)\right] \\
        &\quad \geq \sqrt{n} \left(\hat{\lambda}_{n,\check{h}}^2 \mathcal{J}_n^2\left(\hat{f}_{n,\check{h}}\right) - \hat{\lambda}_{n,\check{h}}^2 \mathcal{J}_n^2\left(f_{n,\check{h}}\right)\right) - \sqrt{n}2^J\delta_n^2 - \sqrt{n}\left[M_n\left(\hat{f}_{n,\check{h}},\check{h}\right)-M_n\left(f_{n,\check{h}},\check{h}\right)\right]\\
        &\quad \geq \sqrt{n} C_{\varepsilon,3}2^{2j}\delta_n^2 - \sqrt{n} 2^J\delta_n^2,
    \end{align*}
    where the first inequality follows from the approximate-maximization property of $(\hat{f}_{n,h},\hat{\lambda}_{n,h})$ in \cref{eq:rate_of_conv_regu_approx_max} and the second inequality follows from \cref{eq:rate_of_conv_regu_1step}.
    Here, for large enough $J$,
    \begin{align*}
        &P\left(\mathbb{G}_n\left(\hat{f}_{n,\check{h}},\check{h}\right)-\mathbb{G}_n\left(f_{n,\check{h}},\check{h}\right) \geq \sqrt{n} C_{\varepsilon,3}2^{2j}\delta_n^2 - \sqrt{n} 2^J\delta_n^2 \right) \\
        &\quad= P\left(\mathbb{G}_n\left(\hat{f}_{n,\check{h}},\check{h}\right)-\mathbb{G}_n\left(f_{n,\check{h}},\check{h}\right) \gtrsim_\varepsilon \sqrt{n} 2^{2j}\delta_n^2 \right).
    \end{align*}
    Also, $(\hat{f}_{n,\check{h}},\hat{\lambda}_{n,\check{h}}) \in S_{n,\check{h},j}$ with $\hat{\lambda}_{n,\check{h}}\geq\lambda_{n}$ implies that $d_{n,\check{h}}\left(\hat{f}_{n,\check{h}}, f_{n,\check{h}}^*\right)\leq 2^j \tau_n^{-1} \delta_n$ and $\mathcal{J}_n\left(\hat{f}_{n,\check{h}}\right) < 2^j\delta_n / \lambda_{n}$.
    Thus, by the set inclusion, we have
    \begin{align*}
        &P\left((\hat{f}_{n,\check{h}},\hat{\lambda}_{n,\check{h}}) \in S_{n,\check{h},j}, d_{n,\check{h}}\left(\hat{f}_{n,\check{h}}, f_{n,\check{h}}^*\right) \geq \tau_n^{-1}C_{\varepsilon,0} \underline{\delta}_n, \hat{\lambda}_{n,\check{h}}\geq\lambda_{n} \right) \\
        & \quad \leq P^*\left(\sup_{\substack{f \in \mathcal{F}_n: d_{n,\check{h}}\left(f, f_{n,\check{h}}^*\right)\leq 2^j \tau_n^{-1}\delta_n,\\ \mathcal{J}_n\left(f\right) < 2^j\delta_n / \lambda_{n}}}\left(\mathbb{G}_n\left(f,\check{h}\right)-\mathbb{G}_n\left(f_{n,\check{h}},\check{h}\right)\right) \gtrsim_\varepsilon \sqrt{n} 2^{2j}\delta_n^2 \right)\\
        &\quad \leq \sup_{h\in\mathcal{H}_{m,\varepsilon}} P\left(\sup_{\substack{f \in \mathcal{F}_n: d_{n,h}\left(f, f_{n,h}^*\right)\leq 2^j \tau_n^{-1}\delta_n,\\ \mathcal{J}_n\left(f\right) < 2^j\delta_n / \lambda_{n}}}\left(\mathbb{G}_n\left(f,h\right)-\mathbb{G}_n\left(f_{n,h},h\right)\right) \gtrsim_\varepsilon \sqrt{n} 2^{2j}\delta_n^2 \right).
    \end{align*}
    
	By Markov's inequality, the main term on the right-hand side is bounded above by some constant depending on $\varepsilon$ times
    \begin{align}
        &\sup_{h\in\mathcal{H}_{m,\varepsilon}}\mathbb{E}\left[\sup_{\substack{f \in \mathcal{F}_n: d_{n,h}\left(f, f_{n,h}^*\right)\leq 2^j \tau_n^{-1} \delta_n,\\ \mathcal{J}_n\left(f\right) < 2^j\delta_n / \lambda_{n}}}\left|\mathbb{G}_n\left(f,h\right)-\mathbb{G}_n\left(f_{n,h},h\right)\right|\right]/(\sqrt{n} 2^{2j}\delta_n^2)\nonumber\\
        &\quad \leq \sup_{h\in\mathcal{H}_{m,\varepsilon}}\mathbb{E}\left[\sup_{\substack{f \in \mathcal{F}_n: d_{n,h}\left(f, f_{n,h}^*\right)\leq 2^j \tau_n^{-1} \delta_n,\\ \mathcal{J}_n\left(f\right) < 2^j\delta_n / \lambda_{n}}}\left|\mathbb{G}_n\left(f,h\right)-\mathbb{G}_n\left(f_{n,h}^*,h\right)\right|\right]/(\sqrt{n} 2^{2j}\delta_n^2)\label{eq:rate_of_conv_regu_decomp1}\\    
        &\qquad + \sup_{h\in\mathcal{H}_{m,\varepsilon}}\mathbb{E}\left[\left|\mathbb{G}_n\left(f_{n,h}^*,h\right)-\mathbb{G}_n\left(f_{n,h},h\right)\right|\right]/(\sqrt{n} 2^{2j}\delta_n^2)\label{eq:rate_of_conv_regu_decomp2}.
    \end{align}
    By the modulus bound in \cref{eq:rate_of_conv_regu_modulus}, the first term \cref{eq:rate_of_conv_regu_decomp1} is proportionally bounded above by
    $\phi_n\left(2^{j+1} \tau_n^{-1} \delta_n\right) /(\sqrt{n} 2^{2j}\delta_n^2)$.
    To bound the second term \cref{eq:rate_of_conv_regu_decomp2}, note that $d_{n,h}\left(f_{n,h}, f_{n,h}^*\right) < \tau_n^{-1} \delta_n$ and $\mathcal{J}_n(f_{n,h})<\delta_n/\lambda_{n}$ by \cref{eq:rate_of_conv_regu_delta_n}.
    Apply \cref{eq:rate_of_conv_regu_modulus} to bound the second term \cref{eq:rate_of_conv_regu_decomp2} proportionally by $\phi_n\left(\tau_n^{-1}\delta_n\right) /(\sqrt{n} 2^{2j}\delta_n^2)$.
    Combining these bounds, we have
    $$
    P\left((\hat{f}_{n,\check{h}},\hat{\lambda}_{n,\check{h}}) \in S_{n,\check{h},j}, d_{n,\check{h}}\left(\hat{f}_{n,\check{h}}, f_{n,\check{h}}^*\right) \geq \tau_n^{-1}C_{\varepsilon,0} \underline{\delta}_n, \hat{\lambda}_{n,\check{h}}\geq\lambda_{n} \right)
    \lesssim_\varepsilon \frac{\phi_n\left(2^{j+1} \tau_n^{-1}\delta_n\right)}{\sqrt{n} 2^{2j}\delta_n^2}.
    $$

	The remainder of the proof proceeds as in the proof of \cref{lem:rate_of_conv}, using the properties of $\phi_n$ in \cref{eq:rate_of_conv_regu_delta_n} and taking $J$ large enough.
    We obtain \cref{eq:rate_of_conv_regu_goal}, which implies that for large enough $J$,
    \begin{align*}
        &\limsup_{n\rightarrow\infty} P\left(\tau_n d_{n,\check{h}}\left(\hat{f}_{n,\check{h}}, f_{n,\check{h}}^*\right) > 2^J \max\left\{\delta_n, \hat{\lambda}_{n,\check{h}} \mathcal{J}_n\left(f_{n,\check{h}}\right)\right\},
        \right.\\
        &\qquad\qquad\qquad\qquad\qquad\qquad\left.
            d_{n,\check{h}}\left(\hat{f}_{n,\check{h}}, f_{n,\check{h}}^*\right) \geq \tau_n^{-1}C_{\varepsilon,0} \underline{\delta}_n, \hat{\lambda}_{n,\check{h}}\geq\lambda_{n}\right)\leq \varepsilon
    \end{align*}
    since $\hat{\lambda}_{n,\check{h}}^2 \mathcal{J}_n^2\left(\hat{f}_{n,\check{h}}\right) \geq 0$.
    As $\delta_n \gtrsim_\varepsilon \underline{\delta}_n$, we can take $J$ large enough that $2^J\delta_n\geq C_{\varepsilon,0}\underline{\delta}_n$. Hence, the event in the next display implies $d_{n,\check{h}}\left(\hat{f}_{n,\check{h}}, f_{n,\check{h}}^*\right)>\tau_n^{-1}C_{\varepsilon,0}\underline{\delta}_n$, and we obtain
    \begin{equation*}
        \limsup_{n\rightarrow\infty} P\left(\tau_n d_{n,\check{h}}\left(\hat{f}_{n,\check{h}}, f_{n,\check{h}}^*\right)  > 2^J \max\left\{\delta_n, \hat{\lambda}_{n,\check{h}} \mathcal{J}_n\left(f_{n,\check{h}}\right)\right\}, \hat{\lambda}_{n,\check{h}}\geq\lambda_{n}\right)
        \leq \varepsilon.
    \end{equation*}
    Thus, on the event $\{\hat{\lambda}_{n,\check{h}}\geq\lambda_{n}\}$, we have
    $$\tau_n d_{n,\check{h}}\left(\hat{f}_{n,\check{h}}, f_{n,\check{h}}^*\right) = O_P(1)\left(\delta_n + \hat{\lambda}_{n,\check{h}} \mathcal{J}_n\left(f_{n,\check{h}}\right)\right).$$
    This completes the proof.
\end{proof}

\subsection{Proof for \cref{lem:curvature_approx}}
\begin{proof}
	Pick an arbitrary $\varepsilon>0$.
    By the assumptions, there exist constants $C_{\varepsilon,1}>0$, $C_{\varepsilon,2}>0$, and $C_{\varepsilon,3}>0$ such that
    \begin{equation}
        \label{eq:approx_lem_epsilon}
        \sup_{h\in\mathcal{H}_{m,\varepsilon}} d_{n,h}\left(f_{n,h}, f_{n,h}^*\right) \leq C_{\varepsilon,1} \underline{\delta}_n,
    \end{equation}
    \begin{equation}
        \label{eq:approx_lem_curvature}
        \sup_{h\in\mathcal{H}_{m,\varepsilon}}\sup_{f \in \mathcal{F}_n:\delta/2 < d_{n,h}\left(f, f_{n,h}^*\right) \leq \delta} M_n(f,h)-M_n\left(f_{n,h}^*,h\right) \leq -C_{\varepsilon,2}\tau_n^2 \delta^2,
    \end{equation}
    and
    \begin{equation}
        \label{eq:approx_lem_delta_n}
        M_n\left(f_{n,h}^*,h\right) - M_n\left(f_{n,h},h\right) \leq C_{\varepsilon,3} d_{n,h}\left(f_{n,h}, f_{n,h}^*\right)^2
    \end{equation} 
    for every $h\in\mathcal{H}_{m,\varepsilon}$.
    Set
    \begin{equation*}
        C_{\varepsilon,0}
        :=
        C_{\varepsilon,1}
        \max\left\{4,\sqrt{32C_{\varepsilon,3}/C_{\varepsilon,2}}\right\},
    \end{equation*}
    which depends on $\varepsilon$ but not on $n$.
    Pick arbitrary $\delta \geq \tau_n^{-1}C_{\varepsilon,0}\underline{\delta}_n$, $h\in\mathcal{H}_{m,\varepsilon}$, and $f\in\mathcal{F}_n$ such that
    $\delta/2<d_{n,h}\left(f,f_{n,h}\right)\leq\delta$.
    By the choice of $C_{\varepsilon,0}$, \cref{eq:approx_lem_epsilon}, and $\tau_n\leq1$, we have
    \begin{align}
        d_{n,h}\left(f,f_{n,h}\right)
        &>\delta/2
        \geq 2d_{n,h}\left(f_{n,h},f_{n,h}^*\right),\label{eq:approx_lem_ineq_1}\\
        d_{n,h}\left(f,f_{n,h}\right)
        &>\delta/2
        \geq \sqrt{8C_{\varepsilon,3}/C_{\varepsilon,2}}\,\tau_n^{-1}d_{n,h}\left(f_{n,h},f_{n,h}^*\right).\label{eq:approx_lem_ineq_2}
    \end{align}
    The reverse triangle inequality and \cref{eq:approx_lem_ineq_1} imply
    \begin{equation}
        d_{n,h}\left(f,f_{n,h}^*\right)
        \geq d_{n,h}\left(f,f_{n,h}\right)-d_{n,h}\left(f_{n,h},f_{n,h}^*\right)
		\geq \frac{1}{2}d_{n,h}\left(f,f_{n,h}\right)>0.
        \label{eq:approx_lem_ineq_3}
    \end{equation}
    Adding and subtracting $M_n\left(f_{n,h}^*,h\right)$, applying \cref{eq:approx_lem_curvature} with shell radius $d_{n,h}\left(f,f_{n,h}^*\right)$, and using \cref{eq:approx_lem_delta_n}, we obtain
    \begin{align*}
        M_n(f,h)-M_n\left(f_{n,h},h\right)
        &\leq -C_{\varepsilon,2}\tau_n^2d_{n,h}\left(f,f_{n,h}^*\right)^2
        +C_{\varepsilon,3}d_{n,h}\left(f_{n,h},f_{n,h}^*\right)^2\\
        &\leq -(C_{\varepsilon,2}/4)\tau_n^2d_{n,h}\left(f,f_{n,h}\right)^2
        +C_{\varepsilon,3}d_{n,h}\left(f_{n,h},f_{n,h}^*\right)^2\\
        &\leq -(C_{\varepsilon,2}/8)\tau_n^2d_{n,h}\left(f,f_{n,h}\right)^2\\
        &\leq -(C_{\varepsilon,2}/32)\tau_n^2\delta^2,
    \end{align*}
    where the second inequality follows from \cref{eq:approx_lem_ineq_3}, the third inequality follows from \cref{eq:approx_lem_ineq_2}, and the last inequality follows from $d_{n,h}\left(f,f_{n,h}\right)\geq\delta/2$. Taking the supremum over the stated shell and over $h\in\mathcal{H}_{m,\varepsilon}$ proves the result.
\end{proof}

\subsection{Proof for \cref{lem:contraction}}
\begin{proof}
    If $L_\phi=0$, the Lipschitz condition and $\phi_i(0)=0$ imply that $\phi_i(u)=0$ for every $u\in\mathbb{R}$ and every $i$, so the conclusion follows immediately.
    Suppose henceforth that $L_\phi>0$.
    By the contraction inequality (\citealp{ledoux1991probability}, Theorem 4.12), for any contraction $\phi_{0,i}:\mathbb{R}\rightarrow\mathbb{R}$, we have
    \begin{equation*}
        \mathbb{E}\left[\sup_{u\in\mathcal{U}} \left|\sum_{i=1}^n \xi_i \phi_{0,i}(u_i)\right|\right] \leq 2 \cdot \mathbb{E}\left[\sup_{u\in\mathcal{U}} \left|\sum_{i=1}^n \xi_i u_i\right|\right].
    \end{equation*}
    For a Lipschitz function $\phi_i:\mathbb{R}\rightarrow\mathbb{R}$ with Lipschitz constant $L_\phi$, we can define $\phi_{0,i}(x)=\phi_i(x)/L_\phi$, which is a contraction. Then,
    \begin{align*}
        \mathbb{E}\left[\sup_{u\in\mathcal{U}} \left|\sum_{i=1}^n \xi_i \phi_i(u_i)\right|\right] 
        &= L_\phi \cdot \mathbb{E}\left[\sup_{u\in\mathcal{U}} \left|\sum_{i=1}^n \xi_i \phi_{0,i}(u_i)\right|\right] \\
        &\leq 2 L_\phi \cdot \mathbb{E}\left[\sup_{u\in\mathcal{U}} \left|\sum_{i=1}^n \xi_i u_i\right|\right].
    \end{align*}
    This completes the proof.
\end{proof}

\subsection{Proof for \cref{lem:unitary_ineq}}
\begin{proof}
    The inequalities directly follow from Theorem II.3.9 of \cite{stewart1990matrix}. 
	The Frobenius norm $\|\cdot\|_F$ and the spectral norm $\|\cdot\|_{\mathrm{sp}}$ are unitarily invariant (\citealp{stewart1990matrix}, p. 74).
	The nuclear norm $\|\cdot\|_*$ is also unitarily invariant (\citealp{horn1994topics}, p. 211, Problem 5).
\end{proof}

\subsection{Proof for \cref{lem:rank_one}}
\begin{proof}
    By the singular value decomposition, we can write $A = \sigma u v'$ for some $\sigma \geq 0$ and unit vectors $u$ and $v$.
	Then, $\|A\|_* = \|A\|_F = \|A\|_{\mathrm{sp}} = \sigma$ since $A$ has either one nonzero singular value $\sigma$ or no nonzero singular values.
    Also, $\sqrt{\operatorname{tr}(A'A)} = \sqrt{\operatorname{tr}(\sigma^2 v u' u v')} = \sqrt{\sigma^2 \operatorname{tr}(v v')} = \sigma$ since $u$ and $v$ are unit vectors.
\end{proof}

\subsection{Proof for \cref{lem:l2_ineq}}
\begin{proof}
    By the definition of the $L_2(P)$ norm,
    \begin{align*}
        \|Aa\|_{P,2}^2 = \int \|Aa(v)\|_2^2 \, dP(v) \leq \int \|A\|_{\mathrm{sp}}^2 \|a(v)\|_2^2 \, dP(v) = \|A\|_{\mathrm{sp}}^2 \|a\|_{P,2}^2,
    \end{align*}
    where the inequality follows from $\|Aa(v)\|_2 = \|Aa(v)\|_F$ and \cref{lem:unitary_ineq}.
	This gives the desired inequality after taking square roots of both sides.
\end{proof}

\subsection{Proof for \cref{lem:spectral_inverse}}
\begin{proof}
	By \cite{harville2008matrix}, Theorem 20.5.6, the Moore-Penrose inverse from the singular value decomposition $A=U\Sigma V'$ is given by $A^+=V\Sigma^+ U'$, where $\Sigma$ is the singular value matrix of $A$ and $\Sigma^+$ is the Moore-Penrose inverse of $\Sigma$ obtained by replacing each nonzero singular value $\sigma$ with $1/\sigma$ and leaving the zero singular values unchanged. Hence, we have $\|A^+\|_{\mathrm{sp}}=\|\Sigma^+\|_{\mathrm{sp}}=1/\sigma_{\min}^+(A)$ as the spectral norm is equal to the largest singular value.
\end{proof}

\subsection{Proof for \cref{lem:effective_dof}}
\begin{proof}
    By the definition of $\mathcal{N}_A(\lambda)$, we have
    \begin{align*}
        &\mathcal{N}_{\tilde{A}}(\lambda)-\mathcal{N}_A(\lambda)\\
        & \quad = \lambda^2 \operatorname{tr}\left((A+\lambda^2 I)^{-1}-(\tilde{A}+\lambda^2 I)^{-1}\right)\\
		& \quad = \lambda^2 \operatorname{tr}\left((\tilde{A}+\lambda^2 I)^{-1}(\tilde{A} - A)(A+\lambda^2 I)^{-1}\right),
    \end{align*}
    where the second equality follows from the matrix identity $X^{-1}-Y^{-1} = Y^{-1}(Y-X)X^{-1}$ for any invertible matrices $X$ and $Y$.
    
    We bound the absolute value of the above expression:
    \begin{align*}
		&\lambda^2 \left|\operatorname{tr}\left((\tilde{A}+\lambda^2 I)^{-1}(\tilde{A} - A)(A+\lambda^2 I)^{-1}\right)\right|\\
		& \quad \leq \lambda^2 \|(\tilde{A}+\lambda^2 I)^{-1}(\tilde{A} - A)(A+\lambda^2 I)^{-1}\|_*\\
		& \quad \leq \lambda^2 \|(\tilde{A}+\lambda^2 I)^{-1}\|_{\mathrm{sp}} \|\tilde{A} - A\|_* \|(A+\lambda^2 I)^{-1}\|_{\mathrm{sp}}\\
		& \quad \leq \|A - \tilde{A}\|_* / \lambda^2,
    \end{align*}
    where the first inequality follows from the singular value decomposition of $(\tilde{A}+\lambda^2 I)^{-1}(\tilde{A} - A)(A+\lambda^2 I)^{-1}$, the cyclic property of the trace, and the fact that the matrices of left and right singular vectors are unitary, the second inequality follows from applying \cref{lem:unitary_ineq} twice, and the last inequality follows from $\|(A+\lambda^2 I)^{-1}\|_{\mathrm{sp}} \leq 1/\lambda^2$ for any positive semi-definite matrix $A$.
\end{proof}

\subsection{Proof for \cref{lem:transformation_dof}}
\begin{proof}
    Define $Q^*=Q/\sqrt{\max\{1,\|Q\|_{\mathrm{sp}}^2\}}$. Then $\|Q^*\|_{\mathrm{sp}}\leq 1$ and
    \begin{equation*}
        QAQ'=\max\{1,\|Q\|_{\mathrm{sp}}^2\}Q^*AQ^{*'}.
    \end{equation*}
    Since $\|Q^*\|_{\mathrm{sp}}\leq 1$, we have $Q^{*\prime}Q^*\preceq I$.
     Therefore, for every $x\in\mathbb{R}^K$,
    \begin{equation*}
        x'A^{1/2}Q^{*\prime}Q^*A^{1/2}x
        \leq
        x'Ax.
    \end{equation*}
	Hence, $A^{1/2}Q^{*\prime}Q^*A^{1/2}\preceq A$, and by the order-reversing property of the matrix inverse for positive definite matrices,
    \begin{equation*}
        (A+\lambda^2 I)^{-1}\preceq (A^{1/2}Q^{*\prime}Q^*A^{1/2}+\lambda^2 I)^{-1}.
    \end{equation*}
    Since $Q^*AQ^{*\prime}$ and $A^{1/2}Q^{*\prime}Q^*A^{1/2}$ have the same nonzero eigenvalues, we have
    \begin{equation*}
        \mathcal{N}_{Q^*AQ^{*\prime}}(\lambda)
        =
        \mathcal{N}_{A^{1/2}Q^{*\prime}Q^*A^{1/2}}(\lambda).
    \end{equation*}
	Using the identity
    \begin{equation*}
        \mathcal{N}_A(\lambda)
        =
        \operatorname{tr}\left(I - \lambda^2(A+\lambda^2 I)^{-1}\right)
        =
		K-\lambda^2\operatorname{tr}\left((A+\lambda^2 I)^{-1}\right),
    \end{equation*}
     we get
    \begin{equation*}
        \mathcal{N}_{Q^*AQ^{*\prime}}(\lambda)
        =
        \operatorname{tr}\left(I - \lambda^2(A^{1/2}Q^{*\prime}Q^*A^{1/2}+\lambda^2 I)^{-1}\right)
        \leq
        \operatorname{tr}\left(I - \lambda^2(A+\lambda^2 I)^{-1}\right)
        =
        \mathcal{N}_A(\lambda).
    \end{equation*}
    Finally, let $\mu_1,\ldots,\mu_K$ denote the eigenvalues of $Q^*AQ^{*\prime}$. Using the definition of $Q^*$, we obtain
    \begin{align*}
        \mathcal{N}_{QAQ'}(\lambda)
        &=
        \mathcal{N}_{\max\{1,\|Q\|_{\mathrm{sp}}^2\}Q^*AQ^{*\prime}}(\lambda)\\
        &=
        \sum_{j=1}^K
        \frac{\max\{1,\|Q\|_{\mathrm{sp}}^2\}\mu_j}{\max\{1,\|Q\|_{\mathrm{sp}}^2\}\mu_j+\lambda^2}\\
        &\leq
        \max\{1,\|Q\|_{\mathrm{sp}}^2\}\sum_{j=1}^K
        \frac{\mu_j}{\mu_j+\lambda^2}\\
        &=
        \max\{1,\|Q\|_{\mathrm{sp}}^2\}\mathcal{N}_{Q^*AQ^{*\prime}}(\lambda)\\
        &\leq
        \max\{1,\|Q\|_{\mathrm{sp}}^2\}\mathcal{N}_A(\lambda).
    \end{align*}
    This completes the proof.
\end{proof}

\section{Multimodal Transfer Learning}
\label{sec:multi-pretrain}
In applications, the economist may have access to several pre-trained models, each trained on a different modality. For example, a computer scientist may train a CNN on images and a transformer on text, after which the economist uses both learned embeddings as inputs to a target model. This section extends our transfer-learning theory to this multimodal setting.

Let $r=1,\ldots,R$ index the modalities, where $R$ is fixed as the sample sizes diverge. For each $r$, write $m_r=m_{r,n}$ and suppose that $m_{r,n}\to\infty$ as $n\to\infty$. The target observations satisfy
\begin{equation*}
	(Y_i,Z_{0,i},Z_{1,i},\ldots,Z_{R,i})\sim_{i.i.d.}P_{\ta},
	\qquad i=1,\ldots,n,
\end{equation*}
where $Z_{0,i}$ denotes the structured covariates and $Z_{r,i}\in\mathcal{Z}_r$ denotes the unstructured input for modality $r$. Write $Z_i=(Z_{0,i},Z_{1,i},\ldots,Z_{R,i})$. The structured covariates may be low- or high-dimensional. For each modality $r$, the computer scientist observes a source sample
\begin{equation*}
	(S_{r,j},Z_{r,j})\sim_{i.i.d.}P_{\so,r},
	\qquad j=1,\ldots,m_r,
\end{equation*}
where the modality-specific source samples and the target sample are mutually independent.
Let $P$ denote the joint distribution of these samples, and let $\mathbb{E}_{\so,r}$ denote expectation under $P_{\so,r}$.

Let $\ell_{\so,r}$ be the source loss for modality $r$, and let $\mathcal{G}_{r,m_r}$ and $\mathcal{H}_{r,m_r}$ be the corresponding classes of source heads and embedding functions. For each modality, define the set of population source-risk minimizers by
\begin{equation}
	\mathcal{R}_{r,m_r}
	=
	\argmin_{g_r\in\mathcal{G}_{r,m_r},h_r\in\mathcal{H}_{r,m_r}}
	\mathbb{E}_{\so,r}[\ell_{\so,r}(g_r\circ h_r(Z_r),S_r)].
\end{equation}
For each modality $r$, let
\begin{equation}
	(\check{g}_r,\check{h}_r)\in\argmin_{g_r\in\mathcal{G}_{r,m_r},h_r\in\mathcal{H}_{r,m_r}}
	\frac{1}{m_r}\sum_{j=1}^{m_r}\ell_{\so,r}(g_r\circ h_r(Z_{r,j}),S_{r,j}).
\end{equation}
Thus, $\check{h}_r$ is the pre-trained embedding estimated from the modality-$r$ source sample.

Define $\mathbf{m}=(m_1,\ldots,m_R)$ and $\mathcal{H}_{\mathbf{m}}=\prod_{r=1}^{R}\mathcal{H}_{r,m_r}$. For every multimodal embedding $h=(h_1,\ldots,h_R)\in\mathcal{H}_{\mathbf{m}}$, define
\begin{equation*}
	(f\circ h)(Z)
	:=
	f(Z_0,h_1(Z_1),\ldots,h_R(Z_R)).
\end{equation*}

After introducing the multimodal representation, we define the population target-risk minimizer correspondence. For every $h\in\mathcal{H}_{\mathbf{m}}$, let
\begin{equation}
	\mathcal{F}_n^{\sub}(h)
	=
	\argmin_{f\in\mathcal{F}_n^{\sub}}
	\mathbb{E}_{\ta}[\ell_{\ta}((f\circ h)(Z),Y)].
\end{equation}
Given the estimated multimodal embedding $\check{h}=(\check{h}_1,\ldots,\check{h}_R)$, the economist estimates
\begin{equation*}
	\hat{f}
	\in
	\argmin_{f\in\mathcal{F}_n}
	\frac{1}{n}\sum_{i=1}^{n}\ell_{\ta}((f\circ \check{h})(Z_i),Y_i).
\end{equation*}

\begin{definition}[Modality-Wise Diversity and Transferability]
	\label{def:multi-diversity}
	Fix population minimizers $(g_{r,m_r},h_{r,m_r})\in\mathcal{R}_{r,m_r}$ for $r=1,\ldots,R$ and $f_n\in\mathcal{F}_n^{\sub}(h_{\mathbf{m}})$, where $h_{\mathbf{m}}=(h_{1,m_1},\ldots,h_{R,m_R})$.
	For each modality $r$ and $\rho\geq0$, define the source-side approximate identified set by
	\begin{equation*}
		\mathcal{H}_{\so,r,m_r}(\rho)
		:=
		\left\{
		h_r\in\mathcal{H}_{r,m_r}:
		\min_{g_r\in\mathcal{G}_{r,m_r}}
		\|g_r\circ h_r-g_{r,m_r}\circ h_{r,m_r}\|_{P_{\so,r},2}
		\leq
		\rho
		\right\}.
	\end{equation*}
	For $h_{-r}=(h_1,\ldots,h_{r-1},h_{r+1},\ldots,h_R)\in\prod_{s\neq r}\mathcal{H}_{s,m_s}$, let $(h_r,h_{-r})$ denote the multimodal embedding whose $r$-th component is $h_r$. A jointly selected family of population target models is a collection $\{f_{n,h}:h\in\mathcal{H}_{\mathbf{m}}\}$ satisfying
	\begin{equation}
		f_{n,h}
		\in
		\mathcal{F}_n^{\sub}(h),
		\qquad
		h\in\mathcal{H}_{\mathbf{m}},
	\end{equation}
	and $f_{n,h_{\mathbf{m}}}=f_n$. Given such a family, define
	\begin{equation*}
		D_{r,n}(h_r\mid h_{-r})
		:=
		\|f_{n,(h_r,h_{-r})}\circ(h_r,h_{-r})
		-f_{n,(h_{r,m_r},h_{-r})}\circ(h_{r,m_r},h_{-r})\|_{P_{\ta},2}
	\end{equation*}
	and the target-side neighborhood induced by this family
	\begin{equation*}
		\mathcal{H}_{\ta,r,n}(\rho\mid h_{-r})
		:=
		\left\{
		h_r\in\mathcal{H}_{r,m_r}:
		D_{r,n}(h_r\mid h_{-r})
		\leq
		\rho
		\right\}.
	\end{equation*}
	Relative to this family, we say that the modality-$r$ source task $g_{r,m_r}$ is $(\rho_{\so,r,m_r},\rho_{\ta,r,n})$-diverse over $f_n$ with respect to $h_{r,m_r}$ if
	\begin{equation*}
		\mathcal{H}_{\so,r,m_r}(\rho_{\so,r,m_r})
		\subseteq
		\mathcal{H}_{\ta,r,n}(\rho_{\ta,r,n}\mid h_{-r})
	\end{equation*}
	holds for every $h_{-r}\in\prod_{s\neq r}\mathcal{H}_{s,m_s}$. Let $\nu_{n,r}>0$ for every $n$ and $r=1,\ldots,R$. Given a sequence of jointly selected families, we say that, relative to this sequence, the modality-$r$ source task $g_{r,m_r}$ is $\nu_{n,r}$-transferable to $f_n$ with respect to $h_{r,m_r}$ at rate $\delta_{\so,r,m_r}$ if, for every nonnegative sequence $\rho_{\so,r,m_r}\asymp\delta_{\so,r,m_r}$, the $(\rho_{\so,r,m_r},\rho_{\so,r,m_r}/\nu_{n,r})$-diversity condition holds for all sufficiently large $n$.
	We say that these modality-wise diversity conditions hold jointly if there exists a single jointly selected family relative to which they hold simultaneously for every $r=1,\ldots,R$ and every $h_{-r}\in\prod_{s\neq r}\mathcal{H}_{s,m_s}$. We say that modality-wise transferability holds jointly if there exists a single sequence of jointly selected families relative to which the preceding transferability condition holds simultaneously for every $r=1,\ldots,R$.
\end{definition}

This definition is the modality-wise analogue of \cref{def:diversity,def:transferability}. The joint condition requires a common selected population target model family at each $n$ such that inclusion in each modality-specific source-side approximate identified set implies inclusion in the corresponding target-side neighborhood for every configuration of the other embeddings. This common selection permits the modalities to be replaced sequentially by their estimated embeddings.

\begin{theorem}[Convergence Rate for Multimodal Transfer Learning]
	\label{thm:multi_transfer}
	Fix $(g_{r,m_r},h_{r,m_r})\in\mathcal{R}_{r,m_r}$ for $r=1,\ldots,R$ and $f_n\in\mathcal{F}_n^{\sub}(h_{\mathbf{m}})$. Define
	\begin{equation*}
		\|f_n\circ h_{\mathbf{m}}-a_{\ta}^*\|_{P_{\ta},2}
		=:
		\epsilon_{\ta,n}.
	\end{equation*}
	Given $\check{h}$, define the $L^2(P_{\ta})$-best approximation to $f_n\circ h_{\mathbf{m}}$ by
	\begin{equation*}
		f_{n,\check{h}}^{\bullet}
		\in
		\argmin_{f\in\mathcal{F}_n^{\sub}}
		\|f\circ\check{h}-f_n\circ h_{\mathbf{m}}\|_{P_{\ta},2}.
	\end{equation*}
	For each modality $r=1,\ldots,R$, suppose that $\nu_{n,r}>0$ for every $n$ and that
	\begin{equation*}
		\|\check{g}_r\circ\check{h}_r-g_{r,m_r}\circ h_{r,m_r}\|_{P_{\so,r},2}
		=
		O_{P_{\so,r}}(\delta_{\so,r,m_r}).
	\end{equation*}
	Suppose also that modality-wise transferability holds jointly in the sense of \cref{def:multi-diversity}; that is, there exists a common sequence of jointly selected population target model families relative to which $g_{r,m_r}$ is $\nu_{n,r}$-transferable to $f_n$ with respect to $h_{r,m_r}$ at rate $\delta_{\so,r,m_r}$ for every $r=1,\ldots,R$. Suppose further that
	\begin{equation*}
		\|\hat{f}\circ\check{h}-f_{n,\check{h}}^{\bullet}\circ\check{h}\|_{P_{\ta},2}
		=
		O_P(r_{\ta,n}).
	\end{equation*}
	Then,
	\begin{equation*}
		\|\hat{f}\circ\check{h}-a_{\ta}^*\|_{P_{\ta},2}
		=
		O_P\left(r_{\ta,n}+\sum_{r=1}^{R}\frac{\delta_{\so,r,m_r}}{\nu_{n,r}}+\epsilon_{\ta,n}\right).
	\end{equation*}
\end{theorem}

\begin{proof}
	Fix a common sequence of jointly selected population target model families witnessing simultaneous transferability for every modality.
	For $r=0,\ldots,R$, define the intermediate multimodal embeddings
	\begin{equation*}
		\widetilde{h}^{(r)}
		:=
		(\check{h}_1,\ldots,\check{h}_r,h_{r+1,m_{r+1}},\ldots,h_{R,m_R}).
	\end{equation*}
	Thus, $\widetilde{h}^{(0)}=h_{\mathbf{m}}$ and $\widetilde{h}^{(R)}=\check{h}$. The triangle inequality yields
	\begin{align*}
		\|\hat{f}\circ\check{h}-a_{\ta}^*\|_{P_{\ta},2}
		&\leq
		\|\hat{f}\circ\check{h}-f_{n,\check{h}}^{\bullet}\circ\check{h}\|_{P_{\ta},2}\\
		&\quad+
		\|f_{n,\check{h}}^{\bullet}\circ\check{h}-f_n\circ h_{\mathbf{m}}\|_{P_{\ta},2}\\
		&\quad+
		\|f_n\circ h_{\mathbf{m}}-a_{\ta}^*\|_{P_{\ta},2}.
	\end{align*}
	The first term on the right-hand side is $O_P(r_{\ta,n})$, and the last term equals $\epsilon_{\ta,n}$. By the definition of $f_{n,\check{h}}^{\bullet}$, the fact that $f_{n,\check{h}}\in\mathcal{F}_n^{\sub}(\check{h})\subseteq\mathcal{F}_n^{\sub}$, and the triangle inequality,
	\begin{align*}
		\|f_{n,\check{h}}^{\bullet}\circ\check{h}-f_n\circ h_{\mathbf{m}}\|_{P_{\ta},2}
		&\leq
		\|f_{n,\check{h}}\circ\check{h}-f_n\circ h_{\mathbf{m}}\|_{P_{\ta},2}\\
		&\leq
		\sum_{r=1}^{R}
		\|f_{n,\widetilde{h}^{(r)}}\circ\widetilde{h}^{(r)}
		-f_{n,\widetilde{h}^{(r-1)}}\circ\widetilde{h}^{(r-1)}\|_{P_{\ta},2},
	\end{align*}
	where the second inequality uses $f_{n,\widetilde{h}^{(0)}}=f_n$ and $f_{n,\widetilde{h}^{(R)}}=f_{n,\check{h}}$.

	For each modality $r$, $\check{g}_r\in\mathcal{G}_{r,m_r}$ and the source-rate condition imply
	\begin{align*}
		\min_{g_r\in\mathcal{G}_{r,m_r}}
		\|g_r\circ\check{h}_r-g_{r,m_r}\circ h_{r,m_r}\|_{P_{\so,r},2}
		&\leq
		\|\check{g}_r\circ\check{h}_r-g_{r,m_r}\circ h_{r,m_r}\|_{P_{\so,r},2}\\
		&=
		O_{P_{\so,r}}(\delta_{\so,r,m_r}).
	\end{align*}
	Each source-rate event depends only on the modality-$r$ source sample, so the same probability bound holds under the joint law $P$. Fix $\varepsilon>0$. For each $r$, the source-rate condition gives a constant $1\leq M_{\varepsilon,r}<\infty$ such that for all sufficiently large $n$ and corresponding $m_r$,
	\begin{equation*}
		P\left(
		\check{h}_r\notin\mathcal{H}_{\so,r,m_r}(M_{\varepsilon,r}\delta_{\so,r,m_r})
		\right)
		\leq
		\frac{\varepsilon}{R}.
	\end{equation*}
	Let $M_\varepsilon=\max_{1\leq r\leq R}M_{\varepsilon,r}$. Since $R$ is fixed, the union bound implies that the event
	\begin{equation*}
		\mathcal{E}_n
		:=
		\bigcap_{r=1}^{R}
		\left\{
		\check{h}_r\in\mathcal{H}_{\so,r,m_r}(M_\varepsilon\delta_{\so,r,m_r})
		\right\}
	\end{equation*}
	satisfies
	\begin{equation*}
		P(\mathcal{E}_n)
		\geq
		1-\sum_{r=1}^{R}\frac{\varepsilon}{R}
		=
		1-\varepsilon
	\end{equation*}
	for all sufficiently large $n$ and corresponding $m_r$.

	The embeddings $\widetilde{h}^{(r)}$ and $\widetilde{h}^{(r-1)}$ differ only in their $r$-th components, and $M_\varepsilon\delta_{\so,r,m_r}\asymp\delta_{\so,r,m_r}$. Therefore, for all sufficiently large $n$, modality-wise transferability applies on $\mathcal{E}_n$ with the remaining components fixed at $(\check{h}_1,\ldots,\check{h}_{r-1},h_{r+1,m_{r+1}},\ldots,h_{R,m_R})$ and yields
	\begin{align*}
		D_{r,n}\left(\check{h}_r\mid
		\left(\check{h}_1,\ldots,\check{h}_{r-1},h_{r+1,m_{r+1}},\ldots,h_{R,m_R}\right)\right)
		&=
		\|f_{n,\widetilde{h}^{(r)}}\circ\widetilde{h}^{(r)}
		-f_{n,\widetilde{h}^{(r-1)}}\circ\widetilde{h}^{(r-1)}\|_{P_{\ta},2}\\
		&\leq
		\frac{M_\varepsilon\delta_{\so,r,m_r}}{\nu_{n,r}}.
	\end{align*}
	It follows that
	\begin{equation*}
		\sum_{r=1}^{R}
		\|f_{n,\widetilde{h}^{(r)}}\circ\widetilde{h}^{(r)}
		-f_{n,\widetilde{h}^{(r-1)}}\circ\widetilde{h}^{(r-1)}\|_{P_{\ta},2}
		=
		O_P\left(\sum_{r=1}^{R}\frac{\delta_{\so,r,m_r}}{\nu_{n,r}}\right).
	\end{equation*}
	Combining these bounds proves the result.
\end{proof}

The population target models $f_{n,h}$ used in modality-wise transferability are target-risk minimizers, whereas $f_{n,\check{h}}^{\bullet}$ in the target-stage condition is the $L^2(P_{\ta})$-best approximation to $f_n\circ h_{\mathbf{m}}$. The best-approximation property links these objects in the proof.

The theorem shows that, under modality-wise transferability, the representation error accumulates additively across modalities. In the image-text special case with $R=2$, the bound becomes
\begin{align*}
	\|\hat{f}\circ\check{h}-a_{\ta}^*\|_{P_{\ta},2}
	&=
	O_P\left(
	r_{\ta,n}
	+\frac{\delta_{\so,\mathrm{image},m_{\mathrm{image}}}}{\nu_{n,\mathrm{image}}}
	+\frac{\delta_{\so,\mathrm{text},m_{\mathrm{text}}}}{\nu_{n,\mathrm{text}}}
	+\epsilon_{\ta,n}
	\right).
\end{align*}
Thus, the convergence rate of the target estimator depends on the sum of the source rates for each modality, adjusted by the corresponding transferability coefficients.

\section{Additional Details for Empirical Application}
\label{sec:additional_empirical_results}
\subsection{Individual-Level Cross-Fitting for Dube's Variables}
\label{sec:individual_cross_fitting}
\cref{tab:dml-dube-individual} repeats the Dube specification with three individual-level folds and five repetitions.
Both the outer cross-fitting folds and the inner tuning or validation samples are split at the individual level, as in \cite{dube2020monopsony}, rather than at the requester level as in the main results.
Standard errors remain requester-clustered and use repetition-variance aggregation as in the main results.
The sample and the $611$ controls are unchanged.

The individual-level split produces OLS, lasso, and ridge duration coefficients of $-0.3557$, $-0.3478$, and $-0.3538$, with standard errors close to $0.052$.
These differ substantially from the requester-level results in \cref{tab:dml-cluster-summary}.
The XGB 1--3 coefficients, $-0.0417$, $-0.0354$, and $-0.0343$, are closer to their requester-level counterparts, but their standard errors fall to approximately $0.004$.
For these learners, reward $R_X^2$ values rise from $0.33$--$0.36$ under requester-level splitting to $0.75$--$0.77$ under individual-level splitting.
This comparison illustrates that clustering the reported standard errors and choosing the unit of cross-fitting are distinct decisions.
Potentially, the very small p-values in \cref{tab:dml-dube-individual} can be explained by the over-fitting of the learners due to within-cluster dependence. To mitigate this concern, we use cross-fitting with cluster-level folds in our analysis reported in \cref{tab:dml-cluster-summary}.

\begin{table}[htbp]
\centering
\footnotesize
\caption{Dube controls with individual-row cross-fitting}\label{tab:dml-dube-individual}
\makebox[\linewidth][c]{%
\resizebox{1.1\linewidth}{!}{%
\begin{tabular}{lrrrrrr}
\toprule
\multicolumn{7}{l}{\textbf{Panel A: Linear and random-forest learners}} \\
 & OLS & CV-Lasso & CV-Ridge & RF 1 & RF 2 & RF 3 \\
\midrule
\multicolumn{7}{l}{\textit{Dube (611 features)}} \\
Estimate (s.e.) & \num{-0.3557} (\num{0.0517}) & \num{-0.3478} (\num{0.0521}) & \num{-0.3538} (\num{0.0518}) & \num{-0.1833} (\num{0.0838}) & \num{-0.2494} (\num{0.0669}) & \num{-0.3692} (\num{0.0382}) \\
$p$-value & \num{5.935e-12} & \num{2.395e-11} & \num{8.243e-12} & \num{0.0287} & \num{0.0002} & \num{4.713e-22} \\
$R_Y^2/R_D^2$ & \num{-3.0073} / \num{-2.6939} & \num{0.0960} / \num{-2.5199} & \num{-3.0079} / \num{-2.6945} & \num{0.2294} / \num{0.2817} & \num{0.2783} / \num{0.3523} & \num{0.4693} / \num{0.6629} \\
\midrule
\multicolumn{7}{l}{\textbf{Panel B: XGBoost and neural-network learners}} \\
 & XGB 1 & XGB 2 & XGB 3 & NN 1 & NN 2 & NN 3 \\
\midrule
\multicolumn{7}{l}{\textit{Dube (611 features)}} \\
Estimate (s.e.) & \num{-0.0417} (\num{0.0040}) & \num{-0.0354} (\num{0.0043}) & \num{-0.0343} (\num{0.0040}) & \num{-0.2320} (\num{0.0291}) & \num{-0.2877} (\num{0.0639}) & \num{-0.2708} (\num{0.0338}) \\
$p$-value & \num{4.542e-25} & \num{2.076e-16} & \num{7.302e-18} & \num{1.574e-15} & \num{6.620e-06} & \num{1.184e-15} \\
$R_Y^2/R_D^2$ & \num{0.8866} / \num{0.7471} & \num{0.8917} / \num{0.7562} & \num{0.8944} / \num{0.7704} & \num{0.5868} / \num{-1.4819} & \num{-17.0267} / \num{-3.7225} & \num{-4.0260} / \num{-3.4895} \\
\bottomrule
\end{tabular}%
}%
}
\par
\begin{minipage}{\linewidth}
\vspace{0.5em}
\footnotesize
\textit{Notes:} Individual-level cross-fitting uses 3 folds and 5 repetitions; 258,352 observations and 24,923 requesters. Estimates are median coefficients on log reward in the log-duration equation. Parentheses report requester-clustered standard errors with split-deviation inflation. Two-sided normal $p$-values test a zero coefficient. The final row lists out-of-sample outcome and treatment $R^2$ scores, averaged across repetitions. 
\end{minipage}
\end{table}

\subsection{Additional Results on Robustness to Pre-Trained Embeddings}
\label{sec:additional_embedding_sensitivity}
\cref{fig:xgb-1-confidence-intervals,fig:xgb-3-confidence-intervals} display the coefficients and confidence intervals for XGB 1 and XGB 3.
These specifications use the same embeddings and estimation procedure as the XGB 2 results in \cref{fig:xgb-2-confidence-intervals}.

\begin{figure}[p]
	\centering
	\includegraphics[width=\linewidth,height=0.84\textheight,keepaspectratio]{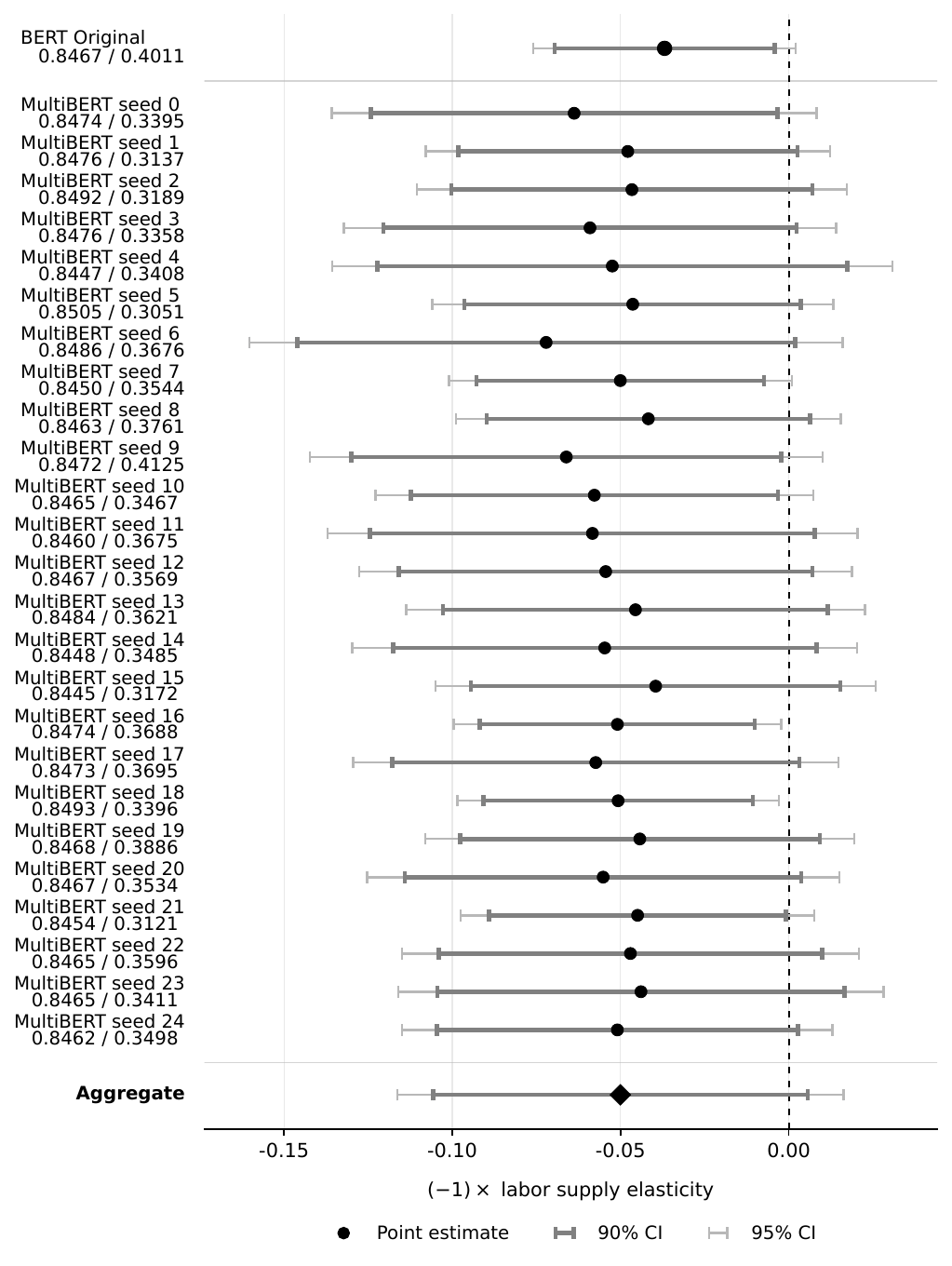}
	\caption{Coefficient robustness to 26 BERT embeddings (Base BERT with Mean pooling) with XGB 1 learners}
	\label{fig:xgb-1-confidence-intervals}
	\captionsetup{font=footnotesize,justification=justified,singlelinecheck=false}
	\caption*{\textit{Notes:} Same as \cref{fig:xgb-2-confidence-intervals}.}
\end{figure}

\begin{figure}[p]
	\centering
	\includegraphics[width=\linewidth,height=0.84\textheight,keepaspectratio]{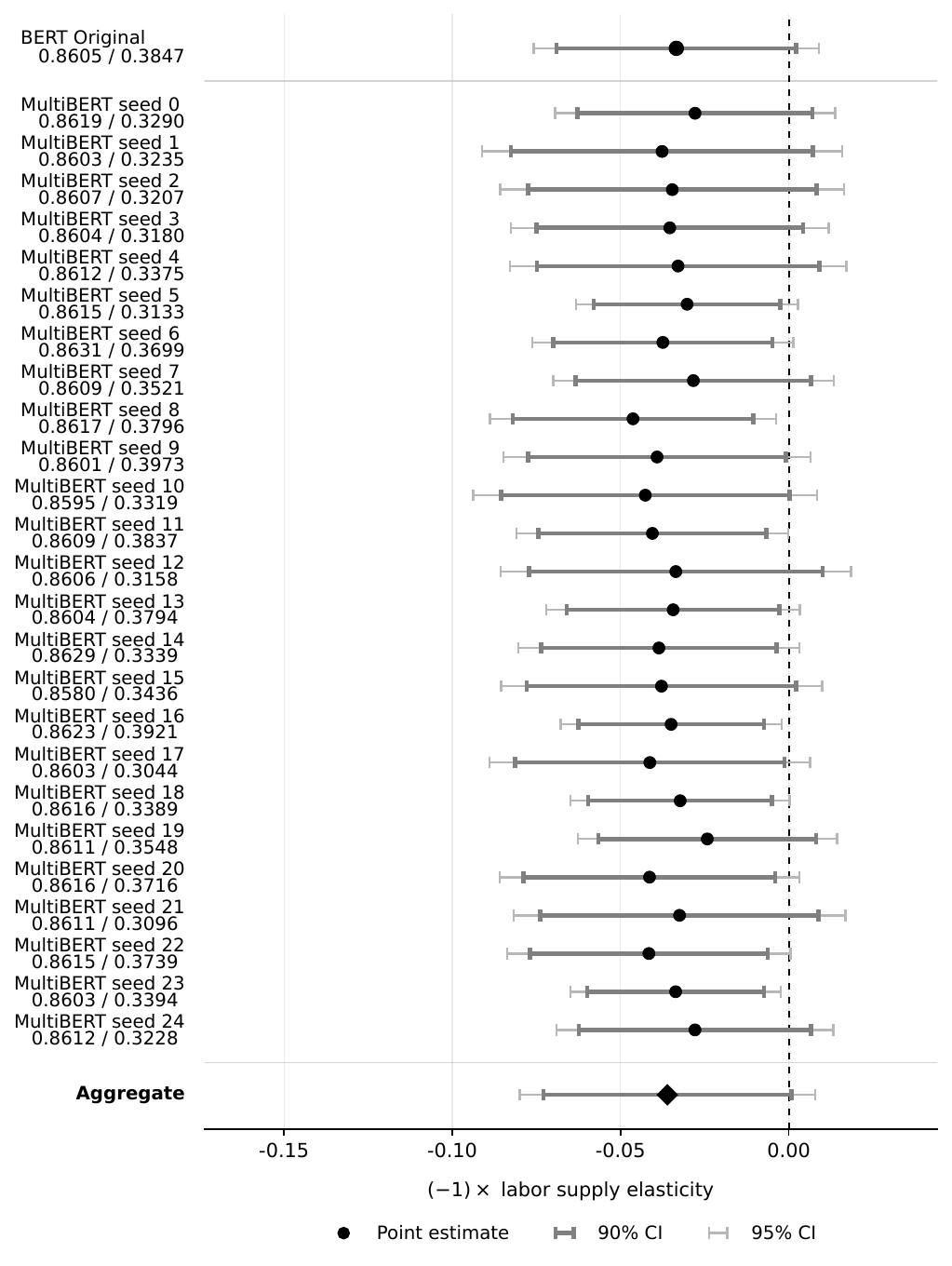}
	\caption{Coefficient robustness to 26 BERT embeddings (Base BERT with Mean pooling) with XGB 3 learners}
	\label{fig:xgb-3-confidence-intervals}
	\captionsetup{font=footnotesize,justification=justified,singlelinecheck=false}
	\caption*{\textit{Notes:} Same as \cref{fig:xgb-2-confidence-intervals}.}
\end{figure}

\subsection{Visualization of Task-Description Embeddings}
\label{sec:visualization_embeddings}
We visualize the task-description embeddings from the empirical application in \cref{sec:empirical_monopsony} using Uniform Manifold Approximation and Projection (UMAP), a nonlinear dimension-reduction method developed by \cite{mcinnes2018umap}.
UMAP constructs a weighted graph of nearby observations in the original embedding space and finds a low-dimensional representation that approximates these neighborhood relationships.
We use two dimensions for visualization; the axes have no direct economic interpretation, and distances between separated groups need not reflect distances in the original embedding space.
It is known that UMAP can reveal meaningful low-dimensional structures in high-dimensional data, especially for the local neighborhood relationships.

\cref{fig:umap_category,fig:umap_reward,fig:umap_duration} display the UMAP projection of task-description embeddings constructed from the BERT mean, colored by task category, log posted reward, and log duration.
These visualizations describe how embeddings capture task characteristics hand-engineered by \cite{dube2020monopsony}, and they are useful controls in subsequent empirical analyses using the log reward as the main regressor and log duration as the outcome.

\begin{figure}[htbp]
	\centering
	\includegraphics[width=\linewidth]{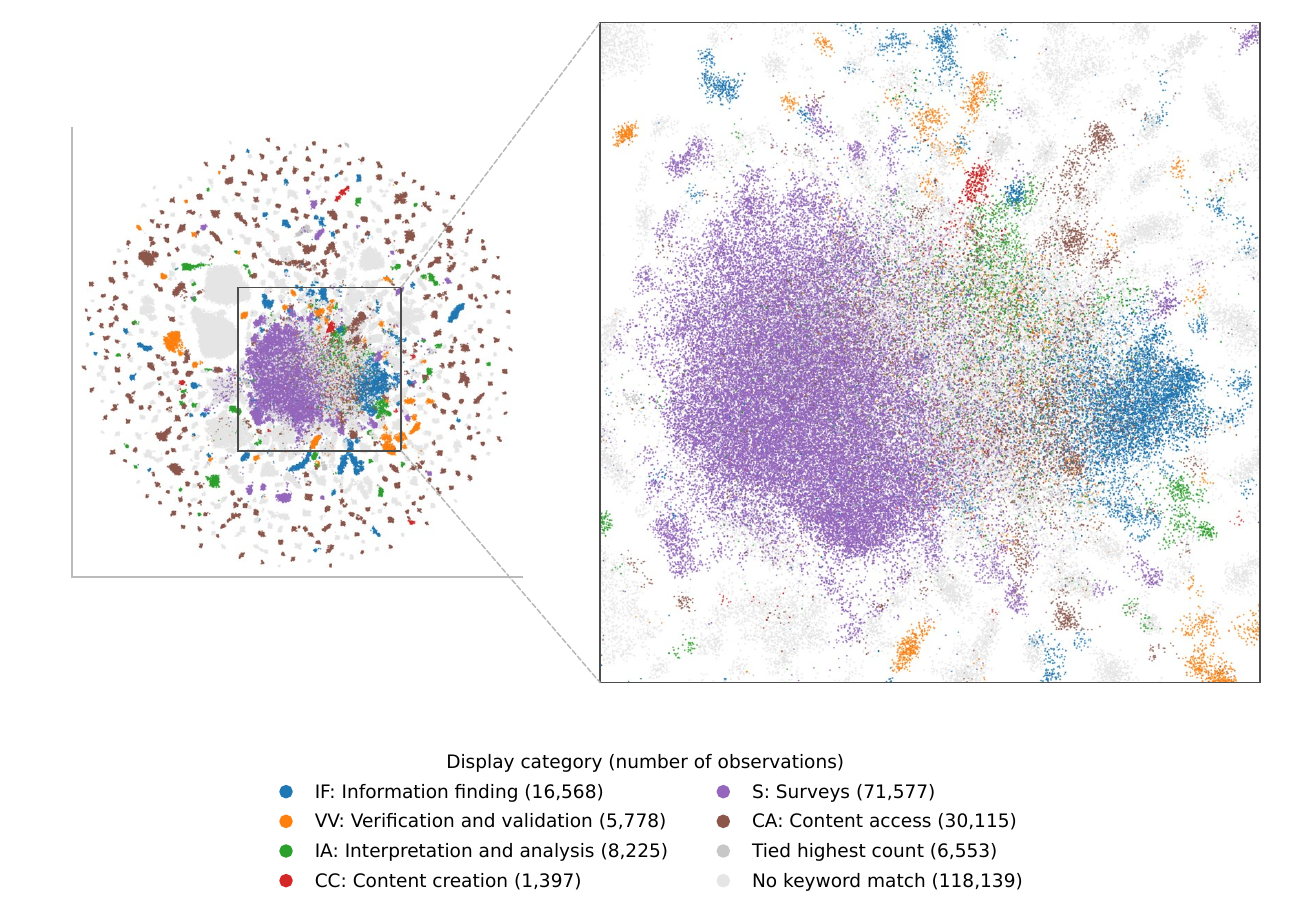}
	\caption{UMAP projection of task-description embeddings, colored by task category used by \cite{dube2020monopsony}}
	\label{fig:umap_category}
\end{figure}

\begin{figure}[htbp]
	\centering
	\includegraphics[width=\linewidth]{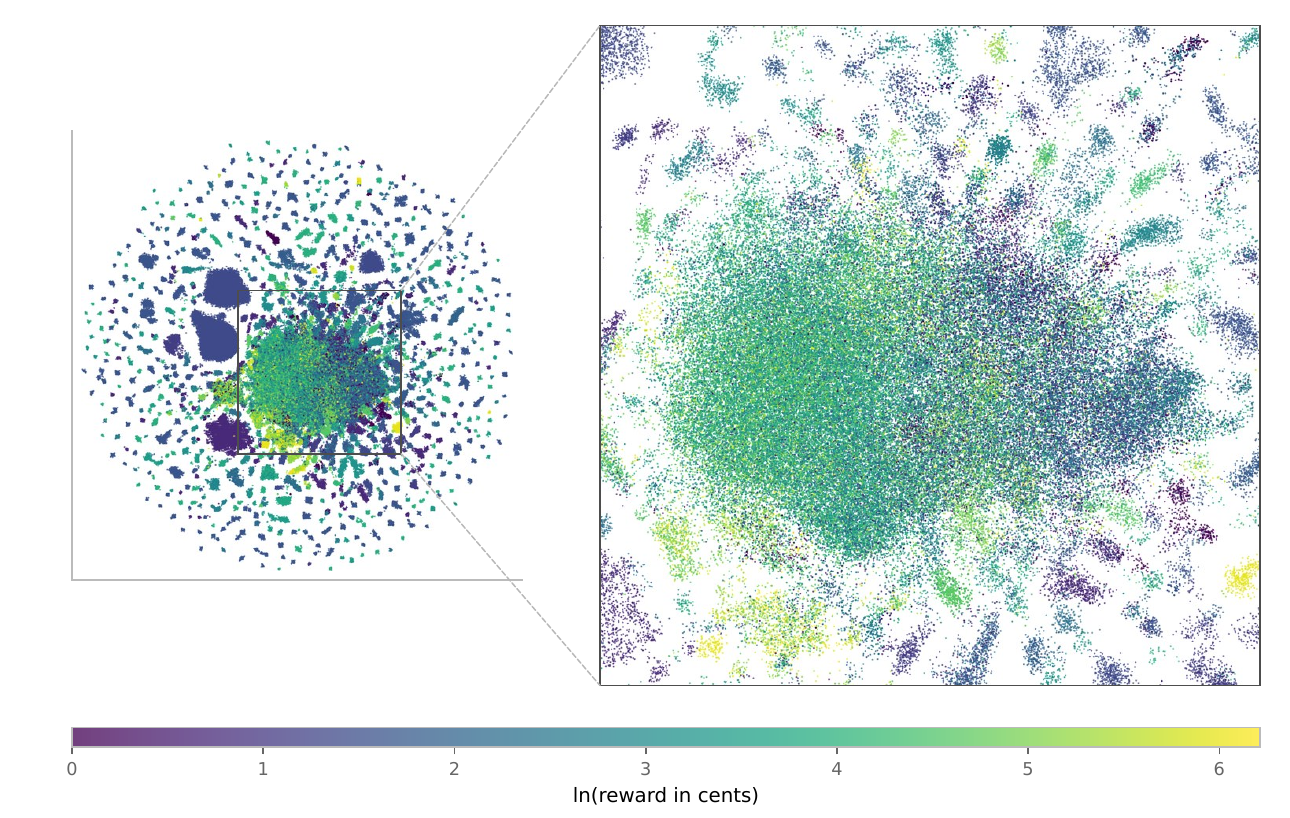}
	\caption{UMAP projection of task-description embeddings, colored by log reward}
	\label{fig:umap_reward}
\end{figure}

\begin{figure}[htbp]
	\centering
	\includegraphics[width=\linewidth]{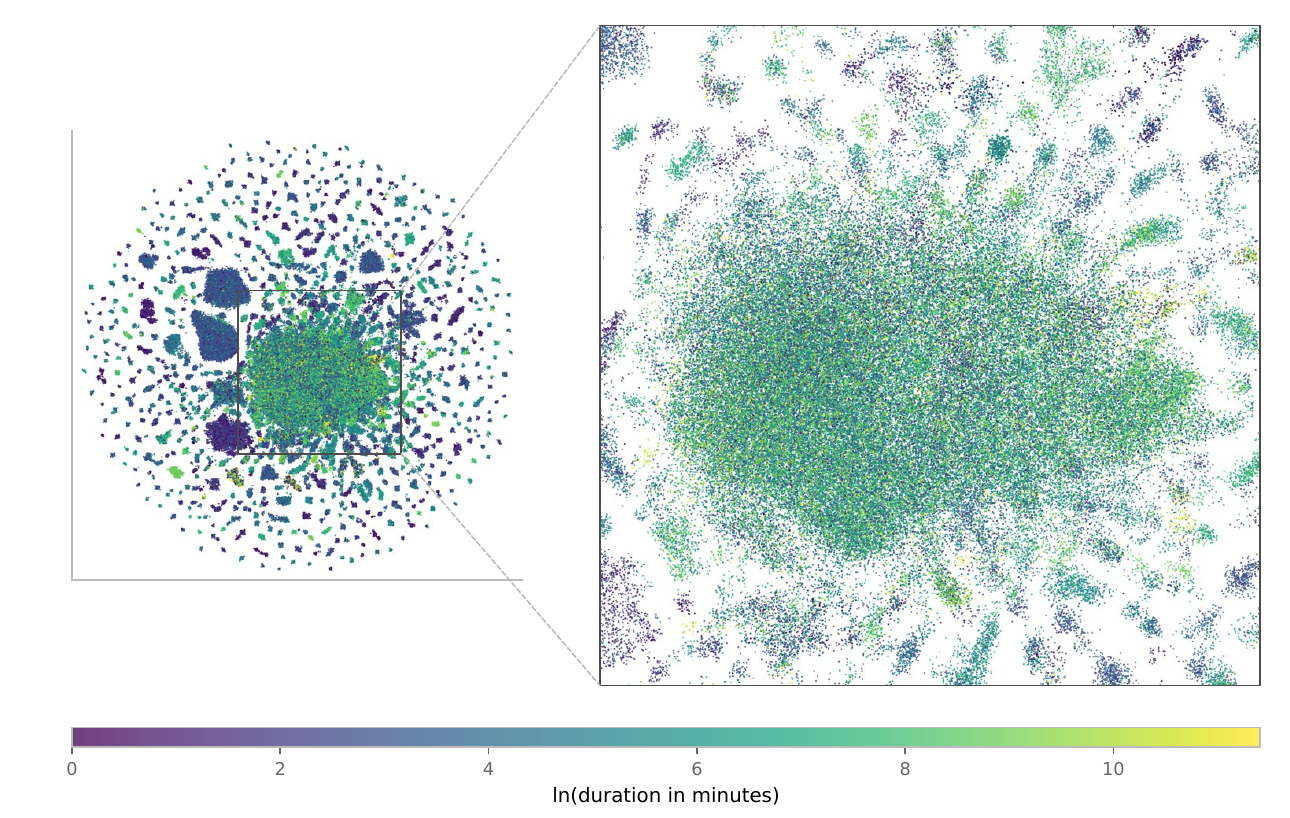}
	\caption{UMAP projection of task-description embeddings, colored by log duration}
	\label{fig:umap_duration}
\end{figure}

\end{document}